\documentclass[11pt]{article}

\usepackage{indentfirst}
\usepackage{amsmath}
\usepackage{amsthm}
\usepackage{color}
\usepackage{eufrak}
\newcommand{\eps}{\varepsilon}
\usepackage{verbatim}
\usepackage[makeroom]{cancel}
\usepackage{rotating}
\usepackage{float}
\usepackage{fullpage}
\usepackage{amssymb}
\usepackage{amsthm}
\usepackage{mathrsfs}
\usepackage{epsfig}
\usepackage[small,margin={1cm,1cm}]{caption}
\usepackage{graphicx}
\usepackage{indentfirst}
\usepackage{framed}
\usepackage[numbers]{natbib}
\usepackage{color}
\usepackage[
  margin=2.5cm,
  includefoot,
  footskip=30pt,
]{geometry}
\usepackage{hyperref}
\usepackage{graphicx}
\usepackage{float}
\usepackage{adjustbox}
\usepackage{bm}
\usepackage{cleveref}
\usepackage[labelfont=bf]{caption}
\usepackage[boxed,linesnumbered]{algorithm2e}
\usepackage{complexity}
\usepackage{bussproofs}
\usepackage{multicol}
\usepackage{enumitem}
\setlist[enumerate]{itemsep=3pt,topsep=3pt}
\setlist[enumerate,1]{label={\rm(}$\roman*${\rm)}}
\usepackage{subcaption}
\usepackage{times}
\usepackage{mdframed}

\usepackage{tikz}
\usetikzlibrary{positioning}
\definecolor{corlinks}{RGB}{0,0,200}
\definecolor{cormenu}{RGB}{0,0,200}
\definecolor{corurl}{RGB}{0,0,200}

\hypersetup{
	colorlinks=true,
	urlcolor=corlinks,
	linkcolor=corlinks,
	menucolor=cormenu,
	citecolor=corlinks,
	pdfborder= 0 0 0
}

\setlist[itemize]{itemsep=2pt,topsep=5pt}
\setlist[enumerate]{itemsep=2pt,topsep=5pt}

\newtheorem{theorem}{Theorem}[section]
\newtheorem*{theorem*}{Theorem}

\newtheorem{lemma}[theorem]{Lemma}
\newtheorem{corollary}[theorem]{Corollary}
\newtheorem{proposition}[theorem]{Proposition}

\newtheorem{fact}[theorem]{Fact}
\newtheorem{claim}[theorem]{Claim}

\theoremstyle{definition}
\newtheorem{definition}[theorem]{Definition}

\newtheoremstyle{reminder}%
{}%
{}%
{\em}%
{}%
{\bf}%
{.}%
{ }%
{Reminder of \thmnote{#3}}

\theoremstyle{reminder}
\newtheorem*{reminder}{Reminder}

\theoremstyle{remark}
\newtheorem{preremark}[theorem]{Remark}
\newenvironment{remark}%
  {\vspace{0.5em}\begin{mdframed}[innertopmargin=-0em,skipabove=0em]\small
	\begin{preremark}\upshape}
{\end{preremark}\end{mdframed}}

\newenvironment{remarkn}[1]%
  {\vspace{0.5em}\begin{mdframed}[innertopmargin=0em,skipabove=0em]\small
	\begin{preremark}[#1]\upshape}
{\end{preremark}\end{mdframed}}

\newtheorem{preexample}[theorem]{Example}
\newenvironment{example}%
  {\vspace{0.5em}\begin{mdframed}[innertopmargin=0em,skipabove=0em]\small
	\begin{preexample}\upshape}
{\end{preexample}\end{mdframed}}

\newenvironment{examplen}[1]%
  {\vspace{0.5em}\begin{mdframed}[innertopmargin=0em,skipabove=0em]\small
	\begin{preexample}[#1]\upshape}
{\end{preexample}\end{mdframed}}

\definecolor{blue-violet}{rgb}{0.54, 0.17, 0.89}
\definecolor{darkorange}{rgb}{1.0, 0.55, 0.0}

\newcommand{\eqdef}{\triangleq}

\newcommand{\PV}{\mathsf{PV}}
\newcommand{\TPV}{\mathsf{T}_\PV}
\newcommand{\LPV}{\mathcal{L}(\PV)}
\newcommand{\LUT}{\mathcal{L}_{\sf UT}}
\newcommand{\LB}{\mathsf{LB}}
\newcommand{\LBw}{\mathsf{LB}_{\mathsf{wst}}}
\newcommand{\APC}{\mathsf{APC}}

\newcommand{\NW}{\mathsf{NW}}
\newcommand{\uhr}{\!\upharpoonright}
\newcommand{\Search}{\text{-}\mathsf{Search}}
\newcommand{\bbQ}{\mathbb{Q}}

\newcommand{\inote}[1]{\textcolor{blue-violet}{(Igor: #1)}}
\newcommand{\jnote}[1]{\textcolor{darkorange}{(Jiatu: #1)}}
\newcommand{\ignore}[1]{}

\newclass{\NSIZE}{NSIZE}
\newclass{\GG}{G3c}
\newclass{\LK}{LK}

\newcommand{\cPi}[1]{\Pi_{#1}\mbox{-}}
\newcommand{\cSig}[1]{\Sigma_{#1}\mbox{-}}
\newclass{\Log}{Log}
\newclass{\EMPTY}{EMPTY}
\newclass{\FZPP}{FZPP}
\newclass{\tPV}{PV}
\newclass{\dWPHP}{dWPHP}

\newcommand{\Tpv}{\mathsf{T}_{\tPV}}
\newcommand{\Upv}{\mathsf{U}_{\tPV}}
\newcommand{\UTpv}{\mathsf{UT}_{\tPV}}

\newcommand{\calL}{\mathcal{L}}
\newcommand{\calM}{\mathcal{M}}
\newclass{\TF}{TF}
\newclass{\sTF}{sTF}

\newcommand{\Lpv}{\mathcal{L}_{\tPV}}
\newcommand{\bbN}{\mathbb{N}}
\newcommand{\bbZ}{\mathbb{Z}}
\newcommand{\bbR}{\mathbb{R}}
\newcommand{\calD}{\mathcal{D}}
\newcommand{\calI}{\mathcal{I}}

\newcommand{\calT}{\mathcal{T}}
\makeatletter
\renewcommand\paragraph{\@startsection{paragraph}{4}{\z@}%
    {1em \@plus1ex \@minus.2ex}%
    {-1em}%
    {\bf\normalsize}}
\makeatother

\DeclareMathOperator*{\Ex}{\mathbb{E}}

\begin{document}
	
\newgeometry{margin=0.82in}	
	
\title{Unprovability of strong complexity lower bounds in bounded arithmetic\vspace{0.4cm}}

\author{
		Jiatu Li\footnote{Email: \texttt{\href{mailto:lijt19@mails.tsinghua.edu.cn}{lijt19@mails.tsinghua.edu.cn}}}\vspace{0.2cm}\\{\small Institute for Interdisciplinary Information Sciences}\\{\small Tsinghua University\vspace{0.3cm}}
		\and	
		Igor C. Oliveira\footnote{Email: \texttt{\href{mailto:igor.oliveira@warwick.ac.uk}{igor.oliveira@warwick.ac.uk}}}\vspace{0.2cm}\\{\small Department of Computer Science}\\
		{\small University of Warwick\vspace{0.4cm}} 
	}

\maketitle

\vspace{-0.6cm}

\begin{abstract}
While there has been progress in establishing the unprovability of complexity statements in lower fragments of bounded arithmetic, understanding the limits of Je\v r\'abek's theory $\APC_1$ \citep{Jerabek07} and of higher levels of Buss's hierarchy $\mathsf{S}^i_2$ \citep{Buss} has been a more elusive task. Even in the more restricted setting of Cook's theory $\PV$ \citep{Coo75}, known results often rely on a less natural formalization that encodes a complexity statement using a collection of sentences instead of a single sentence. This is done to reduce the quantifier complexity of the resulting sentences so that standard witnessing results can be invoked. 

In this work, we establish unprovability results for \emph{stronger theories} and for \emph{sentences of higher quantifier complexity}. In particular, we unconditionally show that $\APC_1$ cannot prove strong complexity lower bounds separating the third level of the polynomial hierarchy. In  more detail, we consider non-uniform average-case separations, and establish that $\APC_1$ cannot prove a sentence stating that
\begin{center}
    $\forall n \geq n_0\;\exists\,f_n \in \cPi{3}\SIZE[n^d]$ that is $(1/n)$-far from every $\cSig{3}\SIZE[2^{n^{\delta}}]$ circuit.
\end{center}
This is a consequence of a much more general result showing that, for every $i \geq 1$, strong separations for $\cPi{i}\SIZE[\poly(n)]$ versus $\cSig{i}\SIZE[2^{n^{\Omega(1)}}]$ cannot be proved in the theory $\Tpv^i$ consisting of all true $\forall \Sigma^b_{i-1}$-sentences in the language of Cook's theory $\PV$. 

Our argument employs a convenient game-theoretic witnessing result that can be applied to sentences of arbitrary quantifier complexity. We combine it with extensions of a technique introduced by Kraj\'{i}\v{c}ek \citep{DBLP:journals/jml/Krajicek11} that was recently employed by Pich and Santhanam \citep{PS21} to establish the unprovability of lower bounds in $\PV$ (i.e., the case $i =1$ above, but under a weaker formalization) and in a fragment of $\APC_1$.
\end{abstract}

\restoregeometry

\newpage

\small
\setcounter{tocdepth}{3}
\tableofcontents

\normalsize

\newpage

\section{Introduction}\label{sec:introduction}

Establishing unconditional lower bounds on the complexity of computations is one of the primary goals of the theory of computational complexity. While the field has seen progress in the setting of restricted computational devices, such as constant-depth Boolean circuits (e.g.,~\citep{DBLP:journals/apal/Ajtai83, DBLP:journals/mst/FurstSS84, DBLP:conf/stoc/Hastad86, Razborov87, DBLP:conf/stoc/Smolensky87,   DBLP:journals/jacm/Williams14}) and monotone Boolean circuits (e.g.,~\citep{Razborov:85a, andreev1985method,  DBLP:journals/combinatorica/AlonB87}), proving super-linear circuit size lower bounds against general (unrestricted) circuits (see, e.g., \citep{FGHK16,DBLP:conf/stoc/Li022}) and separating complexity classes remain longstanding challenges.

Several barrier results have been proposed to explain why techniques that have been successful in certain settings cannot lead to stronger results. The most well known of them are relativization \citep{DBLP:journals/siamcomp/BakerGS75}, natural proofs \citep{DBLP:journals/jcss/RazborovR97}, and algebrization \citep{DBLP:journals/toct/AaronsonW09} (see also \citep{DBLP:conf/stoc/FanL022, DBLP:journals/jacm/ChenHOPRS22} for recent examples). While knowledge of these results provides a practical way to check if some approaches are likely to fail, each of these barriers is formulated in an ad-hoc way and is limited in scope. For instance, the natural proofs barrier does not consider a standard notion of ``proof'' and can be circumvented using simple reductions (see, e.g.,~\citep{AK10, OS18_mag_first, DBLP:conf/focs/ChenJW19, DBLP:journals/jacm/ChenHOPRS22}). In general, the aforementioned barriers don’t really tell if we simply haven’t been clever enough to design a better lower bound technique,
or if there is a deeper, more fundamental reason for the difficulty of establishing complexity lower
bounds and separations.

This motivates the development of a more principled approach to investigate the difficulty of analysing computations and, perhaps more importantly, the intriguing possibility that strong complexity lower bounds might be unprovable from certain mathematical axioms. In order to implement this plan, the first step is to try to understand which logical theories are able to formalise a significant number of results in algorithms and complexity theory. There has been a long and highly successful line of research showing that certain fragments of Peano Arithmetic collectively known as \emph{Bounded Arithmetic} offer a robust class of theories for the formalization of both basic and advanced results in these areas. 

\vspace{0.2cm}
 
\begin{remarkn}{Bounded Arithmetic}
Theories of bounded arithmetic aim to capture mathematical proofs that manipulate concepts from a given complexity class (e.g.,~a proof by induction whose inductive hypothesis can be checked in polynomial time). Notable examples include Cook's theory $\PV$ \cite{Coo75}, which formalises polynomial-time reasoning, Je\v r\'abek's theory $\APC_1$ \cite{Jerabek07}, which formalises probabilistic polynomial-time reasoning, and Buss's theories $\mathsf{S}^i_2$ and $\mathsf{T}^i_2$, which correspond to the levels of the polynomial-time hierarchy \cite{Buss}. 

The correspondence between these theories and the complexity classes is reflected in several ways. For instance, certain \emph{witnessing results} show that every provably total function in a given theory $\mathsf{T}_\mathcal{C}$ (i.e.,~when $\forall x~\exists! y~\varphi(x,y)$ is provable, for certain formulas $\varphi$) is computable within the corresponding complexity class $\mathcal{C}$ (i.e., the function $y = f(x)$ is in $\mathcal{C}$). There are also close relationships between theories of bounded arithmetic and propositional proof systems, e.g., propositional translations between proofs of certain sentences in $\mathsf{PV}$ and $\mathsf{S}^1_2$ and polynomial-size proofs in the extended Frege proof system (see, e.g, \citep{beyersdorff2009correspondence} and references therein).

Weaker theories corresponding to more fine-grained complexity classes such as $\TC^0$ and $\NC^1$ and the mathematical theorems provable in each of them have also received considerable attention. For instance, key properties of the elementary integer arithmetic operations can be established in theory $\mathsf{VTC}^0$ \citep{jevrabek2022iterated},  expander graphs can be constructed and analyzed in theory $\mathsf{VNC}^1$ \citep{DBLP:journals/apal/BussKKK20}, and several results from linear algebra can be formalised in theory $\mathsf{VNC}^2$ \citep{DBLP:journals/jacm/TzameretC21}. We refer to Cook and Nguyen \citep{cook_nguyen_2010} and Kraj\'{i}\v{c}ek \citep{Krajicek-book, krajicek_2019} for more information about bounded arithmetic and the logical foundations of complexity theory. 
\end{remarkn} 

\vspace{1em}

\noindent \textbf{Complexity Lower Bounds in Bounded Arithmetic.} The study and formalization of complexity lower bounds proofs in bounded arithmetic dates back to Razborov \citep{razborov1995unprovability, Razborov-switching-l}. We refer to Pich \citep{DBLP:journals/apal/Pich15} and to M\"{u}ller and Pich \citep{DBLP:journals/apal/MullerP20} for a comprehensive survey of the area. In particular, the latter paper identifies Je\v r\'abek's theory $\APC_1$ \citep{Jerabek07} for probabilistic reasoning as a suitable theory for the formalization of several existing circuit lower bounds. (Informally, $\APC_1$ is defined as the extension of Cook's theory $\PV$ \citep{Coo75} with the dual weak pigeonhole principle for polynomial-time functions.) For instance, $\APC_1$ can prove super-polynomial lower bounds against bounded-depth circuits and against monotone circuits \citep{DBLP:journals/apal/MullerP20}, establish the PCP Theorem \citep{DBLP:journals/corr/Pich14} (also provable in $\PV$), and formalize randomized matching algorithms \citep{DBLP:journals/corr/abs-1103-5215}.

Given the expressive power of $\PV$ and its extensions,
\emph{unconditionally} showing that these theories cannot prove a given result is a non-trivial task. Remarkably, in a recent work, Pich and Santhanam \citep{PS21} employed a technique introduced by Kraj\'{i}\v{c}ek \citep{DBLP:journals/jml/Krajicek11} and further elaborated in \cite{DBLP:journals/apal/Pich15} to establish that $\PV$ cannot show strong complexity lower bounds separating $\NP$ and $\coNP$. More precisely, for each fixed non-deterministic polynomial-time machine $M$, they showed that $\PV$ cannot prove an average-case lower bound for $L(M)$ against co-nondeterministic circuits of size $2^{n^{o(1)}}$. 

In the same work, \citep{PS21} showed that this unprovability result extends to a certain fragment of $\APC_1$ (see \citep{PS21} for the details and for additional results). They left open the status of the provability of strong complexity lower bounds in $\APC_1$. This theory has also been identified in other papers (e.g.,~\citep{CKKO21}) as an important test case for unconditional unprovability results. This is unsurprising, given the number of advanced results from algorithms and complexity that can be formulated and proved in $\APC_1$ and its mild extensions (see \citep{thesis_KO, cook_nguyen_2010, thesis_DTML, thesis_Pich,  DBLP:journals/apal/MullerP20} for many additional examples).\\

\noindent \textbf{Witnessing Theorems and Quantifier Complexity.} The approach of \citep{PS21} crucially relies on the KPT Witnessing Theorem \citep{KrajicekPT91}, a result that can be used to extract computational information from a proof of a sentence with a small number of quantifier alternations. This and similar results (e.g.,~Herbrand's Theorem and Buss's Witnessing Theorem) have been extremely useful tools in unprovability results (see, e.g.,~\citep{DBLP:journals/jsyml/CookK07, DBLP:journals/tocl/Krajicek21, CKKO21}). In order to apply the usual formulation of these witnessing theorems, it is crucial to consider sentences with up to four blocks of quantifiers. In particular, for this reason, the machine $M$ in the aforementioned result from \citep{PS21} is quantified outside of the sentence (i.e.,~in the meta-theory). A similar challenge is faced in other papers that consider the unprovability of complexity statements in bounded arithmetic (see, e.g.,~\citep{KO17} and the subsequent papers \citep{DBLP:journals/aml/BydzovskyM20, BKO20}).\\

\noindent \textbf{Our Contributions.} We 
 obtain (unconditional) unprovability results for \emph{stronger theories}
and for \emph{sentences of higher quantifier complexity}. We can summarize our main contributions as follows.

\begin{enumerate}
    \item\label{enum: contribution 1} Building on previous works \cite{DBLP:journals/jml/Krajicek11, DBLP:journals/apal/Pich15, PS21}, we establish the unprovability of strong complexity lower bounds in $\APC_1$ and in more expressive theories of bounded arithmetic. The lower bound sentences showed unprovable refers to separations between the levels of the polynomial hierarchy, where we consider a non-uniform setting and an average-case lower bound against sub-exponential size circuits.
    \item We consider a more natural (and of higher quantifier complexity) formalization of complexity lower bounds compared with \cite{DBLP:journals/jml/Krajicek11, DBLP:journals/apal/Pich15, PS21}. To achieve this, we employ a convenient game-theoretic witnessing theorem that allows us to extract computational information from proofs of sentences with an arbitrary number of quantifiers. While extensions of the KPT Witnessing Theorem for formulas with more quantifiers have found applications in bounded arithmetic  \citep{pudlak1992relations, pudlak2006consistency, DBLP:journals/jsyml/BussKT14}, to our knowledge this is the first time that such a result is used for the unprovability of complexity bounds.
\end{enumerate}

In the next section, we discuss our results in detail.

\subsection{Results}

Before formally stating our main unprovability result, we introduce the theories $\Tpv^i$ and their common language (vocabulary) $\Lpv$.\\

\noindent \textbf{Theory $\Tpv^i$ and Language $\Lpv$.} We let $\Lpv$ contain the constant symbols $0$ and $1$, and a function symbol $f$ for every function in $\FP$, the class of polynomial-time computable functions.\footnote{For the reader familiar with bounded arithmetic, we note that in our setup considering polynomial-time functions is equivalent to considering polynomial-time algorithms. See \Cref{sec:prelim_logic} for more details.} In particular, $\Lpv$ contains function symbols for the length function $|x|$, addition $+$, etc. $\Lpv$ contains the equality predicate $=$ as its only relation symbol. Note that one can define any polynomial-time computable predicate through its characteristic function, equality, and the constant symbol $1$. 

For each integer $i\geq 1$, we let $\Tpv^i$ denote the theory of all true (with respect to the standard model $\mathbb{N}$) $\forall \Sigma^b_{i-1}$-sentences over the language $\Lpv$.\footnote{This is a standard class of sentences in bounded arithmetic. Informally, it means that the sentence starts with a block of universal quantifiers, followed by $i-1$ blocks of \emph{bounded} quantifiers, i.e., $\forall x\le t$ or $\exists x\le t$ for some term $t$. The formal definition will be given in \Cref{sec:prelim_logic}. (For the specialist, we note that allowing sharply bounded quantifiers would not change our unprovability results.)} In particular, the theory $\Tpv^1$ (which is called $\Tpv$ in \cite{PS21}) is strictly stronger than Cook's theory $\PV$.\footnote{We use $\PV$ to refer to its first-order formalization \citep{Coo75, KrajicekPT91}, also denoted by $\PV_1$ by some authors.} We provide some examples of sentences provable in $\Tpv^i$ after stating our main result.\\

\noindent \textbf{Formalization of Lower Bounds.} In order to consider the provability of a strong complexity lower bound separating the $i$-th level of the (non-uniform) polynomial hierarchy, we introduce a sentence $\LB^i(s_1,s_2,m,n_0)$ stating that, for every input length $n \geq n_0$, there is a $\Pi_i$-circuit $C$ of size $\leq s_1(n)$ such that, for every $\Sigma_i$-circuit $D$ of size $\leq s_2(n)$, we have
$$
\Pr_{x \sim \{0,1\}^n}\Big [C(x) = D(x) \Big] \;\leq\; 1 - \frac{m(n)}{2^n}.
$$
Here a $\Pi_i$-circuit $C$ (similarly for $\Sigma_i$ circuits) is simply a standard deterministic Boolean circuit $C(x,z_1,\ldots, z_i)$, where we define that 
$$
C(x) =1 \quad \text{if and only if}\quad \forall z_1\;\exists z_2 \; \ldots \; Q_i z_i~C(x,z_1,\ldots, z_i) = 1\;.
$$

Formally, let $\cSig{i}\SIZE[s(n)]$ and $\cPi{i}\SIZE[s(n)]$ refer to $\Sigma_i$-circuits and $\Pi_i$-circuits of size $s(n)$, respectively. Let $\LB^i(s_1,s_2,m,n_0)$ denote the following  $\Lpv$-sentence: 
\begin{align*}
&\forall n\in\Log\Log\text{ with }n\ge n_0~\exists C\in\cPi{i}\SIZE[s_1(n)]~\forall D\in\cSig{i}\SIZE[s_2(n)] \\ 
&\exists m=m(n)\mbox{ distinct $n$-bit strings }x^1,\dots,x^m \mbox{ s.t. } \mathsf{Error}(C,D,x^i)\mbox{ for all }i\in[m],
\end{align*}
where $\mathsf{Error}(C,D,x)$ means that the circuits $C$ and $D$ do not agree on the input $x$. For the reader that is not familiar with bounded arithmetic, the notation $n \in \Log\Log$ essentially means that all bounded quantifiers refer to objects of length up to $\mathsf{poly}(2^n)$. As in \citep{PS21}, this makes the unprovability result  stronger. Many existing circuit lower bound proofs can be formalized in $\APC_1$ without ever quantifying over objects of length larger than $\poly(n)$ \citep{DBLP:journals/apal/MullerP20}.

It's easy to see that $\mathsf{Error}(C,D,x)$ is the disjunction of a $\Sigma_i^b$-formula (stating that $C(x) = 0 \wedge D(x)=1$) and a $\Pi_i^b$-formula (stating that $C(x) = 1 \wedge D(x)=0$). We note that, already for $i = 1$, $\LB^i(s_1,s_2,m,n_0)$ is a $\forall \Sigma^b_4$-sentence. In particular, widely used witnessing results such as the KPT Theorem \citep{KrajicekPT91} (see \Cref{sec:witnessing_review} for a review) cannot be directly applied to it.

\vspace{1em}

\noindent \textbf{Main Unprovability Result.} Next, we state our main theorem on the unprovability of complexity lower bounds in $\Tpv^i$ and its corollary for $\APC_1$.  

\begin{theorem}[Main Theorem]\label{thm:intro_main_unprov}
For every $i\ge 1$, $n_0\in\bbN$, $\delta\in\bbQ\cap(0,1)$, and $d\ge 1$, $$\Tpv^i\nvdash\LB^i(s_1,s_2,m,n_0)\;,$$ where $s_1(n)=n^d$, $s_2(n)=2^{n^\delta}$, and $m=2^n/n$.
\end{theorem}

 \Cref{thm:intro_main_unprov} extends the  result of \cite{PS21} in two directions. Firstly, it establishes the unprovability of strong complexity lower bounds in theories believed to be much stronger than $\Tpv^1$. Secondly, \citep{PS21} considered a weaker formalization that instead of quantifying over the circuit $C$ (inside the sentence) considers a collection of sentences $\{\LB^1_M\}_{M}$, one for each uniform non-deterministic machine $M$ (quantified over outside the theory). 

\vspace{0.2cm}
 
 \begin{examplen}{The Strength of Theory $\Tpv^i$}
These theories are quite strong already at small values of $i$, say $i=3$. Below we give some examples (see \Cref{sec:provability_in_Tpv} for a related discussion).
\begin{enumerate}
    \item Fermat's Little Theorem, which states that if $a^{p}\not\equiv a\pmod{p}$ then there is $1<d<p$ such that $d\mid p$, is a true $\forall \Sigma^b_1$-sentence in $\Lpv$ and consequently an axiom of $\Tpv^2$. It is unprovable in $\Tpv^1$ (therefore also unprovable in $\PV$) unless factoring is easy (see, e.g.,~\citep{DBLP:journals/iandc/KrajicekP98, cook_nguyen_2010}). 
    \item The Pigeonhole Principle, which states that for every circuit $C \colon [n+1]\to[n]$ there exists $x\ne y$ such that $C(x)=C(y)$, is also an axiom of $\Tpv^2$. It is not hard to show that even the weaker version of this principle (in which the circuit $C \colon [2n] \to [n]$) is unprovable in $\Tpv^1$ unless there is no (public-key) collision-resistant hash functions (see, e.g., \cite{krajivcek2001weak, Buss08}).
    \item The dual Pigeonhole Principle, which states that for every circuit $C \colon [n]\to[n+1]$ there exists $y\in[n+1]$ such that for all $x\in[n]$ we have $C(x)\ne y$, is in $\Tpv^3$. Even the weak version of this principle (in which the circuit $C \colon [n] \to [2n]$) is unprovable in $\Tpv^1$ unless EMPTY \citep{DBLP:conf/focs/Korten21} (also known as Range Avoidance \citep{DBLP:journals/eccc/RenSW22}) can be solved in polynomial time with $O(1)$ circuit-inversion oracle queries. %
    \item The induction principle for $\Sigma^p_i$-predicates is provable in $\Tpv^{i+2}$, while even the induction principle for $\NP$-predicates is unprovable in $\Tpv^1$ unless the polynomial-time hierarchy collapses \citep{KrajicekPT91, DBLP:journals/apal/Buss95, DBLP:journals/jsyml/Zambella96}.
\end{enumerate} 
\end{examplen}

\vspace{0.2cm}

Since every axiom of $\APC_1$ is implied by a true $\forall \Sigma^b_2$-sentence over the language $\Lpv$ in theory $\Tpv^3$ (see \Cref{sec:prelim} for the definition of $\APC_1$), every sentence provable in $\APC_1$ is also provable in $\Tpv^3$. Consequently, we get the following corollary to \Cref{thm:intro_main_unprov}, which shows that $\APC_1$ cannot establish strong complexity lower bounds separating the third level of the (non-uniform) polynomial hierarchy.

\begin{corollary}[Unprovability of Strong Complexity Lower Bounds in {$\APC_1$}]\label{cor:intro_APC1}
For every $n_0\in\bbN$, $\delta\in\bbQ\cap(0,1)$, and $d\ge 1$, $$\APC_1 \nvdash\LB^3(s_1,s_2,m,n_0)\;,$$ where $s_1(n)=n^d$, $s_2(n)=2^{n^\delta}$, and $m=2^n/n$.
\end{corollary}

\Cref{cor:intro_APC1} establishes the first unconditional result showing the unprovability of strong complexity lower bounds in $\APC_1$. Previously, \cite{PS21} obtained an extension of their result to a fragment of $\APC_1$, but left open the provability of the same collection of sentences in $\APC_1$. Our result is incomparable to theirs in this case, since  \Cref{cor:intro_APC1} refers to $\LB^3$ (the third level of the non-uniform polynomial hierarchy) instead of $\{\LB^1_M\}_{M}$.

\vspace{0.2cm}

\begin{remarkn}{\textbf{Relevance to the Logical Foundations of Complexity Theory}} The hypothesis that $\mathsf{P} \neq \mathsf{PH}$ (which is equivalent to $\mathsf{P} \neq \mathsf{NP}$) can be interpreted as the statement that polynomial time computations cannot simulate a finite number of bounded quantifier alternations. Our unconditional unprovability result, on the other hand, establishes that $\Tpv^i$, the strongest (sound) theory operating with $\forall \Sigma^b_{i-1}$ axioms over $\Lpv$, cannot strongly separate the $i$-th level of the polynomial hierarchy. 

If the lower bound stated by the $\LB^i$ sentence is true, our result indicates the existence of a fundamental limitation of this theory in reasoning about computations at the $i$-th level of the hierarchy and above. In contrast to previous works, which were restricted to subtheories of $\APC_1$, a significant aspect of \Cref{thm:intro_main_unprov} is showing that this phenomenon is not caused by a potential weakness of the theory at hand.
\end{remarkn}

\subsection{Techniques}
\label{sec:general-ph:intuition}

In order to prove \Cref{thm:intro_main_unprov}, we formulate a game-theoretic witnessing theorem that can be applied to sentences of high quantifier complexity, such as $\LB^i(s_1,s_2,m,n_0)$. Our general framework is  similar to an extension of the KPT Witnessing Theorem considered by Buss, Kołodziejczyk, and Thapen \citep{DBLP:journals/jsyml/BussKT14}.\\

\noindent \textbf{A Game-Theoretic Witnessing Theorem for General Formulas.} For a language (vocabulary) $\calL$, let $\varphi(x)$ be a bounded $\calL$-formula defined as
\begin{align*}
  \varphi(x)~\triangleq~&\exists y_1\le t_1(x)~\forall x_1\le s_1(x,y_1)~\exists y_2\le t_2(x,y_1,x_1)\dots \forall x_{k-1}\le s_{k-1}(x,y_1,x_1,\dots,y_{k-1})\\
  ~&\exists y_k\le t_k(x,y_1,x_1,\dots,y_{k-1}, x_{k-1})~\forall x_{k}\le s_k(x,y_1,x_1,\dots,y_{k})~\phi(x,x_1,\dots,x_k, y_1,\dots,y_k),
\end{align*}
where $\phi(x, \vec{x}, \vec{y})$ is a quantifier-free $\calL$-formula. We would like to extract computational information from the provability of $\forall x\, \varphi(x)$ in a theory $\calT$. We achieve this by showing that the provability of this sentence is \emph{equivalent} to the existence of a \emph{winning strategy} in a certain \emph{game}. Moreover, the winning strategy will be computable using \emph{terms} of $\calL$. Consequently, if the interpretation of each term in a given model $\calM$ of $\calT$ has limited computational complexity, we obtain a \emph{computationally bounded} winning strategy. For simplicity, we discuss the game only informally below, deferring the formal details to \Cref{sec:witnessing-1}.

We consider an interactive game between two players, the \emph{truthifier} (associated with existential quantifiers in $\varphi$) and the \emph{falsifier} (associated with universal quantifiers in $\varphi$). A \emph{board} is defined as a pair $(\calM,n_0)$, where $\calM$ is a structure over  $\calL$ such that $\calM\vDash\calT$, and $n_0\in\calM$ is an element of its domain. The \emph{evaluation game} for the formula $\varphi(x)$ on the board $(\calM,n_0)$ is played as follows: in the $i$-th round of the game ($1\le i\le k$), the truthifier firstly chooses an assignment $m_i\in\calM$ for $y_i$ such that $m_i\le t_i(n_0,m_1,n_1\dots,m_{i-1}, n_{i-1})$, then the falsifier chooses an assignment $n_{i}\in\calM$ for $x_i$ such that $n_i\le s_{i}(n_0,m_1,n_1, \dots,m_{i})$. The truthifier \emph{wins} if and only if $\phi(x/n_0,\vec x/\vec n,\vec y/\vec m)$ holds in $\calM$. (Note that when playing on a board $(\calM, n_0)$ we set $x$ in $\varphi(x)$ to $n_0$.)

We will also consider a more general game called the \emph{tree exploration game}. In more detail, we allow the truthifier and falsifier to simultaneously play different evaluation games over the same board $(\calM,n_0)$. The truthifier has a positional advantage over the falsifier: it can decide where to make the next move, i.e., by either
\begin{enumerate}
    \item making the next move in one of the current games; or
    \item starting a new evaluation game over the board $(\calM,n_0)$; or
    \item playing differently some earlier play, which creates a new game from that position but maintains the existing game plays.
\end{enumerate}
The falsifier must respond to the move of the truthifier  in the corresponding evaluation game. Note that the next assignment selected by each player now depends on previous plays in all concurrent games. The truthifier \emph{wins} the tree exploration game if there is a node $u$ in the current partial game tree that is a winning node for the truthifier, that is, the concatenation of the pairs of elements labelling the edges on the root-to-$u$ path forms a winning transcript of the truthifier in the evaluation game of $\varphi(x)$ on the board $(\calM,n_0)$. The \emph{tree exploration game of $\varphi(x)$} is defined as the tree exploration game starting from a partial game tree containing only the root node. See \Cref{fig:tree-exploration-game} for an example of a transcript of the tree exploration game.

\begin{figure}
    \centering
    \begin{subfigure}{0.21\linewidth}
    \begin{tikzpicture}[nodes={circle}, ->]
        \node (r) [draw] {1}
        child { 
            node (n1) [draw] {2} edge from parent [->]; 
        };
        \path (r) -- (n1) node[midway,left] {$(4,\cdot)$};
    \end{tikzpicture}
    \caption{The truthifier adds a child $(2)$ from the root and a label $4$.}
    \end{subfigure}
    ~
    \begin{subfigure}{0.21\linewidth}
    \begin{tikzpicture}[nodes={circle}, ->]
        \node (r) [draw] {1}
        child { 
            node (n1) [draw] {2} edge from parent [->]; 
        };
        \path (r) -- (n1) node[midway,left] {$(4,2)$};
    \end{tikzpicture}
    \caption{The falsifier's response is $2$. Node $(2)$ is not a winning node.}
    \end{subfigure}
    ~
    \begin{subfigure}{0.21\linewidth}
    \begin{tikzpicture}[nodes={circle}, ->]
        \node (r) [draw] {1}
        child { 
            node (n1) [draw] {2} edge from parent [->]; 
        }
        child {
            node (n2) [draw] {3} edge from parent [->];
        };
        \path (r) -- (n1) node[midway,left] {$(4,2)$};
        \path (r) -- (n2) node[midway,right] {$(5,\cdot)$};
    \end{tikzpicture}
    \caption{The truthifier adds a child (3) from the root and a label $5$.}
    \end{subfigure}
    ~
    \begin{subfigure}{0.21\linewidth}
    \begin{tikzpicture}[nodes={circle}, ->]
        \node (r) [draw] {1}
        child { 
            node (n1) [draw] {2} edge from parent [->]; 
        }
        child {
            node (n2) [draw,double] {3} edge from parent [->];
        };
        \path (r) -- (n1) node[midway,left] {$(4,2)$};
        \path (r) -- (n2) node[midway,right] {$(5,3)$};
    \end{tikzpicture}
    \caption{Clearly a winning node is reached regardless of the falsifier's move.}
    \end{subfigure}
    \caption{A transcript of the tree exploration game for $\varphi(x)=\exists y\le 2x~\forall z< y~(y\ge x\land (z=1\lor z\nmid y))$ (``there is a prime number within $[x,2x]$'') on the board $(\bbN, 3)$. The truthifier wins by reaching node $(3)$.} %
    \label{fig:tree-exploration-game}
\end{figure}
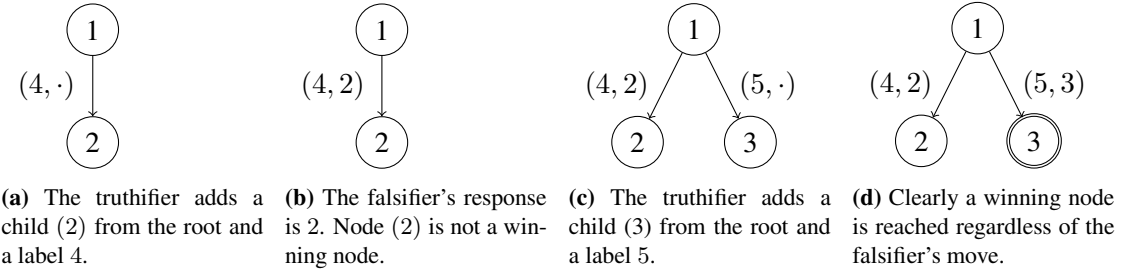

An \emph{$\calL$-strategy} of the truthifier in the tree exploration game is described by a sequence of $\calL$-terms, where each term describes the next move of the truthifier. Finally, a length-$\ell$ $\calL$-strategy is said to be a \emph{universal winning strategy} if the truthifier wins within $\ell$ moves against all strategies (not necessarily $\calL$-strategies) of the falsifier on any board $(\calM,n_0)$. (The ``universality'' of the strategy comes from the fact that it succeeds over any board $(\calM,n_0)$ and against any strategy of the falsifier. Moreover, the location of the next move of the truthifier in the game tree will be independent of the board and of the strategy of the falsifier.)

Recall that a theory $\calT$ is said to be a universal theory if every axiom of $\calT$ is of the form $\forall \vec{z}\,\psi(\vec{z})$, where $\psi(\vec{z})$ is a formula free of quantifiers. We show that the provability of the sentence $\forall x \, \varphi(x)$ in a universal theory $\calT$ with a certain closure property is equivalent to the existence of a universal winning $\calL$-strategy of length $O(1)$ for the truthifier in the tree exploration game of $\varphi(x)$.

\begin{theorem}[Game-Theoretic Witnessing Theorem]\label{thm:intro_witnessing-tree-exploration}
    Let $\calT$ be a universal bounded theory  with vocabulary $\calL$ that is closed under if-then-else (see \Cref{def:closed-under-if-then-else}). Let $\varphi$ be a bounded $\calL$-formula  of the form
    \begin{align*}
    \varphi(x)~\triangleq~&\exists y_1\le t_1(x)~\forall x_1\le s_1(x,y_1)~\exists y_2\le t_2(x,y_1,x_1)\dots \forall x_{k-1}\le s_{k-1}(x,y_1,x_1,\dots,y_{k-1})\\
  ~&\exists y_k\le t_k(x,y_1,x_1,\dots,y_{k-1}, x_{k-1})~\forall x_{k}\le s_k(x,y_1,x_1,\dots,y_{k})~\phi(x,x_1,\dots,x_k, y_1,\dots,y_k),
\end{align*}
where $\phi(x, \vec{x}, \vec{y})$ is a quantifier-free $\calL$-formula. Then $\calT\vdash\forall x~\varphi(x)$ if and only if there is a universal winning $\calL$-strategy of length $O(1)$ for the truthifier in the corresponding tree exploration game of $\varphi(x)$. 
\end{theorem}

Beyond its applicability to sentences with an arbitrary number of quantifiers, we stress that two keys aspects of \Cref{thm:intro_witnessing-tree-exploration} are that the winning strategy is computed by $\calL$-terms and that the truthifier wins in constantly many rounds. (In practice, in order to use this result to obtain computational information from a proof, one typically fixes a particular strategy of the falsifier, which depends on the context and intended application.)

    It is possible to show that \Cref{thm:intro_witnessing-tree-exploration} is a generalization of the KPT Witnessing Theorem \citep{KrajicekPT91}: If the formula $\varphi(x)$ is an $\exists\forall$-formula, the evaluation game for $\varphi$ has only one round; this means that the tree exploration game for $\varphi$ is essentially  a  sequential repetition of the evaluation game (which is equivalent to the Student-Teacher game given by KPT Witnessing Theorem; see \Cref{thm:KPT} and \cite{KrajicekPT91,DBLP:journals/apal/Pich15}). Indeed, KPT witnessing can also be derived from a less general result that we present in \Cref{sec:witnessing_special_case} as a corollary of \Cref{thm:intro_witnessing-tree-exploration} and that is sufficient for the proof of \Cref{thm:intro_main_unprov}.

    We discuss \Cref{thm:intro_witnessing-tree-exploration} in detail in \Cref{sec:witnessing-1} and \Cref{sec:appendix_witnessing}. In contrast to the model-theoretic approach of \citep{DBLP:journals/jsyml/BussKT14}, we establish \Cref{thm:intro_witnessing-tree-exploration} using techniques from proof theory. As we explain below, in our main application we will actually work with a simplified and more convenient framework that might be of independent interest.

    \vspace{0.3cm}

\begin{example}\label{example: non-constructive}
    Why does the provability in a universal theory $\calT$ correspond to the tree exploration game instead of the simpler evaluation game? As a conceptual example, one may consider the well-known non-constructive proof of the existence of two irrational numbers $x,y$ such that $x^y$ is rational. By the Law of Excluded Middle (i.e., $A$ or $\neg A$), one can easily argue that either $(x,y)=(\sqrt 2,\sqrt 2)$ or $(x,y)=((\sqrt 2)^{\sqrt 2},\sqrt 2)$ will be the required pair of irrational numbers. However, we cannot figure out which one of these two possibilities is the correct answer from the structure of this proof. Nevertheless, we can convince any opponent that the original statement is true by a two-round ``tree exploration game'': we first propose $(x,y)=((\sqrt 2)^{\sqrt 2},\sqrt 2)$ and, in case that the opponent argues that $(\sqrt 2)^{\sqrt 2}$ is rational, we propose $(\sqrt 2,\sqrt 2)$ instead. Similarly, the truthifier's strategy extracted from the $\GG$-proof is not guaranteed to witness the existential quantifiers in one shot; it might need to interact with the falsifier for constantly many rounds to produce a correct answer (and each current move of the truthifier can depend on previous moves of both players). 
\end{example}
\vspace{1em}

\noindent \textbf{Unprovability of Strong Complexity Lower Bounds.} The proof of \Cref{thm:intro_main_unprov} extends the approach of \cite{PS21}, which explores a technique from \cite{DBLP:journals/jml/Krajicek11, DBLP:journals/apal/Pich15}. The main challenge for us is that we must consider the significantly more powerful theory $\Tpv^i$ and the (un)provability of a sentence $\LB^i(s_1,s_2,m,n_0)$ with a larger number of quantifier alternations. In particular, while \cite{PS21} considered the provability of a strong complexity lower bound against a fixed machine $M$, the sentence $\LB^i(s_1,s_2,m,n_0)$ merely states that there exists a strong separation between $\Pi_i$ circuits vs $\Sigma_i$ circuits. This introduces an additional technical difficulty that requires us to also revisit and extend the approach of \cite{DBLP:journals/jml/Krajicek11, DBLP:journals/apal/Pich15}.  

Suppose, toward a contradiction, that  $$\Tpv^i\vdash\LB^i(s_1,s_2,m,n_0)\;,$$ where $s_1(n)=n^d$, $s_2(n)=2^{n^\delta}$, and $m=2^n/n$. In other words, we assume that the theory $\Tpv^i$ proves that for every $n\geq n_0$ there is a $\Pi_i$-circuit $C_n$ of size $\leq n^d$ such that, for every $\Sigma_i$-circuit $D_n$ of size $\leq 2^{n^\delta}$,
$$
\Pr_{x \sim \{0,1\}^n}\Big [C_n(x) = D_n(x) \Big] \;\leq\; 1 - \frac{1}{n}.
$$

The key idea behind the argument is that the proof of a strong complexity \emph{lower bound} in bounded arithmetic yields a corresponding complexity \emph{upper bound}. We then argue that the lower bound and the upper bound \emph{contradict each other}. From this, the unprovability of the lower bound sentence follows. 

In more detail, our high-level strategy is as follows:
\begin{itemize}
    \item[(\emph{i})] The provability of the average-case lower bound sentence $\LB^i(s_1,s_2,m,n_0)$ implies the provability in $\Tpv^i$ of a \emph{worst-case} lower bound for $\cPi{i}\SIZE[n^d]$ vs $\cSig{i}\SIZE[2^{n^\delta}]$. The latter is formalized by a sentence $\LBw^i(s_1,s_2,n_0)$.
    \item[(\emph{ii})] From any $\Tpv^i$-proof of $\LBw^i(s_1,s_2,n_0)$, we show how to extract a \emph{complexity upper bound} for an \emph{arbitrary} $\Pi_i$-circuit $E_m(x)$ over an input $x$ of length $m$ and of size at most $\poly(m)$. (This is done outside the theory $\Tpv^i$.) More precisely, we show that there is a deterministic circuit $B_m$ with $\Sigma^p_{i-1}$-oracle gates and of size $\leq 2^{m^{o(1)}}$ such that
    $$
        \Pr_{x \sim \{0,1\}^m}[E_m(x) = B_m(x)] \geq 1/2 + 2^{-m^{o(1)}}.
    $$
    \item[(\emph{iii})] We invoke a hardness amplification result for the (non-uniform) polynomial hierarchy to conclude that, on any large enough input length $n$, \emph{every} $\Pi_i$-circuit $C_n$ of size $\leq n^d$ agrees with some $\Sigma_i$-circuit $D_n$ of size $\leq 2^{n^{\delta}}$ on more than a $1 - 1/n$ fraction of the inputs. (If this is not the case, we would be able to use hardness amplification to contradict the previous item.)
\end{itemize}
Since $\Tpv^i$ is a \emph{sound} theory, i.e., every theorem of $\Tpv^i$ is a true sentence, Item (\emph{iii}) is in contradiction with the complexity lower bound stated in $\LB^i(s_1,s_2,m,n_0)$. Consequently, $\Tpv^i$ does not prove this sentence.

Item (\emph{i}) is trivial, since the provability of an average-case lower bound immediately yields the provability of a worst-case lower bound against circuits of the same size. Item (\emph{iii}) requires an extension of a hardness amplification result of Healy, Vadhan, and Viola \cite{DBLP:journals/siamcomp/HealyVV06} to higher levels of the polynomial hierarchy. We verify that this is possible in \Cref{sec:hard_ampl} and \Cref{sec: lmms for hardness amp}. The most challenging step of the proof is Item (\emph{ii}), which we discuss next.\\

\noindent \emph{General upper bounds from the provability of a complexity lower bound.} In Item (\emph{ii}) we aim to extract computational information from a proof of $\LBw^i(s_1,s_2,n_0)$ in $\Tpv^i$. For this, we would like to invoke our game-theoretic witnessing theorem (\Cref{thm:intro_witnessing-tree-exploration}). Since this result can only be applied to a \emph{universal theory}, the first step is to introduce a convenient universal theory that is at least as powerful as $\Tpv^i$. Using standard techniques from logic and similarly to \citep{KrajicekPT91}, we construct a universal theory $\UTpv^i$ with all the necessary properties (see \Cref{thm:universal_theory} in \Cref{sec:universal_theory}). While the axioms of $\UTpv^i$ are structurally simpler (i.e., universal sentences), the terms of $\UTpv^i$ no longer correspond to polynomial-time functions. However, a careful construction of $\UTpv^i$ ensures that its terms (when interpreted over the standard model) correspond to functions in 
 $\mathsf{FP}^{\Sigma^p_{i-1}}$, which will be sufficient for our purposes. In addition to the (syntactic) simplification of the axioms of $\Tpv^i$, a benefit of $\UTpv^i$ is that the worst-case lower bound sentence $\LBw^i(s_1,s_2,n_0)$, whose quantifier complexity grows with $i$, simplifies to a $\forall \Sigma_4^b$-sentence $\mathsf{ULB}_{\mathsf{wst}}^i(s_1,s_2,n_0)$ in the vocabulary of $\UTpv^i$. %

 Since $\mathsf{ULB}_{\mathsf{wst}}^i(s_1,s_2,n_0)$ is also provable in the universal theory $\UTpv^i$, we can invoke the game-theoretic witnessing theorem  with $\calT = \UTpv^i$ and on the formula $\varphi(x)$ corresponding to $\mathsf{ULB}_{\mathsf{wst}}^i(s_1,s_2,n_0)$. (For this overview, think of $x$ as the input length $n$.) Consequently, there is a universal winning $\calL(\UTpv^i)$-strategy for the truthifier (existential player) in the \emph{tree exploration game} of $\varphi(x)$. In particular, for every input length $n \geq n_0$, the truthifier has a winning strategy computed by functions in $\mathsf{FP}^{\Sigma^p_{i-1}}$ that succeeds within $O(1)$ plays in producing a $\Pi_i$-circuit $C_n$ of size $\leq n^d$ that cannot be computed (in the worst case) by $\Sigma_i$-circuits of size $\leq 2^{n^{\delta}}$. 
 
 The plan for the remainder of the proof is to fix a \emph{particular strategy of the falsifier}, which will depend on the circuit $E_m$ from Item (\emph{ii}) that we would like to approximate, and to show that using the $\mathsf{FP}^{\Sigma^p_{i-1}}$-computable winning strategy of the truthifier we can obtain a good circuit $B_m$ for $E_m$.

Similarly, in the simpler context of the Student-Teacher game obtained from the KPT Witnessing Theorem and for a worst-case lower bound sentence that refers to a fixed machine $M$, \cite{DBLP:journals/jml/Krajicek11, DBLP:journals/apal/Pich15, PS21} showed that an average-case complexity upper bound follows from the provability of a worst-case lower bound. We provide a simple example of how this can be done in \Cref{sec:PichSanthanam}, when we discuss Student-Teacher games with a single round in the context of \citep{PS21}. For games with more than one round, techniques from pseudorandomness and a more elaborated strategy that employs the Nisan-Wigderson generator \citep{NW} play an important role in the argument from \cite{DBLP:journals/jml/Krajicek11, DBLP:journals/apal/Pich15, PS21}.

In our context, the following difficulties arise:
\begin{itemize}
    \item[(1)] We need to consider the considerably more complicated tree exploration game played between the truthifier and the falsifier.
    \item[(2)] The machine $M$ becomes an arbitrary circuit $C'$ that the falsifier proposes as a candidate hard function, and different circuits can be proposed until the winning strategy of the truthifier succeeds in producing a hard circuit $C_n$. 
\end{itemize}

 We are able to avoid a difficult analysis in Item (1) by considering a simpler setting of the tree exploration game that is sufficient for our purposes. In more detail, when considering the strategy for the falsifier based on the circuit $E_m$ that we would like to approximate, the play of the falsifier in the current node of the game tree only depends on the partial play of the \emph{evaluation game} corresponding to the moves of both players in the root-to-node path of the \emph{tree exploration game}. This simpler framework is developed in \Cref{sec:witnessing_special_case}, and we believe that it might find more applications in the investigation of the logical foundations of algorithms and complexity theory. 

Finally, in order to address Item (2), we show that it is possible to modify the use of the Nisan-Wigderson generator in \cite{DBLP:journals/jml/Krajicek11, DBLP:journals/apal/Pich15}  when defining the strategy of the falsifier so that even if the truthifier changes the candidate hard circuit $O(1)$ times when we execute its winning strategy, we are still able to obtain a non-trivial complexity upper bound for $E_m$. We refer to \Cref{sec:unprovability_main_section} for the technical details.  %

\subsection{Organization}

We intended to make the exposition accessible to a broad audience and in particular to someone that might not be so familiar with bounded arithmetic. The remaining sections of the paper are organised as follows:
\begin{itemize}[leftmargin=*,topsep=1em,itemsep=0pt]
    \item[--] \Cref{sec:prelim}  fixes notation and presents some basic definitions and useful tools in logic and complexity. \vspace{0.1cm}
    \item[--] \Cref{sec:witnessing-1} formalizes the game-theoretic witnessing theorem  (\Cref{thm:intro_witnessing-tree-exploration}). We defer its proof to \Cref{sec:appendix_witnessing}. A simpler version that is sufficient for the proof of \Cref{thm:intro_main_unprov} is derived in \Cref{sec:witnessing_special_case}.\vspace{0.1cm}
    \item[--] \Cref{sec:PichSanthanam} provides an exposition of Kraj\'{i}\v{c}ek's technique \citep{DBLP:journals/jml/Krajicek11} (further elaborated in \citep{DBLP:journals/apal/Pich15}) and of the main unprovability result from Pich and Santhanam \cite{PS21} in a language that will be more convenient when discussing our proofs. \vspace{0.1cm}
    \item[--] \Cref{sec:unprovability_main_section} combines and extends results from the previous sections in order to establish \Cref{thm:intro_main_unprov}. \vspace{0.1cm} 
    \item[--] \Cref{sec:provability_in_Tpv} discusses provability in the theories $\Tpv^i$ and relates their strength to certain  computational assumptions. \vspace{0.1cm}
    \item[--] \Cref{sec:appendix_witnessing} contains the proofs of the  witnessing theorems. A proof of \Cref{thm:intro_witnessing-tree-exploration} using sequent calculus appears in \Cref{sec:proof_herbrand}. \Cref{sec:no-counterexample} provides a self-contained proof of the witnessing theorem presented in \Cref{sec:witnessing_special_case} using Herbrandization instead of sequent calculus. \vspace{0.1cm}
    \item[--] Appendices \ref{sec: lmms for hardness amp}, \ref{sec:universal-theory}, and \ref{sec:counting_lemma} contain omitted proofs from Sections  \ref{sec:prelim} and \ref{sec:unprovability_main_section}. \vspace{0.1cm}
\end{itemize}

\vspace{0.3cm}

\noindent \textbf{Acknowledgements.} We would like to thank the anonymous STOC reviewers for their helpful comments and for bringing to our attention the witnessing result presented in \citep[Section 2.3]{DBLP:journals/jsyml/BussKT14}. In addition, we thank J\'{a}n Pich for answering questions about \citep{PS21}, Anupam Das for a discussion on extracting computational content from proofs,  and Junhua Yu  for comments on the game-theoretic witnessing argument. We also thank Jan  Kraj\'{i}\v{c}ek for providing valuable feedback on  different parts of the paper. Finally, we are grateful to Marco Carmosino, Emil Je\v r\'abek, Valentine Kabanets, Rahul Santhanam, Antonina Kolokolova, and Lijie Chen  for related discussions. This work received support from the Royal Society University Research Fellowship URF$\setminus$R1$\setminus$191059, the EPSRC New Horizons Grant EP/V048201/1, and the Centre for Discrete Mathematics and its Applications (DIMAP) at the University of Warwick.

\section{Preliminaries}\label{sec:prelim}

This section fixes notation and presents some basic definitions and useful tools in logic and complexity. 

\subsection{Complexity theory}

Given a function $t \colon \mathbb{N} \to \mathbb{N}$, we generalize the definition of each level of the polynomial hierarchy to machines that run in time $t(n)$ in the natural way. For a fixed $i \geq 1$, we let $\cPi{i}\TIME[t]$ denote the set of languages $L$ that admit a deterministic machine $A$ running in time $t(n)$ such that, for every $x \in \{0,1\}^n$, 
$$x \in L \quad \Longleftrightarrow \quad \forall z_1 \in \{0,1\}^{\leq t(n)}\; \exists z_2 \in \{0,1\}^{\leq t(n)} \ldots \;Q_i z_i \in \{0,1\}^{\leq t(n)}\;A(x,z_1, \ldots, z_i) = 1.$$
The class $\cSig{i}\TIME$ is defined in an analogous way. This generalises the classes $\Sigma^p_i$ and $\Pi^p_i$ corresponding to the $i$-th level of the polynomial hierarchy.

We consider (non-uniform) Boolean circuits over a standard set of gates of fan-in at most two, such as  $\{\land,\lor,\lnot\}$. The size of a circuit is the number of gates in the circuit. We adopt this convention only for concreteness, as our results are robust and do not depend on specific details of the circuit model. We let $\mathsf{SIZE}[s]$ denote the set of languages that admit non-uniform Boolean circuits of size $s(n)$.

We also consider circuits and corresponding circuit classes obtained by extending deterministic circuits to circuits with a constant number of alternations. For a fixed $i \geq 1$, we say that a language $L \in \Sigma_i\text{-}\mathsf{SIZE}[s]$ if there is a sequence $\{C_n\}_{n \geq 1}$ of deterministic Boolean circuits $C_n$ of size $s(n)$ such that, for every $x \in \{0,1\}^n$,
$$
x \in L \quad \Longleftrightarrow \quad \exists z_1 \in \{0,1\}^{s(n)}\; \forall z_2 \in \{0,1\}^{s(n)} \ldots \;Q_i z_i \in \{0,1\}^{s(n)}\;C_n(x,z_1, \ldots, z_i) = 1.
$$
The class $\Pi_i\text{-}\mathsf{SIZE}[s]$ is defined in an analogous way. For convenience, we might refer to $\Sigma_1\text{-}\mathsf{SIZE}[s]$ as $\mathsf{NSIZE}[s]$, i.e., the set of languages computed by non-deterministic circuits of size at most $s(n)$. When we write $C(x) = 1$ for a non-deterministic circuit $C$ and input $x$, we implicitly refer to its acceptance condition, i.e., that there is an input $z$ such that $C(x,z) = 1$. We adopt the analogous convention for co-nondeterministic circuits and for  circuit classes with additional alternations.

We will also consider languages that are computed by circuits with oracle gates. For an oracle $\mathcal{O}$, we let $\mathsf{SIZE}^{\mathcal{O}}[s]$ denote the set of languages computed by circuits of size at most $s$ that can also make use of $\mathcal{O}$-oracle gates.

Finally, for convenience we often abuse notation and associate the size of a Boolean circuit to its bit-length, i.e., its description length under a reasonable encoding.

\subsection{Logic and bounded arithmetic}\label{sec:prelim_logic}

We refer to \citep{buss-survey} for an introduction to bounded arithmetic and to the textbooks \citep{Krajicek-book, cook_nguyen_2010} for a comprehensive treatment. Below we review the relevant definitions and fix notation. 

We use $\calL(\mathcal{T})$ to denote the language (vocabulary) of a theory $\mathcal T$. 

For a structure $\calM$ over a language $\calL$, we often write $\mathcal{M} = (\mathcal{D}, \mathcal{I})$ to explicitly refer to its domain $\mathcal{D}$ and interpretations $\mathcal{I}$. As usual, the $\mathcal{M}$-interpretation of a function symbol $f \in \calL$ will be denoted by $f^{\mathcal{M}}$ (similarly for relations and constants). The $\calM$-interpretation of an $\calL$-term $t$ is also denoted by $t^{\calM}$. 

Given a formula $\psi$, we write $\psi(y)$ to explicitly indicate that $y$ \emph{may} be a free variable in $\psi$. For a formula $\varphi(x)$ and a term $t$, we write $\varphi(x/t)$ for substitution of the free variable $x$ with $t$, or simply $\varphi(t)$ if it is clear from the context. Similarly, we use $s(x)$ to denote a term $s$ that may contain $x$ as a free variable, and $s(x/t)$ to denote the substitution of the free variable $x$ with $t$, or simply $s(t)$ if it is clear.\\

\noindent \textbf{The language $\Lpv$.} 
In theoretical computer science one typically considers functions and predicates that operate over binary strings. For the computational models considered in this paper, this is equivalent to operations on integers, by identifying each non-negative integer with its binary representation. For convenience, we adopt the latter perspective when introducing the language (vocabulary) $\Lpv$ of theories $\Tpv^i$.

Let $\mathbb{N}$ denote the set of non-negative integers. For $a \in \mathbb{N}$, we let $|a| = \max\{\lceil \log_2(a+1) \rceil, 1\}$ denote the length of the binary representation of $a$. For a constant $k \geq 1$, we say that a function $f \colon \mathbb{N}^k \to \mathbb{N}$ is computable in polynomial time if $f(x_1, \ldots, x_k)$ can be computed in time polynomial in $|x_1|, \ldots, |x_k|$. Recall that $\FP$ denotes the set of polynomial time functions. While this definition refers to a particular model of computation (Turing machines), Cobham \citep{Cob64} proved that $\FP$ can be introduced in a machine independent way as the closure of a set of base functions under composition and limited recursion on notation. We briefly review this construction.\footnote{This is not strictly needed in our presentation. We include it here because it provides more intuition about the language of theories $\Tpv^i$ and the typical choice in bounded arithmetic of defining $\FP$ over non-negative integers instead of binary strings.} 

Consider the following class $\mathcal{F}_0$ of base functions:
$$
c(x) = 0, \quad s_0(x) = 2 \cdot x, \quad s_1(x) = 2x + 1, \quad \pi^i_\ell(x_1, \ldots, x_\ell) = x_i, \quad x \# y = 2^{|x| \cdot |y|}
$$
We say that a function $f(\vec{x},y)$ is defined from functions $g(\vec{x})$, $h_0(\vec{x},y,z)$, $h_1(\vec{x},y,z)$, and $k(\vec{x},y)$ by \emph{limited recursion on notation} if
\begin{eqnarray}
f(\vec{x},0) & = & g(\vec{x}) \nonumber \\
f(\vec{x},s_0(y)) & = & h_0(\vec{x},y,f(\vec{x},y)) \nonumber \\
f(\vec{x},s_1(y)) & = & h_1(\vec{x},y,f(\vec{x},y)) \nonumber \\
f(\vec{x},y) & \leq & k(\vec{x},y) \nonumber
\end{eqnarray}
for every sequence $\vec{x}$ and $y$ of natural numbers. Let $\mathcal{F}$ be the least class of functions that contains $\mathcal{F}_0$ and is closed under composition and limited recursion on notation.  Cobham \citep{Cob64} proved that $f \in \mathcal{F}$ if and only if $f \in \FP$.

We let $\Lpv$ contain the constant symbols $0$ and $1$, and a function symbol $f$ for every function in $\FP$. In particular, $\Lpv$ contains function symbols for the length function $|x|$, $\leq$, $+$, etc.\footnote{It is also possible to include in $\Lpv$ a function symbol for every polynomial time \emph{algorithm}, where an algorithm can be described from the base functions and   operations allowed in Cobham's characterisation. However, this is inessential in our context. The theories $\Tpv^i$ will contain all true universal sentences, and polynomial time algorithms that compute the same function are provably equivalent in these theories.}

We use the standard notation $n\in\Log$ and $n\in\Log\Log$ for $\exists N~n=|N|$ and $\exists N~n=||N||$, respectively. We define $\forall n\in\Log$ and $\forall n\in\Log\Log$ as $\forall N~\forall n=|N|$ and $\forall N~\forall n=||N||$, respectively.
\\

\noindent \textbf{Bounded formulas and theories $\bm{\Tpv^i}$.} A \emph{bounded quantifier} is a quantifier of the form $Qx \leq t$, where $Q \in \{\exists, \forall\}$ and $t$ is an $\Lpv$-term that does not involve $x$.\footnote{Bounded quantifiers can be expressed with the usual quantifiers from first-order logic. For instance, a formula $\psi(y)$ of the form $\forall x \leq t(y)\;\varphi(x,y)$ is equivalent to $\forall x \;(x \leq t(y) \rightarrow \varphi(x,y))$. On the other hand, a formula  of the form $\exists x \leq t(y)\;\varphi(x,y)$ is equivalent to $\exists x \;(x \leq t(y) \wedge \varphi(x,y))$.} An $\Lpv$-formula $\psi$ is \emph{bounded} if every quantifier in $\psi$ is bounded.

We will need to introduce a hierarchy of bounded formulas to define the theories $\Tpv^i$. We let $\Sigma^b_0 = \Pi^b_0$ be the set of quantifier-free $\Lpv$-formulas. We then recursively define sets $\Sigma^b_i$ and $\Pi^b_i$ of formulas as follows. For each $i \geq 1$, $\Sigma^b_i$ and $\Pi^b_i$ constitute the smallest class of $\Lpv$-formulas such that the following conditions hold:
\begin{itemize}
    \item[1.] $\Sigma^b_{i-1} \cup \Pi^b_{i-1} \subseteq \Sigma^b_{i} \cap \Pi^b_{i}$;
    \item[2.] both $\Sigma^b_{i}$ and $\Pi^b_{i}$ are closed under Boolean connectives $\wedge$ and $\vee$;
    \item[3.] if $\psi(\vec{x}) \equiv \exists y \leq t(\vec{x})\;\varphi(\vec{x},y)$ is a bounded formula and $\varphi \in \Sigma^b_i$, then $\psi \in \Sigma^b_i$;
    \item[4.] similarly, if $\psi(\vec{x}) \equiv \forall y \leq t(\vec{x})\;\varphi(\vec{x},y)$ is a bounded formula and $\varphi \in \Pi^b_i$, then $\psi \in \Pi^b_i$;
    \item [5.] the negation $\neg \psi$ of a formula $\psi$ from $\Sigma^b_i$ is in $\Pi^b_i$ and vice versa.
\end{itemize}
We note that these classes of sentences are often referred to as \emph{strict} $\Sigma^b_{i}$ and $\Pi^b_{i}$ formulas in the literature, as we do not include sharply bounded quantifiers  between bounded quantifiers. 

For convenience, we sometimes describe formulas with the implication symbol $\rightarrow$, implicitly assuming that it is expressed using the Boolean connectives appearing above.

Note that to each $\Lpv$-formula $\phi(x_1, \ldots, x_k)$ we can associate a language $L_\phi \subseteq \{0,1\}^*$ consisting of binary encodings of all tuples $(a_1, \ldots, a_k) \in \mathbb{N}^k$ such that $\mathbb{N} \models \phi(a_1, \ldots, a_k)$. It is known that $\phi \in \Sigma^b_i$ if and only if $L_\phi \in \Sigma^p_i$  \citep{DBLP:journals/tcs/Stockmeyer76, DBLP:journals/tcs/Wrathall76, DBLP:journals/tcs/KentH82}, where $\Sigma^p_i$ denotes the $i$-th level of the polynomial hierarchy.

For $j \geq 0$, we let $\forall \Sigma^b_{j}$ denote the set of $\Lpv$-sentences of the form $\forall \vec{y} \,\varphi(\vec{y})$, where $\varphi$ is a $\Sigma^b_j$-formula. We sometimes write $\Sigma^b_i(\calL)$, $\Pi^b_i(\calL)$, and $\forall \Sigma^b_i(\calL)$ to emphasize the underlying language $\calL$ of a class of formulas.

As expected, the intended model of theories $\Tpv^i$ is $\mathbb{N}$, with the interpretation of each function symbol $f \in \Lpv$ as the corresponding polynomial time function. We will refer to $(\mathbb{N}, 0^{\mathbb{N}}, +^{\mathbb{N}}, \ldots)$ as the \emph{standard model}.

\begin{definition}[Theories $\Tpv^i$]\label{def:theoryTiPV}
For each integer $i\geq 1$, we let $\Tpv^i$ denote the theory of all true (with respect to $\mathbb{N}$) $\forall \Sigma^b_{i-1}$ sentences over the language $\Lpv$.
\end{definition}

In particular, $\Tpv^1$ is the theory of true universal sentences, and we might refer to $\Tpv^1$ just as $\Tpv$.

Note that the definition of $\Tpv^i$ consists of only true $\Sigma^b_i$-sentences without sharply bounded quantifiers as axioms. However, as we observe in \Cref{sec:sharply-bounded}, this is inessential in our unprovability results, given that the introduction of sharply bounded quantifiers would not make the theories $\Tpv^i$ any stronger.\\

In order to simplify the presentation of some results, we introduce the following definition.

\begin{definition}[Closure under if-then-else]\label{def:closed-under-if-then-else}
A theory $\calT$ is \emph{closed under if-then-else} if for every quantifier-free formula $\varphi(\vec x)$ and terms $t_1(\vec x)$ and $t_2(\vec x)$, there exists a term $t(\vec x)$ such that 
$$\calT\vdash \big (t(\vec x)=t_1(\vec x)\land\varphi(\vec x)\big )\lor \big(t(\vec x)=t_2(\vec x)\land\lnot \varphi(\vec x)\big).$$ 
\end{definition}

We note that in such a theory the provability of a disjunction $\psi(x,t_1(x)) \lor \psi(x,t_2(x)) \lor \ldots \lor \psi(x, t_k(x))$ yields the provability of $\psi(x,t(x))$, for a quantifier-free formula $\psi(x,y)$. Typical theories of bounded arithmetic (e.g.,~$\mathsf{S}^1_2$ and $\Tpv^i$) are closed under if-then-else or admit a suitable extension that is closed under this property.\\

\noindent \textbf{Theory $\bm{\APC_1}$.} In order to formalize certain probabilistic methods and randomised algorithms, Je\v{r}\'abek \cite{Jerabek04, Jerabek07} (following \citep{krajivcek2001weak}) introduced the theory $\APC_1$ by extending $\PV$ with the \emph{dual Weak Pigeonhole Principle} for $\PV$ functions, an axiom scheme postulating that there is no $\PV$ function $f:[2^n]\to[(1+1/n)\cdot 2^n]$ that is surjective.\footnote{The size of the codomain (with respect to the size of the domain) affects the power of the dual Weak Pigeonhole Principle. This can be a subtle point, as the equivalence between dual Weak Pigeonhole Principles with different codomain sizes is not known to be provable in $\PV$ (see \cite{Jerabek07-dWPHP} for more details).} (The notation $\APC_1$ was proposed by \citep{DBLP:journals/jsyml/BussKT14}.) More formally, we define $\dWPHP(f)$ for a function $f$ 
 (with extra parameters) as the sentence\footnote{Note that the additional parameter $\vec z$ is crucial in the definition of $\APC_1$. If we remove this parameter in the definition of $\dWPHP$, denoted by $\dWPHP'$, the theory $\PV+\dWPHP'(\PV)$ will be a (possibly) weaker fragment of $\APC_1$ (see, e.g., \cite{PS21}).}
\begin{equation}\label{equ: dWPHP}
\dWPHP(f) \eqdef \forall n\in\Log~~\forall\vec z~~\exists y<(1+1/n)\cdot 2^n~~\forall x<2^n~~f(\vec z,x)\ne y.
\end{equation}
Let $\dWPHP(\PV)\eqdef\{\dWPHP(f)\mid f\in\Lpv\}$. Then $\APC_1\eqdef\PV+\dWPHP(\PV)$. (For a definition of theory $\mathsf{PV}$, see \citep{Krajicek-book} or the equivalent presentation from \citep{jerabek:sharply-bounded}.) Je\v{r}\'abek \cite{Jerabek04,Jerabek-phd,Jerabek07} developed a sophisticated (but intuitive) framework for approximate counting in $\APC_1$ built on an elegant formalisation of the Nisan-Wigderson PRG \cite{NW} in this theory. 

By counting the quantifier alternations in \Cref{equ: dWPHP}, it is easy to see that $\dWPHP(f)$ is a $\forall\Sigma^b_2$-sentence in $\Lpv$. As a result, $\APC_1$ is a subtheory of $\Tpv^3$. We note that our unprovability result for $\APC_1$ (\Cref{cor:intro_APC1}) is quite robust and works with any non-trivial codomain size in the definition of $\dWPHP(f)$, since this does not increase the quantifier complexity of the corresponding sentences.

\subsection{Total search problems and the polynomial hierarchy}
\label{sec:general-ph:complexity-basic}

In this section, we define complexity classes and circuit classes associated with total search problems in the polynomial hierarchy and explore their basic properties.

Recall that a relation $R \subseteq \{0,1\}^* \times \{0,1\}^*$ is in $\TFNP$ and if there is a polynomial $p(n)$ and a polynomial time machine $A$ such that
\begin{itemize}
    \item For every $x \in \{0,1\}^*$ there is $y \in \{0,1\}^{\leq p(|x|)}$ such that $R(x,y)$ holds. Moreover, any such $y$ is of length at most $p(|x|)$.
    \item For every pair $(x,y)$ of strings $x,y \in \{0,1\}^*$, $(x,y) \in R$ if and only if $A(x,y) = 1$. 
\end{itemize}

The next definition is a standard generalisation of this class.

\begin{definition}
 For $i \geq 1$, we say that a relation $R \in \TF\Sigma_i^p$  if there is a polynomial $p(n)$ and a polynomial time   machine $A$  such that the following conditions hold:
 \begin{itemize}
    \item For every $x \in \{0,1\}^*$ there is $y \in \{0,1\}^{\leq p(|x|)}$ such that $R(x,y)$ holds.
    \item For every pair $(x,y)$ of strings $x,y \in \{0,1\}^*$, 
    $$
    R(x,y) \; \Longleftrightarrow \; \forall z_1 \in \{0,1\}^{p(|x|)} \; \exists z_2 \in \{0,1\}^{p(|x|)} \ldots Q_{i - 1} z_{i-1} \in \{0,1\}^{p(|x|)} \; A(x,y,z_1, \ldots, z_{i-1}). 
    $$
    In other words, $R \in \Pi^p_{i - 1}$.
\end{itemize}    
\end{definition}

We will need the following simulation result.

\begin{theorem}
  \label{thm:oracle-circuit-to-nondet-circuit}
  For every $i\ge 1$ and $s(n)\ge n$, $\SIZE^{\Sigma_{i-1}^p}[s(n)]\subseteq \cSig{i}\SIZE[\poly(s(n))]$. 
\end{theorem}

\begin{proof}
The proof is similar to the well-known inclusion $\mathsf{P}^{\Sigma^p_{i-1}} \subseteq \Sigma^p_i$ (see, e.g.,~\citep[Chapter 17]{Papa}), and we omit the details.
\end{proof}

\subsection{The Nisan-Wigderson generator}\label{sec:NW}

In this section, we review basic properties of the Nisan-Wigderson \citep{NW} pseudorandom generator and fix notation. For an introduction to this generator and to computational pseudorandomness, see \citep{Vadhan-Survey}.

\begin{definition}
A collection $\mathcal{S} = \{S_1, \ldots, S_k\}$ of sets $S_i$ is said to be an $(m, \ell, a)$-\emph{design} if 
\begin{itemize}
    \item for every $i \in [k]$, $S_i \subseteq [m]$;
    \item for every $i \in [k]$, $|S_i| = \ell$; and
    \item for every $i \neq j \in [k]$, $|S_i \cap S_j| \leq a$.\footnote{Designs are also called combinatorial designs by some authors. We will use both terms interchangeably.}
\end{itemize}
The \emph{size} of a design $\mathcal{S}$ is defined as the number of sets in $\mathcal{S}$. 
\end{definition}

\begin{lemma}[Explicit designs; see, e.g.,~\citep{NW, Vadhan-Survey}]\label{lem:designs}
For every constant $c \geq 2$ and every sufficiently large $n \in \mathbb{N}$, there exists an $(n^{c},n^{c/2},n)$-design $\mathcal{S}_{c,n}$ of size $2^n$. Moreover, for every fixed $c$, there is an algorithm  that, given a large enough $n$ and an index $i \in [2^n]$, outputs the $i$-th set $S_i \in \mathcal{S}_{c,n}$ in time $\mathsf{poly}(n)$.
\end{lemma}

Recall that, given an $(m, \ell, a)$-design $\mathcal{S}$ of size $N$ and a function $f \colon \{0,1\}^{\ell} \to \{0,1\}$, the \emph{Nisan-Wigderson generator} (NW generator) maps a \emph{seed} $w \in \{0,1\}^m$ into the $N$-bit string 
\[ 
  f(w|_{S_1}) f(w|_{S_2}) \ldots f(w|_{S_N}),
\]  
where $w|_{S_i}$ is the string of length $\ell$ obtained from $w$ be selecting the bits indexed by $S_i \in \mathcal{S}$.

It will be convenient to view the NW generator as a Boolean function and to introduce additional notation. For a large constant $c \geq 1$, given a function $h \colon \{0,1\}^{n^{c/2}} \to \{0,1\}$, we will use the NW generator to define a function $\NW_h \colon \{0,1\}^{n^c} \times \{0,1\}^{n} \to \{0,1\}$. More precisely,
\begin{itemize}
\item The seed length is $n^c$.
\item The corresponding design is described by a $2^n \times n^c$ Boolean matrix $A$ where each row has exactly $n^{c/2}$ entries set to $1$, and the $1$ entries in distinct rows overlap in at most $n$ columns. As stated in \Cref{lem:designs},
designs with these parameters are known to exist. Given a pair $(i,j) \in [2^n] \times [n^c]$, the $(i,j)$-entry of the corresponding design matrix can be  explicitly computed by circuits of size $\mathsf{poly}(n)$ \cite{NW}.
\item For a row index $x \in \{0,1\}^n$ of $A$, we use $J_x \subseteq [n^c]$ to denote the set of $n^{c/2}$ columns of the $x$-th row of $A$ set to $1$.
\item It will often be convenient to consider an $n^{c}$-bit string $w$ as a function in $\{0,1\}^{[n^c]}$ that maps $[n^c]$ to $\{0,1\}$. If $S_1, S_2 \subseteq [n^c]$ is a partition of $[n^c]$, $a \in \{0,1\}^{S_1}$, and $u \in \{0,1\}^{S_2}$, we let $w = u \cup a$ denote the corresponding $n^c$-bit string obtained from $a$ and $u$.\footnote{This notation is consistent with the standard set-theoretic definition of a function as a set of pairs.}
\item For $x \in \{0,1\}^n$ and strings $a \in \{0,1\}^{n^c - n^{c/2}}$ and $u \in \{0,1\}^{n^{c/2}}$, we let $r_{x}(a,u)$ denote the string $w = u \cup a$ of length $n^c$ obtained by viewing $a \in \{0,1\}^{[n^{c}] \setminus J_{x}}$ and $u \in \{0,1\}^{J_{x^*}}$.
\item By fixing the seed $w \in \{0,1\}^{n^c}$ in the NW generator and the function $h \colon \{0,1\}^{n^{c/2}} \to \{0,1\}$, we obtain a function $\NW_h(w) \colon \{0,1\}^n \to \{0,1\}$ in the natural way. Similarly, we can obtain a family $\{\NW_h(w)\}_{w \in \{0,1\}^{n^c}}$ of functions, one for each possible seed $w$.
\end{itemize}

\subsection{Hardness amplification in the polynomial hierarchy}\label{sec:hard_ampl}

In order to relax the average-case complexity parameter in our unprovability results, we need a hardness amplification theorem for the polynomial hierarchy. The result stated below can be seen as the ``relativised'' version of \citep{DBLP:journals/siamcomp/HealyVV06} (see also \citep[Section 3.3]{PS21}). For completeness, we sketch their proof and explain how to adapt the result to our purpose in \Cref{sec: lmms for hardness amp}.

\begin{theorem}\label{thm:hardamp-ppoly} There is a constant $\gamma>0$ and $\ell=\ell(n)=\poly(n)$ such that the following holds for every $i\ge 1$. Let $s_1,s_2 \colon \mathbb{N} \to \mathbb{N}$ be non-decreasing functions, where $s_2(n) = n^{\omega(1)}$, and suppose there is a function $f_n \colon \{0,1\}^n \to \{0,1\}$ computable by $\cSig{i}\mathsf{SIZE}[s_1(n)]$ circuits (resp.~$\cPi{i}\mathsf{SIZE}[s_1(n)]$ circuits) such that each $\Sigma^p_{i-1}$-oracle circuit $A_n$ of size at most  $s_2(n)$ satisfies 
$$
\Pr_{x \in \{0,1\}^n}[f_n(x) = A_n(x)] \;\leq\; 1 - \frac{1}{n}.
$$ 
Then there exist a function $h_\ell \colon \{0,1\}^\ell \to \{0,1\}$ computable by $\cSig{i}\mathsf{SIZE}[\poly(\ell)\cdot s_1(\ell)]$ circuits (resp. $\cPi{i}\SIZE[\poly(\ell)\cdot s_1(\ell^\gamma)]$ circuits) such that each $\Sigma^p_{i-1}$-oracle circuit $B_\ell$ of size at most $s_2(\ell^{\gamma})^\gamma$ satisfies
$$
\Pr_{y \in \{0,1\}^\ell}[h_\ell(y) = B_\ell(y)] \;\leq\; \frac{1}{2} + \frac{1}{s_2(\ell^{\gamma})^{\gamma}}.
$$
\end{theorem}

\subsection{Herbrand's Theorem and the KPT Witnessing Theorem}\label{sec:witnessing_review}

In this section, we review standard  witnessing theorems previously used to show unprovability results in bounded arithmetic (see, e.g., \cite{CKKO21,PS21}). In all results, we consider a universal theory $\calT$ with vocabulary $\calL$.\footnote{Recall that a theory $\calT$ is \emph{universal} if all its axioms are universal formulas, i.e., a formula of the form $\forall \vec{x}\, \varphi(\vec{x})$, where $\varphi$ is free of quantifiers.} (As a concrete example, one can take $\calT=\PV$ and $\calL=\Lpv$.) 

\begin{paragraph}{Two quantifiers ($\forall \exists$).} The well-known Herbrand's theorem is the simplest witnessing result and can be applied to $\forall\exists$-sentences (see, e.g., Section 2 of \cite{applied-proof-theory}).
\end{paragraph}

\begin{theorem}[Herbrand's Theorem]\label{thm:herbrand} Let $\calT$ be a universal theory with vocabulary $\calL$. 
  If $\calT\vdash\forall x\,\exists y\,\varphi(x,y)$ for a quantifier-free $\calL$-formula $\varphi$, there exist a constant $\ell\ge 1$ and a sequence $t_1,t_2,\dots,t_\ell$ of $\calL$-terms such that
  \[
    \calT\vdash\forall x~\big (\varphi(x,t_1(x))\lor\varphi(x,t_2(x))\lor\dots\lor \varphi(x,t_\ell(x))\big ).
  \]
  In particular, if $\calT$ is closed under \emph{if-then-else}, then there is an $\calL$-term $t$ such that $\calT\vdash\forall x~\varphi(x,t(x))$. 
\end{theorem}

\paragraph{Three quantifiers ($\forall \exists \forall$).} The KPT Witnessing Theorem \cite{KrajicekPT91} extends Herbrand's Theorem by providing witnessing functions for the existential quantifier in a provable $\forall\exists\forall$-sentence. 

\begin{theorem}[KPT Witnessing \cite{KrajicekPT91}]\label{thm:KPT} Let $\calT$ be a universal theory with vocabulary $\calL$. 
  Suppose that, for a quantifier-free $\mathcal{L}$-formula $\varphi$, $\calT \vdash \forall x \,\exists y \,\forall z \; \varphi(x, y, z)$. Then there exist a constant $\ell \geq 1$ and a sequence $t_1, \ldots, t_\ell$ of $\mathcal{L}$-terms such that
  \begin{gather*}
    \calT \vdash \forall x \;\forall \vec{z}\;\;
    \big ( \varphi(x,t_1(x),z_1) \lor \varphi(x,t_2(x,z_1),z_2) 
    \lor \dots \lor \varphi(x,t_\ell(x,z_1, \dots ,z_{\ell-1}), z_\ell) \big ).
  \end{gather*}
\end{theorem}

KPT witnessing has a well-known computational interpretation as an interactive game between a student and a teacher (see, e.g.,~\cite{DBLP:journals/apal/Pich15}). In the first round, the student is given an arbitrary input $x$, and computes according to the term $t_1(x)$. This computation provides a candidate object $y_1$. The teacher then replies with an arbitrary ``counterexample'' $z_1$ such that $\neg \varphi(x,y_1,z_1)$ holds, whenever such $z_1$ exists. Note that the next move of the student takes into account previously presented counterexamples, i.e., the term $t_2$ depends on both $x$ and $z_1$. According to Theorem \ref{thm:KPT}, the game ends in at most $\ell$ rounds, and the student is guaranteed to succeed, i.e., to output $y$ such that $\varphi(x,y,z)$ holds for every $z$.

\begin{example}\label{example: non-constructive proof}
    An example of the interactive game is the proof of the existence of two irrational numbers $x,y$ such that $x^y$ is rational (see \Cref{example: non-constructive}), formalized (in some appropriate theory for real numbers) as: 
    \begin{align*}
    &\exists x,y\in\bbR~\exists p,q\in\bbZ~\forall p',q'\in\bbZ~\psi(x,y,p,q,p',q'), \text{ where }\\ 
    &\psi(x,y,p,q,p',q')\triangleq x^y=p/q\land x\ne p'/q'\land y\ne p'/q'
    \end{align*}
    The student wants to learn $x,y,p,q$ such that $\psi(x,y,p,q,p',q')$ holds for every $p',q'$, with the help of a teacher that finds counterexamples $p',q'$ making $\psi(x,y,p,q,p',q')$ false. The student's strategy (say, extracted from the proof using KPT witnessing) is that: 
    \begin{itemize}
        \item In the first round, try $x=(\sqrt 2)^{\sqrt 2},y=\sqrt 2,p=2,q=1$, and ask for a counterexample $p',q'$ from the teacher if it failed.   
        \item Since $y\ne p'/q'$, the student knows that $x=p'/q'$. The student can then propose in the second round that $x=\sqrt 2,y=\sqrt 2,p=p',q=q'$.
    \end{itemize}
\end{example}

\paragraph{Four quantifiers ($\forall \exists \forall \exists$).} It is also known that one can prove a witnessing theorem for $\forall\exists\forall\exists$-sentences using the standard model-theoretical proof of the KPT witnessing theorem. 

\begin{theorem}[KPT Witnessing for $\forall \exists \forall \exists$-Sentences \cite{KrajicekPT91}]\label{thm:KPT4} Let $\calT$ be a universal theory with vocabulary $\calL$. Let $\varphi$ be a quantifier-free $\mathcal{L}$-formula, and suppose that $\mathcal{T} \;\vdash\; \forall x\, \exists y \, \forall z \, \exists w\;\varphi(x,y,z,w)$. Then there is an $\ell\ge 1$ and a finite sequence $t_1, \dots, t_\ell$ of $\calL$-terms such that
\[
\mathcal{T} \;\vdash\;  \forall x,z_1,\dots,z_k\;
\bigl(\psi(z,t_1(z),z_1) \lor \psi(x,t_2(x,z_1),z_2) \lor \dots \lor \psi(x,t_\ell(z_1,\dots,z_{\ell-1}),z_\ell)\bigr),
\]
where $\psi(x,y,z) \triangleq \exists w~\varphi(x,y,z,w)$.
\end{theorem}

\paragraph{Five or more quantifiers?} Unlike the case of four quantifiers, there is no obvious direct generalization of the KPT witnessing theorem to five or more quantifiers. The intuitive reason is that there is more than one universal quantifier within the outermost existential quantifier that we would like to witness, so  the interaction pattern of the student and the teacher, which can provide counterexamples for all but the outermost universal quantifier, becomes much more complicated. This can be mitigated with the use of Herbrandization, as done in \Cref{sec:no-counterexample}, but the corresponding witnessing results become significantly more involved. %

\subsection{A universal theory for $\Tpv^i$}\label{sec:universal_theory}

There are two immediate issues when trying to show the unprovability of the lower bound sentence $\LB^i$ in $\Tpv^i$. Firstly, $\LB^i$ contains more quantifier alternations than a typical witnessing theorem can handle (see \Cref{sec:witnessing-1}).  Secondly, $\Tpv^i$ is not a universal theory if $i > 1$, which violates a common assumption in these results. To address the latter, the first step of our argument is to turn $\Tpv^i$ into a universal theory by introducing Skolem functions. In turn, this will allow us to reduce the quantifier complexity of $\LB^i$ so that the techniques developed in \Cref{sec:witnessing-1} can be applied (see \Cref{sec:circuit-vs-circuit}).

\paragraph{Theory $\Upv^i$ and Language  $\Lpv^i$.} Let $i\ge 1$. For each $(\Pi_{i-1}^b\cup\Sigma_{i-1}^b)$-formula $\alpha(\vec z)$ over $\Lpv$, we introduce a function symbol $f_{\alpha}$ interpreted (in the standard model) as the Boolean-valued function 
\[
f_{\alpha}^\bbN(\vec z)=\begin{cases}
1\quad & \text{if}~\alpha^\bbN(\vec z)\text{ holds}; \\ 
0 & \text{otherwise.}
\end{cases}
\]
Furthermore, when $i\ge 2$, for each $\Sigma_{i-1}^b$-formula $\beta(\vec x,y)$ and term $t$ in $\Lpv$, we introduce a function symbol $g_{\beta,t}$ that is interpreted (in the standard model) as the function\footnote{If the reader is somewhat uncomfortable with the possibility that the smallest $y$ satisfying the condition below might be $0$, we stress that this is not going to be an issue in our construction -- see, e.g., the statement of  \Cref{lmm:defining-axiom-g}.}
\[
  g_{\beta,t}^\bbN(\vec x)=\begin{cases}
                     \text{smallest }y \in \mathbb{N}\text{ s.t.~}\beta^\bbN(\vec x,y)\quad & \text{if}~\exists y\le t^{\mathbb{N}}(\vec x)~\beta^\bbN(\vec x,y); \\
                     0 & \text{otherwise.}
                   \end{cases}
\]
Denote by $\Lpv^i$ the language of $\Lpv$ supplemented with the new function symbols. Let $\Upv^i$ be the theory consisting of all universal true sentences (over the standard model) in $\Lpv^i$.

\paragraph{Correctness of the extension $\Upv^i$.} Now we show that $\Upv^i$ is an extension of $\Tpv^i$, that is, every sentence provable in $\Tpv^i$  is also provable in $\Upv^i$. We state two useful lemmas (see \Cref{sec:universal-theory} for proofs).

\begin{lemma}[Defining Axioms of $g_{\beta,t}$]\label{lmm:defining-axiom-g}
    Let $i\ge 2$, $\beta(\vec{x},y)$ be any $\Sigma_{i-1}^b$-formula in $\Lpv$, and $t$ be any term in $\Lpv$. Then $\Upv^i\vdash \forall\vec x~((\exists y\le t(\vec x)~f_{\beta}(\vec x, y)=1)\leftrightarrow f_{\beta}(\vec x,g_{\beta,t}(\vec x))=1)$.
\end{lemma}

\begin{lemma}[Defining Axioms of $f_\alpha$]\label{lmm:defining-axiom-f}
    For every $i\ge 1$ and $(\Pi_{i-1}^b\cup\Sigma_{i-1}^b)$-formula $\alpha(\vec z)$ in the language $\Lpv$, $\Upv^i\vdash\forall\vec z~(\alpha(\vec z)\leftrightarrow f_\alpha(\vec z)=1)$. 
\end{lemma}

\begin{theorem}\label{thm:Upv-extends-Tpv}
    For every $i\ge 1$ and $\Lpv$-sentence $\varphi$, if $\Tpv^i\vdash \varphi$, then $\Upv^i\vdash\varphi$. 
\end{theorem}

\begin{proof}
To prove this lemma, it is sufficient to show that for every $\varphi\in \Tpv^i$, $\Upv^i\vdash\varphi$. Let $\varphi=\forall \vec{x}~\alpha(\vec{x})$ be an axiom of $\Tpv^i$, where $\alpha(\vec{x})$ is a $\Sigma_{i-1}^b$-formula. By Lemma \ref{lmm:defining-axiom-f}, we only need to show that $\Upv^i$ proves $\forall \vec{x}~f_\alpha(\vec{x})=1$. This follows directly from the fact that $\forall \vec{x}~f_\alpha(\vec{x})=1$ is a true universal sentence in the standard model. 
\end{proof}

\paragraph{Complexity of the function symbols in $\Lpv^i$.}
As we discussed in Section \ref{sec:general-ph:intuition}, we will extract a KPT-style student-teacher game from the provability of the lower bound sentence in the universal theory $\Upv^i$. In this step, the complexity of the student is determined by the complexity of the standard interpretations of the function symbols in the language $\Lpv^i$, which consists of both the polynomial-time computable functions (i.e. the symbols in $\Lpv$) and the new function symbols $f_\alpha$ and $g_{\beta,t}$. Now we determine the complexity of the functions $f_\alpha$ and Skolem functions $g_{\beta, t}$. 

\begin{lemma}\label{lmm:complexity-f} Let $i \geq 1$. 
    For every function symbol $f_\alpha$ in $\Lpv^i$, $f_\alpha^\bbN:\bbN\to\{0,1\}$ is the characteristic function of a language in $\Pi^p_{i-1}\cup\Sigma^p_{i-1}$. 
\end{lemma}

\begin{proof}
    Recall that each $f_\alpha$ is introduced for a $(\Pi_{i-1}^b\cup\Sigma_{i-1}^b)$-formula $\alpha(\vec z)$ with language $\Lpv$ such that $f_\alpha^\bbN$ is the characteristic function of $\alpha^\bbN$, i.e., for every $\vec m\in\vec\bbN$, $f_\alpha^\bbN(\vec m)=1$ if and only if $\alpha^\bbN(\vec m)$ holds. Since $\alpha$ is a bounded formula and the initial function symbols and relation symbols, when interpreted in the standard model, are polynomial-time computable, it is not hard to see that $\alpha^\bbN\in \Pi^p_{i-1}\cup\Sigma^p_{i-1}$.
\end{proof}

\begin{lemma}\label{lmm:complexity-g} Let $i \geq 2$. 
    For every function symbol $g_{\beta,t}$ in $\Lpv^i$, $g_{\beta,t}^\bbN\in\FP^{\Sigma_{i-1}^p}$. 
\end{lemma}

\begin{proof}
  Recall that $g_{\beta,t}$ is introduced for every $\Sigma_{i-1}^b$-formula $\beta$ and term $t$ in the language $\Lpv$, and that $g_{\beta,t}^\bbN(\vec x)$ finds the minimum $y^*$ such that $\beta^\bbN(\vec x,y^*)$ holds if there is $y\le t(\vec x)$ such that $\beta^\bbN(\vec x,y)$ holds, or outputs $0$ otherwise. Note that using a $\Sigma^p_{i-1}$ oracle we can decide for $0\le l\le r\le t(\vec x)$ whether there exists $y\in [l,r]$ such that $\beta^\bbN(\vec x,y)$ holds. So we can perform a binary search over $[0,t(\vec x)]$ to find the minimum $y^*$ such that $\beta^\bbN(\vec x,y^*)$ holds or detect that no such element exists. This is an $\FP^{\Sigma^p_{i-1}}$ computation for every $i \geq 2$.  
\end{proof}

\begin{theorem}\label{thm:lpvi-term-complexity}
  Let $i \geq 1$. For every $\Lpv^i$-term $t(x_1,\dots,x_\ell)$, we have  $t^\mathbb{N}(x_1,\dots,x_\ell)\in\FP^{\Sigma^p_{i-1}}$.  
\end{theorem}
  
\begin{proof}
  This directly follows from \Cref{lmm:complexity-f} and \Cref{lmm:complexity-g}.
\end{proof}

Theory $\Upv^i$ has almost all properties needed for the proof of our results, except that it is not necessarily closed under if-then-else (\Cref{def:closed-under-if-then-else}). This is desirable as it simplifies the statement of our witnessing result and its proof. For this reason, we further modify $\Upv^i$ to guarantee this property.

\paragraph{Theory $\UTpv^i$ and Language $\mathcal{L}^i_{\sf UT}$.} Let $i \geq 1$, and consider the language $\Lpv^i$ introduced before. We extend $\Lpv^i$ as follows. For every $k \geq 1$ and  function $f \colon \mathbb{N}^k \to \mathbb{N}$ in $\FP^{\Sigma_{i-1}^p}$, we introduce a new function symbol $f_{\sf UT}$. Then, we let 
$$\LUT^i \eqdef \Lpv^i \cup \{f_{\sf UT} \mid f \in \FP^{\Sigma_{i-1}^p}\}.$$
Given $\LUT^i$, we define $\UTpv^i$ as the theory of all universal sentences in $\LUT^i$ that are true in the standard model.

\begin{theorem}[Main Properties of {$\UTpv^i$}] \label{thm:universal_theory}
For every $i \geq 1$, the theory  $\UTpv^i$ satisfies the following properties:
\begin{enumerate}
    \item $\UTpv^i$ is a universal theory.
      \item Every $\Lpv^i$-sentence provable in $\Upv^i$ is also provable in $\UTpv^i$.
    \item\label{enum: provable translate} Every $\Lpv$-sentence provable in $\Tpv^i$ is also provable in $\UTpv^i$.
    \item Let $t$ be an arbitrary $\LUT^i$-term, and consider its interpretation $t^{\mathbb{N}} \colon \mathbb{N}^k \to \mathbb{N}$ over the standard model. Then $t^{\mathbb{N}} \in \FP^{\Sigma_{i-1}^p}$.
    \item $\UTpv^i$ is closed under if-then-else.
    \item $\UTpv^i$ is sound, i.e., every sentence provable in  $\UTpv^i$ is true over $\mathbb{N}$.
\end{enumerate}
\end{theorem}

The proof of the theorem is deferred to \Cref{sec:universal-theory}.

\section{Witnessing Theorems for General Formulas}
\label{sec:witnessing-1}

In this section, we introduce a convenient witnessing theorem that works for sentences of arbitrarily high quantifier complexity. As explained in \Cref{sec:introduction}, the result is used in the proof that strong complexity lower bounds cannot be established in $\Tpv^i$. While it is possible to obtain a general witnessing result that holds for an arbitrary universal theory, due to our main applications we restrict our attention to theories of bounded arithmetic. %

\subsection{A game-theoretic witnessing theorem}\label{sec:game_theoretic_witnessing}

Let $\calT$ be a universal bounded theory over the vocabulary $\calL$. Let $\varphi(x)$ be a bounded $\calL$-formula defined as
\begin{align*}
  \varphi(x)~\triangleq~&\exists y_1\le t_1(x)~\forall x_1\le s_1(x,y_1)~\exists y_2\le t_2(x,y_1,x_1)\dots \forall x_{k-1}\le s_{k-1}(x,y_1,x_1,\dots,y_{k-1})\\
  ~&\exists y_k\le t_k(x,y_1,x_1,\dots,y_{k-1}, x_{k-1})~\forall x_{k}\le s_k(x,y_1,x_1,\dots,y_{k})~\phi(x,x_1,\dots,x_k, y_1,\dots,y_k),
\end{align*}
where $\phi(x, \vec{x}, \vec{y})$ is a quantifier-free $\calL$-formula.\\

\noindent \textbf{The Evaluation Game.} We consider an interactive game between two players, the \emph{truthifier} (associated with existential quantifiers in $\varphi$) and the \emph{falsifier} (associated with universal quantifiers in $\varphi$). A \emph{board} is defined as a pair $(\calM,n_0)$, where $\calM$ is a structure over  $\calL$ such that $\calM\vDash\calT$, and $n_0\in\calM$ is an element of its domain.\footnote{For a concrete example, think of $\calM=(\bbN,\leq^\bbN, +^\bbN, \times^\bbN, \ldots)$.} The \emph{evaluation game} for the formula $\varphi(x)$ on the board $(\calM,n_0)$ is played as follows: in the $i$-th round of the game ($1\le i\le k$), the truthifier firstly chooses an assignment $m_i\in\calM$ for $y_i$ such that $m_i\le t_i(n_0,m_1,n_1\dots,m_{i-1}, n_{i-1})$, then the falsifier chooses an assignment $n_{i}\in\calM$ for $x_i$ such that $n_i\le s_{i}(n_0,m_1,n_1, \dots,m_{i})$. The truthifier \emph{wins} if and only if $\phi(x/n_0,\vec x/\vec n,\vec y/\vec m)$ holds in $\calM$. 

\newcommand{\ttf}{\mathtt{f}}
\newcommand{\ttt}{\mathtt{t}}

The \emph{transcript} of a game given strategies $\tau^\ttt$ for the truthifier and $\tau^\ttf$ for the falsifier, denoted by $\left<\tau^\ttt:\tau^\ttf\right>$, is a pair $(\vec n,\vec m)$ of sequences that records the moves for both players.\footnote{For convenience, we might also write the transcript as $(m_1, n_1, \ldots, m_k, n_k)$. The moves of each player will always be clear in each context.} A \emph{partial transcript} is a prefix of a transcript. A partial transcript is \emph{valid} if all elements $m_i$ and $n_i$ respect the corresponding upper bounds (in $\calM$) given by functions $t_i$ and $s_i$. A strategy $\tau^\ttt$ for the truthifier is said to \emph{beat} a strategy $\tau^\ttf$ for the falsifier (w.r.t.~a given board and formula) if the truthifier wins in the transcript $\left<\tau^\ttt:\tau^\ttf\right>$. A \emph{strategy} for a player is defined in the natural way, i.e., as a function that produces the next assignment given a partial transcript of the game. Equivalently, since we will consider games with only a fixed number of rounds, one can describe a strategy as a finite sequence of functions of the form $f \colon \calM^i \to \calM$, for appropriate values $i \leq 2k$.\\

We will consider games that are played in a more general setting. Roughly speaking, we allow the truthifier and falsifier  to simultaneously play different evaluation games over the same  board $(\calM,n_0)$. The truthifier has a positional advantage over the falsifier: it can decide where to make the next move, i.e., by either making the next move in one of the current games or starting a new game over the board $(\calM,n_0)$ or playing differently some earlier play, which creates a new game from there but maintains the existing game plays. The falsifier must respond to that move in the corresponding game. \emph{Crucially, the next assignment selected by each player now depends on previous plays in all games.} The formal details are provided next.\\

\noindent \textbf{The Tree Exploration Game.} A \emph{partial game tree} $T = (V,E, \gamma)$ (where $(V,E)$ is a directed rooted tree and $\gamma \colon E \to \calM \times \calM$) of the evaluation game for $\varphi$ on the board $(\calM,n_0)$ is defined as a finite rooted tree where each edge $e \in E(T)$ is labeled with a pair $(m,n)$ of elements of $\calM$ and such that, for every node $u \in V(T)$, the concatenation of each pair of elements labelling the edges on the root-to-$u$ path is a prefix of a valid transcript of the evaluation game of $\varphi(x)$ on the board $(\calM,n_0)$. More precisely, if the pairs labelling the edges from the root to $u$ are $(m_1,n_1),(m_2,n_2),\dots,(m_i,n_i)$, then $(m_1,n_1,m_2,n_2,\dots, m_i,n_i)$ is a valid partial transcript of the evaluation game, i.e., for all $j\in[i]$, $\calM\vDash m_j\le t_j(n_0,m_1,n_1,\dots,m_{j-1}, n_{j-1})$ and $\calM\vDash n_j\le s_j(n_0,m_1,n_1,\dots,m_j)$. Note that if $\calM$ is the standard model then a partial game tree of the evaluation game is a finite upper part of the (exponential size) complete game tree of the evaluation game. 

Let $T$ be a partial game tree of $\varphi$ and $(\calM,n_0)$ be a board. The \emph{tree exploration game} starting from $T$ on $(\calM,n_0)$ is played as follows. In each \emph{round}, first the truthifier chooses a node $u$ from $T$ (not necessarily a leaf) and an element $m\in\calM$, then the falsifier chooses an element $n\in\calM$. This creates a child of $u$ and a corresponding directed edge labeled by $(m,n)$. Note that when playing each round of the tree exploration game both players should guarantee that the new partial game tree is always a valid partial game tree, i.e., $m$ and $n$ should satisfy the corresponding inequalities. The \emph{size} of a partial game tree $T$ is given by $|T(V)|$.

The truthifier \emph{wins} the tree exploration game if there is a node in the current partial game tree that is a winning node for the truthifier, that is, the concatenation of the pairs of elements labelling the edges on the root-to-$u$ path forms a winning transcript of the truthifier in the evaluation game of $\varphi(x)$ on the board $(\calM,n_0)$. The \emph{tree exploration game of $\varphi(x)$} is defined as the tree exploration game starting from a partial game tree containing only the root node. We refer to \Cref{fig:tree-exploration-game} for an example of the tree exploration game.\\

Recall that $\calL$ is the vocabulary of the  universal (bounded) theory $\calT$. The main result established in this section shows the existence of a  ``computationally bounded'' winning strategy for the truthifier from a $\mathcal{T}$-proof of $\varphi$, i.e., the strategy can be computed by $\calL$-terms. In addition, the strategy is universal, in the sense that it is specified by $\calL$-terms that are independent of the board $(\calM, n_0)$. Finally, the location of each play of the truthifier in the partial game tree is fixed in advance and does not depend on the strategy of the falsifier nor on the board $(\calM, n_0)$. (The elements selected by the truthifier depend on the previous plays of the truthifier and falsifier.) This means that the trees in the sequence of partial game trees are fixed in advance.\\

\noindent \textbf{$\bm{\calL}$-Strategies for the Tree Exploration Game.} An $\calL$-\emph{quasi-strategy} of the truthifier of \emph{length} $\ell\in\bbN$ and initial tree size $d$ is a sequence $\tau=\left<p_1,r_1,p_2,r_2,\dots,p_\ell,r_\ell\right>$, where each $p_i$ is an $\calL$-term and each $r_i\in\bbN$ is such that $1\le r_i<d+i$. Let $(\calM,n_0)$ be a board and $T$ be a partial game tree on this board with $V(T) = \{1,2,\dots,d\}$. The strategy for the tree exploration game starting from the partial game tree $T$ \emph{induced} by $\tau$ proceeds as follows: 
\begin{itemize}
    \item In the $i$-th move, the truthifier introduces a node numbered $d+i$ as a child of the node $r_i$ and chooses the element $v_i\triangleq p_i^\calM(n_0,\Gamma)\in\calM$, where $\Gamma$ is the sequence of $\calM$-elements chosen by the players in previous rounds (i.e., $v_1,\dots,v_{i-1}$ and the falsifier's moves $w_1, \ldots, w_{i-1}$). 
\end{itemize}
Note that an arbitrary $\calL$-\emph{quasi-strategy} might induce an invalid move $v_i = p_i^\calM(n_0,\Gamma)$ that violates the desired upper bound on $v_i$, depending on the moves of the falsifier. We say that an $\calL$-quasi-strategy of the truthifier is an \emph{$\calL$-strategy} if for every board $(\calM,n_0)$ the resulting partial game trees are valid for every valid strategy of the falsifier.

Finally, a length-$\ell$ $\calL$-strategy is said to be a \emph{universal winning strategy} if the truthifier wins within $\ell$ moves against all valid strategies (not necessarily $\calL$-strategies) of the falsifier on any board $(\calM,n_0)$. Note that the falsifier's strategy is a function of the board $(\calM,n_0)$, partial game tree $T=(V,E,\gamma)$ (which includes all moves from previous rounds), and the move of the truthifier in the current round.

\begin{theorem}[Game-Theoretic Witnessing Theorem]\label{thm:witnessing-tree-exploration}
    Let $\calT$ be a universal bounded theory  with vocabulary $\calL$ that is closed under if-then-else. Let $\varphi$ be a bounded $\calL$-formula  of the form
    \begin{align*}
    \varphi(x)~\triangleq~&\exists y_1\le t_1(x)~\forall x_1\le s_1(x,y_1)~\exists y_2\le t_2(x,y_1,x_1)\dots \forall x_{k-1}\le s_{k-1}(x,y_1,x_1,\dots,y_{k-1})\\
  ~&\exists y_k\le t_k(x,y_1,x_1,\dots,y_{k-1}, x_{k-1})~\forall x_{k}\le s_k(x,y_1,x_1,\dots,y_{k})~\phi(x,x_1,\dots,x_k, y_1,\dots,y_k),
\end{align*}
where $\phi(x, \vec{x}, \vec{y})$ is a quantifier-free $\calL$-formula. Then $\calT\vdash\forall x~\varphi(x)$ if and only if there is a universal winning $\calL$-strategy of length $O(1)$ for the truthifier in the corresponding tree exploration game of $\varphi(x)$. 
\end{theorem}

We defer the proof of \Cref{thm:witnessing-tree-exploration} to  \Cref{sec:appendix_witnessing}.

\subsection{A special case: Falsifiers with oblivious strategies}\label{sec:witnessing_special_case}

In this section, we present a  special case of  game-theoretic witnessing  (\Cref{thm:witnessing-tree-exploration}) that involves \emph{sequential} invocations of the \emph{evaluation game} played against an \emph{oblivious falsifier}. This version is sufficient to show the unprovability of strong circuit lower bounds in bounded arithmetic (\Cref{sec:circuit-vs-circuit}). 

We assume familiarity with the notation from \Cref{sec:game_theoretic_witnessing}. In particular, let $\calT$, $\calL$, and $\varphi(x)$ be defined as in \Cref{sec:game_theoretic_witnessing}. The main difference is that here we consider the evaluation game (as opposed to the tree exploration game) in the presence of \emph{ancillary information for the truthifier}, as explained next.\\

\noindent \textbf{Strategies with Ancillary Information.} Let $\calM=(\calD,\calI)$ be a model for $\calT$. An $\calL$-strategy for the truthifier \emph{with ancillary information} in the evaluation game of $\varphi(x)$ is a sequence $\tau^{\ttt}=(p_1,p_2,\dots,p_k)$ of $k$ $\calL$-terms, where $p_i\triangleq p_i(n_0,m_1,n_1,\dots,m_{i-1},n_{i-1},\vec a)$ means given the \emph{ancillary information} $\vec a$ (constantly many elements from $\calD$), $n_0 \in \calD$, and moves $m_1,n_1,\dots,m_{i-1},n_{i-1}\in\calD$, the truthifier chooses $m_i = p_i^\calM(n_0,m_1,n_1,\dots,m_{i-1},n_{i-1},\vec a)$ as the current move. For every $\vec a\in\vec\calD$, the strategy induced by $\tau^{\ttt}$ given $\vec a$ as ancillary information is denoted by $\tau^{\ttt}[\vec a]$. In particular, if the $\calL$-strategy has no ancillary information, the induced strategy is denoted by $\tau^{\ttt}[\varnothing]$. Similarly to \Cref{sec:game_theoretic_witnessing}, the \emph{transcript} of a game given strategies $\tau^\ttt$ for the truthifier (possibly with ancillary information) and $\tau^\ttf$ for the falsifier, denoted by $\left<\tau^\ttt:\tau^\ttf\right>$, is a pair $(\vec n,\vec m)$ of sequences that records the moves of both players.

\begin{theorem}[Winning strategies against oblivious falsifiers]\label{thm:witnessing-general}
  Let $\calT$ be a universal theory over the language $\calL$ that is closed under if-then-else. Let $\varphi(x)$ be the formula
  \begin{align*}
    \varphi(x)~\triangleq~&\exists y_1\le t_1(x)~\forall x_1\le s_1(x,y_1)~\exists y_2\le t_2(x,y_1,x_1)\dots \forall x_{k-1}\le s_{k-1}(x,y_1,x_1,\dots,y_{k-1})\\
  ~&\exists y_k\le t_k(x,y_1,x_1,\dots,y_{k-1}, x_{k-1})~\forall x_{k}\le s_k(x,y_1,x_1,\dots,y_{k})~\phi(x,x_1,\dots,x_k, y_1,\dots,y_k),
  \end{align*}
  where $\phi(x, \vec{x}, \vec{y})$ is a quantifier-free $\calL$-formula.
  If $\calT\vdash\forall x~\varphi(x)$, then there  is a constant $\ell\in\bbN$ and  $\calL$-strategies $\tau^\ttt_1,\tau^\ttt_2,\dots,\tau^\ttt_\ell$ \emph{(}with ancillary information\emph{)} such that, for any board $(\calM,n_0)$ and  evaluation game of $\varphi(x)$ on $(\calM, n_0)$, for every strategy $\tau^\ttf$ of the falsifier:
  \begin{itemize}[itemsep=0pt]
  \item either $\hat\tau^\ttt_1\triangleq \tau^\ttt_1[\varnothing]$ beats $\tau^\ttf$,
  \item or $\hat\tau^{\ttt}_2\triangleq \tau^\ttt_2[\left<\hat\tau^\ttt_1:\tau^\ttf\right>]$ beats $\tau^\ttf$,
  \item or $\hat\tau^\ttt_3\triangleq \tau^\ttt_3[\left<\hat\tau^\ttt_1:\tau^\ttf\right>,\left<\hat\tau^\ttt_2:\tau^\ttf\right>]$ beats $\tau^\ttf$,
  \item $\dots$,
  \item or $\hat \tau^\ttt_\ell\triangleq \tau^\ttt_\ell[\left<\hat\tau^\ttt_1:\tau^\ttf\right>,\left<\hat\tau^\ttt_2:\tau^\ttf\right>,\dots,\left<\hat\tau^\ttt_{\ell-1}:\tau^\ttf\right>]$ beats $\tau^\ttf$.
  \end{itemize}
\end{theorem}

\vspace{0.2cm}

Before we establish this result, a few comments are in order. First, notice that the moves of the falsifier only depend on the previous moves in the \emph{current  game}. On the other hand, the truthifier gets as ancillary information the transcripts of all previous games, and  succeeds in beating the strategy of the falsifier after (sequentially) playing at most $\ell = O(1)$ games. Intuitively, the falsifier is \emph{oblivious}, since its moves in the current game do not depend on the moves from any previously completed or different game played in parallel, as in the tree exploration game described in \Cref{sec:game_theoretic_witnessing}. Consequently, when extracting computational information from proofs (where one defines appropriate strategies for the falsifier and considers the behaviour of the truthifier), \Cref{thm:witnessing-general} is more limited than \Cref{thm:witnessing-tree-exploration}.

\begin{proof}[Proof of \Cref{thm:witnessing-general}]

Intuitively, as explained above, the meaning of the theorem is that the truthifier has a winning strategy (with ancillary information) in $\ell$ sequential plays of the evaluation game when the falsifier's strategy is fixed. We will obtain such a strategy from a strategy for the truthifier that succeeds in the tree exploration game. This is not entirely obvious, since there is a mismatch between the games: the next play of the truthifier in the tree exploration game depends on all previous plays in the game tree, while in the evaluation game there is no game tree and they play a sequence of evaluation games.

    Let $\calT,\calL$, and $\varphi(x)$ be defined as above. By Theorem \ref{thm:witnessing-tree-exploration}, there exists an $\ell = O(1)$ length $\calL$-term winning strategy $\tau^{\sf tree}$ for the tree exploration game of $\varphi(x)$.  Let $(\calM,n_0)$ be a board. Consider the tree exploration game when the falsifier plays with a \emph{fixed strategy of the evaluation game}, that is, there exist functions $f_1(x,y_1),f_2(x,y_1,y_2),\dots,f_k(x,y_1,y_2,\dots,y_k)$ such that: 
    
    \begin{itemize}[topsep=5pt]
        \item In the $i$-th step, if the truthifier adds a node $v$ to the partial game tree as a child of $u$ and chooses $m$ as the label and $(m_1,n_1),(m_2,n_2),\dots,(m_d,n_d)$ are the labels on the length-$d$ path from the root to $u$, then the falsifier's move is $f_{d+1}(n_0,m_1,m_2,\dots,m_d,m)$. 
    \end{itemize}
    
    We say that a falsifier's strategy of this form in the tree exploration game is \emph{oblivious}, i.e., the next move of the falsifier only considers moves in the corresponding root-to-node path. Since $\tau^{\sf tree}$ is a winning strategy for the tree exploration game, it beats all strategies of the falsifier,  including oblivious strategies. 
    
    We would like to simulate  $\tau^{\sf tree}$, a strategy for the tree exploration game, in the context of \Cref{thm:witnessing-general}, where the evaluation game is played sequentially and the truthifier has ancillary information. The main idea is to play each \emph{round} in the \emph{tree exploration game} as a new \emph{game} in the \emph{evaluation game} that simulates the current root-to-node path. This guarantees when translating strategies that all the necessary information from the tree exploration game appears in the transcript of previous plays (ancillary information) during the next evaluation game. (If the root-to-node path at the end of a round in the tree exploration game is only a partial play of the corresponding evaluation game, the truthifier simply outputs $0^{\mathcal{M}}$ in the current evaluation game until a new game can be started.) In other words, when the truthifier adds a node $v$ as the child of $u$, it can ``replay'' the path from the root to $v$ using the moves $m_1,m_2,\dots,m_d$ on the path, and the oblivious falsifier will choose the moves $n_1,n_2,\dots,n_d$ as response. Therefore the truthifier can simulate the winning strategy for the tree exploration game by sequentially playing the evaluation game $\ell$ times and beating the falsifier in at least one of the games. 
    
    We now describe in more detail the translation of an $\calL$-term universal winning strategy in the tree exploration game into $\calL$-strategies (with ancillary information) for the evaluation game. Consider the strategy $\tau^{\sf tree}=\left<p_1,r_1,p_2,r_2,\dots,p_\ell,r_\ell\right>$, where $\ell$ is a constant. Recall that the location of each play of the truthifier in the tree exploration game is fixed, and that $r_1, \ldots, r_\ell$ describe the nodes to which a new child is added in each play. For each $i\in[\ell]$, we define an $\calL$-term strategy for the evaluation game $\tau^{\sf eval}_i$ as follows: 
    \begin{itemize}[topsep=5pt,itemsep=5pt]
               \item Let $r_i$ be the node of the game tree that is extended during the $i$-th play of the tree exploration game. Suppose this node is at the $d_i$-th level of the tree, and let $p_{i_1}, \ldots, p_{i_{d_i}}$ be the $\calL$-terms corresponding to the moves of the truthifier in the root-to-$r_i$ path, including the current move.
        \item Define the following $\calL$-strategy $\tau^{\sf eval}_i$ of the evaluation game: (1) parse the ancillary information as a sequence $\Gamma$ of transcripts derived from playing strategies $\tau^{\sf eval}_1,\dots,\tau^{\sf eval}_{i-1}$ with the ancillary information described in the statement of the theorem; (2) in the $j$-th step (during the $i$-th evaluation game), where $j\in[k]$, if $j \leq d_i$ play according to $p_{i_j}$ using that all plays from previous rounds of the tree exploration game are available in the transcript $\Gamma$. Otherwise, choose $0^{\calM}$ (i.e., the $j$-th term defining the strategy is the constant term $0$).
    \end{itemize}
    
    From the discussion above, the correctness of the translation is clear: if the strategies $\tau^{\sf eval}_1,\tau^{\sf eval}_2,\dots,\tau^{\sf eval}_{\ell}$ cannot beat a fixed falsifier strategy $\tau^{\ttf}$ in $\ell$ sequential plays of the evaluation game, we can use the oblivious strategy defined by $\tau^{\ttf}$ in the tree exploration game to show that the truthifier does not win the tree exploration game withing $\ell$ moves.  
\end{proof}

\begin{remark}
Instead of viewing \Cref{thm:witnessing-general} as a special case of the game-theoretic witnessing theorem that employs the tree exploration game (\Cref{thm:witnessing-tree-exploration}), we can also establish the result in a more direct way using a technique known as the \emph{no-counterexample interpretation}. We present a self-contained proof in Appendix \ref{sec:no-counterexample}.
\end{remark}

\addtocontents{toc}{\protect\setcounter{tocdepth}{2}}
\section{Warm-up: Kraj\'{i}cek's Technique and the Pich-Santhanam Result}\label{sec:PichSanthanam}

In this section, we provide a detailed exposition of the unprovability result from Santhanam and Pich \citep{PS21}, which relies on a technique introduced by Kraj\'{i}cek \citep{DBLP:journals/jml/Krajicek11} and further investigated by Pich \citep{DBLP:journals/apal/Pich15}. Their result (intuitively) means that strong average-case circuit lower bounds against co-nondeterministic circuits are not provable in $\TPV$. Concretely, for every $L\in\NTIME[2^{n^{o(1)}}]$, $\delta\in(0,1)\cap\bbQ$, and $n_0\in\bbN$, $\TPV$ cannot prove that: 
\begin{quote} 
    For every $n> n_0$ and every co-nondeterminisetic circuit $C:\{0,1\}^n\to\{0,1\}^1$ of size $2^{n^{\delta}}$, $C(x)=L(x)$ on at most $\frac{1}{2}+\frac{1}{2^{n^\delta}}$ fraction of $x\in\{0,1\}^n$. 
\end{quote} 
Since our unprovability results are obtained by extending the original ideas of Pich and Santhanam \cite{PS21} and Kraj\'{i}cek \citep{DBLP:journals/jml/Krajicek11} in combination with our new witnessing theorem, this section might be particularly helpful for a reader that is unfamiliar with these methods.

\subsection{Formalization of complexity lower bounds}

\ignore{\noindent \textbf{Theory PV and its extensions.} We will consider an extension of $\PV$, and establish an unprovability result for this stronger theory. Let $\mathcal{L}(\PV)$ be the language of $\PV$, which contains a function symbol for every polynomial-time computable function $f \colon \mathbb{N}^k \to \mathbb{N}$, where $k \geq 1$ is fixed. (Whenever this is needed, we employ a canonical efficiently computable map between boolean strings and natural numbers.) We use $\TPV$ to denote the theory consisting of all true (over the standard model) universal sentences in the language $\mathcal{L}(\PV)$.\\}

While the unprovability result of \cite{PS21} is robust to some details of the formalization, we will make a few comments here about the way it is done. First, we can represent any natural number $a \in \mathbb{N}$ by an $\mathcal{L}(\PV)$-term, e.g., $a = 1 + 1 + \ldots + 1$, where $+ \colon \mathbb{N}^2 \to \mathbb{N}$ is the $\LPV$ function symbol for addition, and $1$ is a constant symbol in $\LPV$.\footnote{Of course, one can consider more efficient encodings (e.g.,~dyadic notation), but this will not make a difference in our argument.} From this, we can introduce representations for other finite objects. For instance, a natural number can represent the code of a Turing machine $M$, while a pair of natural numbers can represent a rational number $\delta \in \mathbb{Q}$. In some cases, we will quantify over all such objects in the meta-language, e.g., if $M$ is a Turing machine (in the usual sense), then we can consider a $\LPV$-sentence $\phi_M$ that refers to machine $M$ via its representation as a natural number.

For a nondeterministic Turing machine $M$, a constant $n_0 \in \mathbb{N}$, and  functions $s, m \colon \mathbb{N} \to \mathbb{N}$, we write $\LB(M,s,m,n_0)$ to denote an $\LPV$-sentence stating that, for every input length $n \geq n_0$ and for every co-nondeterminstic circuit $D_n(x,z)$ of size $\leq s(n)$, there are at least $m = m(n)$ distinct input strings $x^1, \ldots, x^m \in \{0,1\}^n$ such that $M(x^i) \neq D_n(x^i)$ for each $1 \leq i \leq m$.\footnote{Here and throughout the exposition, we use that $M(x)= 1$ if and only if there exists $y$ such that $M(x,y) = 1$ (as $M$ computes \emph{nondeterministically}), while $D_n(x) = 1$ if and only if for every $z$ we have $D_n(x,z) = 1$ (since $D$ is a \emph{co-nondeterministic} circuit).} A bit more formally, this sentence can be expressed in $\LPV$ in the following way, where we assume that $M$ on input length $n$ runs in time $\leq t(n)$ for some efficiently computable time bound $t(n) \leq N(n) = 2^n$ and that $s(n)$ and $m(n)$ are  efficiently computable and bounded by $N = 2^n$:
\begin{eqnarray}
    \LB(M,s,m,n_0) & \triangleq & \forall v \;\forall N=|v| \;\forall n = |N|~\text{such that}~n\geq n_0~~~(\text{in other words,}~n \in \mathsf{LogLog}) \nonumber \\
    & ~ & \forall\;\text{co-nondet.~circuit~}D_n\text{~of size~}\leq s(n) \nonumber \\
    & ~ & \exists m = m(n)~\text{distinct~}n\text{-bit strings}~x^1, \ldots, x^{m}~\text{s.t.}~\mathsf{Error}_{M,D_n}(x^i)~\text{for all}~i \in [m],\nonumber 
\end{eqnarray}
where we let $\mathsf{Error}_{M,D_n}(x)$ denote the following $\LPV$-formula:
$$
\mathsf{Error}_{M,D_n}(x)  \;\equiv\; \Big [ \exists y \; \exists z\;M(x,y)=1 \,\wedge\, D_n(x,z) = 0  \Big ] \;\vee\; \Big [\forall y'\;  M(x,y') = 0 \,\wedge\, \forall z' \; D_n(x,z') = 1 \Big ],
$$
with the length of $y, y'$ and $z, z'$ bounded by the running time of $M$ and the size of $D_n$, respectively.

The definition above can be made formal by the use of explicit $\LPV$-function symbols that evaluate circuits and machines on a given input and that perform other necessary checks, e.g., deciding when a given object represents a circuit of size at most $s(n)$. All this can be done without increasing the quantifier complexity of the resulting sentence, since $n \in \mathsf{LogLog}$ and polynomial-time computations over $N = 2^n$-bit strings are feasible. For the same reason, the quantification over $i \in [m]$ does not increase quantifier complexity, using that $m(n) \leq N$. Indeed, in the sentence it is enough to existentially quantify over $m(n)$ strings $x^i$ and over $m(n)$ strings $y^i,z^i$ followed by a universal quantification over $m(n)$ strings $y'^i, z'^i$, and the remaining error conditions can be expressed using a single $\LPV$-function symbol that gets as input the encoding of each collection of strings (formally, each family of $m$ strings is a single object, and the strings are decoded from it). Overall, we get that $\LB(M,s,m,n_0)$ is a $\forall \Sigma^b_2$-$\LPV$ sentence. 

\begin{theorem}[$\TPV$ doesn't prove strong a.e.~average-case co-nondeterministic lower bounds for $\mathsf{NP}$]\label{thm:unprov_lb}
For every $n_0 \in \mathbb{N}$ and $\delta \in \mathbb{Q} \cap (0,1)$, if $M$ is a nondeterministic machine whose running time is bounded by some constructive function $t(n) = 2^{n^{o(1)}}$, then
$$
\TPV \nvdash \LB(M,s, m,n_0),
$$
where $s(n) = 2^{n^\delta}$ and $m(n) = 2^n/2 - 2^n/2^{n^{\delta}}$.
\end{theorem}

In particular, for every language $L \in \mathsf{NP}$ and $\delta > 0$ it is consistent with $\TPV$ that there are infinitely many input lengths $n$ and a co-nondeterministic circuit $D_n$ of size $\leq 2^{n^{\delta}}$ such that 
$$
\Pr_{x \sim \{0,1\}^n}[L(x) = D_n(x)] \geq 1/2 + 2^{-n^{\delta}}.
$$

A strengthening of \Cref{thm:unprov_lb} is discussed in \Cref{sec:hard_amplif}.

\subsection{Proof of Theorem \ref{thm:unprov_lb}}\label{sec:proof_main}

Let $n_0$, $\delta$, $M$, $t(n)$, $s(n)$, and $m(n)$ be as in the statement of Theorem  \ref{thm:unprov_lb}. Arguing as in \citep{PS21}, we assume towards a contradiction that 
$$
\TPV \vdash \LB(M,s, m,n_0).
$$
Let $L \subseteq \{0,1\}^*$ be the language defined by $M$. We argue as follows.
\begin{enumerate}
    \item\label{enum: step 1 approx} From the provability of this almost-everywhere average-case lower bound against co-nondeterministic circuits, it follows by the soundness of $\TPV$ that (in the standard model) for every sequence $\{E_n\}_{n \geq 1}$ of \emph{deterministic} circuits $E_n$ of size $\leq 2^{n^{\delta}}$, if $n \geq n_0$ then 
    $$
    \Pr_{x \sim \{0,1\}^n}[L(x) = E_n(x)] \leq 1/2 + 2^{-n^{\delta}}.
    $$
    \item From the provability of the sentence $\LB(M,s, m,n_0)$ it trivially follows that $\TPV$ proves a sentence $\LBw(M,s,n_0)$ which states a \emph{worst-case} lower bound for $M$ against co-nondeterministic circuits of the same size. We then show that the provability of $\LBw(M,s,n_0)$ in $\TPV$ implies that, in the standard model, for every fixed $k\geq 1$ and for every large enough $n$, there is a deterministic circuit $A$ defined over $n^k$ input variables and of size $2^{O(n)}$ such that
    $$
    \Pr_{w \sim \{0,1\}^{n^k}}[L(w) = A(w)] \geq 1/2 + 2^{-O(n)}.
    $$
    Taking $k > 1/\delta$ contradicts \Cref{enum: step 1 approx} above.
\end{enumerate}

Note that the only remaining step is to show that:

\begin{quote} 
$(\star)$ The provability of a worst-case lower bound against \emph{co-nondeterministic} circuits allows us to non-trivially approximate $L$ using \emph{deterministic} circuits of bounded size.
\end{quote}

Before proceeding with the proof of this result, we describe the aforementioned worst-case lower bound sentence in a convenient way.
\begin{eqnarray}
    \LBw(M,s,n_0) & \equiv & \forall n \in \mathsf{LogLog}~\text{with}~n\geq n_0,~\forall ~\text{co-nondet. circuit}~D~\text{of size}~\leq s(n) \nonumber \\
    & ~ & \exists x \in \{0,1\}^n~\exists y \in \{0,1\}^{t(n)}~\exists z \in \{0,1\}^{s(n)}~\text{such that}~\mathsf{Error}(x,y,z), \nonumber
\end{eqnarray}
where here $\mathsf{Error}(x,y,z)$ denotes the following $\LPV$-formula:
\begin{equation}\label{eq:Error}
    \mathsf{Error}(x,y,z) \;\equiv\; \Big [ M(x,y)=1 \,\wedge\, D(x,z) = 0  \Big ] \;\vee\; \Big [\forall y'\;  M(x,y') = 0 \,\wedge\, \forall z' \; D(x,z') = 1 \Big ],
\end{equation}
where the lengths of $y'$ and $z'$ are bounded as before. Observe that $\LBw(M,s,n_0)$ is also a $\forall \Sigma^b_2$-$\LPV$ sentence. 

It is easy to see that, under any reasonable formalization, if $m(n) \geq 1$ then $\TPV$ derives the worst-case lower bound sentence $\LBw(M,s,n_0)$ from the average-case lower bound sentence $\LB(M,s,m,n_0)$. Consequently, it is sufficient for us to prove the following lemma, which formalizes statement $(\star)$.

\begin{lemma}[Non-trivial correlation from the provability of a worst-case lower bound]\label{lem:main}
Let $n_0 \in \mathbb{N}$, $\delta \in \mathbb{Q} \cap (0,1)$, $M$ be a nondeterministic machine whose running time is bounded by some constructive function $t(n) = 2^{n^{o(1)}}$, and $s(n) = 2^{n^{\delta}}$. If
$$
\TPV \vdash \LBw(M,s,n_0),
$$
then for every $k \geq 1$ and sufficiently large $n$, there is a deterministic circuit $B:\{0,1\}^{n^k}\to\{0,1\}$ of size $2^{O(n)}$ such that
    $$
    \Pr_{w \sim \{0,1\}^{n^k}}[L(w) = B(w)] \geq 1/2 + 2^{-O(n)},
    $$
    where $L$ is the language decided by $M$.
\end{lemma}

Lemma \ref{lem:main} and its proof have appeared in \citep{DBLP:journals/jml/Krajicek11, DBLP:journals/apal/Pich15}. We provide next a detailed exposition of this technique.

\subsubsection{A simpler case: $\ell = 1$ in the KPT student-teacher protocol}

Note that we can apply \Cref{thm:KPT} to sentence $\LBw$ and theory $\TPV$, since $\TPV$ is a universal theory and it is not difficult to see that $\LBw$ can be written in the form required by Theorem \ref{thm:KPT}. Since $\LPV$-terms correspond to polynomial-time computable functions, the corresponding student computes in time polynomial in the length of its input. Using that $n \in \mathsf{LogLog}$ in sentence $\LBw$, we obtain uniform algorithms $f_1, \ldots, f_\ell$ that compute in time $2^{O(n)}$ and satisfy the conclusion of Theorem \ref{thm:KPT}. Note that we cannot control the constant $\ell$. In this section, we discuss the simpler case when we get $\ell = 1$ in the application of Theorem \ref{thm:KPT} to $\TPV \vdash \LBw(M,s,n_0)$.

Omitting auxiliary variables in the input to $f_1$ and highlighting the relevant parameters,\footnote{The condition $n \in \mathsf{LogLog}$ means that $n$ is the length of $N$, while $N$ is the length of some universally quantified variable $v$. For this reason, formally, $v$ is an input to $f_1$, and not $n$. However, over $\mathbb{N}$ we will run $f_1$ on a fixed input $v$, such as $1^N$, where $N = 2^n$. For this reason, $n$ is the parameter that  controls the length of the remaining inputs.} $f_1(n,D)$ receives $n$ and an \emph{arbitrary} co-nondeterministic circuit $D$ of size $\leq s(n) = 2^{n^{\delta}}$, and outputs a triple $(x,y,z)$ such that $\mathsf{Error}(x,y,z)$ holds. Assuming that $n \geq n_0$ and $\ell = 1$ in the KPT Witnessing, it follows that this triple witnesses that $D(x) \neq M(x)$ over the standard model $\mathbb{N}$. In other words, $x \in \{0,1\}^n$, $y \in \{0,1\}^{t(n)}$ for $t(n) = 2^{n^{o(1)}}$, $z \in \{0,1\}^{s(n)}$, and the following holds:
$$
\Big [ M(x,y)=1 \,\wedge\, D(x,z) = 0  \Big ] \;\vee\; \Big [\forall y'\;  M(x,y') = 0 \,\wedge\, \forall z' \; D(x,z') = 1 \Big ].
$$

To prove Lemma \ref{lem:main} when $\ell = 1$, let $k \geq 1$, and assume that $n$ is sufficiently large. We will use $f_1$ to construct a deterministic circuit $B$ defined over $n^k$ input variables and of size $2^{O(n)}$ such that
    $$
    \Pr_{w \sim \{0,1\}^{n^k}}[L(w) = B(w)] \geq 1/2 + 2^{-O(n)},
    $$
where $L$ is the language computed by the nondeterministic machine $M$. As a key point, note that we can invoke $f_1(n,D)$ on any co-nondeterministic circuit $D(x, \cdot)$ over $n$ input variables and of size $\leq 2^{n^{\delta}}$. In order to construct the deterministic circuit $B$, we will also use that $f_1(n,D) = (x,y,z)$ computes in time $2^{O(n)}$ and therefore can be simulated by a circuit of size $2^{O(n)}$. For the $\ell = 1$ case, we will not need the output strings $y$ and $z$ during the construction of $A$. From now on, we simplify notation and write $f_1(D) = x$ to denote the relevant input and output of $f_1$ for this case.

Let $C_{n^k}(w,z)$ be a Boolean circuit that computes as $M$ on inputs of length $n^k$, where $z$ corresponds to the nondeterministic input. Since $M$ runs in time at most $2^{n^{o(1)}}$, $C_{n^k}$ has size at most $2^{n^{\delta}}$ when $n$ is sufficiently large. We partition its first input as $w=x\|w'$, where $|x| = n$ and $|w'| = n^k - n$. Now for a fixed string $w' \in \{0,1\}^{n^k - n}$, consider the circuit $D_{w'}(x,z) \triangleq \neg C_{n^k}(x\|w',z)$. Viewing $D_{w'}$ as a co-nondeterministic circuit, we get that
\begin{equation}\label{eq:def_Dw'}
D_{w'}(x) = 1 \;\Longleftrightarrow\; \forall z\, D_{w'}(x,z) = 1 \;\Longleftrightarrow\; \forall z\, C_{n^k}(x\|w',z) = 0 \;\Longleftrightarrow\; x\|w' \notin L(M). 
\end{equation}

Intuitively, if we ``learn'' the output bit of $D_{w'}(x)$ for some pair $(w',x)$, we also ``learn'' if the input string $x\|w'$ is in $L(M)$. As a consequence, the collection $\{D_{w'}(x)\}_{w' }$ of co-nondeterministic circuits $D_{w'}$ (defined over $n$ input bits) captures the computation of $L$ on inputs of length $n^k$. 

As each $D_{w'}$ has size at most $2^{n^{\delta}}$, we can invoke $f_1$ on them. Since $f_1(D_{w'})$ finds mistakes with respect to the nondeterministic computation of $M$, we know that $D_{w'}(x^*) \neq M(x^*)$ for $x^*\triangleq f_1(D_{w'})$. Since there are only $2^n$ possibilities for the output of $f_1$, the following holds.

\begin{fact}\label{fact:xstar1}
There is a string $x^* \in \{0,1\}^n$ such that
$$
\Pr_{w' \in \{0,1\}^{n^k - n}}[f_1(D_{w'}) = x^*] \geq 2^{-n}.
$$
\end{fact}
Following our informal discussion from above, from the knowledge that $f_1(D_{w'}) = x^*$ (and of a bit encoding if $x^* \in L(M)$) we ``learn'' how to compute $L(M)$ on the input $w = x^*\|w'$. Since this will happen with non-trivial probability over a random choice of $w'$ and $x^*$, this can be used to non-trivially approximate $L(M)$ over input length $n^k$. 

Formally, let $b^* \in \{0,1\}$ be $1$ if and only if $M(x^*) = 1$, i.e., if $x^* \in L(M)$. The string $x^*$ and the bit $b^*$ will be stored as non-uniform advice in the deterministic circuit $B$ that we show to be correlated with $L(M)$ on input length $n^k$. First, consider the following randomized circuit $B'$:

\begin{algorithm}[htb]
    \SetKwInOut{Input}{Input}
    \SetKwInOut{Advice}{Advice}
    \caption{Randomized Circuit $B'$ for $L(M)$}\label{algo:B-for-l=1}
    \Input{The input $w \in \{0,1\}^{n^k}$ and a random $r \in \{0,1\}$}
    \Advice{$x^*\in\{0,1\}^n$ and $b^*\in\{0,1\}$}
    Let $w = x\|w'$ and compute $f_1(D_{w'})$\;
    \tcp{Note that we can construct a description of $D_{w'}$ from $w$ and $C_{n^k}$.}
    If $f_1(D_{w'})\ne x^*$, \Return{$r$}\;
    If $x\ne x^*$, \Return{$r$}\;
    Otherwise, \Return{$b^*$}\;
\end{algorithm}

Note that $B'$ can be computed with $2^{O(n)}$ gates, since $f_1$ runs in time $2^{O(n)}$. Next, we show that the randomized circuit $B'$ non-trivially correlates with $L(M)$ on inputs of length $n^k$. After that, fixing the random bit $b$ in $B'$ yields the desired deterministic circuit $A$.

\begin{fact}\label{fact:correct1}
If $B'$ reaches Line 4 on an input string $w$, then $B'(w,b) = M(w)$, i.e., $B'$ correctly decides $L(M)$ on input $w$.
\end{fact}

\begin{proof}
Under the assumption that $B'$ reaches Line $4$ on an input string $w$, it follows that $w = x^*\|w'$ and $f_1(D_{w'}) = x^*$. Moreover, observe that the random bit $b$ does not affect the output of $B'$ in this case. We have
\begin{align*}
    B'(w,r) = 1 & \Longleftrightarrow b^* = 1  \\
     & \Longleftrightarrow M(x^*) = 1 
     \tag{by the definition of~$b^*$} \\
    & \Longleftrightarrow D_{w'}(x^*) = 0 \tag{using that $f_1$ finds a mistake} \\
    & \Longleftrightarrow w = x^*\|w' \in L(M) \tag{by Equation \ref{eq:def_Dw'}} \\ 
    & \Longleftrightarrow M(w)=1. \tag{since $M$ computes $L(M)$}
\end{align*}
In other words, $B'(w,r) = M(w)$.
\end{proof}

In addition, by Fact \ref{fact:xstar1},
$$
p \eqdef \Pr_{r, w}[B'~\text{reaches Line}~4] = \Pr_{x,w'}[x = x^* \wedge f_1(D_{w'}) = x^*] \geq 2^{-n} \cdot 2^{-n}.
$$
On the other hand, when $B'$ does not reach Line 4 it outputs a random bit that is independent of the input string $w$. Therefore, using Fact \ref{fact:correct1} and the lower bound on $p$,
$$
\Pr_{r,w}[B'(w,r) = M(w)] \geq p \cdot 1 + (1-p) \cdot 1/2 = 1/2 + p/2  \geq 1/2 + 2^{-2n + 1} = 1/2 + 2^{-O(n)}.
$$
Fixing the random bit $r$ in the best way maintains this advantage and completes the proof of Lemma \ref{lem:main} when $\ell = 1$.

\begin{remark}\label{remark:general_approx}
The same argument can be used to approximate \emph{any} nondeterministic circuit of size $2^{n^{k\delta}}$ defined over $n^k$ bits by a deterministic circuit of size $2^{O(n)}$, instead of just for $L(M) \cap \{0,1\}^{n^k}$. In other words, by connecting $D_{w'}$ to the computation of the appropriate co-nondeterminisetic circuit, ``learning'' output bits of $D_{w'}$ via $f_1$ translates into a non-trivial approximation \emph{(}using exactly the same strategy\emph{)}. This will also hold when analysing the case $\ell > 1$. In particular, from $\TPV \vdash \LBw(M,s,n_0)$ we are able to non-trivially approximate any language in $\mathsf{NSIZE}[2^{n^{o(1)}}]$ and not just $L(M)$.
\end{remark} 

\subsubsection{The case $\ell = 2$ via the Nisan-Wigderson generator}\label{sec:ell2case}

In this section, we consider the case where the disjunction obtained from KPT Witnessing (\Cref{thm:KPT}) has size $\ell = 2$. This essentially covers all difficulties in the general case. Before handling $\ell = 2$, it is instructive to highlight some key points of the proof when $\ell = 1$:

\begin{enumerate}
    \item\label{enum: step 1 ell=2} We implicitly relied on the ability of \emph{certifying} when an input $x$ is a mistake. More precisely, when $\ell = 1$, if $f_1(D_{w'}) = (x,y,z)$, we have the \emph{guarantee} that $D_{w'}(x) \neq M(x)$. This is because there is a \emph{single} round in the corresponding Student-Teacher protocol. 
    \item By an averaging argument, we fixed a good string $x^* \in \{0,1\}^n$ (\Cref{fact:xstar1}), which eventually allowed us to compute $M(x^*w')$ on a non-trivial fraction of $w'$, by storing $x^*$ and the corresponding bit $b^* = M(x^*)$.
    \item This was accomplished by considering a family $\{D_{w'}(x)\}_{w'}$ of co-nondeterministic circuits over $n$-bit inputs that compute according to a  circuit defined over input length $n^k$ that is related to the language we would like to approximate. 
    \item On an input $w \in \{0,1\}^{n^k}$ with $w = x\|w'$ for which the witnessing provided by $f_1(D_{w'})$ was inconsistent with the actual input part $x$ (we can easily detect this), we output a random bit.
\end{enumerate}

Note that this approach no longer works when $\ell > 1$: the first term obtained from KPT Witnessing might not succeed in finding a mistake. For this reason, we cannot assume in \Cref{enum: step 1 ell=2} that if $f_1(D_{w'}) = (x,y,z)$ then $D_{w'}(x) \neq M(x)$.

Let $f_1$ be the first term in the KPT disjunction when $\ell > 1$. Note that we can still fix a popular \emph{candidate} mistake $x^* \in \{0,1\}^n$, as in \Cref{fact:xstar1}. Recall that $f_1(D_{w'}) = (x_{w'}, y_{w'}, z_{w'})$ (we did not have to  use $y_{w'}$ and $z_{w'}$ in the argument for $\ell = 1$). We can check whether $x_{w'} = x^*$, as before, and we would like to use $y_{w'}$ and $z_{w'}$ together with some hard-coded information to decide if $x^* = x_{w'}$ is indeed a mistake for $D_{w'}$. While both $y_{w'}, z_{w'} \in \{0,1\}^{\leq 2^n}$, there are $2^{\Omega(n^k)}$ possible strings $w'$. Unfortunately, it is unclear how to store enough information in the non-uniform circuit $B'$ to certify that a mistake has been found by $f_1$ while maintaining a circuit size bound of $2^{O(n)}$. 

To reduce the amount of advice needed in $B'$ and address this difficulty, the solution \citep{DBLP:journals/jml/Krajicek11, DBLP:journals/apal/Pich15} is to employ a more sophisticated family $\{D_{w'}\}_{w'}$ of circuits constructed via the Nisan-Wigderson generator \citep{NW}.

For a nondeterministic machine $M$ that decides a language $L(M)$, we use the notation $\{\NW_{\overline{L(M)}}(w)\}_w$ to denote the collection of functions obtained from the Nisan-Wigderson generator when instantiated with the Boolean function $h$ that corresponds to the negation of $L(M)$ over inputs of length $n^{c/2}$.

\begin{fact} Let $M$ be a nondeterministic machine that runs in time $2^{m^{o(1)}}$ on inputs of length $m$.
For any constant $c \geq 1$ and every large enough $n$, each function in $\{\NW_{\overline{L(M)}}(w)\}_w$ can be computed by a co-nondeterministic circuit $D_w(x)$ of size at most $2^{n^{\delta}}$. 
\end{fact}

\paragraph{The case $\ell = 1$ via the NW generator.} Before handling the case $\ell = 2$, we sketch the proof of the case $\ell = 1$ using the collection $\{D_w\}_{w \in \{0,1\}^{n^c}}$ obtained from the nondeterministic machine $M$ and the NW generator, with parameters as above.

Consider the function $f_1(D_w) = (x,y,z)$ obtained by applying \Cref{thm:KPT}, and assume that $\ell = 1$. Again, we will not inspect $y$ and $z$ when $\ell = 1$.  Recall that $f_1(D_w)$ computes in time $2^{O(n)}$. We show how to  decide $L(M)$ on inputs of length $n^{c/2}$ by a deterministic circuit of size $2^{O(n)}$ that agrees with $L(M)$ with probability $\geq 1/2 + 2^{-O(n)}$ over a uniformly random input string.

Similarly to \Cref{fact:xstar1}, by a standard averaging argument we can establish the following fact.

\begin{fact}\label{fact:NWxstar1}
There is a string $x^* \in \{0,1\}^n$ such that
$$
\Pr_{w \in \{0,1\}^{n^c}}[f_1(D_w) = x^*] \geq 2^{-n}.
$$
\end{fact}

Recall that $J_{x^*}$ denotes the subset of $[n^c]$ of size $n^{c/2}$ corresponding to the $x^*$-row of the design in our NW generator; for $a\in\{0,1\}^{n^c-n^{c/2}}$ and $u\in\{0,1\}^{n^c}$, $r_{x}(a,u)$ denotes the ``concatenated'' string $a\cup u$ obtained by viewing $a\in\{0,1\}^{[n^c]\setminus J_{x}}$ and $u\in\{0,1\}^{J_x}$. By another averaging argument, we get the following consequence.

\begin{fact}\label{fact:fixinga}
There is a string $a \in \{0,1\}^{[n^c] \setminus J_{x^*}}$ of length $n^{c} - n^{c/2}$ such that
$$
\Pr_{\stackrel{u \sim \{0,1\}^{J_{x^*}}}{w \triangleq u \cup a}}[f_1(D_w) = x^*] \geq 2^{-n}.
$$
\end{fact}

We can view $D_w = \NW_{\overline{L(M)}}(w)$ as a co-nondeterministic circuit for computing $\overline{L(M)}$ over inputs of length $n^{c/2}$ derived from the seed $w$: 
$$
D_w(x) = 1 \quad \Longleftrightarrow \quad w|_{J_x} \in \overline{L(M)}.
$$
Given the previous discussion, we are interested in seeds $w \in \{0,1\}^{n^{2c}}$ of the form $w = a \cup u$, where $a\in\{0,1\}^{[n^c]\setminus J_{x^*}}$ is fixed, $u \in \{0,1\}^{J_{x^*}}$, and $f_1(D_w) = x^*$. We know that a non-trivial fraction of strings $u$ will satisfy this condition. Since $f_1$ witnesses mistakes with respect to $L(M)$ over inputs of length $n$ (note that $D_w$ is a conondeterministic circuit over $n$-bit inputs), whenever $f_1(D_w) = x^*$ we are guaranteed that
$$
D_w(x^*) = 1 \quad \Longleftrightarrow \quad M(x^*) = 0,
$$
which implies that $M(x^*) = 0$ if and only if $w|_{J_{x^*}} \notin L(M)$. Now $x^*$ is fixed, so the equality $M(x^*) = 0$ does not depend on other conditions. For instance, if $M(x^*) = 0$, we can conclude that on any input string $u \in \{0,1\}^{n^{c/2}}$, if for $w = a \cup u$ we have $f_1(D_w) = x^*$, then $u = w|_{J_{x^*}}$ is not in $L(M)$. Consequently, this allows us to correctly compute $L(M)$ on any such input $u \sim \{0,1\}^{n^{c/2}}$, which constitute a non-trivial fraction of inputs. Moreover, we can check whether an input $u$ satisfies $f_1(D_w) = x^*$ using a deterministic circuit of size $2^{O(n)}$.

Formally, consider the fixed strings $x^* \in \{0,1\}^n$ and  $a \in \{0,1\}^{[n^c] \setminus J_{x^*}}$ from above, and let $b^* \triangleq M(x^*) \in \{0,1\}$. We hardcode $x^*$, $a$, and $b^*$ in the randomised circuit $B(u)$ described below:

\begin{algorithm}[htb]
    \SetKwInOut{Input}{Input}
    \SetKwInOut{Advice}{Advice}
    \caption{Randomised Circuit $B$ for $L(M)$}\label{algo:B-for-l=2}
    \Input{The input $u \in \{0,1\}^{n^{c/2}}$ and a random $r \in \{0,1\}$}
    \Advice{$x^*\in\{0,1\}^n$ and $b^*\in\{0,1\}$}
    Let $w = r_{x^*}(a,u)$\;
    Let $x = f_1(D_w)$\;
    If $x \neq x^*$, output the random bit $r$\;
    Otherwise, \Return{$b^*$}\;
\end{algorithm}

\noindent Given the aforementioned discussion, it is easy to see that
$$
\Pr_{u, r}[B(u,r) = M(u)] \geq p \cdot 1 + (1-p) \cdot 1/2 = 1/2 + p/2 \geq 1/2 + 2^{-n + 1},
$$
where $p$ is the probability in the LHS of \Cref{fact:fixinga}. Consequently, by an averaging argument over the random bit $r$, there is a deterministic circuit of size $2^{O(n)}$ that computes $L(M)$ on inputs of length $n^{c/2}$ with the same advantage.

\paragraph{The case $\ell = 2$ via the NW generator.} Recall that 
$$
 \mathsf{Error}(x,y,z) \;\equiv\; \Big [ M(x,y)=1 \,\wedge\, D(x,z) = 0  \Big ] \;\vee\; \Big [\forall y'\;  M(x,y') = 0 \,\wedge\, \forall z' \; D(x,z') = 1 \Big ].
 $$
 We now have a function $f_1(D) = (x, y, z)$ that attempts to produce a triple $(x,y,z)$ satisfying $\mathsf{Error}(x,y,z)$, and a function $f_2(D, y',z')$ which given a pair $y', z'$ for which 
 \begin{equation}\label{eq:witness}
     \Big [ M(x,y)=0 \,\vee\, D(x,z) = 1  \Big ] \;\wedge\; \Big [ M(x,y') = 1 \,\vee\,  D(x,z') = 0 \Big ]
 \end{equation}
 is  able to produce an input $x'$ such that $D(x') \neq M(x')$. 
 
Again, we consider the family $\{D_w\}_{w \in \{0,1\}^{n^{c}}}$ of conondeterministic circuits $D_w$ of size $\leq 2^{n^\delta}$ that compute $\NW_{\overline{L(M)}}(w) \colon \{0,1\}^n \to \{0,1\}$ for a fixed seed $w$, with parameters as described above. In particular, this generator is instantiated with respect to the Boolean function $h$ corresponding to $\overline{L(M)}$ over inputs of length $n^{c/2}$, for a fixed but arbitrarily large constant $c \geq 1$. 

By an averaging argument, the following claim holds.

\begin{fact}\label{fact:NWx1}
There is a string $x_1 \in \{0,1\}^n$ such that
$$
\Pr_{w \in \{0,1\}^{n^c}}[f_1(D_w) = x_1] \geq 2^{-n}.
$$
\end{fact}

Fix this $x_1$. We define the sets $S_{x_1}^\mathsf{mist}\subseteq S_{x_1}\subseteq \{0,1\}^{n^c}$ as follows:
\begin{eqnarray}
    S_{x_1} & \triangleq & \Big \{w \in \{0,1\}^{n^c}  \mid f_1(D_w) = x_1 \Big \}, \nonumber \\
    S_{x_1}^\mathsf{mist} & \triangleq & \Big \{ w \in S_{x_1}  \mid D_w(x_1) \neq M(x_1) \Big \}, \nonumber 
\end{eqnarray}
and consider the density of $S_{x_1}^\mathsf{mist}$ with respect to its superset $S_{x_1}$.\\

\noindent \textbf{Case 1.} $|S_{x_1}^\mathsf{mist}| > (2/3) \cdot |S_{x_1}|$. We can essentially proceed as in the case of $\ell= 1$, with the exception that one needs to be careful when invoking an analogue of \Cref{fact:fixinga}. This is because fixing a  string $a \in \{0,1\}^{[n^c] \setminus J_{x_1}}$ might keep the density of $S_{x_1}$ at least $2^{-n}$ but could significantly decrease the relative density of the set $S_{x_1}^\mathsf{mist}$ after the restriction.

To handle this, we introduce the following notation. For $m \geq 1$, a set $S \subseteq \{0,1\}^{[m]}$, and a string $a \in \{0,1\}^I$, where $I \subseteq [m]$, we define the \emph{restriction of} $S$ \emph{with respect to} $a$ as the set 
$$
S\uhr_a \eqdef \{w \in S \mid w|_I = a\}.
$$
Under the assumption that $|S_{x_1}^\mathsf{mist}| > (2/3) \cdot |S_{x_1}|$, it is possible to show by a counting argument (see, e.g.,~\Cref{lmm:counting}) that there exists a string $a \in \{0,1\}^{[n^c] \setminus J_{x_1}}$ such that
\begin{equation}\label{eq:goodstringa}
p \eqdef \frac{|S_{x_1}\uhr_a\!|}{2^{n^{c/2}}} \geq \frac{1}{n} \cdot 2^{-n} \quad \text{and} \quad \frac{|S_{x_1}^{\mathsf{mist}}\uhr_a\!|}{|S_{x_1}\uhr_a\!|} \geq \frac{2}{3} - \frac{1}{n}.
\end{equation}
While it is not clear how to decide in size $2^{O(n)}$ if a string $w \in S_{x_1}^{\mathsf{mist}}\uhr_a$, we can check whether $w \in S_{x_1}\uhr_a$. Since $S_{x_1}^{\mathsf{mist}}$ is dense in $S_{x_1}\uhr_a$, this is enough to adapt the original strategy used for $\ell = 1$. 

Formally, fix strings $x_1 \in \{0,1\}^n$ and $a \in \{0,1\}^{[n^c] \setminus J_{x_1}}$ as above, and let $b_1 \eqdef M(x_1) \in \{0,1\}$. We hardcode $x_1$, $a$, and $b_1$ in the randomised circuit $B_1(u,r)$ described below.

\begin{algorithm}[htb]
    \SetKwInOut{Input}{Input}
    \SetKwInOut{Advice}{Advice}
    \caption{Randomized Circuit $B_1$ for $L(M)$ when $|S^\mathsf{mist}_{x_1}|>(2/3)\cdot|S_{x_1}|$.}\label{algo:B-for-l=1-NL}
    \Input{The input $u \in \{0,1\}^{n^{c/2}}$ and a random $r \in \{0,1\}$}
    \Advice{$x_1\in\{0,1\}^n$, $a\in\{0,1\}^{n^c-n^{c/2}}$, and $b_1\in\{0,1\}$}
    Let $w = r_{x_1}(a,u)$\;
    Let $x = f_1(D_w)$\;
    If $x \neq x_1$, output the random bit $r$\;
    Otherwise, output the fixed bit $b_1 = M(x_1)$\;
\end{algorithm}

Clearly, $B_1$ is computed by a randomised circuit of size $2^{O(n)}$. To analyse its success probability, first note that if $u$ is such that $w = r_{x_1}(a,u) \notin S_{x_1}\uhr_a$, then $B_1(u) = M(u)$ with probability $1/2$. On the other hand, for those $u$ such that $w = r_{x_1}(a,u) \in S_{x_1}\uhr_a$, at least a $2/3 - 1/n$ fraction of them are in $S_{x_1}^{\mathsf{mist}}\uhr_a$, in which case $B_1(u)$ is correct. Since $S_{x_1}\uhr_a$ has density at least $1/n \cdot 2^{-n}$, it follows that 
$$
\Pr_{u, r}[B_1(u,r) = M(u)] = (1-p) \cdot \frac{1}{2} + p \cdot \! \left ( \frac{2}{3} - \frac{1}{n}\right ) = \frac{1}{2} + p \cdot \! \left  ( \frac{1}{6} - \frac{1}{n} \right ) = \frac{1}{2} + \Omega \! \left ( \frac{2^{-n}}{n} \right ),
$$
which is $1/2 + 2^{-O(n)}$. Fixing the random bit $r$ in the best way yields the desired deterministic circuit.\\

\noindent \textbf{Case 2.} $|S_{x_1}^\mathsf{mist}| < (2/3) \cdot |S_{x_1}|$. In this case, the mistakes of at least a $1/3$ fraction of the circuits $D_w$ for $w \in S_{x_1}$ must be witnessed by $f_2$. To make sure the output of $f_2(D_w, y',z')$ is indeed a string $x_2$ for which $M(x_2) \neq D_w(x_2)$, we must provide a pair $y', z'$ such that 
 \begin{equation}\label{eq:witness2}
     \Big [ M(x_1,y_1)=0 \,\vee\, D_w(x_1,z_1) = 1  \Big ] \;\wedge\; \Big [ M(x_1,y') = 1 \,\vee\,  D_w(x_1,z') = 0 \Big ],
 \end{equation}
 where $f_1(D_w) = (x_1, y_1, z_1)$. We consider the Teacher that to each $w \in S_{x_1} \setminus S_{x_1}^{\mathsf{mist}}$ and corresponding $(y_1,z_1)$ assign the lexicographic first pair $(y'_w,z'_w)$ for which \Cref{eq:witness2} holds. Note that such a pair always exists, since in this case for $x_1 = f_1(D_w)$ we have $M(x_1) = D_w(x_1)$. 
 
 By an averaging argument, the following claim holds.
 
 \begin{fact}\label{fact:densitySx1x2}
 Under this fixed Teacher, there is a string $x_2 \in \{0,1\}^n$ such that the set
 $$
 S_{x_1, x_2} \eqdef \{ w \in S_{x_1} \mid D_w(x_1) = M(x_1) \wedge f_2(D_w, y'_w, z'_w) = x_2\}
 $$
 has density at least $(1/3) \cdot 2^{-2n}$ in $\{0,1\}^{n^c}$.
 \end{fact}
 
 Note that, by construction, if $w \in S_{x_1, x_2}$ then for $x_2 = f_2(D_w, y'_w, z'_w)$ we have $D_w(x_2) \neq M(x_2)$. Note that $x_1\ne x_2$ because otherwise we have $S_{x_1,x_2}=\varnothing$.\footnote{Assume it is not the case, there is a $w\in S_{x_1,x_2}$ such that $D_w(x_2)\ne M(x_2)$. However, we know that $D_w(x_1)=M(x_1)$ by the definition of $S_{x_1,x_2}$, which is impossible when $x_1=x_2$.} Moreover, the set $S_{x_1, x_2}$ has enough density for our purposes. However, for this to be useful we must verify that a given circuit $D_w$ satisfies $w \in S_{x_1, x_2}$ using a deterministic circuit of size $2^{O(n)}$.

 By another averaging argument, we have the following result.
 
 \begin{fact}\label{fact:fixa-case2-ell2}
 There is a string $a \in \{0,1\}^{[n^c] \setminus J_{x_2}}$ such that
 $$
 \frac{|S_{x_1, x_2} \uhr_a \!|}{2^{n^{c/2}}} \;\geq\; \frac{1}{3} \cdot 2^{-2n}.
 $$
 \end{fact}
 
Fix this string $a\in\{0,1\}^{[n^c]\setminus J_{x_2}}$ together with the strings $x_1$ and $x_2$. We will assume that the following computation is possible in order to complete the proof, returning to it later on:
 
\begin{itemize}
    \item[($\nabla$)] There is a deterministic circuit $E(w)$ of size $2^{O(n)}$ as follows: Given a $w \in S_{x_1}$ of the form $a \cup u$ such that $D_w(x_1) = M(x_1)$, it outputs the lexicographic first pair $(y'_w, z'_w)$ for which \Cref{eq:witness2} holds, where $(x_1, y_1, z_1) = f_1(D_w)$.\footnote{Note that in this case $f_2(D_w, y'_w, z'_w)$ outputs a mistake of $D_w$, since $\mathsf{Error}(x_1, y_1, z_1)$ does not hold and correct witnesses for this are provided.}
\end{itemize}

Consider strings $x_1, x_2 \in \{0,1\}^n$ and $a \in \{0,1\}^{[n^c] \setminus J_{x_2}}$ as above, and let $b_2 \eqdef M(x_2) \in \{0,1\}$. We hardcode this information in the randomised circuit $B_2(u)$ described below, which includes the circuit $E(w)$ from ($\nabla$) as a subroutine:\\

\begin{algorithm}[htb]
    \SetKwInOut{Input}{Input}
    \SetKwInOut{Advice}{Advice}
    \caption{Randomized Circuit $B_2$ for $L(M)$ when $|S_{x_1}^\mathsf{mist}|\le (2/3)\cdot|S_{x_1}|$.}\label{algo:B-for-l=2-NL}
    \Input{The input $u \in \{0,1\}^{n^{c/2}}$ and a random $r \in \{0,1\}$}
    \Advice{$x_1,x_2\in\{0,1\}^n$, $a\in\{0,1\}^{n^c-n^{c/2}}$, and $b_2\in\{0,1\}$}
    Let $w = r_{x_2}(a,u)$\;
    Let $(x, y_1, z_1) = f_1(D_w)$\;
    If $x \neq x_1$, output the random bit $r$\;
    Let $(y'_w, z'_w) =  E(w)$\;
    If the tuple $(x_1, y_1, z_1, y'_w, z'_w)$ satisfies \Cref{eq:witness2} and $f_2(D_w, y'_w, z'_w)= x_2$, output $b_2$\;
    Otherwise output the random bit $r$.
\end{algorithm}

Note that, under assumption ($\nabla$), $B_2$ can be computed by a randomised circuit of size $2^{O(n)}$. Moreover, it follows from our discussion and from the density of $S_{x_1, x_2} \uhr_a$ that 
$$
\Pr_{u, r}[B_2(u,r) = M(u)] \;\geq\; \frac{1}{2} + \Omega\!\left ( 2^{-2n} \right ). 
$$
\noindent This yields a deterministic circuit with the same advantage.\footnote{Note that we cannot really guarantee that $D_w(x_1) = M(x_1)$ when invoking ($\nabla$), since this cannot be easily decided in deterministic size $2^{O(n)}$. This means that more inputs $u$ than those leading to strings $w \in S_{x_1, x_2} \uhr_a $ might reach Line 5 and be assigned output value $b_2$. Nevertheless, $B_2$ will be correct on any such input $u$, by virtue of the two checks performed in Line 5. Put another way, the argument ``covers'' the inputs $u$ leading to strings $w \in S_{x_1, x_2} \uhr_a$.}

\begin{proof}[Proof of $(\nabla)$] We will now use the main property of the combinatorial design behind the Nisan-Wigderson generator: the sets $J_{x_1}$ and $J_{x_2}$ overlap in at most $n$ coordinates. This will allow us to hardcode all relevant pairs $(y'_w, z'_w)$ using circuit size $2^{O(n)}$.

To implement ($\nabla$), we are given a string $w = a \cup u$, where $a \in \{0,1\}^{[n^c] \setminus J_{x_2}}$ is fixed and $u \in \{0,1\}^{J_{x_2}}$, such that the following conditions hold:
\begin{itemize}
    \item Let $(x, y_1, z_1) = f_1(D_w)$, then $x = x_1$.
    \item $D_w(x_1) = M(x_1)$.
\end{itemize}
Our goal is to output the lexicographic first pair $(y'_w, z'_w)$ such that:
\begin{equation*}
    \Big [ M(x_1,y_1)=0 \,\vee\, D_w(x_1,z_1) = 1  \Big ] \;\wedge\; \Big [ M(x_1,y') = 1 \,\vee\,  D_w(x_1,z') = 0 \Big ]. 
\end{equation*}
Note that such pair must exist since we assume that $D_w(x_1)=M(x_1)$. 

Recall that $D_w(x_1)=\NW_{\overline{L(M)}}(w,x_1)$. The crucial observation that leads to the use of NW generator is that the desired pair $(y'_w, z'_w)$ only depends on $w|_{J_{x_1}}$, which contains at most $n$ bits of the input $u\in\{0,1\}^{n^{c/2}}$. This is because $w=a\cup u$ is a concatenation of a fixed $a\in\{0,1\}^{[n^c]\setminus J_{x_2}}$ and $u$ viewed as $u\in\{0,1\}^{J_{x_2}}$, which means that  
\[
w|_{J_{x_1}} = (a\cup u)|_{J_{x_1}} = a|_{J_{x_1}}\cup u|_{J_{x_1}},    
\]
where $a|_{J_{x_1}}$ is fixed and $u|_{J_{x_1}}$ only consists of the indices within $J_{x_1}\cap J_{x_2}$ of size at most $n$. 

As $E(w)$ depends on at most $n$ bits of the input $u\in\{0,1\}^{n^c}$, we can implement it as a circuit that store all the answers for all $2^n$ possibilities, which requires at most $\poly(2^n)=2^{O(n)}$ gates. Concretely, the circuit works as follows: Given $w\in\{0,1\}^{n^c}$, we firstly obtain $u\in\{0,1\}^{J_{x_2}}$ such that $w=a\cup u$; let $u'=u|_{J_{x_1}}$ be of length at most $n$, we look up the table to find the answer corresponding to $u'$. 
\end{proof}

\ignore{First, arguing by contradiction, it easily follows from the second bullet that the LHS of \Cref{eq:witness2} holds, since we define the value $M(x_1)$ according to a nondeterministic computation and the value $D_w(x_1)$ according to a co-nondeterministic computation. On the other hand, in the RHS of \Cref{eq:witness2}, if $M(x_1) = 1$, we can fix the lexicographic first $y'$ such that $M(x_1, y') = 1$ and set $z'$ to be a sequence of zeroes. This can be hardcoded in the circuit $E(w)$ and works for all relevant inputs $w = a \cup u$ in this case. Using $D_w(x_1) = M(x_1) = 1$, it is easy to see that this is the lexicographic first pair that satisfies \Cref{eq:witness2}. In the case where $M(x_1) = 0$, we set $y'$ to be a sequence of zeroes (the left term in the RHS cannot be satisfied), and let $z'_w$ be the lexicographic first string such that $D_w(x_1, z'_w) = 0$. This string exists because $D_w(x_1) = M(x_1) = 0$. The crux of the argument is that in the computation of $D_w(x_1, z'_w)$, the first step is to disregard $z'_w$ and simply compute $w|_{J_{x_1}}$. Since $w = a \cup u$, $a \in \{0,1\}^{[n^c] \setminus J_{x_1}}$ is fixed, $u \in \{0,1\}^{J_{x_2}}$, and $|J_{x_1} \cap J_{x_2}| \leq n$, as $u$ varies there are at most $2^n$ possible outcomes for $w|_{J_{x_1}}$. For this reason, there is no need to store more than $2^n$ different strings $z'_w$ in $E(w)$. Overall, we get that the non-uniform circuit $E(w)$ from ($\nabla$) can be implemented with $2^{O(n)}$ many gates.}

\begin{remark}\label{remark:general_approx2} As in \Cref{remark:general_approx}, we note that the argument can be easily adapted to approximate any Boolean function $g$ defined over $n^k$ bits computable by a nondeterministic circuit of size $2^{n^{k\delta}}$ using a deterministic circuit of size $2^{O(n)}$, instead of for just $L(M) \cap \{0,1\}^{n^k}$. The provability of a circuit lower bound for a single language $L(M)$ provides non-trivial circuits for any such $g$.

Based on this, we can also prove that under the same assumption \emph{(}i.e., the provability of worst-case circuit lower bound in $\Tpv^i$\emph{)}, for every constant $\epsilon\in(0,1)$, $s=s(m)=2^{m^{o(1)}}$, and sufficiently large $m$, any Boolean function $g:\{0,1\}^m\to\{0,1\}$ that can be computable by a nondeterminisetic circuit of size $s$ can also be approximated by a co-nondeterministic circuit $D$ of size $2^{m^\epsilon}$, that is: 
\[
\Pr_{x\sim\{0,1\}^m}\Big[C(x)=D(x)\Big] \ge \frac{1}{2}+\frac{1}{2^{m^{\epsilon}}}.
\]
This can be done by setting $k=\lceil 20/\epsilon\rceil$, padding dammy inputs to $g:\{0,1\}^m\to\{0,1\}$ to obtain $g':\{0,1\}^{m'}\to\{0,1\}$, where $m'=\lceil m^{1/k}\rceil^k\le 2m$ for sufficiently large $m$, and applying the observation above to $g'$ with $n=\lceil m^{1/k}\rceil$. 
\end{remark} 

\subsubsection{Sketch of the general case}

We now sketch how the argument presented in \Cref{sec:ell2case} can be generalised to the case that the Student-Teacher protocol runs for $\ell\ge 3$ rounds. Recall that sets $S_{x_1},S_{x_1}^\mathsf{mist},S_{x_1,x_2}$ in \Cref{sec:ell2case} are defined as 
\begin{eqnarray}
    S_{x_1} & \eqdef & \Big \{w \in \{0,1\}^{n^c}  \mid f_1(D_w) = x_1 \Big \} \nonumber \\
    S_{x_1}^\mathsf{mist} & \eqdef & \Big \{ w \in S_{x_1}  \mid D_w(x_1) \neq M(x_1) \Big \} \nonumber  \\
 S_{x_1, x_2} & \eqdef & \{ w \in S_{x_1}\setminus S_{x_1}^\mathsf{mist} \mid f_2(D_w, y'_w, z'_w) = x_2\} \nonumber
\end{eqnarray}
In the general case, we will define a sequence of $x_1,x_2,\dots,x_{\ell}\in\{0,1\}^n$ as well as the sets  
\[ 
    S_1,S_1^\mathsf{mist}\subseteq S_1,S_2\subseteq S_1\setminus S_1^\mathsf{mist},S_2^\mathsf{mist}\subseteq S_2,\dots, S_\ell\subseteq S_{\ell-1}\setminus S_{\ell-1}^\mathsf{mist}, S_\ell^\mathsf{mist}\subseteq S_\ell. 
\]
For instance, if $\ell = 3$, we proceed as follows. %
\begin{enumerate}
\item We initially argue as in \Cref{sec:ell2case} with $\ell=2$. In Case 1 (i.e., $|S_{x_1}^{\mathsf{mist}}| > (2/3) \cdot |S_{x_1}|$), we can simply apply the aforementioned circuit $B_1$ to approximate $L(M)$. However, we can no longer conclude in its Case 2 (i.e.,~$|S_{x_1}^{\mathsf{mist}}| < (2/3) \cdot |S_{x_1}|$) that $x_2$ is a mistake of $D_w$ for every $w \in S_{x_1, x_2}$. To address this, we define the set
$$
S_{x_1, x_2}^{\mathsf{mist}} \eqdef \{ w \in S_{x_1, x_2} \mid  D_w(x_2) \neq M(x_2)\}\subseteq S_{x_1,x_2}
$$
and consider its density in $S_{x_1,x_2}$. 

\item If $|S_{x_1, x_2}^{\mathsf{mist}}|/|S_{x_1,x_2}| \geq 2/3$, we know that for at least a $2/3$ fraction of $w\in S_{x_1,x_2}$, $x_2$ is a mistake of $D_w$. As in Case 1 of \Cref{sec:ell2case}, we apply \Cref{lmm:counting} (instead of a direct counting argument in \Cref{fact:fixa-case2-ell2}) to find a ``good'' $a \in \{0,1\}^{[n^c] \setminus J_{x_2}}$ such that $S_{x_1,x_2}\uhr_a/2^{n^{c/2}}\ge \Omega(2^{-2n})$ and the density of $S_{x_1,x_2}^\mathsf{mist}\uhr_a$ in $S_{x_1,x_2}\uhr_a$ is at least $2/3-1/100$. By plugging in this $a$ into the circuit $B_2$, we will achieve agreement $\geq 1/2 + \Omega(2^{-2n})$ with $L(M)$.

\item Otherwise, we assume that $|S_{x_1, x_2}^{\mathsf{mist}}|/|S_{x_1,x_2}| \geq 2/3$. Let $(x_2,y_2,z_2) = f_2(D_w,y'_w,z'_w)$. Similar to $y'_w$ and $z'_w$, for every $w\in S_{x_1,x_2}\setminus S_{x_1,x_2}^\mathsf{mist}$, we define $(y''_w, z''_w)$ as the lexicographic first pair such that 
\begin{equation*}
    \Big [ M(x_2,y_2)=0 \,\vee\, D_w(x_2,z_2) = 1  \Big ] \;\wedge\; \Big [ M(x_2,y') = 1 \,\vee\,  D_w(x_1,z') = 0 \Big ], 
\end{equation*}
that is, $(y''_w,z''_w)$ is the output of the canonical Teacher in the second round of the Student-Teacher protocol. Since $S_{x_1,x_2}$ has density at least $\Omega(2^{-2n})$, we can find a string $x_3\in\{0,1\}^{n}$ such that the following set
$$
S_{x_1, x_2, x_3} \eqdef \{ w \in S_{x_1, x_2}\setminus S_{x_1,x_2}^\mathsf{mist} \mid f_3(D_w, y'_w, z'_w, y''_w, z''_w) = x_3 \},
$$
has density at least $\Omega(2^{-3n})$. Since $x_3$ must be a mistake of $D_w$ when $\ell = 3$ and $w \in S_{x_1,x_2,x_3}$, and this set is sufficiently dense, we can obtain a deterministic circuit of size $2^{O(n)}$ that achieves agreement $\geq 1/2 + \Omega(2^{-3n})$ with $L(M)$. 
\end{enumerate}

The argument can be generalised in the natural way, which allows us to obtain a circuit of size $2^{O(n)}$ that approximates $L(M)$ with advantage $\geq 1/2 + \Omega(2^{-\ell n})$ in the case of a disjunction of length $\ell$ in the application of the KPT Witnessing (see \Cref{thm:KPT}). This deterministic circuit computes $L(M)$ on inputs of length $n^{c/2}$, where $c$ is an arbitrary constant.

\begin{remark}
Note that the approach breaks down in theories where the number of rounds in the Student-Teacher game obtained from \Cref{thm:KPT} is polynomial in the relevant parameter, as in the case of Buss's theory $S^{1}_2$ (see, e.g.,~\citep{Krajicek92}). In the latter case, one can get up to $\ell = \mathsf{poly}(2^n)$ rounds in the corresponding witnessing theorem, and the advantage of the resulting deterministic circuit under a naive extension of the presented proof becomes trivial.
\end{remark}

\subsection{Extensions of the technique and unprovability of weaker lower bounds}\label{sec:hard_amplif}

As noted in \cite{PS21}, one can use hardness amplification  to weaken the average-case hardness in the unprovability result (\Cref{thm:unprov_lb}). By an adaptation of the proof of \Cref{thm:unprov_lb} via Remarks \ref{remark:general_approx} and \ref{remark:general_approx2} and an application of \Cref{thm:hardamp-ppoly}, we can obtain the following unprovability result.

\begin{theorem}\label{thm:unprov_lb2}
For every $n_0 \in \mathbb{N}$ and $\delta \in \mathbb{Q} \cap (0,1)$, if $M$ is a nondeterministic machine whose running time is bounded by some constructive function $t(n)=2^{n^{o(1)}}$, then\footnote{The original statement in \cite{PS21} is slightly weaker: they require the nondeterministic machine $M$ to be in polynomial-time instead of $t(n)$ time. We obtain such quantitative improvement by explicitly computing the complexity overhead of the hardness amplification in \cite{DBLP:journals/siamcomp/HealyVV06} (see \Cref{thm:hardamp-ppoly}).}
$$
\TPV \nvdash \LB(M,s, m,n_0),
$$
where $s(n) = 2^{n^\delta}$ and $m(n) = 2^{n}/n$.
\end{theorem}

As a consequence, for every language $L \in \NTIME[2^{n^{o(1)}}]$ and $\delta > 0$ it is consistent with $\TPV$ that there are infinitely many input lengths $n$ and a co-nondeterministic circuit $D_n$ of size $\leq 2^{n^{\delta}}$ such that 
$$
\Pr_{x \sim \{0,1\}^n}[L(x) = D_n(x)] \geq 1 -1/n.
$$

\begin{proof}[Proof of \Cref{thm:unprov_lb2}]
Let $n_0$, $\delta$, $M$, $s(n) = 2^{n^{\delta}}$, and $m(n) = 2^n/n$ be as above. Assume towards a contradiction that 
$$
\TPV \vdash \LB(M,s, m,n_0).
$$
Let $L\triangleq L(M)$ be the language defined by $M$. We argue as follows.
\begin{enumerate}
    \item Under the provability of an almost-everywhere average-case lower bound against conondeterministic circuits, it follows by the soundness of $\TPV$ that (in the standard model) for every sequence $\{E_n\}_{n \geq 1}$ of \emph{deterministic} circuits $E_n$ of size $\leq 2^{n^{\delta}}$, if $n \geq n_0$ then 
    $$
    \Pr_{x \sim \{0,1\}^n}[L(x) = E_n(x)] \leq 1 -1/n.
    $$
    \item From the provability of $\LB(M,s, m,n_0)$, it follows that $\TPV$ proves the sentence $\LBw(M,s,n_0)$ which states a \emph{worst-case} lower bound for $M$ against conondeterministic circuits of the same size. By adapting the argument presented in Section \ref{sec:proof_main} (see Remarks \ref{remark:general_approx} and \ref{remark:general_approx2}), the provability of $\LBw(M,s,n_0)$ in $\TPV$ implies that, in the standard model, for \emph{every} sequence $\{g_n\}_{n\ge 1}$ of functions in $\NSIZE[2^{n^{o(1)}}]$, $\varepsilon > 0$, and large enough $n$, there is a deterministic circuit $C'$ defined over $n$ input variables and of size $2^{n^{\varepsilon}}$ such that
    \begin{equation}\label{eq:approxLprime}
    \Pr_{x \sim \{0,1\}^{n}}[g_n(x) = C'(x)] \geq 1/2 + 2^{-n^{\varepsilon}}.
    \end{equation}
    \item Let $\{f_n\}_{n \geq 1}$ be the sequence of functions in $\NTIME[2^{n^{o(1)}}]$ obtained from $L$, i.e., $f(x)=1$ if and only if $x\in L$. Note that this sequence satisfies the hypothesis of \Cref{thm:hardamp-ppoly} for $s_1(n)=2^{n^{o(1)}}$ and $s_2(n) = 2^{n^{\delta}}$ for sufficiently large $n$. Let $\{h_m\}_{m \geq 1}$ be the sequence of functions in $\mathsf{NSIZE}[2^{m^{o(1)}}]$ obtained by an application of this result, we know that for sufficiently large $n$ and any deterministic circuit $C$ of size $(2^{m^{\gamma\delta}})^\gamma$, it holds that 
    \[
    \Pr_{x\sim\{0,1\}^m}[h_m(x) = C(x)]\le 1/2 + 2^{-\gamma m^{\gamma\delta}}. 
    \]  
    
    Now the hardness of $h_m$ according to \Cref{thm:hardamp-ppoly} contradicts the upper bound provided in \Cref{eq:approxLprime}, if we take $\varepsilon = (1/2) \cdot \delta \cdot \gamma$ and consider large enough input lengths.
    \end{enumerate}
This shows that $\TPV \nvdash \LB(M,s, m,n_0)$, as desired. 
\end{proof}
\addtocontents{toc}{\protect\setcounter{tocdepth}{3}}

\section{Unprovability of Strong Complexity Lower Bounds in Bounded Arithmetic}\label{sec:unprovability_main_section}

In this section, we establish the unprovability of strong $\Sigma_i^p$-vs-$\Pi^p_i$-style lower bounds in bounded arithmetic. Our result generalises a previous unprovability result from \cite{PS21} in two directions: (1) it holds for stronger theories $\Tpv^i$ instead of only $\Tpv^1$; and (2) the lower bound sentence in our unprovability result is more natural in the sense that the hard problem is quantified within the theory, instead of in the meta-theory. 

Due to the complexity of the argument, we will first show in \Cref{sec:general-ph} how to generalise the unprovability result in \cite{PS21} to $\Tpv^i$. Then in \Cref{sec:circuit-vs-circuit} we combine this extension with the new game-theoretic witnessing theorem and with other ideas to obtain our main result, which has both features mentioned above.

\subsection{Unprovability of lower bounds in expressive theories}
\label{sec:general-ph}

For $i \geq 1$, recall that $\Tpv^i$ is the theory consisting of all true (in the standard model) $\forall \Sigma_{i-1}^b(\PV)$ sentences. For instance, $\Tpv^1$ is the universal true theory of $\tPV$. We want to generalize the unprovability of strong nondeterministic circuit lower bounds in $\Tpv^1$ to $\Tpv^i$ for all $i \geq 1$, stated as follows.\footnote{While in \Cref{sec:PichSanthanam} we considered a lower bound for a nondeterministic machine against co-nondeterministic circuits, it will be more convenient for us in this section to phrase the statement as $\Pi_i$-machines against $\Sigma_i$-circuits. Note that this is inconsequential, as the results are equivalent via complementation.}

\begin{theorem}\label{thm:unprovability-in-PH}
     Fix $i \geq 1$. Let $t(n)=2^{n^{o(1)}}$ be a constructive time bound, and $M$ be a $\cPi{i}\TIME[t(n)]$ machine and $\LB^i(M,s,m,n_0)$ be the $\Lpv$-sentence: for all $n\in\Log\Log$ with $n> n_0$ and $C\in\cSig{i}\SIZE[s(n)]$, there exist $m$ distinct inputs $x_1,\dots,x_m$ such that $M(x_j)\ne C(x_j)$ for all $j\in[m]$. Then
    \[
    \Tpv^i\nvdash\LB^i(M,s,m,n_0)
    \]
    for $s(n)=2^{n^\delta}$, $m(n)=2^n/2-2^n/2^{n^\delta}$, and $\delta\in \mathbb{Q}\cap (0,1)$.
\end{theorem}

To obtain the unprovability of strong complexity lower bounds, we  rely on a witnessing theorem that extracts computational information from a proof of the lower bound sentence $\LB^i(M,s,m,n_0)$. 
We discuss the quantifier complexity of (the worst-case complexity analogue of) the $\LB^i(M,s,m,n_0)$ sentence in \Cref{sec:LBi-witnessing}. As its formalization results in a $\forall \Sigma^b_{i+1}(\PV)$ sentence, note that when $i > 1$ we can no longer directly apply the KPT Witnessing Theorem, as in  \Cref{sec:PichSanthanam}. (In addition, for $i > 1$ the theory $\Tpv^i$ is not universal, which is needed when applying this result.) A key aspect of our argument is to introduce an appropriate universal theory with the right abstractions and term complexity (see \Cref{sec:universal_theory}).

\subsubsection{Witnessing for $\Pi_i$ vs $\Sigma_i$  lower bounds}\label{sec:LBi-witnessing}

\ignore{
We will need the following four quantifier KPT witnessing theorem as a technical tool, which immediately follows from the KPT witnessing theorem for universal theory, given the fact that the newly introduced functions $f_\alpha$ and  $g_\beta$ are computable in $\TF\Sigma_{i-1}^p$. Here $\TF\Sigma_{i}^p$ denotes the set of total multi-functions with graphs computable in $\Pi_{i-1}^p$. 

\begin{theorem}
    Let $\varphi$ is any open formula in $\calL(\Upv^i)$, suppose that 
    \[
    \Upv^i\vdash\forall x\exists y\le t\forall z\exists w\le t'\varphi(x,y,z,w).    
    \]
    Then there is a finite sequence $t_1,\dots,t_k$ of $\FP^{\TF\Sigma^p_{i-1}}$ functions such that
    \[
    \mathbb{N}\vdash\forall x,z_1,\dots,z_k(\psi(x,t_1(x),z_1)\lor\psi(x,t_2(x,z_1),z_2)\lor\dots\lor\psi(x,t_k(x,z_1,\dots,z_{k-1}),z_k)),
    \] 
    where $\psi(x,y,z)\triangleq\exists w\le t'\varphi(x,y,z,w)$.  
\end{theorem}
}

Let $\LBw^i(M,s,n_0)$ be the following \emph{worst-case} lower bound sentence in the language $\Lpv$: 
\begin{quote}\it 
    For all $n\in\Log\Log$ with $n>n_0$ and circuit $D\in\cSig{i}\SIZE[s(n)]$, there exists an input $x$ of length $n$, such that $D(x)\ne M(x)$.
\end{quote} 
More formally, we have
\begin{eqnarray}
    \LBw^i(M,s,n_0) & \triangleq & \forall n \in \Log\Log~\text{with}~n >  n_0, \forall ~\text{circuit}~D\in\cSig{i}\SIZE[s(n)] \nonumber \\
    & ~ & \exists x \in \{0,1\}^n~\text{such that}~\mathsf{Error}(D,x), \nonumber
\end{eqnarray}
where $\mathsf{Error}(D,x)$ is a sentence stating that $M(x)\ne D(x)$. Since $M$ is a $\Pi_i$-machine and $D$ is a $\Sigma_i$-circuit, in the language $\Lpv$, the sentence $\phi_1(D,x)\triangleq (M(x)=1\land D(x)=0)$ is in $\Pi^b_i$ and the sentence $\phi_2(D,x)\triangleq (M(x)=0\land D(x)=1)$ is in $\Sigma^b_i$.  

\begin{fact}\label{fact:ave-to-worst}
Let $m(n) \geq 1$. If $\Tpv^i \vdash \LB^i(M,s,m,n_0)$ then $\Tpv^i \vdash \LBw^i(M,s,n_0)$.
\end{fact}

\begin{proof}
  This is immediate for any reasonable formalization of the sentence $\LB^i(M,s,m,n_0)$, since it states an average-case lower bound (at least $m(n) \geq 1$ mistakes) while $\LBw^i(M,s,n_0)$ states a worst-case lower bound (i.e.~at least one mistake).  
\end{proof}

Assume that 
\begin{align*}
  \phi_1(D,x)&\triangleq \forall y\in\{0,1\}^{O(s(n))}~\phi_1'(D,x,y), \\
  \phi_2(D,x)&\triangleq \exists z\in\{0,1\}^{O(s(n))}~\phi_2'(D,x,z), 
\end{align*}
for some $\Sigma^b_{i-1}$-formula $\phi_1'$ and $\Pi^b_{i-1}$-formula $\phi_2'$, respectively. Note that the lengths of the strings $y$ and $z$ are bounded by $O(s(n))$ since we obtain from them parts of the computation of the circuit $D$ (of size $s(n)$) and of the machine $M$ (with running time $2^{n^{o(1)}} < s(n)$). Then $\mathsf{Error}(D,x)\triangleq \phi_1(D,x)\lor\phi_2(D,x)$ is logically equivalent to the  formula
\[
\mathsf{Error}'(D,x)\triangleq \exists z\in\{0,1\}^{O(s(n))}~\forall y\in\{0,1\}^{O(s(n))}~(\phi_1'(D,x,y)\lor\phi_2'(D,x,z)).
\]

\newfunc{\ULB}{ULB}
\newcommand{\ULBw}{\mathsf{ULB}_{\mathsf{wst}}}
Next, consider the universal theories $\Upv^i$ and $\UTpv^i$ introduced in \Cref{sec:universal_theory}. 

\begin{lemma}\label{lmm:ulb-equiv-lbw}
    Let $\ULBw^i(M,s,n_0)$ be a $\Pi^b_3$-sentence in $\calL(\Upv^i)$ defined as follows:
    \begin{align*}
    \ULBw^i(M,s,n_0)\triangleq &\forall n\in\Log\Log\text{ with }n\ge n_0,~\forall\text{ circuit }D\in\cSig{i}\SIZE[s(n)] \\
    &\exists x\in\{0,1\}^n~\exists z\in\{0,1\}^{O(s(n))} \\ 
    &\forall y\in\{0,1\}^{O(s(n))}~(f_{\phi_1'}(D,x,y)=1\lor f_{\phi_2'}(D,x,z)=1).
    \end{align*}
    Then $\Upv^i$ proves $\LBw^i(M,s,n_0)\leftrightarrow\ULBw^i(M,s,n_0)$. Moreover, $\UTpv^i$ proves $\LBw^i(M,s,n_0)\leftrightarrow\ULBw^i(M,s,n_0)$.
\end{lemma}

\begin{proof}
    By the discussion above, $\LBw^i(M,s,n_0)$ is logically equivalent to 
    \begin{align*}
    &\forall n\in\Log\Log\text{ with }n\ge n_0,~\forall\text{ circuit }D\in\cSig{i}\SIZE[s(n)] \\
    &\exists x\in\{0,1\}^n~\exists z\in\{0,1\}^{O(s(n))} \\ 
    &\forall y\in\{0,1\}^{O(s(n))}~(\phi_1'(D,x,y)=1\lor \phi_2'(D,x,z)=1),
    \end{align*}
    which is further equivalent to $\ULBw^i(M,s,n_0)$ in $\Upv^i$ by Lemma \ref{lmm:defining-axiom-f}. The provability of the same sentence in $\UTpv^i$  follows from \Cref{thm:universal_theory}.
\end{proof}

Note that the $\calL(\Upv^i)$-sentence $\ULBw^i(M,s,n_0)$ has low quantifier complexity. 
By exploring the connection between $\Tpv^i$ and the universal theory $\UTpv^i$, we can show the following witnessing result.

\begin{lemma}[Witnessing lemma for $\LB(M,s,m,n_0)$]\label{lmm:witnessing-for-uniform-vs-nonuniform}
    Let $i\ge 1$, $M$ be a $\cPi i\TIME[2^{n^{o(1)}}]$ machine, $\delta\in (0,1)$, $n_0\ge 1$, $s(n)=2^{n^{\delta}}$, and $m(n)=2^n/2-2^n/2^{n^\delta}$. If $\Tpv^i\vdash \LB^i(M,s,m,n_0)$, then there exist  $\ell\in\bbN$ and $\ell$ algorithms $A_1,A_2,\dots,A_\ell$ such that:
    \begin{itemize}
        \item Every $A_i$ is computable in $\FP^{\Sigma_{i-1}^p}$ over inputs of length of order $N = 2^n$.   
        \item For every $i\in[\ell]$, the input of $A_i$ consists of $1^N$, $1^n$ for $n = \log N$, an $n$-input circuit $D\in\cSig{i}\SIZE[s(n)]$, and $i-1$ strings $y_1,\dots,y_{i-1}$; the output of $A_i$ is a pair $(x_i,z_i)\in\{0,1\}^n\times\{0,1\}^{O(s(n))}$.\footnote{Formally, since $n \in \Log\Log$ in our formalization, each $A_i$ has access to an input $\alpha$ of length $|\alpha| = N = 2^n$. For convenience of notation, when discussing $A_1, \ldots, A_\ell$ we often omit the input $1^N$ and concentrate on $n$, which is the key parameter.}
        \item Let $h:(n,D,x)\mapsto y$ be the following function. Given $n$, a string $x\in\{0,1\}^n$, and a circuit $D\in\cSig i\SIZE[s(n)]$, output a $y$ such that $\lnot\phi_1'(D,x,y)$ if such $y\in\{0,1\}^{O(s(n))}$ exists, or $0$ otherwise.
        \item For all $n >  n_0$ and circuit $D\in\cSig{i}\SIZE[s(n)]$, let
          \[
          \begin{matrix}
            & (x_1,z_1)\triangleq A_1(1^n,D) & y_1\triangleq h(n,D,x_1)  & \nonumber \\ 
            & (x_2,z_2)\triangleq A_2(1^n,D,y_1) & y_2\triangleq h(n,D,x_2) & \nonumber \\
            & \vdots & \vdots & \nonumber \\ 
            & (x_\ell,z_\ell)\triangleq A_\ell(1^n,D,y_1,\dots,y_{i-1}) & y_\ell\triangleq h(n,D,x_\ell). & \nonumber
          \end{matrix}
        \]
        Then there is an index $v \in [\ell]$ such that $D(x_v) \neq M(x_v)$.
    \end{itemize}
\end{lemma}

\begin{proof}
Let $i,M,\delta,n_0,s(n),m(n)$ be defined as above. Assume that $\Tpv^i\vdash\LB^i(M,s,m,n_0)$. Then, by  \Cref{fact:ave-to-worst}, it follows that  $\Tpv^i\vdash\LBw^i(M,s,m,n_0)$. Using \Cref{thm:Upv-extends-Tpv} and Lemma \ref{lmm:ulb-equiv-lbw}, we get that $\UTpv^i\vdash\ULBw^i(M,s,m,n_0)$ . 

Since $\ULBw^i(M,s,m,n_0)$ is a $\forall\Sigma^b_2$-sentence and $\UTpv^i$ is a universal theory, we can invoke the KPT witnessing theorem (\Cref{thm:KPT}) to obtain constantly many $\Lpv^i$-terms $A_1,A_2,\dots,A_\ell$ witnessing the existential quantifier given counter-examples to the innermost  universal quantifier. By Theorem \ref{thm:lpvi-term-complexity} and using that $n \in \Log\Log$, each $A_i$ is computable in $\FP^{\Sigma_{i-1}^p}$ over an input of order $N = 2^n$. Furthermore, since the function $h$ is a valid counter-example oracle for the innermost universal quantifier, it is easy to check that the conclusion of the lemma follows from the guarantee provided by KPT witnessing. 
\end{proof}

\ignore{
\paragraph{Warm up.} Now we assume that $\Tpv^i\vdash\LB^i(M,s,m,n_0)$. Since $\Upv^i$ is an extension of $\Tpv^i$, we know that $\Upv^i\vdash\LB^i(M,s,m,n_0)$. This means the following worst-case lower bound statement
\begin{quote}\it 
    for all $n\in\Log$ with $n>n_0$ and circuit $C\in\cSig{i}\SIZE[s(n)]$, there exists an input $x$ of length $n$, such that $C(x)\ne M(x)$,
\end{quote} 
which can be expressed as a $\forall\exists\forall\exists$-formula\footnote{The other $i-2$ quantifiers are replaced by a formula $f_\alpha$.} in $\calL(\Upv^i)$, is provable from $\Upv^i$. By the KPT-witnessing theorem above, there exists a constant-round KPT-learning protocol that finds an $x$ given $n>n_0$ and $C\in\cSig{i}\SIZE[s(n)]$ with $\FP^{\TF\Sigma^p_{i-1}}$ learner.

We claim that to utilize the original analysis we only need to find a $\Sigma_i^p$ function that could \emph{deterministically} simulate a $\FP^{\TF\Sigma^p_{i-1}}$ algorithm. In other words, the function should be in the form $\exists\Pi_{i-1}^p$, and in addtion, there exists \emph{exactly} one $y=(w,z)$ witnessing the existential quantifier, with $z$ to be a possible output of the $\FP^{\TF\Sigma^p_{i-1}}$ algorithm.
}

\subsubsection{Proof of Theorem \ref{thm:unprovability-in-PH}}

\begin{theorem*}[Theorem \ref{thm:unprovability-in-PH}, restated]
    Fix $i \geq 1$. Let $M$ be a $\cPi{i}\TIME[2^{n^{o(1)}}]$ machine and $\LB^i(M,s,m,n_0)$ be the $\Lpv$-sentence: for all $n\in\Log\Log$ with $n> n_0$ and $\cSig{i}\SIZE[s(n)]$-circuit $C$, there exist $m$ distinct inputs $x_1,\dots,x_m$ such that $M(x_j)\ne C(x_j)$ for all $j\in[m]$. Then
    \[
    \Tpv^i\nvdash\LB^i(M,s,m,n_0)
    \]
    for $s(n)=2^{n^\delta}$, $m(n)=2^n/2-2^n/2^{n^\delta}$, and $\delta\in \mathbb{Q}\cap (0,1)$.
\end{theorem*}

\begin{proof}
  Suppose that $\Tpv^i\vdash\LB^{i}(M,s,m,n_0)$ for some $M\in\cPi i\TIME[2^{n^{o(1)}}]$, $n_0\in\bbN$, $s(n)=2^{n^\delta}$, $m(n)=2^n/2-2^n/2^{n^\delta}$, and $\delta \in \mathbb{Q}\cap (0,1)$, there exist an $\ell\in\bbN$ and $\ell$ algorithms $A_1,A_2,\dots,A_\ell$ as described by Lemma \ref{lmm:witnessing-for-uniform-vs-nonuniform}. Similar to \cite{PS21}, we will utilize the algorithms $A_i$ to show that $M$ can be non-trivially approximated by $\cSig i\SIZE[2^{n^\delta}]$ circuits for some $n>n_0$, leading to a contradiction to the soundness of $\Tpv^i$.

  Let $c$ be a constant to be determined later, and $\NW_f(w,x)$ be the Nisan-Wigderson generator with seed length $|w|=n^c$, output length $2^n$, ``hard'' function $f:\{0,1\}^{n^{c/2}}\to\{0,1\}$ (therefore each subset in the combinatorial design has size $n^{c/2}$), $|x|=n$, and any two distinct subsets in the combinatorial design with intersection of size at most $n$. For every seed $w\in\{0,1\}^{n^c}$, let $D_w:\{0,1\}^n\to\{0,1\}$ be a $\Sigma_i$-circuit computing $D_w(x)\triangleq\NW_{\overline{M}}(w,x)$, which is of size at most $2^{n^{o(c)}}\le 2^{n^\delta}$ for sufficient large $n$. We will find some $w\in\{0,1\}^{n^c}$ and use $D_w$ as $C$ in Lemma \ref{lmm:witnessing-for-uniform-vs-nonuniform} to obtain a $\cSig{i}\SIZE[2^{O(n)}]$ circuit $B$ approximating $M$ on input length $n^{c/2}$, i.e., $\Pr_{u\in\{0,1\}^{n^{c/2}}}[B(u)=M(u)]\ge\frac{1}{2}+2^{-O(n)}$. Then by choosing $c>2/\delta$ and sufficiently large $n$, we can prove the theorem.

  \paragraph{Case 1.} Recall that in Lemma \ref{lmm:witnessing-for-uniform-vs-nonuniform}, $A_1$ takes $1^n$ and an $n$-input circuit $D\in\cSig{i}\SIZE[s(n)]$ as input and output a pair $(x,y)\in\{0,1\}^n\times\{0,1\}^{O(s(n))}$. By an averaging argument, there is an $x_1\in\{0,1\}^n$ such that for a uniformly random $w\in\{0,1\}^{n^c}$, with probability at least $2^{-n}$, $A_1(1^n,D_w)$ outputs $(x_1,\cdot)$. Fix this $x_1$ and let
  \begin{align*}
    S_1&\triangleq\Big\{w\in\{0,1\}^{n^c}~\Big\vert~ A_1(1^n,D_w)=(x_1,\cdot)) \Big\}, \\
    S_1^{\sf mist} &\triangleq\Big\{w\in S_1~\Big\vert~ D_w(x_1)\ne M(x_1)\Big\}.
  \end{align*}
  By the definition of $x_1$, we know that $|S_{x_1}|/2^{n^c}\ge 2^{-n}$.

  In Case 1 we assume that $|S_1^{\sf mist}| >  (2/3)\cdot|S_1|$, handling the other scenario in a subsequent case. For any $w\in\{0,1\}^{n^c}$, we know that $D_w(x_1)=\NW_{\overline{M}}(w,x_1)=\overline{M}(w|_{J_{x_1}})$, where $J_{x_1}$ is the subset of indices corresponding to the $x_1$-th row of the combinatorial design. By Lemma \ref{lmm:counting}, there is an assignment $a\in\{0,1\}^{[n^c]\setminus J_{x_1}}$ for the indices outside of $J_{x_1}$ such that $|S_1\uhr_a|/2^{n^{c/2}}\ge 2^{-O(n)}$ and $|S_1^{\sf mist}\uhr_a|/|S_1\uhr_a|\ge 3/5$. Fix an $a\in\{0,1\}^{n^c\setminus J_{x_1}}$ as above. Let $b_1\triangleq M(x_1)\in\{0,1\}$. We define a randomized circuit $B_1:\{0,1\}^{n^{c/2}}\times\{0,1\}\to\{0,1\}$, where the second input is regarded as a random bit, as follows (see Algorithm \ref{algo:b1-uniform} and recall the notation from \Cref{sec:NW}).

  \begin{algorithm}[htb]
    \SetKwInOut{Input}{Input}
    \SetKwInOut{Advice}{Advice}
    \caption{Randomized Circuit $B_1$ for $M$}\label{algo:b1-uniform}
    \Input{The input $u\in\{0,1\}^{n^{c/2}}$ for $M$ and a bit $r\in\{0,1\}$}
    \Advice{$x_1\in\{0,1\}^n$, $a\in\{0,1\}^{[n^c]\setminus J_{x_1}}$, and $b_1=M(x_1)$}
    Let $w=r_{x_1}(a,u)$ and $(x,\cdot)=A_1(1^n,D_w)$\;
    If $x\ne x_1$, \Return{$r$}\;
    Otherwise, \Return{$b_1$}.
  \end{algorithm}

  We first analyse the complexity of $B_1$. Let $m=n^{c/2}=|u|$ be the input length. Since $A_1$ is computable in $\FP^{\Sigma_{i-1}^p}$, it is easy to see that $B_1\in \SIZE^{\Sigma_{i-1}^p}[2^{O(n)}]$. By Theorem \ref{thm:oracle-circuit-to-nondet-circuit}, we get that $B_1\in \cSig{i}\SIZE[2^{O(n)}]$. So we only need to show that for an random bit $r\in\{0,1\}$, $B(x,r)$ approximates $M(x)$ well.

  For any input $u\in\{0,1\}^{n^{c/2}}$ such that $u\in S_1\uhr_a$, we know that
  \begin{align*}
    B(u,r)=M(u) & \iff M(x_1)=M(u) \tag{$x=x_1$ by the definition of $S_{1}$, $B(u,r)=b_1 = M(x_1)$} \\ 
    & \iff M(x_1)\ne D_w(x_1) \tag{$D_w(x_1)=\NW_{\overline{M}}(w,x_1)=\overline{M}(w|_{J_{x_1}})=\overline{M}(u)$} \\ 
    & \iff u\in S_{1}^{\sf mist}\uhr_a.
  \end{align*}
  This means that $B(u,r)$ and $M$ agree on at least $3/5$ of the inputs $u\in S_1\uhr_a$. In the other case, the circuit $B$ outputs the random bit $r$, therefore for some fixed bit $r^*\in\{0,1\}$, $B_1(u,r^*)$ and $M(u)$ agree on at least $1/2$ of the inputs $u\notin S_1\uhr_a$. Since $|S_1\uhr_a|/2^{n^{c/2}}\ge 2^{-O(n)}$, we obtain that
  \[
    \Pr_{u\in\{0,1\}^{n^{c/2}}}\Big[B_1(u,r^*)=M(u)\Big]\ge \frac{3}{5}\cdot \frac{|S_{1}\uhr_a|}{2^{n^{c/2}}} + \frac{1}{2}\cdot\left(1-\frac{|S_{1}\uhr_a|}{2^{n^{c/2}}}\right) = \frac{1}{2}+2^{-O(n)}.
  \]

  \paragraph{Case 2.} Assume that $|S_1^{\sf mist}|\le (2/3)\cdot|S_1|$ instead. For every $w\in S_1$, we define $y_1(w)\triangleq h(n,D_w,x_1)$ to be the output of the counter-example oracle $h$ in Lemma \ref{lmm:witnessing-for-uniform-vs-nonuniform}. Again by an averaging argument, there is an $x_2\in\{0,1\}^n$ such that for a uniformly random $w\in S_1\setminus S_1^{\sf mist}$, with probablity at least $2^{-n}$, $A_2(1^n,D_w,y_1(w))=(x_2,\cdot)$. Fix this $x_2$. Let $S_2$ an $S_2^{\sf mist}$ be the sets defined as follows:
  \begin{align*}
    &S_2\triangleq\Big\{w\in S_1\setminus S_1^{\sf mist}~\Big|~A_2(1^n,D_w,y_1(w))=(x_2,\cdot)\Big\}, \\
    &S_2^{\sf mist}\triangleq\{w\in S_2\mid D_w(x_2)\ne M(x_2)\}. 
  \end{align*}
  By the definition of $x_2$ we know that $|S_2|/2^{n^c}\ge (1/3)\cdot 2^{-O(n)}=2^{-O(n)}$.

  In this case, we further assume that $|S_2^{\sf mist}| > (2/3)\cdot |S_2|$. By construction, for any $w\in\{0,1\}^{n^c}$, $D_w(x_2)=\overline M(w|_{J_{x_2}})$. By Lemma \ref{lmm:counting}, there is an assignment $a\in\{0,1\}^{[n^c]\setminus J_{x_2}}$ for the indices outside of $J_{x_2}$ such that $|S_2\uhr_a|/2^{n^{c/2}}\ge 2^{-O(n)}$ and $|S^{\sf mist}_2\uhr_a|/|S_2\uhr_a|\ge 3/5$. Fix this $a$. We will need the following subroutine to complete this case.\\

   \vspace{-0.2cm}
 
   \noindent ($\nabla$) Given $w \in S_{1}$ of the form $a \cup u$ ($u\in\{0,1\}^{J_{x_2}}$), there is a deterministic circuit $E(w)$ of size at most $2^{O(n)}$ that outputs $(y_1(w), e_1(w))$, where $y_1(w)=h(n,D_w,x_1)$ and $e_1(w)\in\{0,1\}$ such that $e_1(w)=1$ if and only if $w\in S_{1}^{\sf mist}$. \\
 
   \vspace{-0.2cm}

   Note that the circuit $E(w)$ is used to simulate the counter-example oracle $h$ in the first round of the KPT-style game. Let $b_2\triangleq M(x_2)$. We construct a randomized circuit $B_2$ as follows (see Algorithm \ref{algo:b2-uniform}), discussing the claim ($\nabla$) later in the proof.

   \begin{algorithm}
     \SetKwInOut{Input}{Input}
     \SetKwInOut{Advice}{Advice}
     \caption{Randomized circuit $B_2$ for $M$}\label{algo:b2-uniform}
     \Input{The input $u\in\{0,1\}^{n^{c/2}}$ for $M$ and random bit $r\in\{0,1\}$}
     \Advice{$x_1,x_2\in\{0,1\}^n$, $a\in\{0,1\}^{[n^c]\setminus J_{x_2}}$ as discussed, $b_2=M(x_2)$, and $\Gamma$ to support the subroutine $(\nabla)$} 
     Let $w = r_{x_2}(a,u)$ and $(\hat x_1,\hat z_1)=A_1(1^n, D_{w})$\;
     If $\hat x_1\ne x_1$, then \Return{the random bit $r$}; 
     \tcp{after this step, $w\in S_{1}$}
     Let $(y_1(w),e_1(w))=E(w)$ by $(\nabla)$\;
     If $e_1(w)=1$, then \Return{the random bit $r$};
     \tcp{after this step, $w\in S_{1}\setminus S_{1}^{\sf mist}$}
     Let $(\hat x_2,\hat z_2)=A_2(1^n,D_{w},y_1(w))$\;
     If $\hat x_2\ne x_2$, then \Return{the random bit $r$}\; 
     Otherwise, \Return{$b_2$}. \tcp{reaching this line if and only if $w\in S_{2}$}
   \end{algorithm}

   Let $m=n^{c/2}=|u|$ be the input length. We first show that $B_2:\{0,1\}^{m}\times\{0,1\}\to\{0,1\}$ can be implemented by a $\Sigma_{i}$-circuit with size $2^{O(n)}$. Both $A_1$ and $A_2$ are $\FP^{\Sigma^p_{i-1}}$ algorithms with input length $\poly(2^n)$, so both of them can be implemented by $\SIZE^{\Sigma^p_{i-1}}[2^{O(n)}]$ circuits. By ($\nabla$) we also know that $E(w)$ can be implemented by a $2^{O(n)}$-size circuit. As a result, $B_2\in \SIZE^{\Sigma^p_{i-1}}[2^{O(n)}]\subseteq\cSig i\SIZE[2^{O(n)}]$, where the last inclusion follows from \Cref{thm:oracle-circuit-to-nondet-circuit}.

   Now we prove the correctness of the algorithm $B_2$. By construction, it is easy to verify that the algorithm reaches the last line if and only if $w=r_{x_2}(a,u)\in S_2$. Therefore $B_2$ will output a random bit when $w\notin S_2$ (i.e., $u\notin S_2\uhr_a$) and output $b_2$ when $w\in S_2$ (i.e., $u\in S_2\uhr_a$). In the former case, $B_2$ agrees with $M$ on $1/2$ of the inputs for an $r^*\in\{0,1\}$, which will be hard-wired into the circuit. In the latter case, with probability at least $3/5$ over $u\in\{0,1\}^{n^{c/2}}$, $u\in S_2^{\sf mist}\uhr_a$, which further means that
   \[
     M(u)=M(w|_{J_{x_2}})=\overline{D_w}(x_2)=M(x_2)=b_2=B_2(u,r).
   \]
   Since $|S_2\uhr_a|/2^{n^{c/2}}\ge 2^{-O(n)}$, we can conclude that $B_2(u,r^*)$ agrees with $M(u)$ on $\frac{1}{2}+2^{-O(n)}$ of the inputs $u\in\{0,1\}^{n^{c/2}}$.

   \paragraph{Case $j\ge 2$.} Using the technique for Case 2, we can in fact deal with all the remaining cases. Let $j\in\{2,3,\dots,\ell\}$. We define the following notations recursively:
   \begin{enumerate}
   \item $y_{j-1}(w)\triangleq h(n, D_w, x_{j-1})$.
   \item $x_j\in\{0,1\}^n$ be the lexicographically first string such that for a uniformly random string $w\in S_{j-1}\setminus S_{j-1}^{\sf mist}$, with probability at least $2^{-n}$, $A_j(1^n,D_w,y_{1}(w),\dots,y_{j-1}(w))=(x_j,\cdot)$. 
   \item Define $S_j$ and $S_j^{\sf mist}$ as the sets
     \begin{align*}
       &S_j\triangleq \Big\{w\in S_{j-1}\setminus S_{j-1}^{\sf mist}~\Big|~ A_j(1^n,D_w,y_{1}(w),\dots,y_{j-1}(w))=(x_j,\cdot)\Big\} \\
       &S_{j}^{\sf mist}\triangleq \{w\in S_j\mid D_w(x_j)\ne M(x_j)\}. 
     \end{align*}
   \end{enumerate}

   In Case $j \geq 2$ we assume that (1) $|S_j^{\sf mist}| > (2/3)\cdot |S_j|$, and (2) for every $i\in\{1,2,\dots,j-1\}$, $|S_i^{\sf mist}|/|S_i| \leq 2/3$. Crucially, by the definition of each of these sets and the conclusion of Lemma \ref{lmm:witnessing-for-uniform-vs-nonuniform}, we get that by reaching $j = \ell$ we necessarily have $S_\ell=S_\ell^{\sf mist}$, so the case analysis is complete. 
   
   The following lemma will be needed later in the proof.
   
   \begin{lemma}\label{lem:different_x}
       For every $1\le i<j$, we have $x_i\ne x_j$.
   \end{lemma}
   
   \begin{proof}
        Suppose that $x_i = x_j$ for $i < j$. First, it follows from the construction that $S_j \cap S_i^{\sf mist} = \emptyset$. Therefore, for any $w \in S_j^{\sf mist} \subseteq S_j$, we have $M(x_i) = D_w(x_i)$. On the other hand, by definition, for any $w \in S_j^{\sf mist}$, we have $M(x_j) \neq D_w(x_j)$. Note that the assumption of Case $j \geq 2$ implies that $S_j^{\sf mist}$ is nonempty. Take any $w^* \in S_j^{\sf mist}$. Under the hypothesis that $x_i = x_j$, the previous claims yield that both $M(x_i) = D_{w^*}(x_i)$ and $M(x_i) \neq D_{w^*}(x_i)$, which is contradictory.
   \end{proof}

  Under assumptions (1) and (2), one can prove by induction that $|S_j|/2^{n^c}=2^{-O(n)}$, therefore by Lemma \ref{lmm:counting}, there is an assignment $a\in\{0,1\}^{[n^c]\setminus J_{x_j}}$ such that $|S_j\uhr_a|/2^{n^{c/2}}\ge 2^{-O(n)}$ and $|S_j^{\sf mist}\uhr_a|/|S_j\uhr_a|\ge 3/5$. Similarly to $(\nabla)$ in Case 2, we need the following computation $(\nabla^i_j)$ for every $i\in\{1,2,\dots,j-1\}$.\\

 \vspace{-0.2cm}
 
 \noindent ($\nabla_j^i$) Given $w \in S_{i}$ of the form $a \cup u$ ($u\in\{0,1\}^{J_{x_j}}$), there is a deterministic circuit $E_i(w)$ of size at most $2^{O(n)}$ that outputs $(y_i(w),e_{i}(w))$, where $y_i(w)=h(n,D_w,x_{i})$ and $e_{i}(w)\in\{0,1\}$ such that $e_{i}(w)=1$ if and only if $w\in S_{i}^{\sf mist}$.\\
 
 \vspace{-0.2cm}

 Note that $(\nabla_2^1)$ is simply $(\nabla)$ in Case 2. With these subroutines we can construct a randomized circuit $B_j$ that approximates $M$ well as follows (see Algorithm \ref{algo:bj-uniform}).

 \begin{algorithm}[ht]
   \SetKwInOut{Input}{Input}
   \SetKwInOut{Advice}{Advice}
   \caption{Randomized circuit $B_j$ for $M$}\label{algo:bj-uniform}
   \Input{The input $u\in\{0,1\}^{n^{c/2}}$ for $M$ and random bit $r\in\{0,1\}$}
   \Advice{$x_1,\dots,x_j\in\{0,1\}^n$, $a\in\{0,1\}^{[n^c]\setminus J_{x_j}}$ as discussed above, $b_j=M(x_j)$, and $\Gamma$ to support the subroutines $(\nabla^i_j)$}
   Let $w=r_{x_j}(a,u)$\;
   \For{$i=1,2,\dots,j$}{
     Let $(\hat x_i,\hat z_i)=A_i(1^n,D_{w},y_{1},\dots,y_{i-1})$\;
     If $\hat x_i\ne x_i$, then \Return{the random bit $r$}\;
     \tcp{reaching this line iff $w\in S_{i}$}
     \If{$i<j$}{
       Let $(y_i(w),e_i(w))=E_i(w)$ by $(\nabla_j^i)$\; 
       If $e_1(w)=1$, then \Return{the random bit $r$}\;
       \tcp{otherwise, $x\in S_{i}\setminus S_{i}^{\sf mist}$}
     }
   }
   \tcp{reaching this line iff $w\in S_{j}$}
   \Return{$b_j$}\; 
 \end{algorithm}

 Now we analyze the complexity and correctness of the algorithm $B_j$.
 \begin{description}
 \item[(Complexity).] Let $m=n^{c/2}$ be the input length. Since every $A_i$ is computable in $\FP^{\Sigma^p_{i-1}}$ with input length $\poly(2^n, s(n))=2^{O(n)}$, and every $E_i$ is computable by a $2^{O(n)}$-size deterministic circuit, we know that $B_j\in \SIZE^{\Sigma^p_{i-1}}[2^{O(n)}]\subseteq\cSig{i}\SIZE[2^{O(n)}]$ (see Theorem \ref{thm:oracle-circuit-to-nondet-circuit}).
 \item[(Correctness).] The correctness of $B_j$ is proved similarly to Case 2. As noted in the comments appearing in the pseudo-code, for any $w = r_{x_j}(a,u)$ with $u \in \{0,1\}^{n^{c/2}}$, we can prove by induction that the algorithm reaches the end of the $i$-th iteration within the for-loop if and only if $w\in S_i\setminus S_i^{\sf mist}$ for every $i\in\{1,2,\dots,j-1\}$. Furthermore, on such inputs the algorithm reaches the last line if and only if $w\in S_j$. This means that, for an appropriate fixed bit $r^*\in\{0,1\}$, on inputs of this form the algorithm agrees with $M$ on at least $1/2$ of $w\notin S_j$ and on at least $3/5$ of $w\in S_j$. Since $|S_j \uhr_a|/2^{n^{c/2}}=2^{-O(n)}$, $B_j(\cdot, r^*)$ achieves an advantage of $2^{-O(n)}$ with $M(\cdot)$.   
 \end{description}

 \paragraph{Implementation of $(\nabla)$.} To finish the proof, we need to upper bound the circuit complexity of the computation $E_i(w)$ in $\nabla^i_j$ for $1\le i < j\le \ell$, which is used to simulate the counter-example oracle $h$ and to check if $w \in S_i^{\sf mist}$, for $w$ of the appropriate form. Let $j\in\{2,3,\dots,\ell\}$ and $i\in\{1,2,\dots,j-1\}$. Thanks to \Cref{lem:different_x}, we know that $x_i\ne x_j$. Recall that $D_w(x)\triangleq \NW_{\overline M}(w,x)$, where $w\triangleq r_{x_j}(a,u)$. Since any two distinct subsets in the combinatorial design of the Nisan-Wigderson generator have intersection size at most $n$, we get that $|J_{x_i}\cap J_{x_j}|\le n$. Notice that for  $u\in\{0,1\}^{n^{c/2}}$, $h(n,D_w,x_i)$ for $w=r_{x_j}(a,u)$ only depends on $w|_{J_{x_i}}$, which contains at most $n$ bits of $u$. As a result, we can hard-wire the answers of all $2^{n}$ cases and construct a $2^{O(n)}$ circuit to compute $y_i(w) = h(n,D_w,x_i)$. The value $e_i(w)$ can be computed in a similar way. Overall, we obtain that $E_i(w) \eqdef (y_i(w),e_i(w))$ can be computed on all relevant inputs by a (non-uniform) deterministic circuit of size $2^{O(n)}$.

 \paragraph{Wrapping things up.} By the case analysis above, for every $c\ge 2$ and every sufficiently large $n$, there always exists a  $\cSig{i}\SIZE[2^{O(n)}]$ circuit $B$ with input length $m=n^{c/2}$ such that
 \[
   \Pr_{u\in\{0,1\}^m}[B(u)=M(u)]\ge \frac{1}{2}+2^{-O(n)}. 
 \]
 By taking $c\ge 2/\delta$ we get, that for infinitely many values of $n$, $L(M)\cap\{0,1\}^n$ can be approximated with advantage at least $2^{-n^{\delta}}$ by $\cSig{i}\SIZE[2^{n^{\delta}}]$ circuits. This leads to a contradiction, since under $\Tpv^i\vdash\LB^i(M,s,m,n_0)$ and from the soundness of $\Tpv^i$ we obtain that, for sufficiently large $n$, $M$ cannot be approximated with advantage $2^{-n^{\delta}}$ by $\cSig{i}\SIZE[2^{n^{\delta}}]$ circuits. 
\end{proof}

As pointed out in \Cref{remark:general_approx} and \Cref{remark:general_approx2}, we note that assuming the provability of \emph{worst-case} circuit lower bound, the approximation of the machine $M\in\cPi{i}\TIME[2^{n^{o(1)}}]$ by deterministic $\Sigma^p_{i-1}$-oracle circuits of small size also works for any sequence $\{f_n\}_{n\ge 1}$ of functions computable in $\cPi{i}\SIZE[2^{n^{o(1)}}]$. Concretely, instead of defining $D_w(x)\eqdef\NW_{\overline{M}}(w,x)$, we define $D_w(x)\eqdef\NW_{f'}(w,x)$ for $f'(x)\eqdef \lnot f(x)$, and proceed the argument as above. By padding dammy input bits as in \Cref{remark:general_approx2}, we obtain the following corollary. 

\begin{corollary}\label{cor: consequence-of-provability-in-PH}
  Fix $i\ge 1$. Let $M$ be a $\cPi{i}\TIME[2^{n^{o(1)}}]$ machine and $\LBw^i(M,s,n_0)$ be the worst-case lower bound sentence defined as above. Assume that for some $\delta\in(0,1)\cap\bbQ$ and $s(n)=2^{n^\delta}$, $\Tpv^i$ proves $\LB^i_w(M,s,n_0)$. Then for every constant $\epsilon\in(0,1)$, every sufficiently large $n$, and circuit $C\in\cPi{i}\SIZE[2^{n^\delta}]$, there is a $\Sigma^p_{i-1}$-oracle circuit $D$ of size $2^{n^\epsilon}$ such that 
  \[
  \Pr_{x\sim\{0,1\}^n}\Big[C(x)=D(x)\Big]\ge \frac{1}{2}+\frac{1}{2^{n^\epsilon}}.   
  \]   
\end{corollary}

\subsubsection{Relaxing the average-case complexity parameter}
\label{sec: relaxing-avg-case-parameter}

Recall that in \Cref{sec:hard_amplif}, we showed how to obtain the unprovability of circuit lower bounds with a weaker average-case complexity parameter via hardness amplification. This will also be the case here, since the hardness amplification in \cite{DBLP:journals/siamcomp/HealyVV06} generalises to all levels in the polynomial hierarchy (see \Cref{thm:hardamp-ppoly}). 

\begin{theorem}\label{thm:unprovability-in-PH-avg}
  Fix $i\ge 1$. Let $M$ be a $\cPi{i}\TIME[t(n)]$ machine for some constructive function $t(n)=2^{n^{o(1)}}$ and $\LB^i(M,s,m,n_0)$ be defined as in \Cref{thm:unprovability-in-PH}. Then for every constant $\delta\in\bbQ\cap(0,1)$, $s(n)\eqdef 2^{n^\delta}$, $m(n)\eqdef 2^n/n$, and $n_0\in\bbN$, $\Tpv^i\nvdash\LB^i(M,s,m,n_0)$. 
\end{theorem}

\begin{proof}
  Towards a contradiction, we assume that $\Tpv^i\vdash\LB^i(M,s,m,n_0)$ and argue as follows. 
  \begin{enumerate}
    \item Under the provability of the almost-everywhere average-case lower bound $\LB(M,s,m,n_0)$, we obtain from the soundness of $\Tpv^i$ that (in the standard model) for every sequence $\{E_n\}_{n\ge 1}$ of $\Sigma^p_{i-1}$-oracle circuits of size $\le 2^{n^\delta}$ and $n\ge n_0$, we have 
    \[
    \Pr_{x\sim\{0,1\}^n}[M(x)=E_n(x)]\le 1-\frac{1}{n}.   
    \] 
    \item From the provability of $\LB(M,s,m,n_0)$, under reasonable formalization, we can also show that $\LBw^i(M,s,n_0)$ is provable in $\Tpv^i$. We then get from \Cref{cor: consequence-of-provability-in-PH} that for every $\epsilon\in(0,1)$, every sufficiently large $n$, and circuit $C\in\cPi{i}\SIZE[2^{n^\delta}]$, there is a $\Sigma^p_{i-1}$-oracle circuit $D$ of size $2^{n^\epsilon}$ such that 
    \begin{equation}\label{equ: universal-approx-in-ph}
      \Pr_{x\sim\{0,1\}^n}[C(x)=D(x)]\ge\frac{1}{2}+\frac{1}{2^{n^\epsilon}}. 
    \end{equation}
    \item Assume that $n$ is sufficiently large and $f_n:\{0,1\}^n\to\{0,1\}$ is defined as $f(x)=M(x)$. Note that this function satisfies the hypothesis of \Cref{thm:hardamp-ppoly} for $s_1(n)=2^{n^{o(1)}}$ and $s_2(n)=2^{n^\delta}$, hence we can obtain a function $h_\ell:\{0,1\}^\ell\to\{0,1\}$ for some $\ell=O(n^2)$ such that for \emph{every} $\Sigma^p_{i-1}$-oracle circuit $D$ of size $2^{\gamma\ell^{\gamma\delta}}$, it holds that 
    \[
    \Pr_{x\in\{0,1\}^\ell}[h_\ell(x)=D(x)] \le \frac{1}{2}+\frac{1}{2^{\gamma\ell^{\gamma\delta}}}.
    \] 
    By setting $\epsilon=(1/2)\cdot\delta\cdot\gamma$, this violates the upper bound in \Cref{equ: universal-approx-in-ph} when $n$ is sufficiently large. 
  \end{enumerate}
  As a result, we know that $\Tpv^i\nvdash\LB^i(M,s,m,n_0)$ for every $i\ge 1$. 
\end{proof}

\newcommand{\Error}{\mathsf{Error}}

\subsection{Unprovability of lower bound sentences of higher quantifier complexity}\label{sec:circuit-vs-circuit}

In this section, we extend the unprovability results to sentences of higher quantifier complexity that formalize separations between non-uniform circuit classes. Recall that $\cSig{i}\SIZE[s(n)]$ and $\cPi{i}\SIZE[s(n)]$ refer to $\Sigma_i$-circuits and $\Pi_i$-circuits of size $s(n)$, respectively. Let $\LB^i(s_1,s_2,m,n_0)$ denote the following  $\Lpv$-sentence: 
\begin{align*}
&\forall n\in\Log\Log\text{ with }n\ge n_0~\exists C\in\cPi{i}\SIZE[s_1(n)]~\forall D\in\cSig{i}\SIZE[s_2(n)] \\ 
&\exists m=m(n)\mbox{ distinct $n$-bit strings }x^1,\dots,x^m \mbox{ s.t. } \Error(C,D,x^i)\mbox{ for all }i\in[m],
\end{align*}
where $\Error(C,D,x)$ means that the circuits $C$ and $D$ do not agree on the input $x$. It's easy to see that $\mathsf{Error}(C,D,x)$ is a disjunction of a $\Sigma_i^b$-formula and a $\Pi_i^b$-formula. Observe that, already for $i = 1$, $\LB^i(s_1,s_2,m,n_0)$ is a $\forall \Sigma^b_4$-sentence.

\begin{theorem}\label{thm:unprovable-ckt-vs-ckt}
    For every $i\ge 1$, $n_0\in\bbN$, $\delta\in\mathbb{Q}\cap (0,1)$ and $d\ge 1$, $\TPV^i\nvdash\LB^i(s_1, s_2,m,n_0)$, where $s_1(n)=n^d$, $s_2(n)=2^{n^\delta}$ and $m(n)=2^n/2-2^n/2^{n^\delta}$.
\end{theorem}

\subsubsection{Witnessing lemma for lower bound sentences}

Similar to the technique we used in the previous section, we need to apply the witnessing theorem to the lower bound sentences. We define the worst-case version of this lower bound to be the following formula $\LBw^i(s_1,s_2,n_0)$. 
\begin{align*}
&\forall n\in\Log\Log\text{ with }n\ge n_0~\exists C\in\cPi{i}\SIZE[s_1(n)]~\forall D\in\cSig{i}\SIZE[s_2(n)] \\ 
&\exists x\in\{0,1\}^n \mbox{ s.t. } \Error(C,D,x).
\end{align*}
Let $\phi_1(C,D,x)\triangleq (C(x)=1\land D(x)=0)$ be a $\Pi_i^b$-formula and $\phi_2(C,D,x)\triangleq (C(x)=0\land D(x)=1)$ be a $\Sigma_i^b$-formula. Note that $\Error(C,D,x)\triangleq \phi_1(C,D,x)\lor\phi_2(C,D,x)$. Assume that 
\begin{align*}
    \phi_1(C,D,x)&\triangleq\forall y\in\{0,1\}^{O(s(n))}\phi_1'(C,D,x,y), \\ 
    \phi_2(C,D,x)&\triangleq\exists z\in\{0,1\}^{O(s(n))}\phi_2'(C,D,x,z),
\end{align*}
where $\phi_1'$ is a $\Sigma_{i-1}^b$-formula and $\phi_2'$ is a $\Pi_{i-1}^b$-formula. Note that the lengths of $y$ and $z$ are bounded by $O(s(n))$ since they are parts of the computation of the circuits $C$ and $D$. 

\begin{lemma}\label{lmm:equiv-lb'-ulb}
    Let $\ULBw^i(s_1,s_2,n_0)$ be a $\forall \Sigma_4^b$-sentence in $\calL(\Upv^i)$ defined as follows:
    \begin{align*}
    \ULBw^i(s_1,s_2,n_0)~\triangleq~&\forall n\in\Log\Log\text{ with }n\ge n_0,~\exists\text{ circuit }C\in\cPi{i}\SIZE[s_1(n)] \\
    &\forall\text{ circuit }D\in\cSig{i}\SIZE[s_2(n)], \\ 
    &\exists x\in\{0,1\}^n~\exists z\in\{0,1\}^{O(s(n))}~\forall y\in\{0,1\}^{O(s(n))} \\ 
    &\left(f_{\phi_1'}(C,D,x,y)=1\lor f_{\phi_2'}(C,D,x,z)=1\right).
    \end{align*}
    Then $\Upv^i$ proves $\LBw^i(s_1,s_2,n_0)\leftrightarrow\ULBw^i(s_1,s_2,n_0)$. Moreover, $\UTpv^i$ proves $\LBw^i(s_1,s_2,n_0)\leftrightarrow\ULBw^i(s_1,s_2,n_0)$. 
\end{lemma}

\begin{proof}
    The provability in $\Upv^i$ follows from the provability of the defining axioms for $f_\alpha$ (see Lemma \ref{lmm:defining-axiom-f}). In turn, the provability in $\UTpv^i$ follows from \Cref{thm:universal_theory}. 
\end{proof}

\begin{lemma}\label{lmm:kpt-for-lb'}
    Assume that $\TPV^i\vdash\LB^i(s_1,s_2,m,n_0)$. There is an integer $\ell\in\bbN$ and $\FP^{\Sigma^p_{i-1}}$ algorithms $P_1,Q_1,P_2,Q_2,\dots,P_\ell,Q_\ell$ such that the following condition holds.\footnote{As in the statement of \Cref{lmm:witnessing-for-uniform-vs-nonuniform}, the input length of these algorithms is of order $N = 2^n$, since in our formalisation $n \in \Log \Log$. In order to be succinct, we simply write $1^n$ as one of the inputs, since $n$ is the key parameter for us.}
    
    Let $n>n_0$, $g$ be a function that maps a $\cPi{i}\SIZE[s_1(n)]$-circuit to a $\cSig{i}\SIZE[s_2(n)]$ circuit, and $D_C\triangleq g(C)$. Let $h:(n,C,D,x)\mapsto y$ be the function such that $y$ is the lexicographic first string in $\{0,1\}^{O(s(n))}$ such that $\neg \phi_1'(C,D,x,y)$ holds or $0$ if such a string does not exist. Let  
    \[
    \begin{matrix}
    P_1(1^n)= C_1 & 
    Q_1(1^n,D_{C_1})= (x_1,z_1) & 
    h(n,C_1,D_{C_1},x_1)=y_1 \\
    
    P_2(1^n,D_{C_1},y_1)=C_2 & 
    Q_2(1^n,D_{C_1},D_{C_2},y_1)=(x_2,z_2) & 
    h(n,C_2,D_{C_2},x_2)=y_2 \\

    \vdots & \vdots & \vdots \\
    
    P_\ell(1^n,D_{C_{1\dots\ell-1}},y_{1\dots\ell-1})=C_\ell & 
    Q_\ell(1^n,D_{C_{1\dots\ell}},y_{1\dots\ell-1})=(x_\ell,z_\ell) & 
    h(n,C_\ell,D_{C_{\ell}},x_\ell)=y_\ell\;.
    \end{matrix}
    \]
    Then there is $k \in [\ell]$ such that $\Error(C_k,D_{C_k},x_k)$ holds.
\end{lemma}

\begin{proof}
    Suppose that $\Tpv^i\vdash\LB^i(s_1,s_2,m,n_0)$. Then we also have $\Tpv^i\vdash\LBw^i(s_1,s_2,n_0)$, which further means by  \Cref{thm:Upv-extends-Tpv},  \Cref{lmm:equiv-lb'-ulb}, and \Cref{thm:universal_theory} that $\UTpv^i\vdash\ULBw^i(s_1,s_2,n_0)$. Recall that $\UTpv^i$ is a universal theory closed under if-then-else (see \Cref{thm:universal_theory}). By Theorem \ref{thm:witnessing-general}, there are an $\ell\in\bbN$ and a sequence of $\ell$ $\calL$-strategies $\tau^\ttt_1,\tau^\ttt_2,\dots,\tau^\ttt_\ell$ for the truthifier such that for any fixed strategy $\tau^\ttf$ of the falsifier, at least one of the strategies beats $\tau^\ttf$ in $\ell$ sequential plays of the evaluation game with $\tau^\ttt_1,\tau^\ttt_2,\dots,\tau^\ttt_\ell$ vs $\tau^\ttf$. 
    
    In particular, consider the following strategy for the falsifier: if the truthifier chooses $C$ in the first round of the game, the falsifier will choose $D_C$; then if the truthifier chooses $x,z$ in the second round, the falsifier will choose $h(n,C,D_C,x)$. It is easy to see that the claim in the lemma is precisely the winning property of $\tau^\ttt_1,\dots,\tau^\ttt_\ell$ against this particular strategy for the falsifier, given the corresponding auxiliary information. 
\end{proof}

\subsubsection{Proof of Theorem \ref{thm:unprovable-ckt-vs-ckt}}

\begin{theorem*}[Reminder of Theorem \ref{thm:unprovable-ckt-vs-ckt}]
    For every $i\ge 1$, $n_0\in\bbN$, $\delta\in\mathbb{Q}\cap (0,1)$ and $d\ge 1$, $\TPV^i\nvdash\LB^i(s_1, s_2,m,n_0)$, where $s_1(n)=n^d$, $s_2(n)=2^{n^\delta}$ and $m(n)=2^n/2-2^n/2^{n^\delta}$.
\end{theorem*}

\begin{proof}
    Assume that $\TPV^i\vdash\LB^i(s_1,s_2,m,n_0)$. We will derive a contradiction to the soundness of $\Tpv^i$ by showing that for sufficiently large $n$ and all $\Pi_i$-circuits $M:\{0,1\}^{n^{c/2}}\to\{0,1\}$ of size $s_1(n^{c/2})$, there is a $\Sigma_i$-circuit $B:\{0,1\}^{n^{c/2}}\to\{0,1\}$ of size at most $s_2(n^{c/2})$ that agrees with $M$ on all but at most $m(n^{c/2})$ inputs, for some constant $c\in\bbN$ which will be determined later.   
    
    Let $\NW_{f}(w,x)$ be the Nisan-Wigderson generator with: $f:\{0,1\}^{n^{c/2}}\to\{0,1\}$, seed length $|w|=n^c$, $|x| = n+n^d$, and any two distinct subsets in the combinatorial design of intersection of size at most $O(n^d)$. Designs with these parameters are known to exist (see \Cref{sec:NW}).
    
    By Lemma \ref{lmm:kpt-for-lb'}, we have $\ell\in\bbN$ and $\FP^{\Sigma^p_{i-1}}$ machines $P_1,P_2,\dots,P_\ell, Q_1,\dots,Q_\ell$ as described. Let $M:\{0,1\}^{n^{c/2}}\to\{0,1\}$ be a $\Pi_i$-circuit of size $s_1(n^{c/2})$ as described above. Let $D_{w,C}:\{0,1\}^n\to\{0,1\}$ be a $\Sigma_i$-circuit of size at most $s_2(n)$ computing $D_{w,C}(x)\triangleq \NW_{\overline M}(w,x\| C)$ for $w\in\{0,1\}^{n^c}$ and $C\in\{0,1\}^{n^d}$.\footnote{We use $u \| v$ to denote the concatenation of binary strings $u$ and $v$. Jumping ahead, the idea of concatenating $x \|C$ when defining the NW generator will allow us to establish an analogue of \Cref{lem:different_x} in this proof.} We would like to find some suitable $w$ and apply Lemma \ref{lmm:kpt-for-lb'} with $g:C\mapsto D_{w,C}$ to obtain a circuit $B\in\SIZE^{\Sigma^p_{i-1}}[2^{O(n^d)}]\subseteq\cSig{i}\SIZE[2^{O(n^d)}]$ approximating $M$, i.e. $\Pr_{u}[B(u)=M(u)]\ge \frac{1}{2}+2^{-O(n^d)}$. By choosing $c$ as a constant much larger than $d$, we can prove the theorem.
    
    \paragraph{Case 1.} Let $C_1\triangleq P_1(1^n)$. By an averaging argument, there is an $x_1\in\{0,1\}^n$ such that for a uniformly random $w\in\{0,1\}^{n^{c}}$, with probability at least $2^{-n}$, the first coordinate of $Q_1(1^n,D_{w,C_1})$ is $x_1$. Fix this $x_1$ and let
    \begin{align*}
    S_{1}&\triangleq\Big\{w\in\{0,1\}^{n^c}\mid Q_1(1^n,D_{w,C_1})=(x_1,\cdot)\Big\}, \\  
    S_{1}^{\sf mist}&\triangleq\Big\{w\in S_1\mid D_{w,C_1}(x_1)\ne C_1(x_1)\Big\}.
    \end{align*}
    By the definition of $x_1$ we get that $|S_{1}|/2^{n^c}\ge 2^{-n}$. 
    
    In this case we assume that $|S_{1}^{\sf mist}|\ge(2/3)\cdot|S_{1}|$, dealing with the other situation in a subsequent case analysis. For any $w$, we know that $D_{w,C_1}(x_1)=\NW_{\overline{M}}(w,x_1\|C_1)=\overline{M}(w|_{J_{x_1\|C_1}})$, where $J_{x_1\|C_1}$ is the subset of indices corresponding to the $(x_1\|C_1)$-th row of the combinatorial design. By Lemma \ref{lmm:counting}, there is an assignment $a\in\{0,1\}^{[n^c]\setminus J_{x_1\|C_1}}$ for the indices outside of $J_{x_1\|C_1}$ such that $|S_{1}\uhr_a|/2^{n^{c/2}} \ge 2^{-O(n)}$ and $|S_{1}^{\sf mist}\uhr_a|/|S_{1}\uhr_a|\ge 3/5$. 
    
    Now we fix $a\in\{0,1\}^{[n^c]\setminus J_{x_1\|C_1}}$ as above. Let $b_1\triangleq C_1(x_1)\in\{0,1\}$. We define a randomized circuit $B_1$ with access to a random bit $r\in\{0,1\}$ as follow (see Algorithm \ref{algo:b1}). 
    
    \begin{algorithm}[htb]
    \SetKwInOut{Input}{Input}
    \SetKwInOut{Advice}{Advice}
    \caption{Randomized circuit $B_1$ for $M$}\label{algo:b1}
    \Input{The input $u\in\{0,1\}^{n^{c/2}}$ for $M$ and random bit $r\in\{0,1\}$}
    \Advice{$x_1\in\{0,1\}^n$, $C_1=P_1(1^n)$, $a\in\{0,1\}^{[n^c]\setminus J_{x_1\|C_1}}$ as discussed, and $b_1=C_1(x_1)$}
    Let $w=r_{x_1\|C_1}(a,u)$ and $(x,\cdot)=Q_1(1^n,D_{w,C_1})$\;
    If $x\ne x_1$, \Return{the random bit $r$}\;
    Otherwise, \Return{$b_1$}.
    \end{algorithm}
    
    Since $Q_1\in\FP^{\Sigma^p_{i-1}}$ and $|D_{w,C_1}|=2^{n^{o(1)}}$, it is clear that $B_1\in\SIZE^{\Sigma^p_{i-1}}[2^{O(n)}]\subseteq \cSig{i}\SIZE[2^{O(n^d)}]$, so we only need to verify that the randomized circuit $B_1$ approximates $M$. For $u\in\{0,1\}^{n^{c/2}}$ such that $u\in S_{1}\uhr_a$, we have that 
    \begin{align*}
    B_1(u,r)=M(u) & \iff C_1(x_1)=M(u) \tag{$x=x_1$ by the definition of $S_{1}$, $B(u,r)=b_1 = C(x_1)$} \\ 
    & \iff C_1(x_1)\ne D_{w,C_1}(x_1) \tag{$D_{w,C_1}(x_1)=\NW_{\overline{M}}(w,x_1\|C_1)=M(u)$} \\  
    & \iff u\in S_{1}^{\sf mist}\uhr_a.
    \end{align*}
    Therefore $B_1$ and $M$ agree on at least $3/5$ of the inputs $u\in S_{1}\uhr_a$. In the other case, the circuit $B$ simply outputs the random bit $r$, therefore for a specific $r^*\in\{0,1\}$, $B_1(u,r^*)$ and $M(u)$ agree on at least $1/2$ of the inputs $u\notin S_{1}\uhr_a$. Since $|S_{1}\uhr_a\!|/2^{n^{c/2}}\ge 2^{-O(n)}$, we obtain that
    \[
    \Pr_{u\in\{0,1\}^{n^{c/2}}}\Big[B_1(u,r^*)=M(u)\Big]\ge \frac{3}{5}\cdot \frac{|S_{1}\uhr_a|}{2^{n^{c/2}}} + \frac{1}{2}\cdot\left(1-\frac{|S_{1}\uhr_a|}{2^{n^{c/2}}}\right) = \frac{1}{2}+2^{-O(n)}.
    \]

\paragraph{Case 2.} Assume that $|S_{1}^{\sf mist}|\le(2/3)\cdot|S_{1}|$. Let $h(n,C,D,x_1)$ be the function described in Lemma \ref{lmm:kpt-for-lb'}. Let $y_1(w)\triangleq h(n,C_1,D_{w,C_1},x_1)$ and $C_2^w=P_2(1^n,D_{w,C_1},y_1(w))$. Again, by an averaging argument, there are $C_2\in\{0,1\}^{n^d}$ and $x_2\in\{0,1\}^n$ such that for a uniformly random $w\in S_{1}\setminus S_{1}^{\sf mist}$, with probability at least $2^{-O(n^d)}$, $C_2=C_2^w$ and $Q_2(1^n,D_{w, C_1}, D_{w,C_2},y_1(w))=(x_2,\cdot)$. Fix this $C_2$ and $x_2$. Let $S_{2}$ and $S_{2}^{\sf mist}$ be sets defined as follows: 
\begin{align*}
&S_{2}\triangleq\Big\{w\in S_{1}\setminus S_{1}^{\sf mist}\mid C_2=C_2^w\land Q_2(1^n,D_{w,C_1}, D_{w,C_2},y(w))=(x_2,\cdot)\Big\} \\
&S_{2}^{\sf mist}\triangleq \{w\in S_{2}\mid D_{w,C_2}(x_2)\ne C_2(x_2)\}
\end{align*}
By the definitions of $C_2$ and $x_2$, we know that $|S_{2}|/2^{n^c}\ge(1/3)\cdot 2^{-O(n^d)}=2^{-O(n^d)}$.

In this case we assume that $|S_{2}^{\sf mist}|\ge (2/3)\cdot |S_{2}|$. Similarly to Case 1, for any $w\in\{0,1\}^{n^c}$, $D_{w,C_2}(x_2)=\overline{M}(w|_{J_{x_2\|C_2}})$. By Lemma \ref{lmm:counting}, there is an assignment $a\in\{0,1\}^{[n^c]\setminus J_{x_2\|C_2}}$ for the indices outside of $J_{x_2\|C_2}$ such that $|S_{2}\uhr_{a}|/2^{n^{c/2}}\ge 2^{-O(n^d)}$ and $|S_{2}^{\sf mist}\uhr_a|/|S_{2}\uhr_a|\ge 3/5$. Fix this string $a$. We will assume the following computation is possible in order to complete this case, returning to it later on:\\

 \vspace{-0.2cm}
 
 \noindent ($\nabla$) Given $w \in S_{1}$ of the form $a \cup u$ ($u\in\{0,1\}^{J_{x_2\|C_2}}$), there is a deterministic circuit $E(w)$ of size at most $2^{O(n^d)}$ that outputs $(y_1(w),e_1(w))$, where $y_1(w)=h(n,C_1,D_{w,C_1},x_1)$ and $e_1(w)\in\{0,1\}$ such that $e_1(w)=1$ if and only if $w\in S_{1}^{\sf mist}$. \\
 
 \vspace{-0.2cm}
 
Note that if $w\in S_{1}\setminus S_{1}^{\sf mist}$, $y_1(w)$ given by $E(w)$ witnesses that $\neg \Error(C_1,D_{w,C_1},x_1)$. Let $b_2\triangleq C_2(x_2)$. We construct a randomized circuit $B_2$ as follows (see Algorithm \ref{algo:b2}).

\begin{algorithm}
\SetKwInOut{Input}{Input}
\SetKwInOut{Advice}{Advice}
\caption{Randomized circuit $B_2$ for $M$}\label{algo:b2}
\Input{The input $u\in\{0,1\}^{n^{c/2}}$ for $M$ and random bit $r\in\{0,1\}$}
\Advice{$x_1,x_2\in\{0,1\}^n$, $C_1,C_2\in\{0,1\}^{n^d}$, $a\in\{0,1\}^{[n^c]\setminus J_{x_2\|C_2}}$ as discussed, $b_2=C_2(x_2)$, and $\Gamma$ to support the subroutine $(\nabla)$}
Let $w = r_{x_2\|C_2}(a,u)$ and $(\hat x_1,\hat z_1)=Q_1(1^n, D_{w,C_1})$\;
If $\hat x_1\ne x_1$, then \Return{the random bit $r$}; 
\tcp{after this step, $w\in S_{1}$}
Let $(y_1(w),e_1(w))=E(w)$ by $(\nabla)$\;
If $e_1(w)=1$, then \Return{the random bit $r$};
\tcp{after this step, $w\in S_{1}\setminus S_{1}^{\sf mist}$}
Let $\hat C_2=P_2(1^n, D_{w,C_1},y_1(w))$ and $(\hat x_2,\hat z_2)=Q_2(1^n,D_{w,C_1},D_{w,C_2},y_1(w))$\;
If $\hat x_2\ne x_2$ or $\hat C_2\ne C_2$, then \Return{the random bit $r$}\; 
Otherwise, \Return{$b_2$}. \tcp{reaching this line if and only if $w\in S_{2}$}
\end{algorithm}

First, we analyze the complexity of $B_2$. Since $Q_1,P_2,Q_2\in\FP^{\Sigma^p_{i-1}}$ and the input length for each of them is of order $2^{O(n)}$, they can be implemented by circuits of size $2^{O(n)}$ with $\Sigma^p_{i-1}$ oracles. We need $2^{O(n^d)}$ gates to support the computation $(\nabla)$. Therefore, $B_2\in\SIZE^{\Sigma^{p}_{i-1}}[2^{O(n^d)}] \subseteq \cSig{i}\SIZE[2^{O(n^d)}]$. 

By construction, it is easy to verify that the algorithm reaches the last line if and only if $w=r_{x_2\|C_2}(a,u)\in S_{2}$. Therefore $B_2$ will output a random bit when $w\notin S_{2}$ and output $b_2$ when $w\in S_{2}$. In the former case, $B_2$ agrees with $M$ on $1/2$ of the inputs for an $r^*\in\{0,1\}$. In the latter case, with probability at least $3/5$, $w|_{J_{x_2\|C_2}}\in S_{2}^{\sf mist}\uhr_a$, which further means that 
\[
M(u) = M(w|_{J_{x_2\|C_2}})=\overline{D_{w,C_2}}(x_2)=C_2(x_2)=b_2=B_2(u,r). 
\]
As a result, it follows that 
\[
\Pr_{u\in\{0,1\}^{n^{c/2}}}\Big[B_2(u,r^*)=M(u)\Big]\ge \frac{3}{5}\cdot \frac{|S_{2}\uhr_a|}{2^{n^{c/2}}} + \frac{1}{2}\cdot\left(1-\frac{|S_{2}\uhr_a|}{2^{n^{c/2}}}\right) = \frac{1}{2}+2^{-O(n^d)}.
\]

\paragraph{Case $j\ge 2$.} Using the technique for Case $2$, we can in fact deal with all the remaining cases. Let $j\in\{2,3,\dots,\ell\}$. We recursively define the following values:
\begin{enumerate}[itemsep=0pt]
    \item $y_{j-1}(w)\triangleq h(n,C_{j-1},D_{w,C_{j-1}},x_{j-1})$. 
    \item $C_j^w\triangleq P_j(1^n, D_{w,C_1},\dots,D_{w,C_{j-1}}, y_1(w),\dots, y_{j-1}(w))$.
    \item Let $C_j\in\{0,1\}^{n^d}$ be the lexicographical first string (encoding an circuit) such that for a uniformly random string $w\in S_{j-1}\setminus S_{j-1}^{\sf mist}$, with probability at least $2^{-O(n^d)}$, $C_j^w=C_j$. The existence of $C_j$ follows from a counting argument.  
    \item Let $x_j\in\{0,1\}^n$ be the lexicographical first string such that for a uniformly random string $w\in (S_{j-1}\setminus S_{j-1}^{\sf mist})\cap\{w\in\{0,1\}^{n^c}\mid C_j^w=C_j\}$, with probability at least $2^{-n}$, 
    \[ 
    Q_j(1^n,D_{w,C_{1}},\dots,D_{w,C_j},y_{1}(w),\dots,y_{j-1}(w))=(x_j,\cdot).
    \]
    Thus, for a uniformly random string $w\in S_{j-1}\setminus S_{j-1}^{\sf mist}$, with probability at least $2^{-O(n^d)}\cdot 2^{-n}=2^{-O(n^d)}$, $C_j^w = C_j$ and $Q_j(1^n,D_{w,C_{1}},\dots,D_{w,C_j},y_{1}(w),\dots,y_{j-1}(w))=(x_j,\cdot)$.
    \item $S_{j}$ and $S_{j}^{\sf mist}$ be sets recursively defined as 
\begin{align*}
    S_{j}\triangleq&\Big\{w\in S_{j-1}\setminus S_{j-1}^{\sf mist}\mid C_j^w=C_j\land\\
    &Q_j(1^n,D_{w,C_1},\dots,D_{w,C_j},y_{1}(w),\dots,y_{j-1}(w))=(x_j,\cdot)\Big\} \\
    S_{j}^{\sf mist}\triangleq&\{w\in S_{j}\mid D_{w,C_j}(x_j)\ne C_j(x_j)\}
\end{align*}
\end{enumerate}

In Case $j$ we will assume that (1) $|S_{j}^{\sf mist}|/|S_{j}|\ge 2/3$ and (2) for any $i\in\{1,2,\dots,j-1\}$, $|S_{i}^{\sf mist}|/|S_{i}|< 2/3$. In particular, by Lemma \ref{lmm:kpt-for-lb'} we know that if we reach $j = \ell$ then $S_{\ell}=S_{\ell}^{\sf mist}$, so all the cases can be resolved in this way. 

The following lemma will be useful later in the proof.

 \begin{lemma}\label{lem:different_Cx}
       For every $1\le i<j$, we have $(C_i,x_i)\ne (C_j,x_j)$.
   \end{lemma}
   
   \begin{proof}
   First, note that $S_j \cap S_i^{\sf mist} = \emptyset$. Also, since we are in case $j$, $S_j^{\sf mist} \neq \emptyset$, given that $|S_j^{\sf mist}| \geq 2/3 \cdot |S_j|$ and the (inductively established) density lower bound for $|S_j|$. Now take any $w^* \in S_j^{\sf mist}$, i.e.,
   \begin{equation}
       \label{eq:first_ineq_Cx}
       C_j(x_j) \;\neq\; D_{w^*, C_j}(x_j).
   \end{equation}
   Since $S_j^{\sf mist} \subseteq S_j$ and $S_j \cap S_i^{\sf mist} = \emptyset$, we have that $w^* \notin S_i^{\sf mist}$, i.e.,
    \begin{equation}
       \label{eq:second_ineq_Cx}
       C_i(x_i) \;=\; D_{w^*, C_i}(x_i).
   \end{equation}
   Now if we had $(C_i,x_i) = (C_j, x_j)$, this would be in contradiction with \Cref{eq:first_ineq_Cx} and \Cref{eq:second_ineq_Cx}. Consequently, either $C_i \neq C_j$ or $x_i \neq x_j$.
   \end{proof}

We can prove by induction that $|S_{j}|/2^{n^c}=2^{-O(n^d)}$, therefore by Lemma \ref{lmm:counting}, there is an assignment $a\in\{0,1\}^{[n^c]\setminus J_{x_j\| C_j}}$ such that $|S_{j}\uhr_a|/2^{n^{c/2}}\ge 2^{-O(n^d)}$ and $|S_{j}^{\sf mist}\uhr_a|/|S_{j}\uhr_a|\ge 3/5$. Fix this string $a$. Similar to $(\nabla)$ in Case 2, we need the following computation $(\nabla_j^i)$ for every $i\in\{1,2,\dots,j-1\}$. \\

 \vspace{-0.2cm}
 
 \noindent ($\nabla_j^i$) Given $w \in S_{i}$ of the form $a \cup u$ ($u\in\{0,1\}^{J_{x_j\|C_j}}$), there is a deterministic circuit $E_i(w)$ of size at most $2^{O(n^d)}$ that outputs $(y_{i}(w),e_{i}(w))$, where $y_{i}(w)=h(n,C_{i},D_{w,C_i},x_{i})$ and $e_{i}(w)\in\{0,1\}$ such that $e_{i}(w)=1$ if and only if $w\in S_{i}^{\sf mist}$.\\
 
 \vspace{-0.2cm}
 
Note that $(\nabla_2^1)=(\nabla)$ by definition. Let $b_j\triangleq C_j(x_j)$. Using the subroutines described above, We can now present a randomized circuit $B_j$ that approximates $M$ (see Algorithm \ref{alg:bj-for-case-j-ckt}).

\begin{algorithm}[ht]
\SetKwInOut{Input}{Input}
\SetKwInOut{Advice}{Advice}
\caption{Randomized circuit $B_j$ for $M$}\label{alg:bj-for-case-j-ckt}
\Input{The input $u\in\{0,1\}^{n^{c/2}}$ for $M$ and random bit $r\in\{0,1\}$}
\Advice{$x_1,\dots,x_j\in\{0,1\}^n$, $C_1,\dots,C_j\in\{0,1\}^{n^d}$, $a\in\{0,1\}^{[n^c]\setminus J_{x_j\|C_j}}$ as discussed, $b_j=C_j(x_j)$, and $\Gamma$ to support the subroutines $(\nabla^i_j)$}
Let $w=r_{x_j}(a,u)$\;
\For{$i=1,2,\dots,j$}{
Let $\hat C_i=P_i(1^n,D_{w,C_1},\dots,D_{w,C_{i-1}},y_{1}(w),\dots,y_{i-1}(w))$\;
If $\hat C_i\ne C_i$, then \Return{the random bit $r$}\;
Let $(\hat x_i,\hat z_i)=Q_i(1^n,D_{w,C_1},\dots,D_{w,C_i},y_{1}(w),\dots,y_{i-1}(w))$\;
If $\hat x_i\ne x_i$, then \Return{the random bit $r$}\;
\tcp{reaching this line iff $w\in S_{i}$}
\If{$i<j$}{
Let $(y_i(w),e_i(w))=E_i(w)$ by $(\nabla_j^i)$\; 
If $e_i(w)=1$, then \Return{the random bit $r$}\;
\tcp{otherwise, $x\in S_{i}\setminus S_{i}^{\sf mist}$}
}
}
\tcp{reaching this line iff $w\in S_{j}$}
\Return{$b_j$}\; 
\end{algorithm}

Similarly to Case 2, we can see that $B_j\in\SIZE^{\Sigma^p_{i-1}}[2^{O(n^d)}]$. Now we prove the correctness of $B_j$. By the definition of $S_{i}$, we can prove by induction that the algorithm reaches the end of the $i$-th iteration within the for-loop if and only if $w\in S_{i}\setminus S_{i}^{\sf mist}$ for any $i\in\{1,2,\dots,j-1\}$. We can further check that the algorithm reaches the last line if and only if $w\in S_{j}$. This means that, by fixing an appropriate bit $r^*$ as the random bit, the algorithm agrees with $M$ on at least $1/2$ of $w\notin S_{j}$ of the form $w = r_{x_j}(a,u)$, and on at least $3/5$ of $w\in S_{j}$ of the form $w = r_{x_j}(a,u)$. As before, this translates into a correlation of $2^{-O(n^d)}$ over a random input $u \in \{0,1\}^{n^{c/2}}$ using the lower bound on the density of $S_j \uhr_a$. 

\paragraph{Implementation of $(\nabla)$.} To complete the proof it is sufficient to show that $(\nabla_j^i)$ in the $j$-th step is computable by $2^{O(n^d)}$-size circuits, for all $j\in\{2,3,\dots,\ell\}$ and $1\le i<j$. Fix any $j\in\{2,3,\dots,\ell\}$ and $i<j$. Recall that $h(n,C_i,D_{w,C_i},x_i)$ finds the minimal $y_i$ such that $\lnot\phi_1'(C_i,D_{w,C_i},x_i,y_i)$ holds if $\lnot\phi_1(C_i,D_{w,C_i},x_i)$, where $\lnot\phi_1(C_i,D_{w,C_i},x_i)$ means that $C_i(x_i)=0\lor D_{w,C_i}(x_i)=1$. In case $C_i(x_i)=0$, we only need to hard-wire a witness of it, since $C_i$ and $x_i$ are fixed with respect to $w$. 

Now we assume that $C_i(x_i)=1$. By the definition of $D_{w,C_i}$, we know that $D_{w,C_i}(x_i)=\NW_{\overline{M}}(w,x_i\|C_i)$, where $w=a\cup u$ for $a\in \{0,1\}^{[n^c]\setminus J_{x_j\|C_j}}$,  $u\in\{0,1\}^{J_{x_j\|C_j}}$. By \Cref{lem:different_Cx}, $(x_i,C_i)\ne (x_j,C_j)$. Therefore, by the definition of the $\NW$ generator, for $w=a\cup u$ with the input $u\in\{0,1\}^{J_{x_j\|C_j}}$, the output $D_{w,C_i}(x_i)$, as well as the desired witness of the outer-most quantified variable in case that $D_{w,C_i}(x_i)=1$, depends on at most $O(n^d)$ bits of $u$. In such case, we can hard-wire all the $2^{O(n^d)}$ answers with a deterministic $2^{O(n^d)}$-size circuit. 

Similarly, it is not hard to show that the computation $e_i(w)$ can also be implemented by a deterministic circuit of at most this size.

\paragraph{Wrapping things up.} Finally we can combine all these facts to conclude this theorem. By assuming $\TPV^i\vdash \LB^i(s_1,s_2,m,n_0)$, we proved that for sufficiently large $n$ and all $\Pi_i$-circuits $M:\{0,1\}^{n^{c/2}}\to\{0,1\}$ of size $s_1(n^{c/2})=n^{dc/2}$, there is a deterministic circuit $B:\{0,1\}^{n^{c/2}}\to\{0,1\}$ with $\Sigma^p_{i-1}$ oracle gates of size at most $2^{O(n^d)}$ that agrees with $M$ on a $1/2+2^{-O(n^d)}$ fraction of inputs. By Theorem \ref{thm:oracle-circuit-to-nondet-circuit}, we know that $B$ can be implemented by $\Sigma_i$-circuits of size $2^{O(n^d)}$. If we choose $c>2d/\delta$, $B$ is of size $\leq 2^{(n^{c/2})^\delta}$ and agrees with $M$ on a $\geq 1/2+2^{-(n^{c/2})^\delta}$ fraction of the inputs $u \in \{0,1\}^{n^{c/2}}$, which means that $\bbN\vDash \lnot \LB(s_1,s_2,m,n_0)$ for the corresponding choice of $m(n)$. This is a contradiction to the soundness of $\TPV^i$.
\end{proof}

Similarly to what was noted in \Cref{remark:general_approx2} and \Cref{cor: consequence-of-provability-in-PH}, the proof presented above shows that one can approximate \emph{every} $\cPi{i}\SIZE[2^{n^{o(1)}}]$ circuit $M$ by small-size $\Sigma^p_{i-1}$-oracle circuits, assuming the provability of the worst-case circuit lower bound sentence $\LBw(s_1,s_2,n_0)$. We simply use $D_{w,C}(x)\eqdef \NW_{\overline{M}}(w,x\| C)$ and proceed as above. By padding dummy input bits as in \Cref{remark:general_approx2}, we can obtain the following corollary. 

\begin{corollary}\label{cor: unprovability-ckt-average-parameter}
    Fix $i\ge 1$. Assume that for some $n_0\in\bbN$, $\delta\in\bbQ\cap(0,1)$, and $d\ge 1$, $\Tpv^i\vdash\LBw^i(s_1,s_2,n_0)$ for $s_1(n)=n^d$ and $s_2(n)=2^{n^\delta}$. Then for every constant $\epsilon>0$, every sufficiently large $n$, and circuit $A\in\cPi{i}\SIZE[t(n)]$ where $t(n)=2^{n^{o(1)}}$ is some constructive funtion, there is a $\Sigma^p_{i-1}$-oracle circuit $B$ of size $2^{n^\epsilon}$ such that 
    \[
    \Pr_{x\sim\{0,1\}^n}[A(x)=B(x)]\ge\frac{1}{2}+\frac{1}{2^{n^\epsilon}}.     
    \]
\end{corollary}

\subsubsection{Relaxing the average-case complexity parameter}

As in \Cref{sec: relaxing-avg-case-parameter}, we now utilize the hardness amplification theorem (see \Cref{thm:hardamp-ppoly}) to relax the average-case complexity parameter. 

\begin{theorem}
    For every $i\ge 1$, $n_0\in\bbN$, $\delta\in\bbQ\cap(0,1)$, and $d\ge 1$, $\Tpv^i\nvdash\LB^i(s_1,s_2,m,n_0)$, where $s_1(n)=n^d$, $s_2(n)=2^{n^\delta}$, and $m=2^n/n$.
\end{theorem}

\begin{proof}
    Suppose that $s_1=s_1(n)$, $s_2=s_2(n)$, $m$, and $n_0$ are defined as above. Towards a contradiction, we assume that $\Tpv^i\vdash\LB^i(s_1,s_2,m,n_0)$. 
    \begin{enumerate}
        \item\label{enum: step 1 approx final} Under the unprovability of the almost-everywhere average-case lower bound $\LB(s_1,s_2,m_0)$, we obtain from the soundness of $\Tpv^i$ that (in the standard model) for sufficiently large $n$, there is a circuit $C\in\cPi{i}\SIZE[s_1(n)]$ such that for every $\Sigma^p_{i-1}$-oracle circuit $D$ of size $2^{n^\delta}$, we have 
        \[
        \Pr_{x\sim\{0,1\}^n}[C(x)=D(x)]\le 1-\frac{1}{n}.     
        \]
        \item By the assumption that $\Tpv^i\vdash\LB^i(s_1,s_2,m,n_0)$, under any reasonable formalization, we know that $\Tpv^i$ also proves the \emph{worst-case} version of the lower bound $\LBw^i(s_1,s_2,n_0)$. Then by \Cref{cor: unprovability-ckt-average-parameter}, we get that for every constant $\epsilon>0$, every sufficiently large $n$, and every $\cPi{i}\SIZE[2^{n^{o(1)}}]$ circuit, there is a $\Sigma^p_{i-1}$-oracle circuit $D$ of size $2^{n^\epsilon}$ such that
        \begin{equation}\label{equ: universal-approx-in-ckt}
        \Pr_{x\sim\{0,1\}^n}[C(x)=D(x)]\ge\frac{1}{2}+\frac{1}{2^{n^{\epsilon}}}.     
        \end{equation}
        \item Now we assume that $n$ is sufficiently large and $f_n:\{0,1\}^n\to\{0,1\}$ is the function computable by $\cPi{i}\SIZE[s_1(n)]$ circuits in \Cref{enum: step 1 approx final} that is hard on average against $\Sigma^p_{i-1}$-oracle circuit. By \Cref{thm:hardamp-ppoly}, there is a function $h_\ell:\{0,1\}^\ell\to\{0,1\}$ for some $\ell=O(n^2)$ that is computable by $\cPi{i}\SIZE[\poly(n)\cdot s_1(n)]$ circuits, such that for every $\Sigma^p_{i-1}$-oracle circuit $D$ of size $2^{\gamma\ell^{\gamma\delta}}$, 
        \[
        \Pr_{x\sim\{0,1\}^\ell}[h_\ell(x)=D(x)]\le \frac{1}{2}+\frac{1}{2^{\gamma\ell^{\gamma\delta}}}.     
        \]
        By setting $\epsilon=(1/2)\cdot \delta\cdot \gamma$, this contradicts \Cref{equ: universal-approx-in-ckt}. 
    \end{enumerate}
    Therefore we conclude that $\Tpv^i\nvdash\LB^i(s_1,s_2,m,n_0)$. 
\end{proof}

\bibliographystyle{alpha}	
\bibliography{references}

\newcommand{\etalchar}[1]{$^{#1}$}
\begin{thebibliography}{KKMP21}

\bibitem[AB87]{DBLP:journals/combinatorica/AlonB87}
Noga Alon and Ravi~B. Boppana.
\newblock The monotone circuit complexity of boolean functions.
\newblock {\em Combinatorica}, 7(1):1--22, 1987.

\bibitem[AB09]{Arora-Barak09}
Sanjeev Arora and Boaz Barak.
\newblock {\em Complexity Theory: {A} Modern Approach}.
\newblock Cambridge University Press, 2009.

\bibitem[Ajt83]{DBLP:journals/apal/Ajtai83}
Mikl{\'{o}}s Ajtai.
\newblock $\sum^1_1$-formulae on finite structures.
\newblock {\em Ann. Pure Appl. Log.}, 24(1):1--48, 1983.

\bibitem[AK10]{AK10}
Eric Allender and Michal Kouck{\'{y}}.
\newblock Amplifying lower bounds by means of self-reducibility.
\newblock {\em J. {ACM}}, 57(3):14:1--14:36, 2010.

\bibitem[And85]{andreev1985method}
Alexander~E. Andreev.
\newblock On a method for obtaining lower bounds for the complexity of
  individual monotone functions.
\newblock {\em Soviet Math. Dokl}, 31(3):530--534, 1985.

\bibitem[AW09]{DBLP:journals/toct/AaronsonW09}
Scott Aaronson and Avi Wigderson.
\newblock Algebrization: {A} new barrier in complexity theory.
\newblock {\em Transactions on Computation Theory \emph{(TOCT)}}, 1(1), 2009.

\bibitem[Bey09]{beyersdorff2009correspondence}
Olaf Beyersdorff.
\newblock On the correspondence between arithmetic theories and propositional
  proof systems -- a survey.
\newblock {\em Mathematical Logic Quarterly}, 55(2):116--137, 2009.

\bibitem[BGS75]{DBLP:journals/siamcomp/BakerGS75}
Theodore~P. Baker, John Gill, and Robert Solovay.
\newblock Relativizatons of the $\mathsf{P}=?\,\,\mathsf{ NP}$ {Q}uestion.
\newblock {\em SIAM J. Comput.}, 4(4):431--442, 1975.

\bibitem[BKKK20]{DBLP:journals/apal/BussKKK20}
Sam Buss, Valentine Kabanets, Antonina Kolokolova, and Michal Kouck{\'{y}}.
\newblock Expander construction in {VNC}$^1$.
\newblock {\em Ann. Pure Appl. Log.}, 171(7):102796, 2020.

\bibitem[BKO20]{BKO20}
Jan Bydzovsky, Jan Kraj{\'{i}}{\v{c}}ek, and Igor~C. Oliveira.
\newblock {Consistency of circuit lower bounds with bounded theories}.
\newblock {\em {Logical Methods in Computer Science}}, 16(2), 2020.

\bibitem[BKT14]{DBLP:journals/jsyml/BussKT14}
Samuel~R. Buss, Leszek~A. Kołodziejczyk, and Neil Thapen.
\newblock Fragments of approximate counting.
\newblock {\em J. Symb. Log.}, 79(2):496--525, 2014.

\bibitem[BM20]{DBLP:journals/aml/BydzovskyM20}
Jan Bydzovsky and Moritz M{\"{u}}ller.
\newblock Polynomial time ultrapowers and the consistency of circuit lower
  bounds.
\newblock {\em Arch. Math. Log.}, 59(1-2):127--147, 2020.

\bibitem[Bus86]{Buss}
Samuel~R. Buss.
\newblock {\em Bounded Arithmetic}.
\newblock Bibliopolis, 1986.

\bibitem[Bus95]{DBLP:journals/apal/Buss95}
Samuel~R. Buss.
\newblock Relating the bounded arithmetic and polynomial time hierarchies.
\newblock {\em Ann. Pure Appl. Log.}, 75(1-2):67--77, 1995.

\bibitem[Bus97]{buss-survey}
Samuel~R. Buss.
\newblock Bounded arithmetic and propositional proof complexity.
\newblock In {\em Logic of Computation}, pages 67--121. Springer Berlin
  Heidelberg, 1997.

\bibitem[Bus98]{buss1998handbook}
Samuel~R Buss.
\newblock {\em Handbook of Proof Theory}.
\newblock Elsevier, 1998.

\bibitem[Bus08]{Buss08}
Samuel~R. Buss.
\newblock Bounded arithmetic, cryptography and complexity.
\newblock {\em Theoria}, 63:147--167, 2008.

\bibitem[CHO{\etalchar{+}}22]{DBLP:journals/jacm/ChenHOPRS22}
Lijie Chen, Shuichi Hirahara, Igor~C. Oliveira, J{\'{a}}n Pich, Ninad Rajgopal,
  and Rahul Santhanam.
\newblock Beyond natural proofs: Hardness magnification and locality.
\newblock {\em J. {ACM}}, 69(4):25:1--25:49, 2022.

\bibitem[CJW19]{DBLP:conf/focs/ChenJW19}
Lijie Chen, Ce~Jin, and Ryan Williams.
\newblock Hardness magnification for all sparse {NP} languages.
\newblock In {\em Symposium on Foundations of Computer Science \emph{(FOCS)}},
  pages 1240--1255, 2019.

\bibitem[CK07]{DBLP:journals/jsyml/CookK07}
Stephen~A. Cook and Jan Kraj{\'{i}}{\v{c}}ek.
\newblock Consequences of the provability of
  \(\mathsf{NP}\subseteq\mathsf{P}/\mathsf{poly}\).
\newblock {\em J. Symb. Log.}, 72(4):1353--1371, 2007.

\bibitem[CKKO21]{CKKO21}
Marco Carmosino, Valentine Kabanets, Antonina Kolokolova, and Igor~C. Oliveira.
\newblock Learn-uniform circuit lower bounds and provability in bounded
  arithmetic.
\newblock In {\em Symposium on Foundations of Computer Science \emph{(FOCS)}},
  2021.

\bibitem[CN10]{cook_nguyen_2010}
Stephen~A. Cook and Phuong Nguyen.
\newblock {\em Logical Foundations of Proof Complexity}.
\newblock Cambridge University Press, 2010.

\bibitem[Cob65]{Cob64}
Alan Cobham.
\newblock The intrinsic computational difficulty of functions.
\newblock {\em Proc. Logic, Methodology and Philosophy of Science}, pages
  24--30, 1965.

\bibitem[Coo75]{Coo75}
Stephen~A. Cook.
\newblock Feasibly constructive proofs and the propositional calculus
  (preliminary version).
\newblock In {\em Symposium on Theory of Computing \emph{(STOC)}}, pages
  83--97, 1975.

\bibitem[CT06]{DBLP:journals/tocl/CookT06}
Stephen~A. Cook and Neil Thapen.
\newblock The strength of replacement in weak arithmetic.
\newblock {\em {ACM} Trans. Comput. Log.}, 7(4):749--764, 2006.

\bibitem[FGHK16]{FGHK16}
Magnus~Gausdal Find, Alexander Golovnev, Edward~A. Hirsch, and Alexander~S.
  Kulikov.
\newblock A better-than-$3n$ lower bound for the circuit complexity of an
  explicit function.
\newblock In {\em Symposium on Foundations of Computer Science \emph{(FOCS)}},
  pages 89--98, 2016.

\bibitem[FLY22]{DBLP:conf/stoc/FanL022}
Zhiyuan Fan, Jiatu Li, and Tianqi Yang.
\newblock The exact complexity of pseudorandom functions and the black-box
  natural proof barrier for bootstrapping results in computational complexity.
\newblock In {\em Symposium on Theory of Computing \emph{(STOC)}}, pages
  962--975, 2022.

\bibitem[FSS84]{DBLP:journals/mst/FurstSS84}
Merrick~L. Furst, James~B. Saxe, and Michael Sipser.
\newblock Parity, circuits, and the polynomial-time hierarchy.
\newblock {\em Math. Syst. Theory}, 17(1):13--27, 1984.

\bibitem[H{\aa}s86]{DBLP:conf/stoc/Hastad86}
Johan H{\aa}stad.
\newblock Almost optimal lower bounds for small depth circuits.
\newblock In {\em Symposium on Theory of Computing \emph{(STOC)}}, pages 6--20,
  1986.

\bibitem[HVV06]{DBLP:journals/siamcomp/HealyVV06}
Alexander Healy, Salil~P. Vadhan, and Emanuele Viola.
\newblock Using nondeterminism to amplify hardness.
\newblock {\em {SIAM} J. Comput.}, 35(4):903--931, 2006.

\bibitem[Imp95]{Impagliazzo95}
Russell Impagliazzo.
\newblock Hard-core distributions for somewhat hard problems.
\newblock In {\em Symposium on Foundations of Computer Science \emph{(FOCS)}},
  pages 538--545. {IEEE} Computer Society, 1995.

\bibitem[Je{\v{r}}04]{Jerabek04}
Emil Je{\v{r}}{\'{a}}bek.
\newblock Dual weak pigeonhole principle, boolean complexity, and
  derandomization.
\newblock {\em Ann. Pure Appl. Log.}, 129(1-3):1--37, 2004.

\bibitem[Je{\v{r}}05]{Jerabek-phd}
Emil Je{\v{r}}{\'{a}}bek.
\newblock {\em Weak pigeonhole principle and randomized computation}.
\newblock PhD thesis, 2005.

\bibitem[Je{\v r}06]{jerabek:sharply-bounded}
Emil Je{\v r}{\'a}bek.
\newblock The strength of sharply bounded induction.
\newblock {\em Mathematical Logic Quarterly}, 52(6):613--624, 2006.

\bibitem[Je{\v{r}}07a]{Jerabek07}
Emil Je{\v{r}}{\'{a}}bek.
\newblock Approximate counting in bounded arithmetic.
\newblock {\em J. Symb. Log.}, 72(3):959--993, 2007.

\bibitem[Je{\v{r}}07b]{Jerabek07-dWPHP}
Emil Je{\v{r}}{\'{a}}bek.
\newblock On independence of variants of the weak pigeonhole principle.
\newblock {\em J. Log. Comput.}, 17(3):587--604, 2007.

\bibitem[Jer22]{jevrabek2022iterated}
Emil Jer{\'{a}}bek.
\newblock Iterated multiplication in {VTC}$^0$.
\newblock {\em Archive for Mathematical Logic}, pages 1--63, 2022.

\bibitem[KH82]{DBLP:journals/tcs/KentH82}
Clement~F. Kent and Bernard~R. Hodgson.
\newblock An arithmetical characterization of {NP}.
\newblock {\em Theor. Comput. Sci.}, 21:255--267, 1982.

\bibitem[KKMP21]{DBLP:conf/innovations/KleinbergKMP21}
Robert Kleinberg, Oliver Korten, Daniel Mitropolsky, and Christos~H.
  Papadimitriou.
\newblock Total functions in the polynomial hierarchy.
\newblock In {\em Innovations in Theoretical Computer Science Conference
  \emph{(ITCS)}}, pages 44:1--44:18, 2021.

\bibitem[KO17]{KO17}
Jan Kraj{\'{i}}{\v{c}}ek and Igor~C. Oliveira.
\newblock {Unprovability of circuit upper bounds in Cook's theory PV}.
\newblock {\em {Logical Methods in Computer Science}}, 13(1), 2017.

\bibitem[Koh08]{applied-proof-theory}
Ulrich Kohlenbach.
\newblock {\em Applied Proof Theory - Proof Interpretations and their Use in
  Mathematics}.
\newblock Springer Monographs in Mathematics. Springer, 2008.

\bibitem[Kor21]{DBLP:conf/focs/Korten21}
Oliver Korten.
\newblock The hardest explicit construction.
\newblock In {\em Symposium on Foundations of Computer Science \emph{(FOCS)}},
  pages 433--444, 2021.

\bibitem[KP98]{DBLP:journals/iandc/KrajicekP98}
Jan Kraj{\'{\i}}cek and Pavel Pudl{\'{a}}k.
\newblock Some consequences of cryptographical conjectures for {$S^1_2$} and
  {$EF$}.
\newblock {\em Inf. Comput.}, 140(1):82--94, 1998.

\bibitem[KPT91]{KrajicekPT91}
Jan Kraj{\'{i}}{\v{c}}ek, Pavel Pudl{\'{a}}k, and Gaisi Takeuti.
\newblock Bounded arithmetic and the polynomial hierarchy.
\newblock {\em Ann. Pure Appl. Log.}, 52(1-2):143--153, 1991.

\bibitem[Kra92]{Krajicek92}
Jan Kraj{\'i}{\v{c}}ek.
\newblock No counter-example interpretation and interactive computation.
\newblock In Yiannis~N. Moschovakis, editor, {\em Logic from Computer Science},
  pages 287--293, New York, NY, 1992. Springer New York.

\bibitem[Kra95]{Krajicek-book}
Jan Kraj{\'{i}}{\v{c}}ek.
\newblock {\em Bounded Arithmetic, Propositional Logic, and Complexity Theory}.
\newblock Encyclopedia of Mathematics and its Applications. Cambridge
  University Press, 1995.

\bibitem[Kra01]{krajivcek2001weak}
Jan Kraj{\'\i}{\v{c}}ek.
\newblock On the weak pigeonhole principle.
\newblock {\em Fundamenta Mathematicae}, 1(170):123--140, 2001.

\bibitem[Kra11]{DBLP:journals/jml/Krajicek11}
Jan Kraj{\'{\i}}cek.
\newblock On the proof complexity of the {Nisan-Wigderson} generator based on a
  hard {NP} {\(\cap\)} {coNP} function.
\newblock {\em J. Math. Log.}, 11(1), 2011.

\bibitem[Kra19]{krajicek_2019}
Jan Krajíček.
\newblock {\em Proof Complexity}.
\newblock Encyclopedia of Mathematics and its Applications. Cambridge
  University Press, 2019.

\bibitem[Kra21]{DBLP:journals/tocl/Krajicek21}
Jan Kraj{\'{\i}}cek.
\newblock Small circuits and dual weak {PHP} in the universal theory of p-time
  algorithms.
\newblock {\em {ACM} Trans. Comput. Log.}, 22(2):11:1--11:4, 2021.

\bibitem[LC11]{DBLP:journals/corr/abs-1103-5215}
Dai Tri~Man Le and Stephen~A. Cook.
\newblock Formalizing randomized matching algorithms.
\newblock {\em Log. Methods Comput. Sci.}, 8(3), 2011.

\bibitem[LY22]{DBLP:conf/stoc/Li022}
Jiatu Li and Tianqi Yang.
\newblock 3.1\emph{n} - \emph{o}(\emph{n}) circuit lower bounds for explicit
  functions.
\newblock In {\em Symposium on Theory of Computing \emph{(STOC)}}, pages
  1180--1193, 2022.

\bibitem[Lê14]{thesis_DTML}
Dai Tri~Man Lê.
\newblock {\em Bounded Arithmetic and Formalizing Probabilistic Proofs}.
\newblock PhD thesis, 2014.

\bibitem[McK10]{mckinley2010sequent}
Richard McKinley.
\newblock A sequent calculus demonstration of {H}erbrand's theorem.
\newblock {\em arXiv preprint arXiv:1007.3414}, 2010.

\bibitem[MP20]{DBLP:journals/apal/MullerP20}
Moritz M{\"{u}}ller and J{\'{a}}n Pich.
\newblock Feasibly constructive proofs of succinct weak circuit lower bounds.
\newblock {\em Ann. Pure Appl. Log.}, 171(2), 2020.

\bibitem[Nis92]{Nisan92}
Noam Nisan.
\newblock Pseudorandom generators for space-bounded computation.
\newblock {\em Comb.}, 12(4):449--461, 1992.

\bibitem[NW94]{NW}
Noam Nisan and Avi Wigderson.
\newblock Hardness vs randomness.
\newblock {\em J. Comput. Syst. Sci.}, 49(2):149--167, 1994.

\bibitem[Oja04]{thesis_KO}
Kerry Ojakian.
\newblock {\em Combinatorics in Bounded Arithmetic}.
\newblock PhD thesis, 2004.

\bibitem[OS18]{OS18_mag_first}
Igor~C. Oliveira and Rahul Santhanam.
\newblock Hardness magnification for natural problems.
\newblock In {\em Symposium on Foundations of Computer Science \emph{(FOCS)}},
  pages 65--76, 2018.

\bibitem[Pap94]{Papa}
Christos~H. Papadimitriou.
\newblock {\em Computational Complexity}.
\newblock Addison Wesley, 1994.

\bibitem[Pic14]{thesis_Pich}
J\'{a}n Pich.
\newblock {\em Complexity Theory in Feasible Mathematics}.
\newblock PhD thesis, 2014.

\bibitem[Pic15a]{DBLP:journals/apal/Pich15}
J{\'{a}}n Pich.
\newblock Circuit lower bounds in bounded arithmetics.
\newblock {\em Ann. Pure Appl. Log.}, 166(1):29--45, 2015.

\bibitem[Pic15b]{DBLP:journals/corr/Pich14}
J{\'{a}}n Pich.
\newblock Logical strength of complexity theory and a formalization of the
  {PCP} theorem in bounded arithmetic.
\newblock {\em Log. Methods Comput. Sci.}, 11(2), 2015.

\bibitem[PS21]{PS21}
J{\'{a}}n Pich and Rahul Santhanam.
\newblock Strong co-nondeterministic lower bounds for {NP} cannot be proved
  feasibly.
\newblock In {\em Symposium on Theory of Computing \emph{(STOC)}}, 2021.

\bibitem[Pud92]{pudlak1992relations}
Pavel Pudl{\'a}k.
\newblock Some relations between subsystems of arithmetic and the complexity
  theory.
\newblock In {\em Proc. Conf. Logic from Computer Science}, pages 499--519.
  Springer-Verlag, 1992.

\bibitem[Pud06]{pudlak2006consistency}
Pavel Pudlák.
\newblock Consistency and games - in search of new combinatorial principles.
\newblock In V.~Stoltenberg-Hansen and J.~Väänänen, editors, {\em Logic
  Colloquium ’03}, volume~24 of {\em Lecture Notes in Logic}, pages 244--281.
  ASL, 2006.

\bibitem[Raz85]{Razborov:85a}
Alexander~A. Razborov.
\newblock Lower bounds on the monotone complexity of some {B}oolean functions.
\newblock {\em Doklady Akademii Nauk SSSR}, 281:798--801, 1985.
\newblock English translation in: Soviet Mathematics Doklady 31:354--357, 1985.

\bibitem[Raz87]{Razborov87}
Alexander~A. Razborov.
\newblock Lower bounds on the size of constant-depth networks over a complete
  basis with logical addition.
\newblock {\em Mathematicheskie Zametki}, 41(4):598--607, 1987.

\bibitem[Raz95a]{Razborov-switching-l}
Alexander~A. Razborov.
\newblock Bounded arithmetic and lower bounds in boolean complexity.
\newblock In P.~Clote and J.~Remmel, editors, {\em Feasible Mathematics II},
  pages 344–--386. Birkh{\"{a}}user, 1995.

\bibitem[Raz95b]{razborov1995unprovability}
Alexander~A Razborov.
\newblock Unprovability of lower bounds on circuit size in certain fragments of
  bounded arithmetic.
\newblock {\em Izvestiya: mathematics}, 59(1):205, 1995.

\bibitem[RR97]{DBLP:journals/jcss/RazborovR97}
Alexander~A. Razborov and Steven Rudich.
\newblock Natural proofs.
\newblock {\em J. Comput. Syst. Sci.}, 55(1):24--35, 1997.

\bibitem[RSW22]{DBLP:journals/eccc/RenSW22}
Hanlin Ren, Rahul Santhanam, and Zhikun Wang.
\newblock On the range avoidance problem for circuits.
\newblock In {\em Symposium on Foundations of Computer Science \emph{(FOCS)}},
  2022.

\bibitem[Smo87]{DBLP:conf/stoc/Smolensky87}
Roman Smolensky.
\newblock Algebraic methods in the theory of lower bounds for {B}oolean circuit
  complexity.
\newblock In {\em Symposium on Theory of Computing \emph{(STOC)}}, pages
  77--82, 1987.

\bibitem[Sto76]{DBLP:journals/tcs/Stockmeyer76}
Larry~J. Stockmeyer.
\newblock The polynomial-time hierarchy.
\newblock {\em Theor. Comput. Sci.}, 3(1):1--22, 1976.

\bibitem[TC21]{DBLP:journals/jacm/TzameretC21}
Iddo Tzameret and Stephen~A. Cook.
\newblock Uniform, integral, and feasible proofs for the determinant
  identities.
\newblock {\em J. {ACM}}, 68(2):12:1--12:80, 2021.

\bibitem[Vad12]{Vadhan-Survey}
Salil~P. Vadhan.
\newblock Pseudorandomness.
\newblock {\em Foundations and Trends® in Theoretical Computer Science},
  7(1–3):1--336, 2012.

\bibitem[Wil14]{DBLP:journals/jacm/Williams14}
Ryan Williams.
\newblock Nonuniform {ACC} circuit lower bounds.
\newblock {\em J. {ACM}}, 61(1):2:1--2:32, 2014.

\bibitem[Wra76]{DBLP:journals/tcs/Wrathall76}
Celia Wrathall.
\newblock Complete sets and the polynomial-time hierarchy.
\newblock {\em Theor. Comput. Sci.}, 3(1):23--33, 1976.

\bibitem[Zam96]{DBLP:journals/jsyml/Zambella96}
Domenico Zambella.
\newblock Notes on polynomially bounded arithmetic.
\newblock {\em J. Symb. Log.}, 61(3):942--966, 1996.

\end{thebibliography}

\appendix

\addtocontents{toc}{\protect\setcounter{tocdepth}{2}}

\section{Provability in $\Tpv^i$}\label{sec:provability_in_Tpv}

\newclass{\ioP}{(i.o.)P}

In this section, we further elaborate on the strength of the theories $\Tpv^i$. Similarly to the relation between the complexity classes $\mathsf{P}$, $\mathsf{NP}$, and the different levels of $\mathsf{PH}$, it is currently open if the theories $\Tpv^i$ form a proper hierarchy, i.e., if $\Tpv^j$ can prove more sentences than $\Tpv^i$ when $j > i$. However, as explained in this section, this is the case under standard computational hardness assumptions. Conversely, separating the theories would lead to new complexity class separations.

In Section \ref{sec:sharply-bounded}, we show that $\Tpv^i$ proves every true $\forall\Sigma^b_{i-1}$-sentence extended with sharply bounded quantifiers. 

In Sections \ref{sec:TPVi_total} and \ref{sec:TPVipolhier}, we relate the relative strength of these theories to the hierarchy of total functions and to the polynomial time hierarchy, respectively. The results presented in these sections are closely related to results from \citep{KrajicekPT91}, which explore the strength of Buss's theories $\mathsf{S}^i_2$ and $\mathsf{T}^i_2$ and related questions. 

In Section \ref{sec:NPnotinioP}, we exhibit a complexity lower bound statement of comparatively higher quantifier complexity that is provable in $\Tpv^2$ (under a minimal assumption). In more detail, we show that if $\NP\nsubseteq\ioP$ is true then it is provable in $\Tpv^2$. %

\subsection{Sentences with sharply bounded quantifiers}
\label{sec:sharply-bounded}

Recall that we have defined $\Tpv^i$ as the theory consisting of all true (strict) $\forall\Sigma^b_{i-1}$-sentences. In some contexts, it can also be desirable to allow sharply bounded quantifiers of the form $\forall x\le |t|$ and $\exists x\le |t|$ to appear arbitrarily in a $\Sigma^b_i$-formula (without increasing its quantifier complexity). It is therefore also reasonable to consider a ``stronger'' theory $\widetilde{\mathsf{T}}_\PV^i$ that consists of all true $\forall\Sigma^b_{i-1}$-sentences where sharply bounded quantifiers can freely appear in the axioms (see a standard reference such as \citep{Krajicek-book} for the formal definition of this more general class of sentences). Note that the introduction of sharply bounded quantifiers could in principle be an issue in our unprovability results, since the \emph{replacement principle} \cite{Buss} that is used to manipulate sharply bounded quantifiers is unlikely to be provable in weak bounded theories \cite{DBLP:journals/tocl/CookT06}. 

In this subsection, we show that $\widetilde{\mathsf{T}}_\PV^i=\Tpv^i$. Consequently, we can use without loss of generality strict $\forall \Sigma^b_{i-1}$-sentences when defining each theory $\Tpv^i$.  

\begin{lemma}\label{lmm:sharply-bounded-over-open}
    For every formula $\varphi(\vec x)$ that contains only sharply bounded quantifiers, there is a quantifier-free formula $\hat\varphi(\vec x)$ such that $\Tpv^1\vdash \forall\vec x~(\varphi(\vec x)\leftrightarrow \hat\varphi(\vec x))$. 
\end{lemma}

\begin{proof}[Proof Sketch]
    We use induction on the number of (sharply bounded) quantifiers in $\varphi(\vec x)$. By replacing sharply bounded quantifiers with polynomial time functions that enumerate over their domains, we can reduce the number of sharply bounded quantifiers while maintaining the equivalence over $\Tpv^1$. We omit the details.
\end{proof}

\begin{lemma}\label{lmm:merge-quantifier}
    For every $i\ge 1$, if $\exists y\le t~\phi$ is a $\Sigma^b_i$-formula without sharply bounded quantifiers, then there exists a $\Pi^b_{i-1}$-formula $\phi'$ without sharply bounded quantifiers and a $\Lpv$-term $s$ such that $\Tpv^1\vdash \exists y\le t~\phi\leftrightarrow \exists z\le s~\phi'$. 
    
    Similarly, if $\forall y\le t~\phi$ is a $\Pi^b_i$-formula without sharply bounded quantifiers, then there exists a $\Sigma^b_{i-1}$-formula $\phi'$ without sharply bounded quantifiers and a $\Lpv$-term $s$ such that $\Tpv^1\vdash \forall y\le t~\phi\leftrightarrow \forall z\le s~\phi'$
\end{lemma}
\begin{proof}[Proof Sketch.]
    We can make the upper bound $s$ sufficiently large, merge all outermost bounded existential quantifiers of $\phi$ into a single existential quantifier $\exists z\le s$ (or $\forall z\le s$ in the other case), and use $\PV$-definable pairing functions to simulate the original block of existential quantifiers. See, e.g.,~the proof of \Cref{lmm:bounded-normal-form} for more details. 
\end{proof}

\begin{lemma}\label{lmm:move-sharply-bounded-in}
    Let $i \geq 2$. For every true $\forall\Sigma^b_{i-1}$-sentence $\varphi$ with sharply bounded quantifiers, there is a true $\forall\Sigma^b_{i-1}$-sentence $\hat\varphi$ such that $\Tpv^i\vdash\hat\varphi(x)\to\varphi(x)$, where no sharply bounded quantifier in $\hat\varphi$ appears outside of a (non-sharply) bounded quantifier.
\end{lemma}
\begin{proof}
By applying prenexification rules, we can obtain a $\forall\Sigma^b_{i-1}$-sentence $\varphi'$ that is logically equivalent to $\varphi$. We will also assume without loss of generality that $\varphi'$ is in negation normal form. A pair of quantifiers $(Q_1,Q_2)$ in $\varphi'$ is said to be a \emph{bad pair} if $Q_1$ is a sharply bounded quantifier, $Q_2$ is a (non-sharply) bounded quantifier, and $Q_1$ quantifies over a subformula containing $Q_2$. We prove the lemma by induction on the number of bad pairs within $\varphi'$. If there is no bad pair, we simply let $\hat\varphi=\varphi'$ and the lemma follows. 

    Now we assume that $\varphi'$ contains $\ell\ge 1$ bad pairs. Fix a bad pair $(Q_1, Q_2)$ such that $Q_2$ is innermost and $Q_2$ is outermost, so that there are no other quantifiers in between and there is no bad pair within the formula quantified by $Q_2$. By \Cref{lmm:sharply-bounded-over-open}, we can further assume without loss of generality that there is no sharply bounded quantifier within the formula quantified by $Q_2$. Consider $Q_1 x\le |t|~Q_2 y\le s~\phi$ as a subformula of $\varphi'$. If $Q_1=Q_2=\forall$ or $Q_1=Q_2=\exists$, we can simply exchange them to obtain a logically equivalent sentence with $\ell-1$ bad pairs, which completes the proof via the induction hypothesis.
    \begin{description}
        \item[Case 1.] Assume that $Q_1=\forall$ and $Q_2=\exists$. By the replacement axiom (see, e.g., \cite{Buss}), there are $\Lpv$ terms $\alpha,\beta$ and a quantifier-free formula $\gamma$ such that 
        \begin{equation}
        \big(\forall x\le |t|~\exists y\le s~\phi\big) \leftrightarrow \big(\exists w\le \alpha(s,t)~\forall x\le |t|~(\phi(y/\beta(x,w))\land \gamma(x,w))\big)\label{equ:replacement-axiom}
        \end{equation}
        holds in the standard model. (Note that $\phi$ may have free variables $\vec v$ other than $x,y$.) Since $\varphi'$ is a $\forall\Sigma^b_{i-1}$-sentence containing $\forall x\le |t|~\exists y\le s~\phi$ as a subformula, we can see that $\exists y\le s~\phi$ must be a $\Sigma^b_{i-1}$-formula. By \Cref{lmm:merge-quantifier}, we can assume without loss of generality that $\phi$ is a $\Pi^b_{i-2}$-formula. Then
        \begin{align*} 
        \Psi&\eqdef \forall \vec v~\Big(\big(\exists w\le \alpha(s,t)~\forall x\le |t|~(\phi(y/\beta(x,w))\land \gamma(x,w))\big)\to \big(\forall x\le |t|~\exists y\le s~\phi\big)\Big) \\ 
        & \Leftrightarrow \forall \vec v~\Big(\big(\forall w\le \alpha(s,t)~\exists x\le |t|~(\lnot\phi(y/\beta(x,w))\lor \lnot\gamma(x,w))\big)\lor \big(\forall x\le |t|~\exists y\le s~\phi\big)\Big)
        \end{align*} 
        is a true $\forall\Sigma^b_{i-1}$-sentence (where $\Leftrightarrow$ is in the meta-language and denotes logical equivalence). Moreover, since $\phi$ contains no sharply bounded quantifiers, we know that $\Psi$ is a $\forall\Sigma^b_{i-1}$-sentence even if we treat sharply bounded quantifiers simply as bounded quantifiers. Therefore $\Tpv^i\vdash\Psi$. 
        
        Let $\varphi''$ be the $\forall\Sigma^b_{i-1}$-sentence (with sharply bounded quantifiers) obtained from $\varphi'$ by replacing the LHS of \Cref{equ:replacement-axiom} with the RHS. Since $\Tpv^i\vdash\Psi$, we know that $\Tpv^i\vdash \varphi''\to\varphi'$ as $\varphi'$ is in negation normal form. Moreover, $\varphi''$ has $\ell-1$ bad pairs. This completes the proof by the induction hypothesis.
    
    \item[Case 2.] Assume that $Q_1=\exists$ and $Q_2=\forall$. Again, by the replacement axiom, we know that 
    \begin{equation}
        \big(\exists x\le |t|~\forall y\le s~\phi\big) \leftrightarrow \big(\forall w\le \alpha(s,t)~\exists x\le |t|~(\phi(y/\beta(x,w))\lor \gamma(x,w))\big)\label{equ:replacement-axiom-rev}
    \end{equation}
    is true in the standard model, where $\alpha,\beta$ are $\Lpv$ terms $\gamma$ is a quantifier-formula. Since $\varphi'$ is a $\forall\Sigma^b_{i-1}$-sentence containing $\exists x\le |t|~\forall y\le s~\phi$ as a subformula, we get that $\forall y\le s~\phi$ is a $\Pi^b_{i-2}$-formula and $i > 2$. By \Cref{lmm:merge-quantifier}, we can assume without loss of generality that $\phi$ is a $\Sigma^b_{i-3}$-formula. Then
    \begin{align*} 
        \Psi&\eqdef \forall \vec v~\Big(\big(\forall w\le \alpha(s,t)~\exists x\le |t|~(\phi(y/\beta(x,w))\lor \gamma(x,w))\big)\to\big(\exists x\le |t|~\forall y\le s~\phi\big)\Big) \\ 
        & \Leftrightarrow \forall \vec v~\Big(\big(\exists w\le \alpha(s,t)~\forall x\le |t|~(\lnot\phi(y/\beta(x,w))\land \lnot\gamma(x,w))\big)\lor\big(\exists x\le |t|~\forall y\le s~\phi\big)\Big)
    \end{align*} 
    is a true $\forall \Sigma^b_{i-1}$-sentence even if we treat sharply bounded quantifiers as bounded quantifiers. Therefore $\Tpv^i\vdash\Psi$. Then we can resolve this case as in Case 1. \qedhere
    \end{description}
\end{proof}

\begin{theorem}
   For every $i \geq 1$, $\Tpv^i$ proves every true $\forall\Sigma^b_{i-1}$-sentence even if sharply bounded quantifiers are allowed to appear arbitrarily in the sentence. In other words, $\widetilde{\mathsf{T}}_\PV^i=\Tpv^i$.
\end{theorem}

\begin{proof}
If $i = 1$, the result immediately follows from \Cref{lmm:sharply-bounded-over-open}. For $i \geq 2$, we can first move sharply bounded quantifiers via \Cref{lmm:move-sharply-bounded-in} so they only appear as innermost quantifiers. We can then remove the sharply bounded quantifiers using \Cref{lmm:sharply-bounded-over-open}.
\end{proof}

\subsection{Strength of $\Tpv^i$ and the hierarchy of total functions}\label{sec:TPVi_total}

In this section, we show that separating the theories $\Tpv^i$ would lead to new complexity class separations. For instance, we prove that $\Tpv^2 = \Tpv^1$ if and only if the search problem of every $\TFNP$ relation can be solved in polynomial time. A related  result holds for every $i \geq 1$ (see \Cref{thm:ph-and-Tpv-hierarchy} below for the precise statement).

The relationship between these theories and the corresponding complexity collapses provides evidence that the theories $\Tpv^i$ form a strict hierarchy. However, it also shows that unconditionally establishing that this is the case will be quite difficult.

For convenience, in the statements below we identify $\TF\Sigma^p_0$ with $\FP$. We refer to \Cref{sec:general-ph:complexity-basic} for definitions and to \citep{DBLP:conf/innovations/KleinbergKMP21} for more information about total functions in the polynomial hierarchy. Abusing notation, in the statements below we view $\TF\Sigma^p_{i}$ as a class of search problems, i.e., given $x$ the goal is to find $y$ such that $R(x,y)$ holds, where $R \in \TF\Sigma^p_{i}$.

We will need  the following lemmas, which can be proved using standard techniques from complexity and logic.

\begin{lemma}
For every $i\ge 1$, $\P^{\TF\Sigma^p_{i-1}}\subseteq \Sigma^p_{i-1}\subseteq \P^{\TF\Sigma^p_i}$. 
\end{lemma}

\begin{lemma}\label{lmm:search-decision-collapse}
For every $i\ge 1$, $\Sigma^p_i\subseteq\Sigma^p_{i-1}$ if and only if $\Sigma^p_i\subseteq\P^{\TF\Sigma^p_{i-1}}$. 
\end{lemma}

\begin{lemma}\label{lmm:witnessing-standard}
    Let $i\ge 1$, $t(x)$ be an $\Lpv$ term, and $\phi(x,y)$ be a $\Pi^b_{i-1}(\Lpv)$-formula. If $\Tpv^i\vdash\forall x~\exists y\le t(x)~\phi(x,y)$, then there exists a $\FP^{\Sigma^p_{i-1}}$ algorithm $A(x)$ such that for every $x\in\bbN$, $\phi^\bbN(x,A(x))$ holds. 
\end{lemma}

The next theorem relates the relative strength of theories $\Tpv^i$ to the computational complexity of the search problems associated with the relations in $\TF\Sigma^p_{j}$.

\begin{theorem}\label{thm:ph-and-Tpv-hierarchy}
For every $i\ge 1$, the following propositions hold:
\begin{enumerate}
    \item If $\TF\Sigma^p_i\subseteq \FP^{\TF\Sigma^p_{i-1}}$, then $\Tpv^i\equiv\Tpv^{i+1}$.
    \item If $\Tpv^i\equiv\Tpv^{i+1}$, then $\TF\Sigma^p_i\subseteq \FP^{\Sigma^p_{i-1}}$.
\end{enumerate} 
In particular, $\TFNP=\FP$ if and only if $\Tpv^1\equiv\Tpv^2$. 
\end{theorem}

\begin{proof}
$(1)$ Assume that $\TF\Sigma^p_i\subseteq \FP^{\TF\Sigma^p_{i-1}}$. We need to show that for every $\varphi\in \Tpv^{i+1}$, $\Tpv^{i}\vdash \varphi$. Since it is enough to prove this for the axioms of $\Tpv^{i+1}$, we can assume without loss of generality that $\varphi=\forall x~\exists y\le t(x)~\phi(x,y)$ for some $\Pi^b_{i-1}$-formula $\phi$. Let $R\subseteq\{0,1\}^*\times\{0,1\}^*$ be the search problem such that $(x,y)\in R$ if and only if $y\le t(x)$ and $\phi^\bbN(x,y)$. Using the assumption in Item (\emph{i}), we get that this search problem can be solved in  $\FP^{\TF\Sigma^p_{i-1}}$. In particular, there is a $\Sigma^b_{i-1}$-formula $\beta(x,y)$ that is total over $\mathbb{N}$ and only accepts a pair $(x,y)$ if $(x,y)\in R$. Thus $\Tpv^i\vdash\forall x~\exists y\le t(x)~\beta(x,y)$ and $\Tpv^i\vdash\forall x~\forall y\le t(x)~(\beta(x,y)\to\phi(x,y))$ by counting the quantifier complexity of these two sentences. It then follows that $\Tpv^i\vdash \varphi$.

$(2)$ Assume that $\Tpv^{i+1}\equiv\Tpv^i$. Let $R\in\TF\Sigma^p_i$ be a total relation such that for every $(x,y)\in R$, $|y|\le |x|^c$. Let $\beta(x,y)$ be a $\Pi^b_{i-1}$-formula that captures over the standard model that $(x,y)\in R$. Then $\Tpv^{i+1}\vdash \forall x~\exists y\in\{0,1\}^{|x|^c}~\beta(x,y)$, which further means by the assumption in Item (\emph{ii}) that $\Tpv^i\vdash\forall x~\exists y\in\{0,1\}^{|x|^c}~\beta(x,y)$. By Lemma \ref{lmm:witnessing-standard}, we get that the search problem corresponding to $R$ can be solved in $\FP^{\Sigma^p_{i-1}}$. 
\end{proof}

\subsection{Strength of $\Tpv^i$ and the polynomial hierarchy}\label{sec:TPVipolhier}

In this section, we relate the collapse of theories $\Tpv^i$ to a collapse of the polynomial hierarchy. More precisely, we  show that if $\Tpv^{i+2} = \Tpv^i$ then $\Sigma^p_{i+1}=\Pi^p_{i+1}$. Consequently, under the widely believed assumption that $\mathsf{PH}$ does not collapse, the theories $\Tpv^i$ can prove more sentences as $i$ increases. 

We will need a technical lemma from \citep{KrajicekPT91} employed there to relate a certain collapse in Buss's hierarchy of theories of bounded arithmetic to a corresponding collapse of the polynomial hierarchy. First, we review the following statement.

\paragraph{Principle $\bm{\Omega(i)}$.} There is a constant $k\in\bbN$ such that the following holds. For every relation $P(x,y)\in\Pi^p_i$, there are $\FP^{\Sigma_i^p}$ functions $f_1(a),f_2(a,b_1),\dots,f_k(a,b_1,\dots,b_{k-1})$ such that:
\begin{itemize}[itemsep=0pt]
    \item Either $\forall z~P^*(a,f_1(a),z)$ is true, or for every $b_1$ s.t. $\lnot P^*(a,f_1(a),b_1)$, it holds that: 
    \item Either $\forall z~P^*(a,f_2(a,b_1),z)$ is true, or for every $b_2$ s.t. $\lnot P^*(a,f_2(a,b_1),b_2)$, it holds that:
    \item Either $\forall z~P^*(a,f_3(a,b_1,b_2),z)$ is true, or $\dots$
    \item $\dots$
    \item $\forall z~P^*(a,f_k(a,b_1,b_2,\dots,b_k),z)$ is true;
\end{itemize}
where $P^*(x,y,z)\triangleq |y|\le |x|\land (y=0\lor P(x,y))\land (|y|<|z|\le |x|\to \lnot P(x,z))$.

\begin{lemma}[\citep{KrajicekPT91}, Lemma 2.2]\label{lmm:kpt-collapse}
    For every $i\ge 0$, if Principle $\Omega(i)$ is true, then $\Sigma^p_{i+1}\subseteq\P^{\Sigma^p_i}_{/\poly}$ and thus also $\Sigma^p_{i+2}=\Pi^p_{i+2}$. 
\end{lemma}

\begin{theorem} 
For every $i\ge 1$, if $\Tpv^i\equiv\Tpv^{i+2}$, then $\Sigma^p_{i}\subseteq\P^{\Sigma^p_{i-1}}_{/\poly}$ and thus also $\Sigma^p_{i+1}=\Pi^p_{i+1}$. 
\end{theorem}

\begin{proof}
By \Cref{lmm:kpt-collapse}, it suffices to show that $\Tpv^i\equiv\Tpv^{i+2}$ implies Principle $\Omega(i-1)$, for each $i\ge 1$. Assume that $\Tpv^i\equiv\Tpv^{i+2}$. For every relation $P(x,y)\in\Pi^p_{i-1}$, consider the $\Pi^b_{i-1}(\Lpv)$-formula $\alpha(x,y)$ that defines it and let $\alpha^*(x,y,z)$ be defined as
\[
\alpha^*(x,y,z)\triangleq |y|\le |x|\land (y=0\lor \alpha(x,y))\land (|y|<|z|\le |x|\to \lnot \alpha(x,z)).
\]
Let $\varphi\triangleq \forall x~\exists |y|\le |x|~\forall |z|\le|x|~\alpha^*(x,y,z)$. Since $\bbN\vDash\varphi$ and $\varphi$ is a $\forall \Sigma^b_{i+1}$ sentence, we obtain that $\Tpv^{i+2}\vdash\varphi$ and thus $\Tpv^i\vdash\varphi$ by the assumption that $\Tpv^i\equiv\Tpv^{i+2}$. By Theorem \ref{thm:Upv-extends-Tpv}, it follows that $\Upv^i\vdash\varphi$. Moreover, by \Cref{lmm:defining-axiom-f}, we know that 
\[
\Upv^i\vdash\forall x~\exists y~\forall z~|y|\le |x|\land (y=0\lor f_\alpha(x,y))\land (|y|<|z|\le |x|\to \lnot f_\alpha(x,z)),
\]
where $f_\alpha^\bbN(x,y)$ is exactly $P(x,y)$. Principle $\Omega(i-1)$ then follows directly from the KPT Witnessing Theorem (\Cref{thm:KPT}) and \Cref{thm:lpvi-term-complexity}.
\end{proof}

\subsection{On the provability of $\NP\nsubseteq\ioP$}\label{sec:NPnotinioP}

Recall that the axioms of $\Tpv^2$ consist of all true $\forall \Sigma^b_1$-sentences in the language $\Lpv$. In this section, we give a simple example of a complexity lower bound  encoded by a collection of  $\forall \Sigma^b_2(\Lpv)$-sentences provable in $\Tpv^2$, assuming the lower bound holds. (Note that the formalization below uses $n \in \mathsf{Log}$, while our unprovability results hold even for $n \in \mathsf{Log} \mathsf{Log}$.)

\newcommand{\Sat}{\mathsf{3SAT}}

\begin{theorem}
Assume that $\NP\nsubseteq\ioP$. For every polynomial-time Turing machine $A$, there is a constant $n_0\in\bbN$ such that $\Tpv^2$ proves
\begin{align*}
\mathsf{Fail}(A)\triangleq\forall n\in\Log~\exists \varphi(x_1,\dots,x_m)\in\{0,1\}^n~\Big(n>n_0\to\mathsf{Error}(A,\varphi)\Big),
\end{align*}
where $\varphi$ is an 3-CNF formula, and 
\[
\mathsf{Error}(A,\varphi)\triangleq (\exists x\in\{0,1\}^m~\varphi(x)=1\land A(\varphi)=0)
\lor(\forall x\in\{0,1\}^m~\varphi(x)=0\land A(\varphi)=1).
\]
\end{theorem}

\begin{proof}
    Assume that $\NP\nsubseteq\ioP$. Then $\Sat\nsubseteq\ioP$, which means that for every polynomial-time Turing machine $A$, there exists a constant $n_0$ such that $A$ does not solve $\Sat$ on instances of length $n>n_0$. Let $A$ be an arbitrary polynomial-time Turing machine. We would like to show that $\Tpv^2\vdash\mathsf{Fail}(A)$. 
    
    \paragraph{Search-to-Decision Reduction.} We firstly use a standard search-to-decision reduction to construct an efficient algorithm $S$ that searches for a satisfying assignment, assuming that $A$ solves $\SAT$. An explicit description of $S$ appears below (see  Algorithm \ref{algo:S-search-sat}: Search-SAT Algorithm $S$).
    
    \begin{algorithm}[htb]
    \SetKwInOut{Input}{Input}
    \SetKwInOut{Advice}{Advice}
    \caption{Search-SAT Algorithm $S$}\label{algo:S-search-sat}
    \Input{A string $\varphi\in\{0,1\}^{n}$ encoding a 3-CNF formula.}
    Let $\varphi_1(x_1,x_2,\dots,x_m)$ be $\varphi$\;
    Let $z\in\{0,1\}^m$ be a string to be determined\;
    If $A(\varphi_1)=0$, \Return{$(0,0)$}\;
    \For{$i=1,2,\dots,m$} {
    \uIf{$A(\varphi_{i}(x_i/0))=1$}{
    Let $z_i=0$ and $\varphi_{i+1}=\varphi_{i}(x_i/0)$\;
    }\uElseIf{$A(\varphi_i(x_i/1))=1$}{
    Let $z_i=1$ and $\varphi_{i+1}=\varphi_{i}(x_i/1)$\;
    }\Else{
    \Return{$(1,\varphi_i)$};
    }
    }
    \Return{$(2,z)$}\;
    \end{algorithm}
    
    Without loss of generality, we assume that for every $\varphi=\varphi_1\in\{0,1\}^n$ and $i \in [m]$, the corresponding formulas $\varphi_i(x_i/0)$ and $\varphi_i(x_i/1)$ can also be encoded as $n$-bit strings. Let $A'$ be the following polynomial-time algorithm: given an instance $\varphi\in\{0,1\}^n$ encoding a 3-CNF formula; run $S(\varphi)=(b,z)$; accept if and only if $b=2$ and $\varphi(z)=1$. 
    
    \begin{claim}\label{claim:error-A'}
    There is a constant $n_1\in\bbN$ such that $\Tpv^2\vdash\forall n\in\Log~\exists \varphi(x_1,\dots,x_m)\in\{0,1\}^n~\exists x\in\{0,1\}^m~(n>n_1\to \varphi(x)=1\land A'(\varphi)=0)$.
    \end{claim}
    \begin{proof}
    By the definition of $A'$ we can see that it has only one-sided error, i.e., for every $\varphi$ such that $A'(\varphi)=1$, $\varphi$ is satisfiable. Since $\Sat\nsubseteq\ioP$, the sentence is a $\forall \Sigma^b_1$-sentence that is true in the standard model provided that $n_1$ is large enough, which further means that it is provable in $\Tpv^2$. 
    \end{proof}
    
    \begin{claim}\label{claim:find-counterexample-from-A'} We have that 
    $\Tpv^2\vdash\forall n\in\Log~\forall \varphi(x_1,\dots,x_m)\in\{0,1\}^n~(A(\varphi)=1\land A'(\varphi)=0\to \exists\varphi'\in\{0,1\}^n~\mathsf{Error}(A,\varphi'))$.
    \end{claim}
    \begin{proof}
     Indeed, it is possible to establish even in $\mathsf{PV}$ that if $\lnot \mathsf{Error}(A,\varphi')$ holds for every $\varphi' \in \{0,1\}^n$ then the search-to-decision reduction  works as desired and consequently $\lnot (A(\varphi) = 1 \land A'(\varphi) = 0)$. We omit the details.
    \end{proof}

    \paragraph{Provability of the Hardness of $\Sat$.} Now we prove in $\Tpv^2$ that $\mathsf{Fail}(A)$ holds for $n_0\triangleq n_1$, where $n_1\in\bbN$ is the constant in Claim \ref{claim:error-A'}. Arguing in the theory, let $n\in\Log$ be larger than $n_1$. Towards a contradiction, assume that for every $\varphi(x_1,\dots,x_m)\in\{0,1\}^n$, $\lnot\mathsf{Error}(A,\varphi)$. Let $\varphi(x_1,\dots,x_m)\in\{0,1\}^n$ be a 3-CNF formula from Claim \ref{claim:error-A'} such that $\exists x\in\{0,1\}^m~(n>n_1\to \varphi(x)=1\land A'(\varphi)=0)$. Since $\varphi$ is satisfiable and by assumption $\lnot\mathsf{Error}(A,\varphi)$, we get that $A(\varphi) = 1$. Consequently, we have both $A(\varphi)= 1 \land A'(\varphi) = 0$. In turn, \Cref{claim:find-counterexample-from-A'} yields the existence of $\varphi' \in \{0,1\}^n$ such that $\mathsf{Error}(A, \varphi')$. This is in contradiction to the initial assumption on the correctness of $A'$ on all instances of length $n$.\qedhere
\end{proof}

\section{Proofs of the Witnessing Theorems}\label{sec:appendix_witnessing}

In this section, we present some omitted proofs for the witnessing theorems discussed in \Cref{sec:witnessing-1}.

\subsection{Proof of \Cref{thm:witnessing-tree-exploration} via Herbrand's Theorem} \label{sec:proof_herbrand}

We now demonstrate a proof of the witnessing theorem using Herbrand's Theorem.\footnote{We thank an anonymous reviewer for suggesting this perspective, which simplified our previous presentation relying on a direct analysis based on a sequent calculus. J.~Kraj\'{i}\v{c}ek has also proposed a simplification of the original proof via Gentzen's Midsequent Theorem that we do not explore here.} The latter appears in different forms in the literature; in this section, we refer to the exposition in \cite{buss1998handbook}.\footnote{See also \cite{mckinley2010sequent} for a correction in the proof of Herbrand's  Theorem presented in \cite{buss1998handbook}.} 

We start by clarifying some basic definitions. We work with connectives and quantifiers $\{\forall,\exists,\land,\lor,\lnot\}$ and define other connectives from them. We always assume that first-order sentences are written in \emph{negation normal form}, i.e., negations are placed only over atoms. A formula is said to be in \emph{prenex normal form} if it can be written as $Q_1 x_1~Q_2 x_2~\dots~Q_k x_k~P$, where $Q_1,\dots,Q_k\in\{\forall,\exists\}$ and $P$ is quantifier-free. We identify formulas that differ only by a renaming of bounded variables. 

\begin{definition}
    Let $\varphi(x)$ be a formula. A \emph{prenexification} of $\varphi(x)$ is a formula in prenex normal form obtained by successive applications of the following operations $Qx~\phi \star \psi\mapsto Qx~(\phi\star\psi)$ and $\phi \star Qx~\psi \mapsto Qx~(\phi\star\psi)$,
    where $Q\in\{\forall,\exists\}$ and $\star\in\{\land,\lor\}$. As usual, variables are renamed whenever necessary.
\end{definition}

\begin{definition} An $\lor$\emph{-expansion} of a formula $\varphi$ is any formula obtained from $\varphi$ through applications of the following operation:
\begin{itemize}
    \item[] If $\psi$ is a subformula of an $\lor$-expansion $\varphi'$ of $\varphi$, replacing $\psi$ in $\varphi'$ with $\psi \lor \psi$ produces another $\lor$-expansion of $\varphi$. 
\end{itemize}
    A \emph{strong $\lor$-expansion} of a formula $\varphi$ restricts $\psi$ to be a subformula where the outermost connective is an existential quantifier. Similarly to the previous definition, multiple applications of the rule are allowed. 
\end{definition}

\begin{definition}
    Let $\calT$ be a theory and $\varphi$ be a formula in prenex normal form: 
    \[
    \varphi\eqdef \forall x_1\dots\forall x_{n_1}~\exists y_1~\forall x_{n_1+1}\dots \forall x_{n_2}~\exists y_2\dots   \forall x_{n_{r-1}+1}\dots\forall x_{n_{r}}~\exists y_r~\forall x_{n_{r}+1}\dots\forall x_{n_{r+1}}~\psi(\vec x, \vec y).
    \]
    A \emph{witnessing substitution for $\varphi$ over $\calT$} is a sequence of terms $t_1,\dots,t_r$ such that $\calT\vdash \forall\vec x~\varphi(\vec x,\vec y/\vec t)$, where $t_i$ contains variables only from $x_1,\dots, x_{n_i}$ for every $i\in[r]$.
\end{definition}

\begin{theorem}[Herbrand's Theorem (see, e.g.,~\cite{buss1998handbook,mckinley2010sequent})]\label{thm:herbrand-general}
    Let $\calT$ be a universal theory and $\varphi$ be a first-order formula. Then $\calT\vdash\varphi$ if and only if there is a prenexification of a strong $\lor$-expansion of $\varphi$ that admits a witnessing substitution over $\calT$. 
\end{theorem}

\begin{reminder}[\Cref{thm:witnessing-tree-exploration}]
    Let $\calT$ be a universal bounded theory  with vocabulary $\calL$ that is closed under if-then-else. Let $\varphi$ be a bounded $\calL$-formula  of the form
    \begin{align*}
    \varphi(x)~\triangleq~&\exists y_1\le t_1(x)~\forall x_1\le s_1(x,y_1)~\exists y_2\le t_2(x,y_1,x_1)\dots \forall x_{k-1}\le s_{k-1}(x,y_1,x_1,\dots,y_{k-1})\\
  ~&\exists y_k\le t_k(x,y_1,x_1,\dots,y_{k-1}, x_{k-1})~\forall x_{k}\le s_k(x,y_1,x_1,\dots,y_{k})~\phi(x,x_1,\dots,x_k, y_1,\dots,y_k),
\end{align*}
where $\phi(x, \vec{x}, \vec{y})$ is a quantifier-free $\calL$-formula. Then $\calT\vdash\forall x~\varphi(x)$ if and only if there is a universal winning $\calL$-strategy of length $O(1)$ for the truthifier in the corresponding tree exploration game of $\varphi(x)$. 
\end{reminder}

The ``if'' direction of the theorem is simpler. Assume that $\calT\nvdash\forall x~\varphi(x)$. Then by the completeness theorem there is a model $\calM=(\calD,\calI)$ and $n_0\in\calD$ such that $\varphi^\calM(n_0)$ is false, which further means that there is a winning strategy of the falsifier in the evaluation game of $\varphi(x)$ on the broad $(\calM,n_0)$. Consider the strategy of the falsifier in the tree exploration game that simply simulates this winning strategy, i.e., after the truthifier adds a node and specifies an element on the edge, the falsifier treats the path from the root to this node as a partial transcript of the evaluation game and chooses an element according to the strategy of the evaluation game. It is clear that the truthifier cannot reach a winning node, thus it does not have a universal winning $\calL$-strategy of the tree exploration game.

In the rest of this sub-section, we prove the ``only if'' direction of the theorem, that is to extract a winning strategy from $\calT\vdash\forall x~\varphi(x)$.

\paragraph{Step 1: Unbounded Tree Exploration Game.} Due to technical reasons, we need to define unbounded variants of the tree exploration games.

Let $\varphi(x)$ be a $\Sigma^b_k$-formula in prenex normal form (with bounded quantifiers). We define the following translation $[\cdot]_{\sf imp}$ that transforms a bounded formula in prenex normal form into a logically equivalent formula with only unbounded quantifiers:
\begin{itemize}[itemsep=0pt]
    \item If $\varphi(\vec x)$ is quantifier free, $[\varphi(\vec x)]_{\sf imp}\triangleq\varphi(\vec x)$. 
    \item If $\varphi(\vec x)=\forall y\le t(\vec x)~\phi(\vec x,y)$ and $[\phi]_{\sf imp}= Q_1 z_1~Q_2 z_2\dots Q_k z_k~\alpha(\vec x,y,\vec z)$, where $Q_i\in\{\forall,\exists\}$ for $i\in[k]$ and $\alpha$ is quantifier-free, then $[\varphi(\vec x)]_{\sf imp}\triangleq \forall y~Q_1 z_1~Q_2 z_2\dots Q_k z_k(\neg( y\le t(\vec x)) \lor \alpha(\vec x,y,\vec z))$.
    \item If $\varphi(\vec x)=\exists y\le t(\vec x)~\phi(\vec x,y)$ and $[\phi]_{\sf imp}= Q_1 z_1~Q_2 z_2\dots Q_k z_k~\alpha(\vec x,y,\vec z)$, where $Q_i\in\{\forall,\exists\}$ for $i\in[k]$ and $\alpha$ is quantifier-free, then $[\varphi(\vec x)]_{\sf imp}\triangleq \exists y~Q_1 z_1~Q_2 z_2\dots Q_k z_k(y\le t(\vec x)\land \alpha(\vec x,y,\vec z))$.
\end{itemize}
We say a formula $\varphi$ is \emph{implicitly bounded} if there is a bounded formula $\psi$ in prenex normal form such that $\varphi=[\psi]_{\sf imp}$.

Let $\varphi(x)=\exists y_1~\forall x_1\dots \exists y_k~\forall x_k~\phi(x,\vec x,\vec y)$ be an implicitly bounded $\calL$-formula as discussed above and $(\calM=(\calD,\calI),n_0)$ be a board. The \emph{unbounded tree exploration game} of $\varphi$ is defined as follows. In each round, the truthifier chooses a node $u$ on the tree (which only consists of the root at the beginning) and specifies a number $m\in\calD$; the falsifier then specifies a number $n\in\calD$; after this round, a child of $u$ is added to the tree by an edge labeled $(m,n)$. The truthifier wins if and only if there is a node on the tree such that the pairs on the path from the root to the node form a satisfying assignment of $\phi(x/n_0,\vec x,\vec y)$ within $\calM$, where the truthifier's moves are for $\vec y$ and the falsifier's moves are for $\vec x$. 

An $\calL$-strategy of the truthifier of length $\ell\in\bbN$ is a sequence $$\tau=\left<p_1,r_1,p_2,r_2,\dots,p_\ell,r_\ell\right>\,,$$  where $p_i$ is an $\calL$-term and $r_i\in\bbN$ such that $1\le r_i\le i$. Let $(\calM,n_0)$ be a board. The game-theoretic strategy for the unbounded tree exploration game induced by $\tau$ is the following strategy: 
\begin{itemize}
    \item In the $i$-th move, the truthifier introduces a node numbered $i+1$ as a child of the node $r_i$, and chooses the element $v_i\triangleq p_i^{\calM}(n_0,T,\Gamma)\in\calM$, where $\Gamma$ describes the moves of previous rounds (including $v_1,\dots,v_{i-1}$ and the falsifier's moves). 
\end{itemize}
A length-$\ell$ $\calL$-strategy is said to be a universal winning strategy if the truthifier playing the induced game-theoretic strategy wins within $\ell$ moves against any strategy of the falsifier on any board $(\calM, n_0)$. 

\begin{lemma}\label{lmm:strategy-unbounded-to-bounded}
    Let $\calT$ be a bounded theory over the language $\calL$ that is closed under if-then-else. If there is an $O(1)$-length $\calL$-strategy that is a universal winning strategy of the truthifier for the unbounded tree exploration game of $\varphi = [\psi]_{\sf imp}$, then there is an $O(1)$-length $\calL$-strategy that is a universal winning strategy of the truthifier for the tree exploration game of $\psi$. 
\end{lemma}
\begin{proof}
Assume that $\tau=\left<p_1,r_1,p_2,r_2,\dots,p_\ell,r_\ell\right>$ is a universal winning $\calL$-strategy of length $\ell\in\bbN$ for the unbounded tree exploration game. Let $p_i'(x,\Gamma)$ be the term defined as follows:

\begin{enumerate}[topsep=0pt, itemsep=0pt]
    \item Parse $\Gamma=(m_1,n_1,m_2,n_2,\dots,m_{i-1},n_{i-1})$ as the moves in previous rounds. 
    \item Define $\hat \Gamma_0$ to be the empty list and $\hat \Gamma_{j+1}$ to be 
    \[
    \hat \Gamma_{j+1}=\begin{cases}
    \hat \Gamma_{j};(m_{j+1},n_{j+1}) & \text{if } m_{j+1}=p_{j+1}(n_0,\hat \Gamma_j) \\ 
    \hat \Gamma_{j};(p_{j+1}(n_0,\hat \Gamma_j),0) & \text{otherwise}
    \end{cases}
    \]
    \item Output $0$ if $p_i(n_0,\hat \Gamma_{i-1})$ is not a valid move (i.e., it violates the inequality for the bounded variable); and output $p_i(n_0,\hat \Gamma_{i-1})$ otherwise.  
\end{enumerate}
Note that such $p_i'$ always exists as $\calT$ is closed under if-then-else. We now argue that the $\calL$-quasi-strategy $\tau'\triangleq \left<p_1',r_1,p_2',r_2,\dots,p_\ell',r_\ell\right>$ is indeed a universal winning $\calL$-strategy for the tree exploration game of $\psi$. Intuitively, $\tau'$ is the following $\calL$-quasi-strategy: it simulates $\tau$ if it gives a valid move; otherwise, it simply outputs $0$ and ``forgets'' the response of the falsifier, pretending that in this round it simulates $\tau$ and the falsifier's response were $0$. 

By the definition of $p_i'$ it is easy to see that $\tau'$ is an $\calL$-strategy for the tree exploration game of $\psi$, since it will never output an invalid move. Towards a contradiction we assume that it is not a universal winning strategy. In such case, there exist a board $(\calM,n_0)$ and a strategy $\tau_\ttf'$ for the falsifier that prevents the truthifier from winning within $\ell$ rounds on the board against the truthifier playing the induced strategy of $\tau'$. We now construct a strategy $\tau_\ttf$ of the falsifier for the unbounded tree exploration game of $\varphi$ on the board $(\calM,n_0)$ that prevents $\tau$ from winning within $\ell$ rounds and thus leads to a contradiction.

\begin{itemize}[topsep=0pt]
    \item Assume that the moves of $\tau_\ttf'$ against $\tau'$ are $n_1',n_2',\dots,n_\ell'$. In the $i$-th move, if the truthifier's move is an invalid move in the (bounded) tree exploration game (i.e., it violates the inequality for the bounded variable), the falsifier chooses $n_i\triangleq 0$; otherwise the falsifier chooses $n_i\triangleq n_i'$.   
\end{itemize}
It is easy to check that against this strategy of the falsifier, $\tau$ cannot win within $\ell$ rounds. This is because the transcript of $\tau_\ttf$ vs $\tau$ is exactly the lists $\hat \Gamma$ in the definition of $\tau'$; and since $\tau'$ cannot win against $\tau'_\ttf$ within $\ell$ rounds, $\tau$ also cannot win against $\tau_\ttf$ within $\ell$ rounds.  
\end{proof}

This lemma shows that to obtain a winning strategy of the tree exploration game, we only need to construct a winning strategy of the unbounded tree exploration game. In practice, this means that we do not need to treat bounded quantifiers in a special way.

\newcommand{\phiexp}{\varphi^{\mathsf{exp}}}
\newcommand{\phipre}{\varphi^{\mathsf{pre}}}
\newcommand{\phipn}{\varphi^{\mathsf{pn}}}

\paragraph{Step 2: Strategy from Herbrand's Theorem.} Let $\varphi(x)$ be any implicitly bounded formula of form $\varphi(x)=\exists y_1~\forall x_1\dots \exists y_k~\forall x_k~\phi(x,\vec x,\vec y)$ and $\calT$ be a universal theory, where $\phi$ is quantifier-free. Assume that $\calT\vdash\forall x~\varphi(x)$. Then by \Cref{thm:herbrand-general} there is a prenexification of a strong $\lor$-expansion of $\forall x~\varphi(x)$ that admits a witnessing substitution over $\calT$. Our goal is to extract a winning strategy for the unbounded tree exploration game of $\varphi(x)$ from the strong $\lor$-expansion, prenexification, and witnessing substitution.  

Let $\phiexp$ be the strong $\lor$-expansion of $\forall x~\varphi(x)$ and $\phipre$ be a prenexification of $\phiexp$. We can see that the existential quantifiers within $\phiexp$ form a tree structure with respect to the sub-formula relation. More formally: we introduce a node $\eps_i$ for each existential quantifier $\exists_i$ in $\phiexp$, and define the node $\eps_i$ to be a child of $\eps_j$ if and only if $\exists_i$ is inside $\exists_j$ within $\phiexp$ and there is no other existential quantifiers in between. We introduce a root node $\eps_0$ corresponding to the entire sentence that has all nodes without parent as children. 

Let $T$ be the tree defined above. It is easy to see that the tree has depth $k$ and every leaf in $T$ is in the $k$-th level (the root is in the $0$-th level). 

We can observe that for every existential quantifier $\exists_i$ in $\phiexp$, there is a universal quantifier $\forall_i$ that immediately follows it (i.e.~$\forall_i$ is the outermost quantifier of the formula quantified by $\exists_i$), and this pair $(\exists_i,\forall_i)$ is a copy of an adjacent pair of quantifiers in $\varphi(x)$. Conversely, every universal quantifier (except for the outermost one) directly follows an existential quantifier. Therefore the quantifiers of $\phiexp$ except for the outermost one are partitioned into disjoint pairs $(\exists_i,\forall_i)$ as defined above. 

Recall that $\phipre$ is a prenexification of $\phiexp$, that is, we turn $\phiexp$ into its prenex normal form by prenexification rules. Moreover, the order of existential quantifiers of $\phipre$ is essentially a traversal of $T$. That is, for every $\eps_i$ and $\eps_j$ in $T$ such that $\eps_j$ is a child of $\eps_i$, where $\eps_i$ corresponds to $\exists_i$ and $\eps_j$ corresponds to $\exists_j$, then $\exists_i$ is to the left of $\exists_j$ in $\phipre$. In addition, for every existential quantifier $\exists_i$, its corresponding universal quantifier $\forall_i$ appears to the right of $\exists_i$. Let $\phipn$ be the sentence obtained from $\phipre$ by moving $\forall_i$ to the immediate right of $\phipre$, for every pair $(\exists_i,\forall_i)$ of corresponding quantifiers, as defined above. We note that if $\phipre$ has a witnessing substitution, then $\phipn$ also has a witnessing substitution. (Intuitively, this is because $\exists x~\forall y ~\phi(x,y)$ implies $\forall y~\exists x~\phi(x,y)$.) 

Now we focus on the structure of $\phipn$. The quantifiers of $\phipn$ (except for the outermost universal quantifier) are obtained from a traversal of $T$, where every universal quantifier immediate follows its existential quantifier. The quantifier-free formula within all the quantifiers is a disjunction of copies of $\phi(x,\vec x,\vec y)$, where: 
\begin{itemize} 
\item $x$ is a bounded variable quantified by the outermost universal quantifier; 
\item $\vec x=(x_1,\dots,x_k)$ and $\vec y=(y_1,\dots,y_k)$ are bounded variables quantified by universal and existential quantifiers within $\phipn$, respectively, where $y_i$ and $x_i$ are quantified by a pair of corresponding pairs of existential and universal quantifiers.
\item Assume that $y_i,x_i$ are quantified by $\exists_i,\forall_i$, respectively, for every $i\in[k]$. Let $\eps_i$ be the node in $T$ corresponding to $\exists_i$. Then $\eps_1,\eps_2,\dots,\eps_k$ forms a path in $T$ from the root to a leaf. 
\item Conversely, for every such path $\eps_1,\eps_2,\dots,\eps_k$ corresponding to $\exists_1,\exists_2,\dots,\exists_k$, there is a copy of $\phi(x,\vec x,\vec y)$ appearing in the disjunction in $\phipn$ such that $\vec y$ are quantified by $\exists_1,\dots,\exists_k$ and $\vec x$ are quantified by the universal quantifiers corresponding to $\exists_1,\dots,\exists_k$. 
\end{itemize}
Therefore, the paths from the root to the leaves in $T$ corresponds to the copies of $\phi(x,\vec x,\vec y)$ in the disjunction in $\phipn$.

Now we spell out the strategy for the (unbounded) tree exploration game of $\varphi(x)$ from the tree $T$, $\phipn$, and the witnessing substitution of $\phipn$. Let 
\[ 
    \phipn=\forall_0 x~\exists_1 y_1~\forall_1 x_1~\exists_2 y_2~\forall_2 x_2\dots\exists_d y_d~\forall_d x_d~\hat\phi,~~~~\hat\phi\eqdef \bigvee_{i=1}^\ell \phi^i~,
\]
where each $\phi^i$ is a copy of $\phi(x,\vec x,\vec y)$ corresponding to a leaf in $T$. (We add subscripts to the quantifiers for  simplicity of the presentation.) Let $t_1,t_2,\dots,t_d$ is a witnessing substitution of $\phipn$, where $t_j$ contains $x,x_1,\dots,x_{j-1}$ as free variables for every $j\in[n]$. The strategy is as follows. 
\begin{itemize}
    \item Fix a model $\calM=(\calD,\calI)$ and any $n_0\in\calD$. In the first round, the truthifier chooses the root, adds a child, and puts $t_1^\calM(n_0)$ on the edge. Suppose that the falsifier puts $n_1\in\calD$ on the edge (so that the edge is labeled with $(t_1(n_0),n_1)$). In the second round, the truthifier works as follows.
    \begin{itemize}
        \item If $\exists_2$ is a child of $\exists_1$ in $T$, then the truthifier chooses the node corresponding to $\exists_1$, adds a child, and puts $t_2^\calM(n_0,n_1)$ on the edge. 
        \item Otherwise, $\exists_2$ must be a child of the root in $T$. Then the truthifier chooses the root, adds a child, and puts $t_2^\calM(n_0,n_1)$ on the edge.
    \end{itemize} 
    \item Suppose that the falsifier's responses in the first $i-1$ rounds are $n_1,n_2,\dots,n_{i-1}$. The parent of $\exists_i$ in $T$ is either $\exists_j$ for some $j<i$ or the root. In the $i$-th round, the truthfier chooses the node introduced in the $j$-th round if $\exists_j$ is the parent of $\exists_i$ in $T$, and chooses the root if the parent of $\exists_i$ is the root in $T$. The truthifer then adds a child, and puts $t_i^{\calM}(n_0,n_1,\dots,n_{i-1})$ on the edge. 
\end{itemize}

It is clear that the strategy can be described by an $\calL$-term strategy, so it remains to show that it wins the game within at most $d$ rounds. Suppose that the falsifier's responses in the first $d$ rounds are $n_1,n_2,\dots,n_d\in\calD$. We observe that: 
\begin{enumerate}
\item\label{obs:tree-struct} The game tree explored by the truthifier is identical to $T$. Moreover, the order of explored nodes follows exactly the traversal of $T$ specified by the order of existential quantifiers in $\phipre$ (and $\phipn$). 
\item \label{obs:edge-label} In the $i$-th round, the truthifier puts $t_i^\calM(n_0,n_1,\dots,n_{i-1})$ on the edge. Therefore in the explored game tree, the edge connecting the $i$-th explored node and its parent is labeled by $(t_i^\calM(n_0,n_1,\dots,n_{i-1}), n_i)$. 
\end{enumerate}

Let $\sigma$ be the assignment of variables $\{x\mapsto n_0,x_i\mapsto n_i, y_i\mapsto t_i^\calM(n_0,n_1,\dots,n_{i-1})\}$. Since $t_1,\dots,t_d$ come from witnessing substitution, $\hat\phi[\sigma]$ is true in $\calM$. In other words, for some $i\in[\ell]$, $\phi^i[\sigma]$ is true in $\calM$. Fix this $i\in[\ell]$. Recall that $\phi^i$ is a copy of $\phi(x,\vec x,\vec y)$ and corresponds to a path in $T$ from the root to a leaf, in the sense that the bounded variables in $\phi^i$ appears on the path. By \Cref{obs:tree-struct} above, it also corresponds to a path in the \emph{explored game tree} from the root to a leaf. 

Suppose that the path corresponds to $\exists_{j_1}, \exists_{j_2}, \dots, \exists_{j_k}$ from the root to the leaf, where $1\le j_1<j_2<\dots<j_k\le d$. By \Cref{obs:edge-label}, the labels on the edges in the path are
\begin{equation}
(t_{j_1}^\calM(n_0,\dots,n_{j_1-1}),n_{j_1}), (t^\calM_{j_2}(n_0,\dots,n_{j_2-1}),n_{j_2}),\dots,(t^\calM_{j_k}(n_0,\dots,n_{j_k-1}),n_{j_k}).     \label{equ:path-labels}
\end{equation}
By the definition of $\sigma$, the path (\ref{equ:path-labels}) is a truthifier's winning transcript in the evaluation game. This  shows that the truthifier wins the tree exploration game. As the responses of the falsifier can be arbitrary, the aforementioned strategy is a winning strategy of the tree exploration for the truthifier. This completes the proof.

\subsection{Oblivious falsifiers: Self-contained proof of Theorem \ref{thm:witnessing-general} via Herbrandization }
\label{sec:no-counterexample}

The \emph{no-counterexample interpretation} (see, e.g., \citep[Section 2.3]{applied-proof-theory} and \citep{Krajicek92}) is a standard tool in proof theory to extract computational content from provable sentences of high quantifier complexity. In this section, we use this perspective to provide a different proof of \Cref{thm:witnessing-general}. We refer to \Cref{sec:witnessing_special_case} for the necessary definitions and notation.

Let $\calT$ be a universal theory over $\calL$, and let $$\varphi(x)\triangleq \exists y_1~\forall x_1~\exists y_2\dots \forall x_{k-1}~\exists y_k~\forall x_{k}~\phi(x,\vec x, \vec y)$$ be an $\calL$-formula, where $\phi$ is quantifier-free. The \emph{Herbrand normal form} of $\varphi(x)$ is defined as
\[
\varphi^H(x)\triangleq\exists y_1~\exists y_2\dots\exists y_k~\phi(x,x_1/f_1(x,y_1),x_2/f_2(x,y_1,y_2),\dots,x_k/f_k(x,y_1,y_2,\dots,y_k),\vec y), 
\]
where $f_1,f_2,\dots,f_k$ are new function symbols not in $\calL$. By a simple model-theoretical argument, $\calT\vdash\forall x~\varphi(x)$ if and only if $\calT\vdash\forall x~\varphi^H(x)$. Under the assumption that  $\calT\vdash\forall x~\varphi(x)$, we can apply Theorem \ref{thm:herbrand} to extract $\calL(f_1,f_2,\dots,f_k)$-terms that witness the existential quantifiers. In particular, if $\calT$ is $\Tpv$ and $\calL$ is $\Lpv$, this witnessing result implies that for every $x \in \mathbb{N}$ and all  interpretations of $f_1,f_2,\dots,f_k$ over $\mathbb{N}$, we can find suitable $y_1,y_2,\dots,y_k \in \mathbb{N}$ in polynomial-time given oracle access to $f_1^{\mathbb{N}},f_2^{\mathbb{N}},\dots,f_k^{\mathbb{N}}$.

Let $\calM$ be a structure over the vocabulary $\calL$ such that $\calM\vDash\calT$ (e.g., $\calT = \Tpv$ and $\calM=\bbN$), and let $n_0$ be an object in the domain of $\calM$. It is instructive to consider the following game on the board $(\calM,n_0)$. There are two players in the game: a truthifier (or student) that claims $\calM\vDash\varphi(n_0)$, and a falsifier (or teacher) that claims $\calM\vDash\lnot\varphi(n_0)$. In the $i$-th step, first the truthifier chooses an element $n_i$ for $y_i$, then the falsifier chooses an element $m_i$ for $x_i$. The truthifier (resp. falsifier) wins if and only if $\calM\vDash\varphi(n_0,m_1,\dots,m_k,n_1,\dots,n_k)$ holds (resp. does not hold). It is easy to see that $\calM\vDash\varphi(n_0)$ if and only if the truthifier has a winning strategy for the game on board $(\calM,n_0)$. The interpretation of the function symbols $f_1,\dots,f_k$ corresponds naturally to a strategy for the falsifier. The no-counterexample interpretation essentially means that if $\calT\vdash\forall x~\varphi(x)$, for every board $(\calM,n_0)$ and every strategy $f_1,\dots,f_k$ of the falsifier, the truthifier has a winning strategy that can be expressed by terms in $\calL(f_1,f_2,\dots,f_k)$. Next, we transform such a strategy into $\calL$-strategies with ancillary information for the truthifier in the evaluation game of $\varphi(x)$. 

\begin{theorem*}[Reminder of Theorem \ref{thm:witnessing-general}]
  Let $\calT$ be a universal theory over the language $\calL$ that is closed under if-then-else. Let $\varphi(x)$ be the formula
  \begin{align*}
    \varphi(x)~\triangleq~&\exists y_1\le t_1(x)~\forall x_1\le s_1(x,y_1)~\exists y_2\le t_2(x,y_1,x_1)\dots \forall x_{k-1}\le s_{k-1}(x,y_1,x_1,\dots,y_{k-1})\\
  ~&\exists y_k\le t_k(x,y_1,x_1,\dots,y_{k-1}, x_{k-1})~\forall x_{k}\le s_k(x,y_1,x_1,\dots,y_{k})~\phi(x,x_1,\dots,x_k, y_1,\dots,y_k),
  \end{align*}
  where $\phi(x, \vec{x}, \vec{y})$ is a quantifier-free $\calL$-formula.
  If $\calT\vdash\forall x~\varphi(x)$, then there  is a constant $\ell\in\bbN$ and  $\calL$-strategies $\tau^\ttt_1,\tau^\ttt_2,\dots,\tau^\ttt_\ell$ \emph{(}with ancillary information\emph{)} such that, for any board $(\calM,n_0)$ and  evaluation game of $\varphi(x)$ on $(\calM, n_0)$, for every strategy $\tau^\ttf$ of the falsifier:
  \begin{itemize}[itemsep=0pt]
  \item either $\hat\tau^\ttt_1\triangleq \tau^\ttt_1[\varnothing]$ beats $\tau^\ttf$,
  \item or $\hat\tau^{\ttt}_2\triangleq \tau^\ttt_2[\left<\hat\tau^\ttt_1:\tau^\ttf\right>]$ beats $\tau^\ttf$,
  \item or $\hat\tau^\ttt_3\triangleq \tau^\ttt_3[\left<\hat\tau^\ttt_1:\tau^\ttf\right>,\left<\hat\tau^\ttt_2:\tau^\ttf\right>]$ beats $\tau^\ttf$,
  \item $\dots$,
  \item or $\hat \tau^\ttt_\ell\triangleq \tau^\ttt_\ell[\left<\hat\tau^\ttt_1:\tau^\ttf\right>,\left<\hat\tau^\ttt_2:\tau^\ttf\right>,\dots,\left<\hat\tau^\ttt_{\ell-1}:\tau^\ttf\right>]$ beats $\tau^\ttf$.
  \end{itemize}
\end{theorem*}

\begin{proof}
    We introduce Herbrandization functions $f_1,f_2,\dots,f_k$ such that in $\calL^*\triangleq\calL\cup\{f_1,\dots,f_k\}$,  
    \[
        \calT\vdash\forall x~\exists \vec y\le\vec t~\phi(x,\vec x^*,\vec y),
    \]
    where $x^*_j=f_j(x,y_1,y_2,\dots,y_j)$ for all $j\in[k]$. By Herbrand's Theorem (\Cref{thm:herbrand}), there is a constant $r\in\bbN$ and $\calL^*$-terms $q^i_j(x)$ ($i\in[r],j\in[k]$) such that 
  \[
    \calT\vdash\forall x\left(\bigvee_{i=1}^r\phi_i(x)\right),    
  \]
  where $\phi_i(x)\triangleq \phi(x,x_1/f_1(x,q^i_1(x)),\dots,x_k/f_k(x,q_1^i(x),\dots,q_k^i(x)),y_1/q^i_1(x),\dots,y_k/q^i_k(x))$.
  
  We will translate $(q^i_1,q^i_2,\dots,q^i_k)$ into $\ell_i$ $\calL$-strategies $\tau^\ttt_{i,1},\tau^\ttt_{i,2},\dots,\tau^\ttt_{i,\ell_i}$ for some $\ell_i\in\bbN$, such that for every board $(\calM=(\calD,\calI),n_0)$ and every interpretation $F_1,F_2,\dots,F_k$ of $f_1,f_2,\dots,f_k$ over $\calD$ derived from $\hat \tau^\ttf$, if $\calM(F_1,F_2,\dots,F_k)\vDash\phi_i(x/n_0)$, then $\tau^\ttt_{i,1},\tau^\ttt_{i,2},\dots,\tau^\ttt_{i,\ell_i}$ will satisfy the conclusion of the theorem against the strategy $\hat \tau^\ttf$. 
    If this is possible, then 
    \[
    \tau^\ttt_{1,1},\tau^\ttt_{1,2},\dots,\tau^\ttt_{1,\ell_1},\tau^\ttt_{2,1},\tau^\ttt_{2,2},\dots,\tau^\ttt_{2,\ell_2},\dots,\tau^\ttt_{r,1},\tau^\ttt_{r,2},\dots,\tau^\ttt_{r,\ell_r}
    \]
    is a sequence of $\calL$-strategies as required. The argument is as follows. Fix any board $(\calM=(\calD,\calI),n_0)$ and any strategy $\tau^\ttf$ of the falsifier. Let $F_1,\dots,F_k$ be the interpretation of $f_1,\dots,f_k$ corresponding to this strategy, i.e., for every $j\in[k]$, 
    \[
    F_j(n,m_1,m_2,\dots,m_k)\triangleq\begin{cases}
        \begin{aligned}
        &\text{the move of }\tau^\ttf \\ 
        &\text{in the }j\text{-th step}
        \end{aligned}~~~~ & \begin{aligned} 
        &\text{if}~n = n_0~\;\text{and}\;~n_1,F_1(n_0,m_1),n_2,F_2(n_0,m_1,m_2), \\
        &\dots,n_{j-1},F_{j-1}(n_0,m_1,\dots,m_{j-1})~\text{is a prefix of a}  \\
        &\text{valid transcript over}~(\mathcal{M},n_0);
        \end{aligned} \\ 
        0 & \text{otherwise.}
    \end{cases}
    \]
    Then there is an index $i\in[r]$ such that $\calM(F_1,F_2,\dots,F_j)\vDash\phi_i(x/n_0)$ holds. 

    Before presenting the translation, we explain the main difficulty and how to address it. The issue is that $(q^i_1, q^i_2, \ldots, q^i_k)$ are $\calL^*$-terms, while the desired strategy in the evaluation game consists of $\calL$-terms only. For simplicity, suppose $q^i_j(x)$ invokes a single function from the list $f_1, \ldots, f_k$ of new function symbols, and assume it is $f_1(x,y_1)$. The idea is to replace the computation $F_1(w_1,w_2)$ over inputs $w_1, w_2$ by forcing the falsifier to compute its value in a previously played game. To achieve this, we use that $\tau^\ttf$ is fixed. In other words, if $w_1 = n_0$, the falsifier must play and reveal $F_1(n_0,w_1)$ if the truthifier plays $w_1$ in the first round. (On the other hand, if $w_1 \neq n_0$ we have $F_1(w_1,w_2) = 0$ by definition.) Consequently, by playing more games we guarantee that the necessary information appears in the transcript, which allows us to replace calls to functions $f_j$ and express the winning strategy using $\calL$-terms. To streamline the presentation, in the description below we omit the trivial case where the first input to a function $f_j$ is different than $x$, the input to the $\calL^*$-terms $q^i_j$ (corresponding to the case $w_1 \neq n_0$ we have just explained).

    Let $\tau^\ttf$ be the strategy specified by $f_1,f_2,\dots,f_k$ (i.e. $f_i$ denotes the falsifier's move in the $i$-th round). We prove by structural induction on the terms that we can decompose each $q^i_j$ ($j\in[k]$), which consists of $\calL$-functions and $f_1,f_2,\dots,f_k$, into finitely many $\calL$-strategies $\tau^{i,j}_1,\tau^{i,j}_2,\dots,\tau^{i,j}_{d_{i,j}}$ and an  $\calL$-term $p^{i,j}$ such that $\calM(F_1,F_2,\dots,F_j)\vDash q^{i}_j(n_0)=p^{i,j}(\Gamma(n_0))$ for every board $(\calM,\calI)$ and strategy $\tau^\ttf$ of the falsifier, where $F_1,F_2,\dots,F_j$ is the interpretation of $\tau^\ttf$ corresponding to the strategy and $\Gamma(n_0)$ is a sequence of transcripts produced as follows.
    \begin{itemize}
    \item For each $u\in[d_{i,j}]$, let $\Gamma_u(n_0)\triangleq\left<\tau^{i,j}_u[\Gamma_1(n_0),\Gamma_2(n_0),\dots,\Gamma_{u-1}(n_0)]:\tau^\ttf\right>$. 
    \item Let $\Gamma(n_0)\triangleq (\Gamma_1(n_0),\Gamma_2(n_0),\dots,\Gamma_{d_{i,j}}(n_0))$.
    \end{itemize} 
    Concretely, we translate each term as follows. Let the term be $g(v_1,v_2,\dots,v_d)$. By induction hypothesis, for every $r\in[d]$, we can decompose the term $v_r$ into a sequence of $c_r\in\bbN$ $\calL$-strategies $\tau^r_1,\tau^r_2,\dots,\tau^r_{c_r}$ and an $\calL$-term $p_r$, such that for every board $(\calM,m)$, $\calM\vDash v_r(n_0)=p_r(\Gamma^r(n_0))$ (where $\Gamma^r(n_0)$ is the transcript of games as described above). 
    \begin{itemize}
        \item If $g(\cdot)$ is a function symbol in the original language $\calL$, it is easy to see that the $\calL$-term $$g(p_1(\Gamma(n_0)),p_2(\Gamma(n_0)),\dots,p_d(\Gamma(n_0)))$$ and the strategies 
        \[
        (\tau^1_1,\tau^1_2,\dots,\tau^1_{c_1},\tau^2_1,\tau^2_2,\dots,\tau^2_{c_2},\dots,\tau^d_1,\tau^d_2,\dots,\tau^d_{c_d})
        \]
        provide what we want, where $\Gamma(n_0)\triangleq (\Gamma^1(n_0),\Gamma^2(n_0),\dots,\Gamma^{c_d}(n_0))$. 
        \item If $g(\cdot)=f_j$ for some $j\in[k]$, we define a new $\calL$-strategy $\tau^{f_j}$ as follows: suppose that the ancillary information  consists of the transcripts $\Gamma$ of $\tau^1_1,\tau^1_2,\dots,\tau^1_{c_1},\tau^2_1,\tau^2_2,\dots,\tau^2_{c_2},\dots,\tau^d_1,\tau^d_2,\dots,\tau^d_{c_d}$ vs $\tau^\ttf$; in the $i$-th round for $i\le j$, the truthifier's move is 
        \[
            \hat p_i(n_0,m_1,n_1,\dots,n_{i-1},\Gamma)\triangleq\begin{cases}
                p_i(\Gamma) & p_i(\Gamma)\le t_i(n_0,m_1,n_1,\dots,n_{i-1}) \\ 
                0 & \text{otherwise}
            \end{cases}
        \]
        while in the remaining $n-i$ rounds the truthifier always chooses $0$. Note that $\hat p_i$ is expressible since $\calT$ is closed under if-then-else. It is clear that the following $\calL$-term $v_g$ that takes $\Gamma$ and the transcript $\left<\tau^{f_j}[\Gamma]:\tau^\ttf\right>$ as input parameters outputs $f_j(v_1(n_0),\dots,v_d(n_0))$:
        \begin{itemize}
            \item If $\bigvee_i p_i(\Gamma)> t_i(n_0,m_1,n_1,\dots,n_{i-1})$ holds, then $v_g$ outputs $0$.
            \item Otherwise, $v_g$ outputs the $j$-th move of the falsifier in the transcript $\left<\tau^{f_j}[\Gamma]:\tau^\ttf\right>$.  
        \end{itemize}
        Therefore, we can obtain a term $v_g$ that simply reads the transcripts $\Gamma,\left<\tau^{f_j}[\Gamma]:\tau^\ttf\right>$ and outputs $f_j(v_1(n_0),\dots,v_d(n_0))$, together with the strategies 
        \[ 
            \tau^1_1,\tau^1_2,\dots,\tau^1_{c_1},\tau^2_1,\tau^2_2,\dots,\tau^2_{c_2},\dots,\tau^d_1,\tau^d_2,\dots,\tau^d_{c_d},\tau^{f_j},
        \] 
        as promised in the induction hypothesis.
    \end{itemize}

    Now we go back to the translation of $(q^i_1,q^i_2,\dots,q^i_k)$ into $\calL$-strategies. Assume that each $q^i_j$ for $j\in[k]$ has been decomposed into strategies $\tau^{i,j}_1,\dots,\tau^{i,j}_{d_{i,j}}$ and a term $p^{i,j}(\Gamma)$ as discussed above. Define an $\calL$-strategy $\tau^{q^i}$ with ancillary information as follows: suppose that the ancillary information is the transcripts $\Gamma(n_0)$ of $\tau^\ttf$ vs
    \[ 
        \tau^{i,1}_1,\tau^{i,1}_2,\dots\tau^{i,1}_{d_{i,1}},\tau^{i,2}_1,\tau^{i,2}_2,\dots\tau^{i,2}_{d_{i,2}},\dots,\tau^{i,k}_1,\tau^{i,k}_2,\dots\tau^{i,k}_{d_{i,k}}
    \]
    in which the latter strategies are given the transcripts of $\tau^\ttf$ vs previous strategies. In the $j$-th round, the truthifier's move is $p^{i,j}(\Gamma(n_0))$. By construction, it is easy to see that for every board $(\calM,n_0)$ and every strategy $\tau^\ttf$ of the falsifier, given correct ancillary information $\Gamma(n_0)$, $\tau^{q^i}$ will choose $q^i_j(n_0)$ in the $j$-th round. Therefore, the strategy will beat $\tau^\ttf$ as long as $\calM(F_1,F_2,\dots,F_j)\vDash\phi_i(n_0)$, where $F_1,\dots,F_j$ constitute the interpretation of $f_1,\dots,f_j$ corresponding to $\tau^\ttf$. This completes the proof by previous discussions.
    \end{proof}

    \ignore{
    Note that the standard KPT-witness theorem is a special-case of this theorem with $k=1$ and $\tau^\ttf$ being any fixed optimal strategy.
    
    \begin{enumerate}
        \item If the theory $\calT$ is the true universal theorem of polynomial-time computation, we can slightly strengthen the witnessing theorem: the last strategy (instead of one of the strategies) will always win. This is because we can find a winning strategy through the transcripts in polynomial-time and replay it as the last strategy. 
        \item The idea of oracle elimination using padding and Nisan-Wigderson generator in \cite{PS21} seems still work in this general case, can we find some natural sentences so that we can make use of both their idea and this general witnessing theorem? 
        \item Is the case when $\tau^\ttf$ is also a polynomial-time strategy meaningful? It seems to me that it corresponds to an ``algorithmic semantic'' of first-order logic, where the univesral quantifiers are not interpreted as a universal search in the domain, but a uniform polynomial-time search. Can this kind of logic be axiomatized? Will it be an ideal playground for cryptography? (Say, e.g., the information-theoretical security of a cipher as a first-order sentence holds in this logic if the it is semantically secure.) 
        \end{enumerate}
      }

\section{Proof of Hardness Amplification in $\PH$}
\label{sec: lmms for hardness amp}

\begin{reminder}[\Cref{thm:hardamp-ppoly}] There is a constant $\gamma>0$ and $\ell=\ell(n)=\poly(n)$ such that the following holds for every $i\ge 1$. Let $s_1,s_2 \colon \mathbb{N} \to \mathbb{N}$ be non-decreasing functions, where $s_2(n) = n^{\omega(1)}$, and suppose there is a function $f_n \colon \{0,1\}^n \to \{0,1\}$ computable by $\cSig{i}\mathsf{SIZE}[s_1(n)]$ circuits (resp.~$\cPi{i}\mathsf{SIZE}[s_1(n)]$ circuits) such that each $\Sigma^p_{i-1}$-oracle circuit $A_n$ of size at most  $s_2(n)$ satisfies 
    $$
    \Pr_{x \in \{0,1\}^n}[f_n(x) = A_n(x)] \;\leq\; 1 - \frac{1}{n}.
    $$ 
    Then there exist a function $h_\ell \colon \{0,1\}^\ell \to \{0,1\}$ computable by $\cSig{i}\mathsf{SIZE}[\poly(\ell)\cdot s_1(\ell)]$ circuits (resp. $\cPi{i}\SIZE[\poly(\ell)\cdot s_1(\ell^\gamma)]$ circuits) such that each $\Sigma^p_{i-1}$-oracle circuit $B_\ell$ of size at most $s_2(\ell^{\gamma})^\gamma$ satisfies
    $$
    \Pr_{y \in \{0,1\}^\ell}[h_\ell(y) = B_\ell(y)] \;\leq\; \frac{1}{2} + \frac{1}{s_2(\ell^{\gamma})^{\gamma}}.
    $$
\end{reminder}

\newcommand{\Bias}{\mathsf{Bias}}
\newcommand{\ExpBias}{\mathsf{ExpBias}}
\newcommand{\NoiseStab}{\mathsf{NoiseStab}}
\newcommand{\U}{\mathcal{U}}
\renewcommand{\epsilon}{\varepsilon}
\newcommand{\Tribes}{\mathsf{Tribes}}

Note that since $\Sigma^p_{i-1}$-oracle circuits are closed under complementation, we only need to prove the case where $f_n$ is computable by $\Sigma_i$ circuits. More formally, given a function $f:\{0,1\}^n\to\{0,1\}$ computable by $\cPi{i}\SIZE[s_1(n)]$ circuits that is hard on average against $\Sigma^p_{i-1}$-oracle circuits of size $s_2(n)$, we can consider $g(x)\eqdef \lnot f(x)$ that is computable by $\cSig{i}\SIZE[s_1(n)]$ circuits and still hard on average against the same class. By the hardness amplification theorem for $\Sigma_i$ circuits, we can obtain a function $h_\ell$ computable by $\cSig{i}\SIZE[\poly(\ell)\cdot s_1(\ell)]$ circuits that is strongly hard on average against $\Sigma^p_{i-1}$-oracle circuits of size $s_2(\ell^\gamma)^\gamma$. The negation of $h_\ell$ is then the required hard function computable in $\cPi{i}\SIZE[\poly(\ell) \cdot s_1(\ell)]$. 

We first fix the notation. 
\begin{itemize}  
\item A \emph{probabilistic function} is a Boolean function with two inputs $h(x;r)$ where the second input is treated as random bits. If the random bits are omitted, a probabilistic function is treated as a function mapping the input to a random variable distributed according to the output of the function over the random bits. 
\item Let $g$ be a function probabilistic function) with input length $n$, the \emph{$k$-th direct product} is defined as the function (resp. probabilistic function) with input length $k\cdot n$ and output length $k$ as follows:
\[
  g^{\otimes k}(x_1,\dots,x_k) \eqdef g(x_1)\| \dots\| g(x_k)\,.  
\]
\item The \emph{bias} of a random variable $X$ is defined as $\Bias(X)\triangleq \big |\!\Pr[X=0]-\Pr[X=1]\big|$. The \emph{bias} of a probabilistic function $h(x;r)$ is defined as the bias of the random variable $h(x;r)$ for a uniformly random $x$ and $r$. The probabilistic function $h$ is said to be \emph{balanced} if $\Bias(h)=0$. 
\item A probabilistic function $h:\{0,1\}^n\times\{0,1\}^r\to\{0,1\}$ is \emph{$\delta$-random} if $h$ is balanced and there is a subset $H\subseteq\{0,1\}^n$ of size $2\delta\cdot 2^n$ such that $h$ is a ``coin flip'' over $H$ and deterministic outside $H$ (i.e.,~$\Pr[h(x)=1]=1/2$ for every $x\in H$, and $h(x)$ is deterministic for every $x\notin H$). 
\item The \emph{expected bias} of a probabilistic function $h$ is defined as $\ExpBias(h)\eqdef\Ex_x\left[\Bias(h(x))\right]$. 
\item The \emph{noise stability} of a Boolean function $C:\{0,1\}^k\to\{0,1\}$ with respect to the noise rate $\delta$ is defined as 
\[
\NoiseStab_\delta(g)\eqdef 2\cdot \Pr_{x,\eta}[C(x)=C(x\oplus \eta)]-1,  
\]
where $x\sim\{0,1\}^k$ and each bit of $\eta$ is $1$ independently with probability $\delta$. By Lemma 3.7 of \cite{DBLP:journals/siamcomp/HealyVV06},  $\ExpBias[C\circ g^{\otimes k}]\le\sqrt{\NoiseStab_\delta[C]}$ for every $\delta$-random probabilistic function $g$.  
\item Two random variables $X_1$ and $X_2$ are said to be $\epsilon$-indistinguishable for size $s$, denoted by $X_1\approx_\epsilon^s X_2$, if for every $\Sigma^p_{i-1}$-oracle circuit $C$ of size $s$, $\big |\!\Pr[C(X_1)=1]-\Pr[C(X_2)=1]\big |\le \epsilon$. Note that our definition of the indistinguishability differs from the original definition in \cite{DBLP:journals/siamcomp/HealyVV06} since we are proving hardness amplification against $\Sigma^p_{i-1}$-oracle circuits.
\item For simplicity, we say a function $f:\{0,1\}^n\to\{0,1\}$ is $\epsilon$-hard for size $s$, if for every $\Sigma^p_{i-1}$-oracle circuit $C$ of size $s$, $C(x)=f(x)$ for at most an $\epsilon$ fraction of $x\in\{0,1\}^n$.
\end{itemize}

We assume that $f_n$ is \emph{balanced}, that is, $\Pr_x[f_n(x)=1]=1/2$ for every $n \geq 1$. This is without loss of generality, since we can first increase the input length by one then use non-uniformity to make the resulting function balanced, without a relevant change of parameters.

The hardness amplification of \cite{DBLP:journals/siamcomp/HealyVV06} proceeds as follows.
  
\paragraph{The Construction.} Fix any $n\ge 1$. Let $f:\{0,1\}^n\to\{0,1\}$ be the hard function and $C:\{0,1\}^k\to\{0,1\}$ be an explicit circuit to be determined later. Let $G:\{0,1\}^\ell\to(\{0,1\}^n)^k$ be an explicit function in the sense that given $\sigma\in\{0,1\}^\ell$ and $i\in[k]$, we can compute the $i$-th block $X_i\in\{0,1\}^n$ of the output of $G(\sigma)$ in $\poly(\ell,\log k)$ time.\footnote{We note that $C$ is used to replace the XOR function in the standard hardness amplification based on Yao's XOR Lemma (see, e.g., Theorem 19.2 of \cite{Arora-Barak09}), while $G$ is used as a pseudorandom generator that (in some sense) ``fools'' $C \circ f^{\otimes k}$.} The amplified function is defined as $\Amp_f:\{0,1\}^\ell\to\{0,1\}$:
\[
  \Amp_f(\sigma)\triangleq C(f(X_1),f(X_2),\dots,f(X_k)),
\]
where $(X_1,X_2,\dots,X_k)\triangleq G(\sigma)$. We need to carefully choose $C$ and $G$ such that $\Amp_f(\sigma)$ is computable in $\cSig{i}\SIZE[\poly(n)\cdot s_1(n)]$ and can amplify the hardness of $f$. 

  \paragraph{The Choice of $G$.} To ensure the hardness of $\Amp_f$, the function $G_k:\{0,1\}^\ell\to(\{0,1\}^n)^k$ should satisfy the following two technical requirements.
  \begin{itemize}
    \item $G_k$ is \emph{indistinguishability-preserving for size $t=k^2$}: Let $f_1,\dots,f_k,g_1,\dots,g_k$ be probabilistic functions such that for every $i\in[k]$, $x \| f_i(x) \approx_\epsilon^s x \| g_i(x)$ for $x\sim\{0,1\}^n$, then 
    \[
    \sigma\| f_1(X_1) \| \dots\| f_k(X_k) \approx^{s-t}_{k\cdot\epsilon} \sigma \| g_1(X_1)\|\dots\| g_k(X_k),
    \] 
    where $\sigma\sim\{0,1\}^\ell$ and $(X_1,\dots,X_k)\eqdef G_k(\sigma)$.
    \item $G_k$ is $2^{-n}$-pseudorandom against (read-once oblivious) branching programs of size $2^n$ and block-size $n$:\footnote{See \citep[Definition 5.4]{DBLP:journals/siamcomp/HealyVV06} for the precise definition of this branching program model.}  for every branching program $B$ of size $2^n$ and block-size $n$, we have 
    \[
    \left|\Pr_{x\sim\{0,1\}^\ell}[B(G_k(x))=1]-\Pr_{y\sim\{0,1\}^{nk}}[B(y)=1]\right|\le 2^{-n}.
    \]
  \end{itemize} 

  \begin{lemma}[Generalized version of {\cite[Theorem 5.12]{DBLP:journals/siamcomp/HealyVV06}}]\label{lmm: generator-g-amp}
    For every $k\le 2^n$, there is an explicit computable generator $G_k:\{0,1\}^\ell\to(\{0,1\}^n)^k$ that satisfies the requirements below: 
    \begin{enumerate}
        \item\label{enum: item 1 hard} There is an algorithm that computes the $i$-th block of $G_k(\sigma)$ in $\poly(\ell,\log k)$ time given $\sigma,i$. 
        \item\label{enum: item 2 hard} $G_k$ is indistinguishability-preserving for size $t=k^2$.
        \item\label{enum: item 3 hard} $G_k$ is $2^{-n}$-pseudorandom against branching programs of size $2^n$ and block-size $n$. 
    \end{enumerate}
  \end{lemma}

\newcommand{\Nisan}{\mathsf{N}}

\begin{proof}
    The only difference between this lemma and \cite[Lemma 5.12]{DBLP:journals/siamcomp/HealyVV06} is that in our definition, the indistinguishability-preserving property holds against $\Sigma^p_{i-1}$-oracle circuits instead of standard circuits, which will not cause any issue since their argument only requires mild closure properties of the adversary. For completeness, we sketch their proof here.

    The generator $G_k$ is defined as the XOR of two generators: a Nisan-Wigderson based generator $\NW_k : \{0,1\}^{\ell_\NW}\to(\{0,1\}^n)^k$ that is efficiently computable and indistinguishability-preserving; and Nisan's unconditional PRG $\Nisan_k:\{0,1\}^{\ell_\Nisan}\to(\{0,1\}^n)^k$ against (probabilistic) branching programs (see, e.g., \cite[Theorem 5.6]{DBLP:journals/siamcomp/HealyVV06} and \cite{Nisan92}). That is, $G_k(x,y)\eqdef \NW_k(x)\oplus \Nisan_k(y)$. Both $\Nisan_k$ and $\NW_k$ have seed length at most $O(n^2)$, hence $G_k$ has seed length $\ell=O(n^2)$. Next, we discuss the properties of the generator.
    \begin{itemize}
        \item Both $\NW_k$ and $\Nisan_k$ are efficiently computable in the sense that given $\sigma$ and $i$, we can compute the $i$-th block of the output in $\poly(\ell,\log k)$ time. Therefore \Cref{enum: item 1 hard} holds.
        \item To prove \Cref{enum: item 2 hard}, we need to show that any indistinguishability-preserving generator XORed with a fixed string is still indistinguishability-preserving. Towards a contradiction, assume that $G_k$ is not indistinguishability-preserving. This means that there are $f_1,\dots,f_k,g_1,\dots,g_k$ such that for every $i\in[k]$, $x\| f_i(x)\approx^s_\epsilon x\| g_i(x)$ for $x\sim\{0,1\}^n$, while for $(\sigma_1,\sigma_2)\sim\{0,1\}^{\ell_\NW}\times\{0,1\}^{\ell_\Nisan}$ and $(X_1,\dots,X_k)\eqdef \NW_k(\sigma_1)\oplus\Nisan_k(\sigma_2)$, 
        \begin{equation*}
        \sigma_1\| \sigma_2\|f_1(X_1)\|\dots\|f_k(X_k) ~\not\approx^{s-t}_\epsilon~ \sigma_1\|\sigma_2\| g_1(X_1)\|\dots\|g_k(X_k). 
        \end{equation*}  
        By an averaging argument, there is a $\sigma_2^*\in\{0,1\}^{\ell_\Nisan}$ such that  
        \begin{equation}
            \sigma_1\|f_1(X_1\oplus y_1)\|\dots\|f_k(X_k\oplus y_k) ~\not\approx^{s-t}_\epsilon~ \sigma_1\|g_1(X_1\oplus y_1)\|\dots\|g_k(X_k\oplus y_k),\label{equ: distinguish-contra}
        \end{equation}  
        where $(y_1,\dots,y_k)\eqdef\Nisan_k(\sigma_2^*)$. Let $f_i'(x)\eqdef f_i(x\oplus y_i)$ and $g_i'(x)\eqdef g_i(x\oplus y_i)$ for $i\in[k]$. Clearly for every $i\in[k]$, $x\|f'_i(x)\approx^s_\epsilon x\|g_i'(x)$,\footnote{There is no loss in the circuit size of the adversary if we define the circuit model so that NOT gates are free.} which is impossible since $\NW_k$ is indistinguishability-preserving but \Cref{equ: distinguish-contra} holds.
        \item Similarly, we can show that since $\Nisan_k$ is $2^{-n}$-pseudorandom against branching programs of size $2^n$, after XORed with another generator, $G_k$ is still $2^{-n}$-pseudorandom against branching programs of size $2^n$. This implies \Cref{enum: item 3 hard}.
    \end{itemize}
    
    It remains to verify that the Nisan-Wigderson based generator $\NW_k$ is indistinguishability-preserving for size $k^2$ against $\Sigma^p_{i-1}$-oracle circuits. Let $\ell=O(n^2)$ and $S_1,S_2,\dots,S_k\subseteq[\ell]$ be an $(\ell,n,\log k)$-design (see \Cref{sec:NW} and \cite{Nisan92}). Then $\NW_k:\{0,1\}^\ell\to(\{0,1\}^n)^k$ is defined as 
    \[
    \NW_k(\sigma) \eqdef (\sigma|_{S_1},\sigma|_{S_2},\dots,\sigma|_{S_k}).
    \] 
    Let $f_1,\dots,f_k,g_1,\dots,g_k$ be probabilistic functions such that for every $i\in[k]$, $x\| f_i(x)\approx^s_\epsilon x|| g_i(x)$ for $x\sim \{0,1\}^n$. Suppose, for the sake of contradiction, that 
    \begin{equation}
    \sigma\| f_1(\sigma|_{S_1})\|\dots\| f_k(\sigma|_{S_k})~\not\approx^{s-k^2}_{k\cdot \epsilon}~ \sigma\| g_1(\sigma|_{S_1})\|\dots\| g_k(\sigma|_{S_k}). \label{equ: distinguishable-f-g}
    \end{equation}
    For every $i\in[0,k]$, we define the hybrid distribution 
    \[
        H_i=\sigma\|g_1(\sigma|_{S_1})\|\dots\|g_i(\sigma|_{S_i})\|f_{i+1}(\sigma|_{S_{i+1}})\|\dots\| f_k(\sigma|_{S_k}). 
    \]
    Then the distinguisher $D$ for \Cref{equ: distinguishable-f-g}, which is a $\Sigma^p_{i-1}$-oracle circuit of size $s-k^2$, can distinguish between $H_i$ and $H_{i+1}$ with advantage at least $\epsilon$ for some $i\in[0,k-1]$. Note that 
    \begin{align*}
        H_i &= \sigma\|g_1(\sigma|_{S_1})\|\dots\|g_i(\sigma|_{S_i})\|f_{i+1}(\sigma|_{S_{i+1}})\|f_{i+2}(\sigma|_{S_{i+2}})\|\dots\| f_k(\sigma|_{S_k}) & \text{and} \\ 
        H_{i+1} &= \sigma\|g_1(\sigma|_{S_1})\|\dots\|g_i(\sigma|_{S_i})\|g_{i+1}(\sigma|_{S_{i+1}})\|f_{i+2}(\sigma|_{S_{i+2}})\|\dots\| f_k(\sigma|_{S_k})
    \end{align*}
    differ only on the $(i+2)$-th part: $H_i$ has $f_{i+1}(\sigma|_{S_{i+1}})$ while $H_{i+1}$ has $g_{i+1}(\sigma|_{S_{i+1}})$. 
    
    By an averaging argument, we can fix all the bits of $\sigma$ outside of $S_{i+1}$ so that $\hat H_i$ and $\hat H_{i+1}$ are still distinguishable with advantage $\epsilon$, where $\hat H_i$ and $\hat H_{i+1}$ refer to the distribution $H_i$ and $H_{i+1}$ after we fix the bits of $\sigma$ outside of $S_{i+1}$. Since for every $j\ne i+1$, $|S_j\cap S_{i+1}|\le \log k$, we can construct a $\Sigma^p_{i-1}$-oracle circuit of size at most $(s-k^2)+2^{\log k}\cdot k = s$ that hardwires all possibilities for the common parts of $H_i$ and $H_{i+1}$ such that:
    \begin{itemize}[itemsep=0pt]
        \item Given the unfixed bits of $\sigma$ and $f_{i+1}(\sigma)$, it generates $\hat H_i$ and outputs $D(\hat H_i)$. 
        \item Given the unfixed bits of $\sigma$ and $g_{i+1}(\sigma)$, it generates $\hat H_{i+1}$ and outputs $D(\hat H_{i+1})$. 
    \end{itemize}  
    Since $D$ can distinguish between $\hat H_i$ and $\hat H_{i=1}$ with advantage $\epsilon$, the circuit above can distinguish bewteen $\sigma\| f_{i+1}(\sigma\| S_{i+1})$ and $\sigma\| g_{i+1}(\sigma\| S_{i+1})$ with advantage $\epsilon$. This leads to a contradiction. 
\end{proof}

  \paragraph{The Choice of $C$.} The outer function $C$, which serves as the counterpart of the XOR function in Yao's XOR Lemma (see, e.g.,~\cite[Theorem 12.9]{Arora-Barak09}), is chosen according to the following lemma. 
  \begin{lemma}[Generalized version of {\cite[Lemma 5.15]{DBLP:journals/siamcomp/HealyVV06}}]\label{lmm: choice of C}
    For every $i\ge 1$, $\delta(n)=1/\poly(n)$, and $k=k(n)$ such that $n^{\omega(1)}\le k\le 2^n$, there is a function $C_k:\{0,1\}^k\to\{0,1\}$ such that: 
    \begin{enumerate}
      \item $\NoiseStab_\delta[C_k]\le 1/k^{\Omega(1)}$; 
      \item For every $f:\{0,1\}^n\to\{0,1\}$ computable by $\cSig{i}\SIZE[s(n)]$ circuits, $(C_k\circ f^{\otimes k})\circ G_k:\{0,1\}^\ell\to\{0,1\}$ is computable by $\cSig{i}\SIZE[\poly(n)\cdot s(n)]$ circuits. 
      \item $C_k$ is computable by a branching program of size $\poly(n)\cdot k$ and by a deterministic circuit of size $\poly(n)\cdot k$.    
    \end{enumerate}
  \end{lemma}

\newcommand{\Maj}{\mathsf{Maj}}
\newcommand{\RMaj}{\mathsf{RMaj}}

\begin{proof}
    Let $\delta=\delta(n)\ge 1/\poly(n)$ and $k=k(n)$ such that $n^{\omega(1)}\le k\le 2^n$. We will define $C_k$ as the composition of two functions defined as follows: 
    \begin{itemize}
        \item The \emph{recursive-majority function} $\RMaj_r:\{0,1\}^{3^r}\to\{0,1\}$ is recursively defined by 
        \begin{eqnarray*}
            \RMaj_1(x_1,x_2,x_3) & \eqdef & \Maj(x_1,x_2,x_3) \\ 
            \RMaj_r(x_1,\dots,x_{3^r}) & \eqdef & \RMaj_{r-1}(\Maj(x_1,x_2,x_3),\dots,\Maj(x_{3^r-2},x_{3^r-1},x_{3^r}))
        \end{eqnarray*} 
        where $\Maj(x_1,x_2,x_3)$ is the majority value among $x_1,x_2,x_3\in\{0,1\}$. 
        \item The \emph{tribes function} of $k$ bits is defined by 
        \[
        \Tribes_k(x_1,\dots,x_k) \eqdef (x_1\land\dots\land x_b)\lor (x_{b+1}\land\dots\land x_{2b})\lor\dots\lor(x_{k-b+1}\land\dots\land x_k),   
        \]
        where $b=O(\log k)$ is the largest integer such that $(1-2^{-b})^{k/b}\ge 1/2$. 
    \end{itemize}
    
    Let $r\eqdef c\cdot\log(1/\delta)$ for a constant $c$ to be determined later. Assuming without loss of generality that $r$ and $k/3^r$ are integers, we define $C_k:\{0,1\}^k\to\{0,1\}$ by 
    \[ 
        C_k\eqdef\Tribes_{k/3^r}\circ\RMaj_r^{\otimes k/3^r}. 
    \] 
    As \cite[Section 5.5]{DBLP:journals/siamcomp/HealyVV06} in the proof of Lemma 5.15, we know that for some sufficiently large constant $c$, the noise stability of $C_k$ is at most $1/k^{\Omega(1)}$. Also they showed that $C_k$ can be computed by a branching program of size $\poly(n)\cdot k$ and a deterministic circuit of size $\poly(n)\cdot k$.
    
    It remains to determine the complexity of $(C_k\circ f^{\otimes k})\circ G_k$ for $f:\{0,1\}^n\to\{0,1\}$ computable by $\cSig{i}\SIZE[s(n)]$ circuits. Consider the following $\cSig{i}$circuit. We first guess (using non-determinism) a clause $K$ of the upper $\Tribes_{k/3^r}$ function that is satisfied. For every $\RMaj_r$ function feeding into this clause (there are $b=O(\log k)=\poly(n)$ such $\RMaj_r$ functions), we guess the input bits of the upper $C_k$ sub-circuit (or equivalently, the output bits of the lower $f$ functions) that are $1$ and
    \begin{enumerate}
        \item we verify that these input bits that are $1$ make the clause $K$ accept, which can be done by a deterministic circuit of size $\poly(3^r)=\poly(n)$ since $\RMaj$ is a monotone function; 
        \item for every guessed input bit of $C_k$ (or equivalently, the output bit of one of $f$ in the middle $f^{\otimes k}$ layer) that is supposed to be $1$, we use the $\cSig{i}\SIZE[s(n)]$ circuit for $f$ to verify that it is indeed $1$. The input to this function $f$ is one of the $n$-bit blocks of the output of $G_k$, which can be computed by a deterministic algorithm in $\poly(\ell,\log k)=\poly(n)$ time (see \Cref{lmm: generator-g-amp}). 
    \end{enumerate}
    The overall $\Sigma_i$-circuit complexity of $(C_k\circ f^{\otimes k})\circ G_k$ is at most $\poly(n)\cdot s(n)$.
\end{proof}
  
  Note that the second item means that the function $(C_k\circ f^{\otimes k})\circ G$ is efficiently computable \emph{even if $k$ is as large as $2^n$}. The argument relies on the explicitness of $C_k$ and $G$ as well as on the power of $\Sigma_i$-circuits. This is crucial for hardness amplification up to $1/2-1/s_2(\ell^\gamma)^\gamma$ (instead of only $1/2-1/\poly(\ell)$). 

  \paragraph{Proof of the Hardness Amplification.} Following \cite[Section 5]{DBLP:journals/siamcomp/HealyVV06}, we now argue that if $f$ is $\delta$-hard for size $s(n)\ge n^{\omega(1)}$, where $\delta\ge 1/\poly(n)$, then  we can construct $\Amp_f:\{0,1\}^\ell\to\{0,1\}$ with $\ell=\poly(n)$ that is $(1/2-1/s(\sqrt \ell)^{\Omega(1)})$-hard for size $s(\sqrt{\ell})^{\Omega(1)}$. To prove this, we need the following two technical lemmas. 

  \begin{lemma}[{\cite[Lemma 5.7 and Lemma 5.12]{DBLP:journals/siamcomp/HealyVV06}}]\label{lmm: bias-and-noise-stab}
    Let $g$ be an $n$-input single output $\delta$-random function, and $C_k$ and $G_k$ be defined as above. Then 
    \[
      \ExpBias[(C_k\circ g^{\otimes k})\circ G_k]\le\sqrt{\NoiseStab_\delta(C_k)+2^{-n+1}}.
    \]    
  \end{lemma}

  \begin{lemma}[Generalized version of {\cite[Lemma 5.2]{DBLP:journals/siamcomp/HealyVV06}}]\label{lmm: indist-exp-bias}
    Assume that $f:\{0,1\}^n\to\{0,1\}$ is $\delta$-hard for size $s=n^{\omega(1)}$. There is a $\delta'$-random function $g$ with $\delta'\in[\delta/2,\delta]$ such that $\Amp_f \colon \{0,1\}^\ell\to\{0,1\}$ has hardness 
    \[
    \frac{1}{2}-\frac{\ExpBias[(C\circ g^{\otimes k})\circ G]}{2}-\frac{k}{s^{1/3}}
    \]
    for size $\Omega(s^{1/3}/\log(s/\delta))-k^2-\poly(n)\cdot k$.  
  \end{lemma}

Before proving this lemma, we need to verify that Impagliazzo's hardcore lemma (see, e.g., \cite[Section 19.1.2]{Arora-Barak09}) holds against adversaries with access to $\Sigma^p_{i-1}$ oracles.

\begin{lemma}[Generalized version of Impagliazzo's Hardcore Lemma]
    Assume that $2n<s<0.001\cdot (\epsilon\delta)^2\cdot 2^n/n$. Let $f:\{0,1\}^n\to\{0,1\}$ be a balanced function that is $\delta$-hard for $\Sigma^p_{i-1}$-oracle circuits of size $s$. There exists a $\delta'$-random function $g:\{0,1\}^n\to\{0,1\}$ such that $X\|f(X)\approx^{s'}_{\epsilon} X\| g(X)$ for $X\sim\{0,1\}^n$, where $s'=\Omega(s\epsilon^2/\log(1/(\delta\epsilon)))$ and $\delta'\in[\delta/2,\delta]$. 
\end{lemma}

\newcommand{\calC}{\mathcal{C}}

\newcommand{\calH}{\mathcal{H}}

\begin{proof}[Proof Sketch.]
    We follow the proof presented in \cite[Section 19.1.2]{Arora-Barak09} based on the \emph{min-max theorem} for  zero-sum games (also see, e.g., \cite{Impagliazzo95}). We say that a distribution $\calH$ over $\{0,1\}^n$ has density $\delta$ if for every $x\in\{0,1\}^n$, $\calH(x)\le 1/(\delta 2^n)$. Let $\delta_1=0.99\delta$. We first show that there is a distribution $\calH$ of density $\delta_1$ such that for every $\Sigma^p_{i-1}$-oracle circuit $C$ of size $s'$, $\Pr[f(x)=C(x)]<1/2+\epsilon/2$ for $x\sim\calH$. 
    
    Towards a contradiction, we assume that such distribution does not exist. By a game-theoretic argument using the min-max theorem, we can construct a distribution $\calC$ over $\Sigma^p_{i-1}$-oracle circuits of size $s'$ such that for every distribution $\calH$ of density $\delta_1$, a random $C\sim\calC$ can approximate $f$ over $\calH$ with error $\le 1/2-\epsilon/2$. 

    An input $x\in\{0,1\}^n$ is said to be \emph{bad} if $\Pr[C(x)\ne f(x)]>1/2-\epsilon/2$ for $C\sim\calC$. It is said to be \emph{good} otherwise. There are at most $\delta_1\cdot 2^n$ bad inputs, since otherwise we can let $\calH$ be the uniform distribution over a set of $\delta_1\cdot 2^n$ bad inputs and violate the aforementioned property of $\calC$. Let $t=O(\epsilon^{-2}\log(1/(\delta\epsilon)))$ and $C$ be the following probabilistic circuit (with $\Sigma^p_{i-1}$ oracles): given input $x$, obtain $t$ independent samples $C_1,\dots,C_t\sim\calC$, and output the majority of $C_1(x),\dots,C_t(x)$. This probabilistic circuit has size at most $t\cdot s'\le s$. By the Chernoff bound, it computes $f(x)$ for any good $x$ with error at most $\exp(-\Omega(\epsilon^2 t))\le 0.001\cdot \delta$. This means that for a uniformly random $x\sim\{0,1\}^n$, the probabilistic $\Sigma^p_{i-1}$-oracle circuit (and also deterministic $\Sigma^p_{i-1}$-orcle circuit by an averaging argument) can approximate $f(x)$ with error at most $\delta_1+\delta/2\le\delta$ for an $x\sim\{0,1\}^n$, which is impossible. 

    We then prove via a probabilistic argument that there is a subset $H$ of size $\delta'\in[\delta/2,\delta]$ such that no $\Sigma^p_{i-1}$-oracle circuit of size $s$ can approximate $f$ on $H$ with advantage $\epsilon$. Let $H$ be a \emph{random} subset defined as follows: for every $x\in\{0,1\}^n$, we let $x\in H$ independently with probability $\calH(x)$. By a ``concentration bound then union bound'' argument, we get with non-zero probability that $H$ has size $\delta'\in[\delta/2,\delta]$ \emph{and} for every $C$ of size $s$, $\Pr[f(x)=C(x)]\le 1/2+1/\epsilon$. This means that the $\delta'$-random function $g$ defined over $H$ satisfies the conditions of the lemma.   
\end{proof}

\begin{proof}[Proof of \Cref{lmm: indist-exp-bias}]
    Assume that $f:\{0,1\}^n\to\{0,1\}$ is $\delta$-hard for size $s=n^{\omega(1)}$. By Impagliazzo's hardcore lemma, there is a $\delta'$-random function $g:\{0,1\}^n\to\{0,1\}$ such that $X\| f(X)\approx^{s'}_\epsilon X\|g(X)$ for $X\sim\{0,1\}^n$, where $s'=\Omega(s\epsilon^2/\log(1/(\delta\epsilon)))$ and $\delta'\in[\delta/2,\delta]$. Since $G$ is indistinguishability-preserving for size $k^2$, we get that
    \[
    \sigma\|f(X_1)\|\dots\|f(X_k)\approx^{s'-k^2}_{k\epsilon}\sigma\| g(X_1)\|\dots\|g(X_k), 
    \]
    where $\sigma\sim\{0,1\}^\ell$ and $(X_1,\dots,X_k)=G(\sigma)$. Since $C_k$ has complexity bounded by $\poly(n)\cdot k$ this further means that 
    \[
        \sigma\|C_k(f(X_1),\dots,f(X_k))\approx^{s''}_{k\epsilon}\sigma\|C_k(g(X_1),\dots,g(X_k)),
    \]
    where $s''=s'-k^2-\poly(n)\cdot k$. Note that $$C_k(f(X_1),\dots,f(X_k))=(C_k\circ f^{\otimes k})\circ G_k(\sigma) \quad \text{and} \quad C_k(g(X_1),\dots,g(X_k))=(C_k\circ g^{\otimes k})\circ G_k(\sigma)\,.$$

    Also we can see that the for every probabilistic function $h$, the statistical distance between $X\| h(X)$ and $X\|b$ for $X\sim\{0,1\}^n$ and $b\sim\{0,1\}$ is exactly $\ExpBias[h]/2$ (see, e.g., \cite[Lemma 3.4]{DBLP:journals/siamcomp/HealyVV06}). Therefore we know that 
    \[
    \Delta(\sigma\|(C_k\circ g^{\otimes k})\circ G_k(\sigma),\sigma\|b) \le \frac{\ExpBias[(C_k\circ g^{\otimes k})\circ G_k]}{2},     
    \]
    where $\sigma\sim\{0,1\}^\ell$ and $b\sim\{0,1\}$. This further means that $\sigma\|(C_k\circ f^{\otimes k})\circ G_k(\sigma)$ and $\sigma\|b$ are $k\epsilon+(1/2)\cdot \ExpBias[(C_k\circ g^{\otimes k})\circ G_k]$ indistinguishable for size $s''$. By setting $\epsilon=s^{-1/3}$, we obtain the lemma.  
    \end{proof}

  Let $k=k(n)=s(n)^{1/7}$, $C_k$ be the function in \Cref{lmm: choice of C}, and $G_k$ be the generator in \Cref{lmm: generator-g-amp} with $\ell=O(n^2)$. Recall that $\Amp_f:\{0,1\}^\ell\to\{0,1\}$ is defined as $\Amp_f\eqdef (C_k\circ f^{\otimes k})\circ G_k$, and note that the upper bound on the complexity of $\Amp_f$ is guaranteed by \Cref{lmm: choice of C}. By \Cref{lmm: indist-exp-bias}, we know that $\Amp_f$ has hardness 
  \begin{equation}
    \frac{1}{2}-\frac{\ExpBias[(C_k\circ g^{\otimes k})\circ G_k]}{2}-\frac{k}{s^{1/3}} \label{equ: hardness-of-amp}
  \end{equation}
  for size $\Omega(s^{1/3}/\log(s/\delta))-k^2-\poly(n)\cdot k=s(\sqrt\ell)^{\Omega(1)}$, where $g$ is some $\delta'$-random function with $\delta'\in[\delta/2,\delta]$. By \Cref{lmm: bias-and-noise-stab}, we can bound \Cref{equ: hardness-of-amp} using the noise stability bound for $C_k$ given in \Cref{lmm: choice of C}: 
  \begin{align*}
    (\ref{equ: hardness-of-amp})\; &\ge\; \frac{1}{2}-\frac{\sqrt{\NoiseStab_{\delta'}(C_k)+2^{-n+1}}}{2}-\frac{k}{s^{1/3}} \\ 
    & \ge\; \frac{1}{2}-\frac{\sqrt{k^{-\Omega(1)}+2^{-n+1}}}{2}-\frac{k}{s(n)^{1/3}} \\ 
    & \ge\; \frac{1}{2} - \frac{1}{s(\sqrt\ell)^{\Omega(1)}}.
  \end{align*}
  This completes the argument. 

\ignore{
We will need a more general version of this result.

\begin{theorem}[Implicit in \citep{DBLP:journals/siamcomp/HealyVV06}]\label{thm:hardamp} Let $i\ge 1$ be an arbitrary constant. There is a constant $\gamma >0$ for which the following holds. Let $s_1\colon \mathbb{N} \to \mathbb{N}$ and $s_2\colon \mathbb{N} \to \mathbb{N}$ be non-decreasing functions, and suppose there is a sequence $f_n \colon \{0,1\}^n \to \{0,1\}$ of functions in $\cPi{i}\SIZE[s_1(n)]$ such that each non-uniform circuit $A_n\in\SIZE^{\sTF\Sigma^b_i}[s_2(n)]$ satisfies 
$$
\Pr_{x \sim \{0,1\}^n}[f_n(x) = A_n(x)] \;\leq\; 1 - \frac{1}{n}
$$ 
for sufficiently large $n$. Then there is sequence $h_m \colon \{0,1\}^m \to \{0,1\}$ of functions computable in $\cPi{i}\SIZE[s_1(m)^{\gamma^{-1}}]$ such that each non-uniform circuit $B_m\in\SIZE^{\sTF\Sigma^b_i}[s_2(m^\gamma)^\gamma]$ satisfies
$$
\Pr_{y \sim \{0,1\}^m}[h_m(y) = B_m(y)] \;\leq\; \frac{1}{2} + \frac{1}{s_2(m^{\gamma})^{\gamma}}
$$
for sufficiently large $n$.
\end{theorem}

\begin{proof}[Sketch of the Proof.]
  \jnote{TODO: Still need careful reexamination of [HVV06].} 
\inote{Adapt statement as we're no longer using $\sTF\Sigma^b_i$ classes.}
\end{proof}
}

\section{A Universal Theory for $\Tpv^i$}
\label{sec:universal-theory}

In this section, we describe the proofs omitted from \Cref{sec:universal_theory}. First, we will need some auxiliary results.

\begin{lemma}\label{lmm:bounded-normal-form}
    Let $i\ge 0$. For every $\Sigma^b_i$-formula \emph{(}resp.~$\Pi^b_i$-formula\emph{)} $\alpha(\vec z)$ in the language $\Lpv$, there exists a formula $\alpha^{\sf norm}(\vec z)=Q_1 x_1\le t_1(\vec z)~Q_2 x_2\le t_2(\vec z)\dots Q_i x_i\le t_i(\vec z)~\phi(\vec z,\vec x)$, where $Q_1=\exists$ \emph{(}resp.~$Q_1=\forall$\emph{)}, $Q_j\in\{\forall,\exists\}$, and $Q_j\ne Q_{j+1}$ for every $j\le i-1$, such that $\Tpv^1\vdash\forall\vec z~(\alpha(\vec z)\leftrightarrow\alpha^{\sf norm}(\vec z))$. 
\end{lemma}

\begin{proof}
    Note that for every $\Sigma^b_i$-formula (resp.~$\Pi^b_i$-formula) $\alpha(\vec z)$, we can firstly find its prenex normal form $\alpha^{\sf pnf}(\vec z)$ with $i-1$ quantifier alternations starting with an existential (resp.~universal) quantifier that is logically equivalent to $\alpha(\vec z)$. Note that $\Tpv^1$ defines pairing and unpairing functions. Concretely, there are functions $\left<\cdot,\cdot\right>,\pi_1(\cdot),\pi_2(\cdot)$ such that $\Tpv^1\vdash\forall x~\forall y~(\pi_1(\left<x,y\right>)=x\land \pi_2(\left<x,y\right>)=y\land |\left<x,y\right>|\le 10\cdot (|x|+|y|))$. It is then easy to see that 
    \begin{align*}
    \Tpv^1&\vdash \Big(\forall x\le s~\forall y\le t~\varphi(x,y)\Big) \leftrightarrow \Big(\forall p\le (s\cdot t)^{10}(\pi_1(p)\le s\land \pi_2(p)\le t\to \varphi(\pi_1(p),\pi_2(p)))\Big) \\ 
    \Tpv^1&\vdash \Big(\exists x\le s~\exists y\le t~\varphi(x,y)\Big) \leftrightarrow \Big(\exists p\le (s\cdot t)^{10}(\pi_1(p)\le s\land \pi_2(p)\le t\land \varphi(\pi_1(p),\pi_2(p)))\Big).
    \end{align*}
    Therefore we can further collapse adjacent quantifiers of the same kind to obtain $\alpha^{\sf norm}(\vec z)$ as described above such that $\Tpv^1\vdash\forall\vec z~(\alpha(z)\leftrightarrow\alpha^{\sf pnf}(\vec z)\leftrightarrow\alpha^{\sf norm}(\vec z))$. 
\end{proof}

\begin{lemma}\label{lmm:defining-axiom-not}
    For every $i\ge 1$ and $(\Pi_{i-1}^b\cup\Sigma_{i-1}^b)$-formula $\alpha(\vec x)$, we have \emph{(1)} $\Upv^i\vdash \forall\vec x~(f_\alpha(\vec x)=1\leftrightarrow f_{\lnot\alpha}(\vec x)=0)$ and \emph{(2)} $\Upv^i\vdash\forall\vec x~(f_\alpha(\vec x)=0\lor f_\alpha(\vec x)=1)$.
\end{lemma}
  
\begin{proof}
    By the definition of each $f_\alpha^\bbN$, we can see that the universal sentences $\forall\vec x~(f_\alpha(\vec x)=1\leftrightarrow f_{\lnot\alpha}(\vec x)=0)$ and $\forall\vec x~(f_\alpha(\vec x)=0\lor f_\alpha(\vec x)=1)$ are both true in the standard model, so they are provable in $\Upv^i$. 
\end{proof}

\begin{reminder}[\Cref{lmm:defining-axiom-g}]
    Let $i\ge 2$, $\beta(\vec{x},y)$ be any $\Sigma_{i-1}^b$-formula in $\Lpv$, and $t$ be any term in $\Lpv$. Then $\Upv^i\vdash \forall\vec x~((\exists y\le t(\vec x)~f_{\beta}(\vec x, y)=1)\leftrightarrow f_{\beta}(\vec x,g_{\beta,t}(\vec x))=1)$.
\end{reminder}

\begin{proof}
  We will show separately that:
  \begin{align}
    &\Upv^i\vdash\forall\vec x~((\exists y\le t(\vec x)~f_{\beta}(\vec x, y)=1)\to f_{\beta}(\vec x,g_{\beta,t}(\vec x))=1), \label{eq:defining-g-case-1}\tag{Case 1}\\
    &\Upv^i\vdash \forall\vec x~(f_{\beta}(\vec x,g_{\beta,t}(\vec x))=1\to (\exists y\le t(\vec x)~f_{\beta}(\vec x, y)=1)) \label{eq:defining-g-case-2}\tag{Case 2}.
  \end{align}
  \begin{itemize}
  \item \ref{eq:defining-g-case-1}: It's easy to see that the sentence we want to prove (in $\Upv^i$) is logically equivalent to the following universal sentence: $\forall\vec x~\forall y\le t(\vec x)~(f_{\beta}(\vec x,y)=1\to f_{\beta}(\vec x,g_{\beta,t}(\vec x))=1)$ ($*$). Furthermore, ($*$) is a true universal sentence in the standard model by the definition of $g^\bbN_{\beta,t}$ and $f^\bbN_\beta$. Therefore $\Upv^i$ proves $(*)$.
  \item \ref{eq:defining-g-case-2}: By the definition of $g^\bbN_{\beta,t}$, the universal sentence $\forall\vec x~g_{\beta,t}(\vec x)\le t(\vec x)$ is true in the standard model, which further means that $\Upv^i\vdash \forall\vec x~g_{\beta,t}(\vec x)\le t(\vec x)$. This sentence logically implies the sentence we want to prove in $\Upv^i$. \qedhere
  \end{itemize}
\end{proof}

\begin{reminder}[\Cref{lmm:defining-axiom-f}]
    For every $i\ge 1$ and $(\Pi_{i-1}^b\cup\Sigma_{i-1}^b)$-formula $\alpha(\vec z)$ in the language $\Lpv$, $\Upv^i\vdash\forall\vec z~(\alpha(\vec z)\leftrightarrow f_\alpha(\vec z)=1)$. 
\end{reminder}

\begin{proof}
    Fix any $i\ge 1$. Let $\varphi_\alpha\triangleq\forall\vec z~(\alpha(\vec z)\leftrightarrow f_\alpha(\vec z)=1)$. We firstly prove that $\Upv^i\vdash\varphi_\alpha$ for every bounded $\Lpv$-formula $\alpha(\vec z)= Q_1 x_1\le t_1(\vec z)~Q_2 x_2\le t_2(\vec z)\dots Q_k x_k\le t_k(\vec z)~\phi(\vec z, x_1,\dots,x_k)$, where $\phi$ is quantifier free, $k\le i-1$, $Q_i\in\{\forall,\exists\}$, and $Q_i\ne Q_{i+1}$ for every $i\in [k-1]$. We will prove this by induction over $k$. 
\begin{itemize}
    \item \textbf{(Case 0).} Assume that $k=0$ and $\alpha(\vec{z})$ is a quantifier-free formula. Then $\varphi_\alpha$ is a universal sentence. Furthermore, by the definition of the interpretation of $f_\alpha$ over the standard model, we know that $\bbN\vDash \varphi_\alpha$, which means that $\Upv^i\vdash \varphi_\alpha$. 
    \item \textbf{(Case 1).} Assume that $\alpha(\vec z)=\forall x\le t(\vec z)~\alpha'(x,\vec z)$. In such case, $i\ge 2$. By the induction hypothesis, $\Upv^i\vdash\varphi_{\alpha'}$. To show that $\Upv^i\vdash\varphi_\alpha$ it is sufficient to prove that $\Upv^i\vdash \forall\vec z~(f_\alpha(\vec z)=1\to \alpha(\vec z))$ and $\Upv^i\vdash \forall \vec z~(\alpha(\vec z)\to f_\alpha(\vec z)=1)$. Now we prove them separately.    
    \begin{enumerate}
        \item Since $\Upv^i\vdash\varphi_{\alpha'}$, we know that $\Upv^i$ proves $\forall\vec z~\forall x\le t(\vec z)~(f_{\alpha'}(x,\vec z)=1\to \alpha'(x,\vec z))$ ($\star$). Consider the universal sentence $\psi\triangleq\forall\vec z~\forall x\le t(\vec z)~(f_\alpha(\vec z)=1\to f_{\alpha'}(x,\vec z)=1)$. By the definition of the interpretations of $f_{\alpha}$ and $f_{\alpha'}$, $\psi$ is a true sentence, therefore $\Upv^i\vdash \psi$ ($\lozenge$). Combining $(\star)$ and $(\lozenge)$ we get that 
        \[
        \Upv^i\vdash \forall\vec z~\forall x\le t(\vec z)~(f_\alpha(\vec z)=1\to \alpha'(x,\vec z)).
        \]
        This means that $\Upv^i\vdash \forall\vec z~(f_\alpha(\vec z)=1\to \alpha(\vec z))$. 
      \item Recall that we need to show that $\Upv^i\vdash\forall\vec z~((\forall x\le t(\vec z)~\alpha'(x,\vec z))\to f_\alpha(\vec z)=1)$. Since $\Upv^i\vdash\varphi_{\alpha'}$, it is sufficient to prove that
          \[
            \Upv^i\vdash\forall\vec z~((\forall x\le t(\vec z)~f_{\alpha'}(\vec z,x)=1)\to f_\alpha(\vec z)=1).
          \]
          Since $\alpha$ is a $\Pi^b_{i-1}$-formula of the form above, $\lnot \alpha'$ is a $\Sigma^b_{i-2} \cup \Pi^b_{i-2}$-formula. By Lemma \ref{lmm:defining-axiom-not}, $\Upv^i\vdash f_{\alpha'}(\vec z,x)=1\leftrightarrow f_{\lnot\alpha'}(\vec z,x)=0$. So we only need to prove that 
          \[
            \Upv^i\vdash\forall\vec z~((\forall x\le t(\vec z)~f_{\lnot\alpha'}(\vec z,x)=0)\to f_\alpha(\vec z)=1).
          \]
          By Lemma \ref{lmm:defining-axiom-g}, we know that $\Upv^i\vdash (\exists x\le t(\vec z)~f_{\lnot\alpha'}(\vec z,x)=1)\leftrightarrow f_{\lnot\alpha'}(\vec z,g_{\lnot\alpha',t}(\vec z))=1$, which means we only need to prove that
        \begin{equation}\label{equ:defining-axiom-2:target-case-2}
            \Upv^i\vdash\forall\vec z~(f_{\lnot\alpha'}(\vec z,g_{\lnot\alpha',t}(\vec z))=0\to f_\alpha(\vec z)=1).
        \end{equation}
        By considering the interpretations of $f_\alpha$, $f_{\lnot\alpha'}$, and $g_{\lnot\alpha',t}$ in the standard model, it follows that the universal sentence (\ref{equ:defining-axiom-2:target-case-2}) is true in the standard model. Therefore it is provable in $\Upv^i$. This completes this case.  
    \end{enumerate}
    \item \textbf{(Case 2).} Assume that $\alpha(\vec z)=\exists x\le t(\vec z)~\alpha'(x,\vec z)$. Let $\overline\alpha(\vec z)$ be the formula obtained by pushing the negation in  $\lnot\alpha(\vec z)$ into the quantifiers. Note that $\vdash \lnot\alpha(\vec z)\leftrightarrow\overline\alpha(\vec z)$. By applying Case 1, we can show that 
    \[
    \Upv^i\vdash\forall\vec z~\left(f_{\overline{\alpha}}(\vec z)=1\leftrightarrow\lnot\alpha(\vec z)\right). 
    \]
    Since $\forall\vec z~(f_{\overline{\alpha}}(\vec z)=1\leftrightarrow f_{\alpha}(z)\ne 1)$ is a universal sentence that is true in the standard model, we know that it is provable in $\Upv^i$, which further implies that 
    \[
    \Upv^i\vdash \forall\vec z~(f_{\alpha}(\vec z)\ne 1\leftrightarrow \lnot\alpha(\vec z)).
    \]
    This yields $\Upv^i\vdash\varphi_\alpha$. 
\end{itemize}

Now we consider the case when $\alpha$ is an arbitrary $(\Pi^b_{i-1}\cup\Sigma^b_{i-1})$-formula. By Lemma \ref{lmm:bounded-normal-form}, we can see that $\Upv^i\vdash \forall\vec z~(\alpha(\vec z)\leftrightarrow\alpha^{\sf norm}(\vec z))$. According the discussion above, we know that $\Upv^i\vdash\forall \vec z~(\alpha^{\sf norm}(\vec z)\leftrightarrow f_{\alpha^{\sf norm}}(\vec z)=1)$. Moreover, we have that $\Upv^i\vdash\forall\vec z~(f_\alpha(\vec z)=1\leftrightarrow f_{\alpha^{\sf norm}}(\vec z)=1)$, since this is a true universal sentence in the standard model. It follows from the provability of these three sentences that $\Upv^i\vdash \varphi_\alpha$, as desired.
\end{proof}

\begin{reminder}[\Cref{thm:universal_theory}]
    For every $i \geq 1$, the theory  $\UTpv^i$ satisfies the following properties:
    \begin{enumerate}
        \item $\UTpv^i$ is a universal theory.
          \item Every $\Lpv^i$-sentence provable in $\Upv^i$ is also provable in $\UTpv^i$.
        \item\label{enum: provable translate} Every $\Lpv$-sentence provable in $\Tpv^i$ is also provable in $\UTpv^i$.
        \item Let $t$ be an arbitrary $\LUT^i$-term, and consider its interpretation $t^{\mathbb{N}} \colon \mathbb{N}^k \to \mathbb{N}$ over the standard model. Then $t^{\mathbb{N}} \in \FP^{\Sigma_{i-1}^p}$.
        \item $\UTpv^i$ is closed under if-then-else.
        \item $\UTpv^i$ is sound, i.e., every sentence provable in  $\UTpv^i$ is true over $\mathbb{N}$.
    \end{enumerate}
\end{reminder}
    
\begin{proof}
    We prove each item in turn.
    \begin{enumerate}
        \item This is immediate from the definition of the theory.
        \item Let $\varphi$ be an $\Lpv^i$-sentence provable in $\Upv^i$. It is enough to argue that every axiom of $\Upv^i$ is provable in $\UTpv^i$. Since $\Upv^i$ is the theory consisting of all universal true sentences (over the standard model) in $\Lpv^i$, $\Lpv^i \subseteq \mathcal{L}^i_{\sf UT}$, and $\UTpv^i$ is the theory of all universal sentences in $\LUT^i$ that are true in the standard model, the result is immediate.
        \item Let $\varphi$ be an $\Lpv$-sentence provable in $\Tpv^i$. It follows from \Cref{thm:Upv-extends-Tpv} that $\varphi$ is provable in $\Upv^i$. Consequently, the claim follows from the previous item.
        \item This follows from \Cref{thm:lpvi-term-complexity}, the definition of $\mathcal{L}^i_{\sf UT}$, and the closure of the functions in $\FP^{\Sigma_{i-1}^p}$ under composition.
        \item To show this, let $\varphi(x_1, \ldots, x_k)$ be a quantifier-free $\LUT^i$-formula, and consider $\LUT^i$-terms $t_1(x_1, \ldots, x_k)$ and $t_2(x_1, \ldots, x_k)$. We must prove that there exists an $\LUT^i$-term $t(x_1, \ldots, x_k)$ such that 
        \begin{equation}\label{eq:prov_ifthenelse}
    \UTpv^i\vdash \big (t(\vec x)=t_1(\vec x)\land\varphi(\vec x)\big )\lor \big(t(\vec x)=t_2(\vec x)\land\lnot \varphi(\vec x)\big).
        \end{equation}
    Consider the interpretations of terms $t_1^{\mathbb{N}}, t_2^{\mathbb{N}} \colon \mathbb{N}^k \to \mathbb{N}$ over the standard model. Let $f \colon \mathbb{N}^k \to \mathbb{N}$ be the function defined as follows:
    $$
    f(\vec{a}) = \begin{cases}
    t_1^{\mathbb{N}}(\vec{a})\quad & \text{if~}\; \varphi^{\mathbb{N}}(\vec{a})~\text{ is true}; \\ 
    t_2^{\mathbb{N}}(\vec{a}) & \text{otherwise.}
    \end{cases}
    $$
    Since $\varphi$ is a quantifier-free formula, thanks to \Cref{enum: provable translate}, it is easy to see that $f \in \FP^{\Sigma_{i-1}^p}$. Consequently, the corresponding function symbol $f_{\sf UT} \in \LUT^i$. Take $t$ as $f_{\sf UT}$. It follows from the definition of $f$ and of $t$ that, for every $\vec{a} \in \mathbb{N}^k$,
    $$
    \mathbb{N} \models \big (t(\vec{a})=t_1(\vec{a})\land\varphi(\vec{a})\big )\lor \big(t(\vec{a})=t_2(\vec{a})\land\lnot \varphi(\vec{a})\big).$$ 
    Since the formula above is free of quantifiers, the definition of $\UTpv^i$ immediately yields \Cref{eq:prov_ifthenelse}.
    \item This is obvious from its definition.\qedhere
    \end{enumerate}
\end{proof}

\section{The Counting Lemma: Existence of a Good Restriction}\label{sec:counting_lemma}

\noindent \textbf{Notation.} Recall that for $m \geq 1$, a set $S \subseteq \{0,1\}^{[m]}$, and a string $a \in \{0,1\}^I$, where $I \subseteq [m]$, we define the \emph{restriction of} $S$ \emph{with respect to} $a$ as the set 
$$
S\uhr_a \eqdef \{w \in S \mid w|_I = a\}.
$$
For a non-empty set $U$ and a set $S \subseteq U$, we define $\mathsf{dens}_U(S) \eqdef |S|/|U|$.\\

For simplicity of the exposition, we consider without loss of generality  restrictions with respect to the first $m_1$ input bits. Let $S \subseteq \{0,1\}^m$, where $m = m_1 + m_2$. Suppose that $\mathsf{dens}_{\{0,1\}^m}(S) = \delta$. Now let $T \subseteq S$ be a set such that $\mathsf{dens}_S(T) > 2/3$. The following result appears implicit in  \citep{DBLP:journals/jml/Krajicek11, DBLP:journals/apal/Pich15}.

\begin{lemma}[Counting Lemma]\label{lmm:counting}
Under these assumptions, there is $a \in \{0,1\}^{m_1}$ such that
\begin{equation}\label{eq:densitya}
    \frac{|S\uhr_a\!|}{2^{m_2}} \geq \frac{1}{100} \cdot \delta \quad \text{and} \quad \frac{|T\uhr_a\!|}{|S\uhr_a\!|} \geq \frac{2}{3} - \frac{1}{100}. 
\end{equation}
\end{lemma}

\begin{proof}
Suppose this is not the case, i.e., for every $a \in \{0,1\}^{m_1}$, at least one of the two inequalities above does not hold. We use this to  contradict $|T| > (2/3) \cdot |S| = (2/3) \cdot \delta \cdot 2^m$. Under the assumption, and using that $T\uhr_a \subseteq S\uhr_a$,
\begin{eqnarray}
    |T| & = & \sum_{a \in \{0,1\}^{m_1}} |T\uhr_a\!|  \nonumber \\
    & \leq & \sum_{a \in \{0,1\}^{m_1}} \left ( \frac{1}{100} \cdot \delta \cdot 2^{m_2} + \left ( \frac{2}{3} - \frac{1}{100} \right ) |S\uhr_a\!| \right ) \nonumber \\
    & = & 2^{m_1 + m_2} \cdot \frac{1}{100} \cdot \delta + \left( \frac{2}{3} - \frac{1}{100} \right) \cdot |S| \nonumber \\
    & = & \frac{\delta}{100} \cdot 2^m + \left( \frac{2}{3} - \frac{1}{100} \right) \cdot \delta \cdot 2^m \nonumber \\
    & = & \frac{2}{3} \cdot \delta \cdot 2^m.
\end{eqnarray}
This completes the proof.
\end{proof}

\newfunc{\cktEval}{cktEval}

\end{document}